\documentclass{article}
\usepackage{amsmath, amssymb, amsthm}
\usepackage{mathtools}				
\usepackage[usenames,dvipsnames]{xcolor}	
\usepackage[letterpaper, margin=1in]{geometry}
\usepackage{rotating}				
\usepackage{oubraces}				
\usepackage{tikz-cd}
\usepackage{cases}
\usepackage{esint}				
\usepackage{dsfont}				
\usepackage[font=small]{caption}            
\usepackage{subcaption}				
\usepackage{hyperref}
\hypersetup{
	colorlinks=true,
	linkcolor=RoyalBlue,
	citecolor=Thistle
}

\usepackage{tikz, pgfplots}
\pgfplotsset{compat=1.18}
\usetikzlibrary{arrows.meta}
\usepgfplotslibrary{fillbetween}

\pgfdeclarelayer{ft}
\pgfdeclarelayer{bg}
\pgfsetlayers{bg, main, ft}

\newcommand{\fref}[2]{\hyperref[#2]{#1 \ref*{#2}}}	
\newcommand{\wavecoord}{\sigma}
\newcommand{\altwavecoord}{\zeta}
\newcommand{\norm}[2]{{\left\lvert \left\lvert #1 \right\rvert \right\rvert}_{#2}} 
\DeclareMathOperator{\im}{im}				
\DeclareMathOperator{\sign}{sign}			
\DeclareMathOperator*{\argmin}{arg\,min}		

\numberwithin{equation}{section}

\theoremstyle{plain}
\newtheorem{thm}{Theorem}[section]
\newtheorem{theorem}[thm]{Theorem}
\newtheorem{prop}[thm]{Proposition}
\newtheorem{lemma}[thm]{Lemma}
\newtheorem{cor}[thm]{Corollary}
\theoremstyle{definition}
\newtheorem{definition}[thm]{Definition}
\newtheorem{remark}[thm]{Remark}

\title{
    Beyond Linear Decomposition: a Nonlinear Eigenspace Decomposition for a Moist Atmosphere with Clouds
}
\author{
    Antoine Remond-Tiedrez
    \thanks{Department of Mathematics, University of Wisconsin--Madison, Madison, WI, USA}
    \and Leslie M. Smith
    \thanks{Department of Mathematics and Department of Atmospheric and Oceanic Sciences, University of Wisconsin--Madison, Madison, WI, USA}
    \and Samuel N. Stechmann
    \footnotemark[2]
}
\date{}

\begin{document}


\maketitle

\abstract{
	A linear decomposition of states underpins many classical systems.
	This is the case of the Helmholtz decomposition, used to split vector fields into divergence-free and potential components, and of the dry Boussinesq system in atmospheric dynamics, where identifying the slow and fast components of the flow
	can be viewed as a decomposition.
	The dry Boussinesq system incorporates two leading ingredients of mid-latitude atmospheric motion: rotation and stratification.
	In both cases the leading order dynamics are linear so we can rely on an eigendecomposition to decompose states.

	Here we study the extension of dry Boussinesq to incorporate another important ingredient in the atmosphere: moisture and clouds.
	The key challenge with this system is that nonlinearities are present at leading order due to phase boundaries at cloud edge.
	Therefore standard tools of linear algebra, relying on eigenvalues and eigenvectors, are not applicable.
	The question we address in this paper is this: in spite of the nonlinearities, can we find a decomposition for this moist Boussinesq system?

	We identify such a decomposition adapted to the nonlinear balances arising from water phase boundaries.
	This decomposition combines perspectives from partial differential equations (PDEs), the geometry, and the conserved energy.
	Moreover it sheds light on two aspects of previous work.
	First, this decomposition shows that the nonlinear elliptic PDE used for potential vorticity and moisture inversion can be used outside the limiting system where it was first derived.
	Second, we are able to rigorously justify, and interpret geometrically, an existing numerical method for this elliptic PDE.
	This decomposition may be important in applications because, like its linear counterparts, it may be used to analyze observational data.
	Moreover, by contrast with previous decompositions, it may be used even in the presence of clouds.
}


\vspace{2em}
\textbf{MSC codes.}
    Primary: 86A10, 76U60; Secondary: 37N10, 53Z05, 65N99.

\vspace{2em}
\textbf{Keywords:}
Potential vorticity inversion, clouds, wave-vortical decomposition, inertia-gravity waves, water vapor, moisture, geophysical fluid dynamics, phase interface.


\newpage
\tableofcontents


\newpage

\noindent
Note to the reader: in \fref{Section}{sec:intro} below we provide a ``shortest path'' to the result, providing just enough background to state and motivate the result.
The main results and conclusions are summarized in \fref{Sections}{sec:summary} and \ref{sec:conclusion}.
For a more detailed discussion of the motivation and context we refer the reader to \fref{Section}{sec:background}.

\section{Introduction}
\label{sec:intro}
	
	Decomposing a vector field, or a state vector, into several constitutive components is ubiquitous in fluid mechanics and atmospheric science.
	For example, the Helmholtz decomposition will split a vector field into its potential and divergence-free components, and it
	is used throughout fluid mechanics \cite{chorin1968numerical,temam1969approximation,majda2001vorticity}. Another example is the vortical--wave decomposition of a state vector which is used in atmospheric and oceanic sciences to understand the fast rotation and strong stratification
	 regime of the Boussinesq equations, e.g., \cite{blumen1972geostrophic,bartello95,majda_2003}.
	Crucially, as discussed in \fref{Section}{sec:background}, both of these decompositions can be viewed from multiple perspectives and share
 salient features such as
 reconstruction of the original field via elliptic PDE,
	a rich geometric interpretation of the decomposition (involving orthogonality and projections), and a connection to slow and fast components of appropriate dynamical systems.

	In this paper we consider a moist Boussinesq system from atmospheric dynamics, which brings additional realism and complexity due to clouds. We ask the following question:
	does this moist, cloudy Boussinesq system also admit such a decomposition?

While these decompositions have applications in slow--fast dynamical systems,
they can be investigated without reference to the dynamical evolution,
by considering the associated operator. 
To see this from a general perspective, consider a dynamical system
written abstractly as
\begin{equation*}
    \partial_t U + \frac{1}{\epsilon}\mathcal{D}_0 U + \mathcal{D}_1 U = 0,
\end{equation*}
where $\epsilon$ is a small parameter,
$\mathcal{D}_0$ is the leading-order operator,
and $\mathcal{D}_1$ is the next-order contribution to the operator for the dynamical system.
To leading order with respect to $\epsilon$, it is the operator $\mathcal{D}_0$
that contains the essential information.
Understanding the dynamics can be tantamount to understanding the operator $\mathcal{D}_0$.
Hence, in the remainder of the paper, we will seldom refer to the actual dynamical evolution,
and instead the main object of interest is the
operator $\mathcal{D}_0$.

In past cases, such as the Helmholtz decomposition or Boussinesq equations, the leading-order operator is linear ($\mathcal{D}_0=\mathcal{L}$), and the decomposition involves a linear operator and linear eigenmodes. On the other hand, for an atmosphere with clouds, additional nonlinearity arises from phase changes of water, and $\mathcal{D}_0=\mathcal{N}$ is a nonlinear operator. 
Hence a substantial challenge in the present paper is from nonlinearity, from the distinction
\begin{equation}
    \mathcal{D}_0=\mathcal{L} 
    \quad\mbox{versus}\quad 
    \mathcal{D}_0=\mathcal{N}.
    \label{eq:LvsN}
\end{equation}
For a nonlinear operator, it is not clear if eigenmodes can be identified. 
Consequently, a decomposition
\begin{equation}
    U=U_1+U_2+\cdots+U_{N_{\text{eig}}}
    \label{eq:Udecomp}
\end{equation}
into $N_{\text{eig}}$ eigenmodes is likely to be impossible.

Despite the challenges of nonlinearity in \eqref{eq:LvsN}--\eqref{eq:Udecomp},
here we forge ahead with a decomposition in the spirit of \eqref{eq:Udecomp} as the main aim
and main theme of the paper.

In \fref{Section}{sec:background} we describe the background and main questions in more detail. \fref{Sections}{sec:decomposition}, \ref{sec:slow_fast_decomp}, \ref{sec:energy_and_decomposition}, and \ref{sec:iterative_methods} describe main results, and \fref{Sections}{sec:summary} and \ref{sec:conclusion} present a summary of main results and conclusions, including links to the main theorems for easy reference.
We point the reader who may skip the background in \fref{Section}{sec:background} to note \fref{Section}{sec:notation} where the simplified notation
used in the main body of the paper is recorded.

\section{Background and motivation}
\label{sec:background}

	In this section we first expound on known decompositions (namely the Helmholtz decomposition and the vortical--wave decomposition of the dry Boussinesq system)
	in order to set the stage for the decomposition of moist Boussinesq states studied in this paper.
	We also motivate in more detail the context in which the dry and moist Boussinesq equations arise.

	We start with the Helmholtz decomposition for two reasons.
	First, we present it here for folks familiar with fluid dynamics but not necessarily familiar with geophysical and atmospheric fluid dynamics.
	Second, we present it here because, as we will detail later, the Helmholtz decomposition plays a central role in the other decompositions, in both the dry and the moist cases.

\subsection{Helmholtz decomposition}
\label{sec:Helmholtz}

	The Helmholtz decomposition of a vector field $u$ splits it into two components:
	\begin{equation*}
		u = \nabla p + \sigma
	\end{equation*}
	where $\sigma$ is divergence-free. One of the main reasons for the ubiquity of this decomposition, especially in the realm of fluid mechanics,
	is that this single decomposition can be viewed through a variety of lenses.
	Here we recall how it can be viewed geometrically, viewed in relation to elliptic PDE, and viewed in relation to slow-fast dynamical systems.
	We emphasize these three perspectives because these are precisely the perspectives that coincide with each other in the dry case but diverge in the moist case.

	We can view the Helmholtz decomposition geometrically: it decomposes the space of three-dimensional vector fields into two orthogonal pieces,
	namely potential vector fields and divergence-free vector fields,
	and thus each component of the decomposition can be understood as a projection onto the corresponding subspace.
	The Helmholtz decomposition is also related to elliptic PDE since the scalar potential $p$ solves
	\begin{equation*}
		\Delta p = \nabla\cdot u
	\end{equation*}
	while, if we write the solenoidal component as $\sigma = \nabla\times\psi + c$, the vector potential $\psi$ solves
	\begin{equation*}
		\nabla\times (\nabla\times \psi) = \nabla\times u
	\end{equation*}
	with $c$ given by $c = \fint u$.
        Here we use the notation $\fint f$ to denote the average over the spatial domain of any scalar or vector function $f$.

Finally we can view the Helmholtz decomposition as a slow-fast decomposition for an appropriate dynamical system.
The relevant system is the low Mach number limit of the compressible Euler equations (for an ideal isentropic gas)
\cite{klainerman1981singular,klainerman1982compressible,majda1984compressible}.
In this setting the slow dynamics correspond to incompressible flow whereas the fast dynamics correspond to the acoustic wave equation.
Crucially: the solenoidal component of the Helmholtz decomposition is precisely the slow piece and the potential component is precisely the fast piece.
We postpone a detailed discussion of this slow-fast dynamics to \fref{Section}{app:Helmholtz}:
this low Mach number limit is well-known, so we only provide details later once it helps us compare and contrast these classical results (pertaining to Helmholtz and dry Boussinesq)
	with the results we have obtained in this paper.

\subsection{Dry Boussinesq and its decomposition}
\label{sec:dry_Bouss_intro}

Now we turn our attention to the dry Boussinesq system and its decomposition. As we describe this decomposition we will highlight the properties that it
shares with the Helmholtz decomposition discussed above.

In oceanic and atmospheric application problems, one may need a decomposition of observational or model data. The latter is often motivated by, or closely related to, normal-mode analysis of a PDE model such as the dry Boussinesq equations, e.g. \cite{lienmuller1992,buhlercalliesferrari2014,lelongetal2020,zagaretal2023}.  Furthermore, a normal-mode decomposition is often the starting point for numerical computations \cite{leith1980,bartello95,smithwaleffe2002,herbertetal2016,earlyetal2020,vas-zagar2021}.   For the Boussinesq equations, the normal-mode analysis is also a wave-vortical decomposition, as further explained below.  Although many numerical computations focused on constant rotation and buoyancy frequencies in periodic domains \cite{bartello95,smithwaleffe2002,herbertetal2016}, other analyses also address more realistic PDE systems, such as 
variable-stratification Boussinesq dynamics with rigid surface and bottom boundaries \cite{earlyetal2020} and 
the shallow water equations on a sphere \cite{kasahara1978, vas-zagar2021}.

Let's stop and discuss this terminology for a second: why is the decomposition of dry Boussinesq states called a vortical-wave decomposition?
It is called this because the decomposition has two components.
The first component is characterised by the potential vorticity. This is why it is referred to as the ``vortical mode''. (In the dry case this corresponds to an honest-to-goodness eigenmode. This is no longer true in the nonlinear moist case, which is why we use the term ``component'' instead of ``mode'' in that case.)
The second component corresponds to the fast piece of the decomposition. It can then be shown that the dynamics obeyed by this fast piece are that of linear oscillations.
This is why this component is known as the ``wave mode''. Together the vortical mode and the wave mode thus constitute the vortical-wave decomposition.

The dry Boussinesq system and its fast rotation and strong stratification limit is one of the success stories of atmospheric science.
In the associated parameter regime the dry Boussinesq system may be written in terms of conservation laws:
\begin{subnumcases}{}
		\partial_t u + u\cdot\nabla u + \frac{1}{\varepsilon} ( e_3\times u + \nabla p - \theta e_3 ) = 0 \label{eq:dry_Boussinesq_u} \text{ and } \\
		\partial_t \theta + u\cdot\nabla \theta + \frac{1}{\varepsilon} u_3 = 0 \label{eq:dry_Boussinesq_theta}
	\end{subnumcases}
	subject the incompressibility assumption $\nabla\cdot u = 0$. Note that for simplicity physical parameters of order one have be set to unity.
	Here the unknowns are the velocity vector field $u$ and the potential temperature $\theta$ (a scalar field).
	The first equation corresponds to conservation of linear momentum while the second equation corresponds to the conservation of potential temperature.
	Of course, as written above \eqref{eq:dry_Boussinesq_theta} makes it look like $\theta$ is \emph{not} conserved.
	This is because there is a large background potential temperature gradient (this is a manifestation of the strong stratification assumption)
	and so the quantity $\theta$ used above is actually the \emph{anomaly} of the potential temperature about that background gradient.

	Here the small parameter indicates that we are particularly interested in the behaviour of this system as the stratification is strong and the rotation is fast.
	In the limit we formally obtain from \eqref{eq:dry_Boussinesq_theta} that the vertical component of the velocity vanishes and we obtain from the \eqref{eq:dry_Boussinesq_u} the
	geostrophic balance $u_h = \nabla^\perp_h p$ and the hydrostatic balance $\theta = \partial_3 p$.
	Here $v_h = (v_1,\, v_2,\, 0)$ denotes the horizontal part of a vector $v$ and $v_h^\perp$ denotes its $\frac{\pi}{2}$--rotation $v_h^\perp = (-v_2,\,v_1,\,0)$.
	Note that the velocity is horizontal since its vertical component vanishes, yet it still depends on all three spatial coordinates
	(in this sense the flow is \emph{not} two-dimensional). 
 The discussion is formal here, but rigorous theorems are available (see \cite{em96} and \cite{me98} for balanced and general initial data, respectively).

	In particular, the slow dynamics are characterised by the potential vorticity $PV = \nabla_h^\perp \cdot u_h + \partial_3 \theta$.
	This follows from two observations which hold in the limit:
	(1) we may compute from \eqref{eq:dry_Boussinesq_u}--\eqref{eq:dry_Boussinesq_theta} that $PV$ is purely advected by the flow
	and (2) the geostrophic and hydrostatic balances tell us that $p$ solves $\Delta p = PV$.	
 This is the crux of why the limiting Boussinesq dynamics on the slow timescale, known as the (dry) quasi-geostrophic equations, are a success story:
	everything comes down to a single conserved quantity, the potential vorticity, from which all other information about the system may be recovered (by way of inverting a Laplacian).
	The potential vorticity not only helps us understand the limit, it also allows us to decompose states of the original dry Boussinesq system into slow and fast components. 

	We now discuss this decomposition, also known as the vortical-wave decomposition.
	This decomposition splits the state-vector $(u,\,\theta)$ into two pieces:
	\begin{equation*}
		\begin{pmatrix}
			u \\ \theta
		\end{pmatrix} = \begin{pmatrix}
			\nabla^\perp_h p \\ \partial_3 p
		\end{pmatrix} + \begin{pmatrix}
			\wavecoord_h^\perp + we_3 \\ \wavecoord_3
		\end{pmatrix}
	\end{equation*}
	where $\wavecoord$ is divergence-free.
	The first component is known as the \emph{vortical mode} and the second component is known as the \emph{wave mode}.
	As in the case of the Helmholtz decomposition above, this decomposition has several equivalent interpretations.
	We can view this decomposition geometrically as projecting into two $L^2$--orthogonal subspaces:
	the first piece is the set of balanced states satisfying both geostrophic and hydrostatic balance
	and the second piece is the set of states with vanishing potential vorticity.
	This is similar to the Helmholtz decomposition where the first piece, the potential component, satisfies the ``balance'' $u = \nabla p$
	while the second piece belongs to the set of states with vanishing divergence.
	Note that $L^2$ is used here, not by default, but because it is precisely the energy conserved by the dry Boussinesq system \eqref{eq:dry_Boussinesq_u}--\eqref{eq:dry_Boussinesq_theta}.
	We can view this decomposition in relation to elliptic PDE since $p$ solves
	\begin{equation*}
		\Delta p = PV.
	\end{equation*}
	Meanwhile, since $\wavecoord$ is divergence-free, we can write it in terms of a vector potential as $\wavecoord = \nabla\times\psi + c$ for some constant $c$ where $\psi$ solves
	\begin{equation*}
		\nabla\times (\nabla\times \psi) = \begin{pmatrix}
			j \\ \partial_3 w
		\end{pmatrix}
	\end{equation*}
	for $j = \partial_3 u_h - \nabla_h^\perp \theta$ denoting the thermal wind imbalance\footnote{
    A thermal wind refers to a velocity field which satisfies both geostrophic and hydrostatic balance.
    This terminology is warranted by the fact that such velocity fields are fully determined by the pressure through geostrophic balance, which in turn is defined by the ``thermal'' variable, namely the potential temperature, through hydrostatic balance.
    Here the thermal wind \emph{imbalance} vanishes when both the geostrophic an hydrostatic balances hold. This imbalance can therefore be thought of as measuring how far a given state $(u,\, \theta)$ is from satisfying both of these balances.
    }, 
	$c = - \fint u_h^\perp + \fint \theta e_3$,
	and finally $w = u_3$.
	Finally we can view this decomposition as a slow-fast decomposition for \eqref{eq:dry_Boussinesq_u}--\eqref{eq:dry_Boussinesq_theta}.
	Here we describe only cursorily how this may be done, postponing a more detailed discussion to \fref{Section}{sec:slow_fast_decomp}. 
	Indeed, the vortical mode corresponds precisely to solutions of the limiting quasi-geostrophic system
	while the wave mode corresponds to the component of the solution with time derivative proportional to $1/\varepsilon$ (and hence fast).
	We discuss in more detail why this decomposition works as a slow-fast decomposition in \fref{Section}{sec:slow_fast_decomp},
	where we recall how the fast-wave averaging framework allows us to make this idea precise in order to then use this framework to discuss the slow-fast decomposition in the moist case.

 \subsection{Moist Boussinesq}
\label{sec:moist_Bouss_intro}

	We have seen in \fref{Sections}{sec:Helmholtz} and \ref{sec:dry_Bouss_intro} above that both the Helmholtz decomposition
	and the vortical-wave decomposition possess several interpretations, be they geometric, related to elliptic PDE, or about the slow-fast nature of appropriate dynamical systems.
	In this section we introduce the moist Boussinesq system, keeping in mind that our goal is to find a decomposition for that system which possesses properties akin to those discussed above.

	First, we motivate the moist Boussinesq system itself.
	Its \emph{dry} counterpart is potent for analytical, computational, and predictive purposes because it focuses on two physical features which dominate mid-latitude atmospheric dynamics:
	stratification and rotation (i.e. the Coriolis effect).
	The \emph{moist} Boussinesq system introduces another important physical feature: moisture.
	Even without incorporating all of the complexities of the water cycle, moisture plays an important role in atmospheric dynamics
	because of the role it plays in the energy budget:
	as moist air cools, condensing water releases latent heat energy (and vice-versa as air warms and liquid water evaporates).
	When incorporating moisture into the dry system from \eqref{eq:dry_Boussinesq_u}--\eqref{eq:dry_Boussinesq_theta}, we obtain 
\begin{subnumcases}{}
		\partial_t u + u\cdot\nabla u + \frac{1}{\varepsilon} ( e_3\times u + \nabla p - \theta e_3 ) = 0 \label{eq:moist_source_Boussinesq_u}, \\
		\partial_t \theta + u\cdot\nabla \theta + \frac{\Gamma_\theta}{\varepsilon} u_3 = C \label{eq:moist_source_Boussinesq_theta}, \\
		\partial_t q_v + u\cdot\nabla q_v - \frac{\Gamma_q}{\varepsilon} u_3 = -C \label{eq:moist_source_Boussinesq_qv} \text{ and } \\
		\partial_t q_l + u\cdot\nabla q_l = C \label{eq:moist_source_Boussinesq_ql}.
	\end{subnumcases}
In comparison to the dry Boussinesq system,
the moist Boussinesq system includes additional evolution equations for the
water vapor mixing ratio, $q_v$, in \eqref{eq:moist_source_Boussinesq_qv}
and for the liquid water mixing ratio, $q_l$, in \eqref{eq:moist_source_Boussinesq_ql}. 
It also includes a source/sink term $C$
for condensation and evaporation, and associated heating and cooling.
Condensation occurs for $C>0$ and represents a
phase transition from vapor to liquid, and 
evaporation occurs for $C<0$ and represents a
phase transition from liquid to vapor.
The constants $\Gamma_\theta$ and $\Gamma_q$ correspond to parameters encoding the strength of the background vertical gradients in potential temperature and water vapor, respectively.
We also note that the precise definition of the water content $q_v$ or $q_l$ as a mixing ratio is measured as kilograms of total water (vapor plus liquid) per kilogram of dry air 
\cite{klein_majda_06, hernandez_duenas}.
While additional cloud microphysics processes will not be
considered here, the moist Boussinesq system above
is valuable for more complex scenarios as well, since it
provides the starting point for extensions that
include rainfall, ice, other precipitation, and other complexities
\cite{grabowski1996two,klein_majda_06,smith2017precipitating,wetzel2019balanced,wetzel2020potential}.

Second, we rewrite the moist Boussinesq system in a way that
facilitates the main goal here of a slow--fast state decomposition.
In particular, note that the source terms $C$ in
\eqref{eq:moist_source_Boussinesq_theta}--\eqref{eq:moist_source_Boussinesq_ql}
appear to potentially cause complications.
However, the system can be rewritten so that the source terms
do not explicitly appear,
through a convenient change of variables.
This reformulation has been used in several past studies
\cite{bretherton1987theory,pauluis2010idealized,hernandez_duenas,smith2017precipitating,marsico2019energy,kooloth2023hamilton}
and has facilitated theoretical advances such as 
an energy decomposition 
\cite{smith2017precipitating,marsico2019energy}
and conservation of potential vorticity \cite{kooloth2023hamilton},
even in the presence of clouds and phase changes.

For the convenient way of rewriting \eqref{eq:moist_source_Boussinesq_u}--\eqref{eq:moist_source_Boussinesq_ql},
we transform to a different set of thermodynamic variables that
is conserved. In particular, define the 
equivalent potential temperature, $\theta_e$,
and total water mixing ratio, $q_t$, as
\begin{subnumcases}{}
    \theta_e = \theta+q_v, 
    \label{eq:thetae-def} 
    \\
    q_t = q_v + q_l. 
    \label{eq:qt-def}
\end{subnumcases}
Note that the equivalent potential temperature $\theta_e$ may be understood physically as the potential temperature that a parcel of air would have if it were heated by converting all of its water vapour into liquid water.
The evolution equations of $\theta_e$ and $q_t$ can be found
by taking appropriate linear combinations of
\eqref{eq:moist_source_Boussinesq_theta}--\eqref{eq:moist_source_Boussinesq_ql},
and they take the form
\begin{subnumcases}{}
\partial_t \theta_e + u\cdot\nabla \theta_e + \frac{1}{\epsilon}u_3 = 0, 
\label{eq:thetae-evol}
\\
\partial_t q_t + u\cdot\nabla q_t - \frac{1}{\epsilon} u_3 = 0,
\label{eq:qt-evol}
\end{subnumcases}
where we have now set the background gradient parameters from
\eqref{eq:moist_source_Boussinesq_theta} and \eqref{eq:moist_source_Boussinesq_qv}
to be $\Gamma_\theta=2$ and $\Gamma_q=1$
without loss of generality within the stably stratified regime,
and for simplicity.
The important feature of \eqref{eq:thetae-evol}--\eqref{eq:qt-evol}
is that the source term of condensation and evaporation, $C$,
has been eliminated.
As a result, \eqref{eq:thetae-evol}--\eqref{eq:qt-evol}
show that $\theta_e+(z/\epsilon)$ and 
$q_t-(z/\epsilon)$ (which combine the anomalies $\theta_e$ and $q_t$
with their background profiles $z/\epsilon$ and $-z/\epsilon$, respectively)
are conserved along fluid parcel trajectories.
Physically, the equivalent potential temperature
is conserved because losses of water vapor $q_v$
are compensated by gains in heat $\theta$,
as indicated by the definition in \eqref{eq:thetae-def};
and the total water mixing ratio, $q_t=q_v+q_l$, is conserved because
losses of water vapor $q_v$ are compensated by gains in
liquid water $q_l$.

To complete the rewriting, the buoyancy term $\theta/\epsilon$
from \eqref{eq:moist_source_Boussinesq_u} must also be
rewritten in terms of $\theta_e$ and $q_t$.
To do so, the following transformation is used:
\begin{subnumcases}{}
    \theta = \theta_e - \min(q_t,q_{vs}),
    \label{eq:theta-def} 
    \\
    q_v = \min(q_t,q_{vs}),
    \label{eq:qv-def} 
    \\
    q_l = \max(0,q_t-q_{vs}),
    \label{eq:ql-def}
\end{subnumcases}
which is the reverse of the transformation in
\eqref{eq:thetae-def}--\eqref{eq:qt-def}.
The saturation water vapor value $q_{vs}$ appears here
for distinguishing between the unsaturated phase ($q_v<q_{vs}$),
when no cloud is present,
and the saturated phase ($q_v=q_{vs}$),
when a cloud is present.
Note that $q_{vs}$ is not appearing here for the first time. It is already implicit in the condensation/evaporation source term $C$ in \eqref{eq:moist_source_Boussinesq_theta}--\eqref{eq:moist_source_Boussinesq_ql}.
The value of $q_{vs}$ will be constant here for simplicity.
For explicit expressions for $\theta$ in each phase,
one can rewrite \eqref{eq:theta-def} as
\begin{equation}
\theta = \left\{
\begin{array}{ll}
\theta_e -q_t 
& \mbox{if}\quad q_t<q_{vs},
\\[4pt]
\theta_e - q_{vs} 
& \mbox{if}\quad q_t\ge q_{vs},
\end{array}
\right. 
\end{equation}
which can be used for another expression of the buoyancy
in terms of the variables $\theta_e$ and $q_t$.

        The end result of the rewriting is
	the moist Boussinesq system in the following form:
	\begin{subnumcases}{}
		\partial_t u + u\cdot\nabla u + \frac{1}{\varepsilon} ( e_3\times u + \nabla p - (\theta_e - {\min}_0\, q_t) e_3 ) = 0 \label{eq:moist_Boussinesq_u},\\
		\partial_t \theta_e + u\cdot\nabla \theta_e = -\frac{1}{\varepsilon} u_3 \label{eq:moist_Boussinesq_theta} \text{ and } \\
		\partial_t q_t + u\cdot\nabla q_t = \frac{1}{\varepsilon} u_3
        \label{eq:moist_Boussinesq_q},
	\end{subnumcases}
	subject to $\nabla\cdot u=0$, and
        following from \eqref{eq:moist_source_Boussinesq_u},
        \eqref{eq:thetae-evol}, \eqref{eq:qt-evol},
        and \eqref{eq:theta-def}.
    Some additional specifications are as follows.
    The notation $\min_0 q_t = \min(0,\,q_t)$ has been introduced
    in \eqref{eq:moist_Boussinesq_u}, 
    and the saturation water vapor parameter $q_{vs}$ has been
    set to zero without loss of generality, and for simplicity.
	Throughout this paper we work on the three-dimensional torus (i.e. with triply-periodic boundary conditions).
    Also recall that we set the background gradient parameters from
    \eqref{eq:moist_source_Boussinesq_theta} and \eqref{eq:moist_source_Boussinesq_qv}
    to be $\Gamma_\theta=2$ and $\Gamma_q=1$ 
    without loss of generality in the stably stratified regime,
    and for simplicity.
One can now see the value in rewriting the moist Boussinesq system:
while the $(\theta,q_v,q_l)$ formulation in
\eqref{eq:moist_source_Boussinesq_u}--\eqref{eq:moist_source_Boussinesq_ql}
is helpful for seeing the additional cloud physics that
is added to the dry system,
the $(\theta_e,q_t)$ formulation in
\eqref{eq:moist_Boussinesq_u}--\eqref{eq:moist_Boussinesq_q}
is helpful for writing the moist system in a way that is
mathematically similar to the dry system.
The main differences between the moist and dry systems are
that the moist system has two thermodynamic variables
($\theta_e$ and $q_t$) instead of one ($\theta$),
and the buoyancy term in the moist case
involves the nonlinear expression
$\theta_e - {\min}_0\, q_t$.

	What makes the moist Boussinesq system particularly interesting
	is that moisture introduces \emph{phase boundaries}.
	In concrete terms: here we have assumed without loss of generality that water saturation occurs at $q_t=0$ and so different physics are in play
	depending on the sign of $q_t$, i.e. depending on whether water is above or below saturation
	(this is discussed in more detail in \fref{Section}{sec:background}).
	We therefore refer to the set $ \left\{ q_t = 0 \right\}$ as the \emph{phase boundary} since it delineates the separation between water in vapour form (outside of a cloud)
	and water in vapour-plus-liquid form (inside of a cloud).

Finally, to end this subsection, 
we review some properties of the moist Boussinesq system
that are potentially relevant to the aim here---i.e., to the search for a state decomposition.

	As $\varepsilon\to 0$, solutions of the moist Boussinesq system \eqref{eq:moist_Boussinesq_u}--\eqref{eq:moist_Boussinesq_q} are forced to satisfy two \emph{balances},
	namely geostrophic balance $u = u_h = \nabla_h^\perp p$ and hydrostatic balance $\theta_e - \min_0 q_t = \partial_3 p$,
	which come from the fast rotation and strong stratification, respectively.
	This warrants the following definition.
	\begin{definition}[State space and balanced set]
	\label{def:state_space}
		We denote by $\mathbb{L}_\sigma^2 \vcentcolon= {(L^2)}^3_\sigma \times L^2 \times L^2$
		the state space of the moist Boussinesq system \eqref{eq:moist_Boussinesq_u}--\eqref{eq:moist_Boussinesq_q}.
		Here $u \in {(L^2)}^3_\sigma$ means that $u\in L^2 ( \mathbb{T}^3; \mathbb{R}^3)$ is divergence-free.
		We then define the \emph{balanced set}
		\begin{equation}
			\mathcal{B} \vcentcolon= \left\{
				(u,\,\theta_e,\,q_t) \in \mathbb{L}^2_\sigma : 
				u_h = \nabla_h^\perp p,\, u_3 = 0,\, \text{and } \theta_e - {\min}_0\, q_t = \partial_3 p
				\text{ for some } p\in \mathring{H}^1
			\right\}
		\label{eq:def_balanced_set}
		\end{equation}
		where $ \mathring{H}^1 $ is the set of scalar fields in $H^1$ with vanishing averages.
	\end{definition}
	We note that a state belongs to the balanced set precisely when it satisfies both geostrophic and hydrostatic balance.

	Proceeding once again by analogy with the dry case, one may now identify \emph{two} slow quantities:
	the \emph{equivalent} potential vorticity ${PV}_e = \nabla_h^\perp \cdot u_h + \partial_3 \theta_e$, defined now in terms of the \emph{equivalent} potential temperature,
	and the moist variable $M = \theta_e + q_t$ \cite{smith2017precipitating}.
	Similar to the dry case, now $p$ and $M$ fully determine a balanced state by way of the ${PV}_e$-and-$M$ inversion
	\begin{equation}
		\Delta p + \frac{1}{2} \partial_3 {\min}_0\, (M - \partial_3 p) = {PV}_e,
        \label{eq:PV-and-M-inversion-section-2}
	\end{equation}
	which is obtained as before by inserting the balances into the definition of ${PV}_e$. This is a nonlinear elliptic PDE whose well-posedness was established recently \cite{remond2024nonlinear}.
	Using this PDE for ${PV}_e$-and-$M$ inversion, one can then phrase the limiting moist (and potentially \emph{precipitating}) quasi-geostrophic system as a nonlinear transport equation for ${PV}_e$ and $M$ \cite{smith2017precipitating}.

    It is important to note that the nonlinear elliptic PDE in
    \eqref{eq:PV-and-M-inversion-section-2}
    only has meaning for the \emph{balanced states} from
    \eqref{eq:def_balanced_set},
    at this point in the story. 
    This is in contrast to the cases of 
    the Helmholtz decomposition and dry Boussinesq system
    from \fref{Sections}{sec:Helmholtz} and \ref{sec:dry_Bouss_intro}, respectively,
    for which an elliptic PDE is part of a state decomposition
    for \emph{general} states.
    One aim of the present paper is to investigate whether
    the nonlinear ${PV}_e$-and-$M$ inversion in 
    \eqref{eq:PV-and-M-inversion-section-2}
    may have additional meaning as part of a state decomposition
    for \emph{general} states.
    
	The question is now this: can we find a decomposition for the moist Boussinesq system 
	which has the same favourable properties as the Helmholtz and vortical--wave decompositions?

\subsection{Literature on atmospheric dynamics with moisture and clouds, and mathematical theory}
\label{sec:literature}

In this section we review some literature on topics related to the
main focus of the paper, i.e., state decomposition for a
moist atmosphere with clouds.
We mention both atmospheric sciences literature and
mathematics literature.

Decompositions have been useful for numerous applications
in atmospheric and science, and they are commonly based on
decomposing the dry portions $(u,\,\theta)$ of the atmospheric state.
Applications include 
analysis of observational or model data \cite{hoskins1985use,zetal15},
data assimilation \cite{lorenc1981global,daley1993atmospheric,derber1999reformulation,zgk04qjrms,bannister2021balance},
and nonlinear wave interactions
\cite{gottwald1999formation,majda2003nonlinear,biello2004effect,raupp2008resonant,ferguson2009two}.
Other references are also described at the beginning of \fref{Section}{sec:dry_Bouss_intro}.
One aim here is to extend the ideas of dry decompositions
to the realm of a moist atmosphere with clouds.

For a moist atmosphere with clouds, decompositions can potentially
provide valuable insight into the behavior of
rainfall and precipitation \cite{khouider2012climate}.
Past work has commonly used linearisation, and either
dry or moist eigenmodes.
By adapting or modifying the dry linear eigenmodes,
some studies have analyzed observational and model data
for understanding moisture-coupled phenomena
such as convectively coupled equatorial waves (CCEWs)
\cite{yang2003convectively,gehne2012spectral,ogrosky2016identifying,marques2018diagnosis,knippertz2022intricacies},
tropical intraseasonal oscillations such as the 
Madden--Julian Oscillation (MJO)
\cite{zagar2015systematic},
and the Walker circulation \cite{stechmann2014walker}.
Moist eigenmodes, obtained by linearisation of
moist dynamics, have also been used for
analyzing observational and model data
for the MJO
\cite{ogrosky2015mjo,stechmann2015identifying},
and for asymptotic analysis of weakly nonlinear wave interactions \cite{chen2015multiscale,chen2016tropical}.
Balanced and unbalanced moisture decompositions have been investigated with a type of linearisation that applies with clouds present \cite{wetzel2019balanced,wetzel2020potential}.
Applications to data assimilation have also investigated linear eigenmodes that decompose moisture into its balanced and unbalanced components \cite{ogrosky2024data}.

Here we aim for a decomposition that does not rely on linearisation,
so that it may potentially include a more faithful representation 
of the nonlinearity and the physics 
associated with clouds and latent heating.

The PDE community has recently provided rigorous treatment of moist atmospheric dynamics with phase changes.  A model with water vapor and liquid water was analyzed in 
\cite{coti_zelati_temam_12,coti_zelati_fremond_temam_tribbia_13, bousquet_coti_zelati_temam_14},
where the velocity field is prescribed.  Then followed results for two-phase flow with dynamical velocity evolving via the primitive equations \cite{coti_zelati_huang_kukavica_temam_ziane,lian_ma,temam_wu_15, teman_wang_16}.  An additional category of water was included in 
\cite{cao_hamouda_temam_tribbia,tan_liu,hittmeir_klein_li_titi_17,hittmeir_klein_li_titi_20}, such that liquid water is divided into cloud water (that does not fall) and precipitating water.  The work on warm-rain microphysics flow may be separated by the microphysics model determining how water is partitioned between vapor, cloud water and rain water, as well as by whether the velocity is prescribed or governed by the primitive equations.  The case of prescribed velocity is considered in \cite{cao_hamouda_temam_tribbia} for Grabowski microphysics \cite{grabowski}, and in \cite{hittmeir_klein_li_titi_17} for Klein-Majda microphysics \cite{klein_majda_06}.  For primitive-equation velocity, theorems for the Grabowski and Klein-Majda models are found, respectively, in \cite{tan_liu} and \cite{hittmeir_klein_li_titi_20}.
Incorporation of ice into water microphysics is discussed in \cite{cao_jia_temam_tribbia}.

We note that the aforementioned rigorous studies of atmospheric flows with phase changes treat the velocity as either prescribed, or as governed by the primitive equations.
Following \cite{cao_titi}, the barotropic-baroclinic decomposition of the velocity is the basis for analysis of the primitive equations and its extensions.  Recently, the current authors obtained existence, uniqueness and regularity results \cite{remond2024nonlinear} for the nonlinear elliptic PDE \eqref{eq:PV-and-M-inversion-section-2} underlying Boussinesq velocity dynamics in the limit of asymptotically fast rotation rate and buoyancy frequencies \cite{smith2017precipitating,zhang_smith_stechmann_22}.  Here, we show that the same nonlinear elliptic PDE is the foundation for decomposition of the Boussinesq equations more generally, in the regime of finite parameter values.

Pertaining to the dry Boussinesq evolution equations in the limit of asymptotically fast rotation rate and buoyancy frequency, a series of rigorous results established the decoupling between slow, balanced motions and fast, unbalanced waves \cite{em96,babinetal1997,em98,me98,majda_2003}. For moist atmospheric dynamics with phase changes, first steps to use the fast-wave averaging framework \cite{em96,em98,me98,majda_2003} appear in \cite{zhang_smith_stechmann_20_asymptotics,Zhang_Smith_Stechmann_2021_JFM}.  The latter references present formal asymptotic analysis \cite{zhang_smith_stechmann_20_asymptotics} and supporting numerical analyses \cite{Zhang_Smith_Stechmann_2021_JFM}. The inclusion of boundary layers and the diabatic layer have been investigated in formal asymptotic analysis as well \cite{klein2022qg,baumer2023scaling}. We note that rigorous fast-wave-averaging analysis for multi-phase atmospheric flows remains open.  In \fref{Section}{sec:slow_fast_decomp}, we explain the path forward for rigorous fast-wave-averaging analysis using the decomposition defined in \fref{Section}{sec:decomposition}. To date, the evidence suggests that nonlinear oscillations lead to coupling between fast and slow flow components, possibly even in the limiting dynamics when distinct phases co-exist (see \fref{Figure}{fig:nonlinear_oscillations} in the appendix and \cite{Zhang_Smith_Stechmann_2021_JFM}).

\subsection{Simplified notation}
\label{sec:notation}

	To simplify the notation used throughout the rest of the paper, we will from now on use $\theta$ to denote the equivalent potential temperature $\theta_e$
	(which is not an issue since the potential temperature, denoted $\theta$ in the text above, will not make any appearances in the sequel).
	Similarly we will write $PV$ for the equivalent potential temperature ${PV}_e$ and $q$ to mean $q_t$, the total water content.

\section{Decomposition}
\label{sec:decomposition}

As discussed in \fref{Sections}{sec:intro} and \ref{sec:background}, we seek a decomposition adapted to the leading-order operator of the moist Boussinesq system.
We can see from \eqref{eq:moist_Boussinesq_u}--\eqref{eq:moist_Boussinesq_q} that this operator is 
	\begin{equation}
		\mathcal{N} (u,\,\theta,\,q) = \begin{pmatrix}
		    \mathbb{P}_L \left( u_h^\perp - (\theta - {\min}_0\, q \right) e_3 \\
            u_3\\
            -u_3
		\end{pmatrix}
    \label{eq:def_N}
	\end{equation}
 where we note again that $\theta_e$ and $q_t$ are written as $\theta$ and $q$ for notational simplicity, as specified in \fref{Section}{sec:notation}. Here $\mathbb{P}_L$ denotes the Leray projector onto the space of divergence-free vector fields. On the three-dimensional torus it takes the simple form $\mathbb{P}_L u= -\nabla \times ( \Delta^{-1} \nabla\times u)$.

Given that the operator $\mathcal{N}$ is nonlinear, it is 
likely impossible to define and characterize 
a full set of the nonlinear eigenfunctions
and eigenvalues. Hence it is likely impossible to 
define a full eigenfunction decomposition such as in
\eqref{eq:Udecomp}.
An alternative route to a decomposition must be sought.

Here we propose to use the notions of null set and image
as concepts that may allow us to extend from the linear to
the nonlinear realm. 
We propose a decomposition that will split any state in the state space $ \mathbb{L}^2_\sigma $ into two pieces:
	one piece residing in the null set of $ \mathcal{N} $ and another piece residing in the image of $ \mathcal{N} $.
	In other words we will perform the decomposition
	\begin{equation}
	\label{eq:decomp_foreshadow}
		\mathbb{L}^2_\sigma = \left\{ \mathcal{N} = 0 \right\} + \im \mathcal{N}.
	\end{equation}
If the operator $ \mathcal{N} $ were \emph{linear}, 
then this decomposition would be a (coarse) eigendecomposition,
into the eigenspace of the zero eigenvalue and the eigenspace of
all nonzero eigenvalues. Such eigenspace ideas are helpful as
motivation for pursuing the null set and image as fundamental concepts
for defining a decomposition here. However, the operator $\mathcal{N}$
is nonlinear, so the full meaning of \eqref{eq:decomp_foreshadow} 
is not immediately clear.
We aim to demonstrate in this section that 
\eqref{eq:decomp_foreshadow} is able to characterize the desired
decomposition in the nonlinear realm.
	
 Note that in this section we will only present and discuss this decomposition itself.
	We will not discus the various interpretations of this decomposition.
	This will be done in later sections.
	In \fref{Section}{sec:slow_fast_decomp} we discuss how this decomposition is related to slow-fast dynamics,
	in \fref{Section}{sec:energy_and_decomposition} we discuss how this decomposition relates to the conserved energy and projections, and
	finally in \fref{Section}{sec:iterative_methods} we discuss iterative methods used to compute this decomposition in practice and their geometric interpretation.

	Since we wish to perform the decomposition foreshadowed in \eqref{eq:decomp_foreshadow} we must begin by understanding both the null set and the image of $ \mathcal{N} $ as well as possible.
	We first turn our attention to the null set.
	As shown in the result below, this null set is precisely the set of states consisting of states in both hydrostatic and geostrophic balance.
	This is why the null set of $ \mathcal{N} $ will be referred to, from now on, as the \emph{balanced set}.

	\begin{prop}[Alternate characterisations of the balanced set]
	\label{prop:alt_char_balanced_set}
		For any state $ \mathcal{X}\in \mathbb{L}^2_\sigma $ the following are equivalent.
		\begin{enumerate}
			\item	$ \mathcal{X}$ is balanced, i.e., $\mathcal{N}(u,\,\theta,\,q)=0$.
			\item	$ \mathcal{X} = (u,\,\theta,\,q)$ satisfies
				\begin{equation}
					\partial_3 u_h - \nabla_h^\perp (\theta - {\min}_0\, q ) = 0,\,
					u_3 = 0,\,
					\fint u_h = 0, \text{ and } 
					\fint \theta - {\min}_0\, q = 0.
				\label{eq:char_balanced_set_as_zero_set}
				\end{equation}
			\item	$ \mathcal{X}$ may be parameterised as $ \mathcal{X} = \Phi (p,\,M)$ for some unique $p\in \mathring{H}^1 $ and $M\in L^2$, where
				\begin{equation*}
					\Phi (p,\,M)
					= \begin{pmatrix}
						\nabla^\perp_h p\\
						\partial_3 p + \frac{1}{2} {\min}_0\, (M - \partial_3 p)\\
						M - \partial_3 p - \frac{1}{2} {\min}_0\, (M-\partial_3 p)
					\end{pmatrix}.
				\end{equation*}
		\end{enumerate}
	\end{prop}

	Note that this proposition characterizes the balanced set in three ways:
	\begin{itemize}
		\item	it is the null set of $ \mathcal{N} $,
		\item	it is characterised as the null set of the three \emph{measurements} $j$, $a$, and $w$ given by
          \begin{align*}
              j(u,\,\theta,\,q) = \partial_3 u_h^\perp - \nabla_h^\perp (\theta - {\min}_0\,q),\,
              w(u,\,\theta,\,q) = u_3\\
              \text{ and } a(u,\,\theta,\,q) = -\fint u_h^\perp + \fint (\theta - {\min}_0\,q ) e_3,
          \end{align*}
          and
		\item	it may be parameterised via the \emph{coordinates} $p$ and $M$.
	\end{itemize}
    The measurement $j$ is commonly referred to as the \emph{thermal wind imbalance} since it quantifies how far the state is from obeying both geostrophic and hydrostatic balanced.
	Further down we will see that the image of $ \mathcal{N} $ also admits three comparable descriptions.
    The measurement $a$ depends on some spatial averages of the state $(u,\,\theta,\,q)$.
    Why it depends on these particular averages (e.g. it depends on the horizontal velocity average but not the vertical velocity average) and why these averages are combined in this way to constitute $a$ is explained in \fref{Remark}{rmk:good_unknown}.

	\begin{proof}[Proof of \fref{Proposition}{prop:alt_char_balanced_set}]
		It will be particularly convenient to consider the new unknown
		\begin{equation*}
			v \vcentcolon= -u_h^\perp + (\theta - {\min}_0\, q)e_3.
		\end{equation*}
		Indeed, we see that using this unknown we may rephrase items 1, 2, and 3 in a much simpler way.
		The first item may now be written as
		\begin{equation*}
			v = \nabla p \text{ and } u_3 = 0
			\text{ for some } p \in \mathring{H}^1
		\end{equation*}
		while the second item may be written, using \fref{Lemma}{lemma:div_and_curl_of_good_unknown}, as
		\begin{equation*}
			\nabla\times v = 0,\, \fint v = 0, \text{ and } u_3 = 0.
		\end{equation*}
		It then follows directly from the Helmholtz decomposition (see \fref{Corollary}{cor:Helmholtz}) that items 1 and 2 are equivalent.
		Finally we note that, by virtue of \fref{Lemma}{lemma:invert_b_and_M}, item 3 is equivalent to
		\begin{equation*}
			\left\{
			\begin{aligned}
				& u = \nabla_h^\perp p,\\
				& \theta  + q = M, \text{ and } \\
				& \theta - {\min}_0\, q = \partial_3 p,
			\end{aligned}
			\right.
		\end{equation*}
		which implies that $v = \nabla p$ and that $u_3 = 0$, from which item 1 follows.
	\end{proof}

	\begin{remark}[The ``good unknown'']
	\label{rmk:good_unknown}
		The so-called ``good-unknown'' $-u_h^\perp + (\theta - {\min}_0\, q) e_3$ will appear throughout this paper.
		This is because this three-dimensional vector field, which combines the horizontal velocity and the buoyancy, turns out to be the one that needs to be decomposed
		according to a three-dimensional Helmholtz decomposition. This contrasts with the approach commonly used in the literature where instead of a three-dimensional Helmholtz decomposition
		of this good unknown (or of its dry counterpart, the good unknown $-u_h^\perp + \theta e_3$), the decomposition carried out is a two-dimensional Helmholtz decomposition
		of only the horizontal velocity. It is precisely because the typical decomposition is two-dimensional (even though the problem is three-dimensional)
		that vertically sheared horizontal flows often require special treatment. When an honest-to-goodness three-dimensional decomposition is used, these vertically sheared horizontal flows
		naturally accounted for (see \fref{Remark}{rmk:VSHF} where this is discussed in more detail).

		Note that the role played by the good unknown in the good decomposition is discussed in more detail in \fref{Remark}{rmk:global_chg_coord_dry_case}.
		For now, in order to illustrate why it is such a convenient unknown to work with, we record a couple of computations from the dry case.
		In the dry case, the good unknown is
		\begin{equation*}
			-u_h^\perp + \theta e_3.
		\end{equation*}
		We then observe from \fref{Lemma}{lemma:div_and_curl_of_good_unknown} that its divergence is precisely the potential vorticity since
		\begin{equation*}
			\nabla\cdot  ( -u_h^\perp + \theta e_3 ) = - \nabla_h^\perp \cdot u_h + \partial_3 \theta = - PV
		\end{equation*}
		while its curl is precisely the thermal wind imbalance since
		\begin{equation*}
			\nabla\times ( -u_h^\perp + \theta e_3 ) = \partial_3 u_h - \nabla_h^\perp \theta.
		\end{equation*}
		The physical relevance of thermal wind imbalance is that it vanishes precisely when the (dry) state $(u,\,\theta)$ is in both hydrostatic and geostrophic balance
		(provided that both $u_h$ and $\theta$ have vanishing averages).
		Once again: \fref{Remark}{rmk:global_chg_coord_dry_case} discusses in more detail the precise role that the \emph{moist} good unknown $-u_h^\perp + (\theta - {\min}_0\, q) e_3$ plays in the decomposition,
		but already the computations above, from the dry case, indicate that essential quantities may readily be computed from the (dry) good unknown.
	\end{remark}

	Now that we understand the null set quite well we turn our attention to the image of $ \mathcal{N} $.
	This set will actually be called the \emph{wave set}.
	This is because the dynamics induced on that set by the moist Boussinesq system are oscillatory
	(in the dry case, the dynamics induced on the dry analog of the wave set by the dry Boussinesq system are precisely the dynamics of linear oscillations, i.e. waves).

	\begin{prop}[Characterisation of the image of $\mathcal{N}$]
	\label{prop:chara_im_N}
		For any state $ \mathcal{X}\in \mathbb{L}^2_\sigma $ the following are equivalent.
		\begin{enumerate}
			\item	$ \mathcal{X}$ belongs to the image of $\mathcal{N}$.
			\item	$ \mathcal{X} = (u,\,\theta,\,q)$ satisfies
			\begin{equation}
	PV = \nabla_h^\perp \cdot u_h + \partial_3 \theta = 0 \text{ and } 
					M = \theta + q = 0.
				\label{eq:char_image_N_as_zero_set}
				\end{equation}
			\item	$ \mathcal{X}$ may be parameterised as $ \mathcal{X} = \Psi (\wavecoord,\,w)$ for some unique $\wavecoord\in H^1_\sigma$ and $w\in L^2$
				satisfying $\nabla_h^\perp \cdot \wavecoord_h = \partial_3 w$, where
				\begin{equation*}
					\Psi (\wavecoord,\,w)
					= \begin{pmatrix}
						\wavecoord_h^\perp + w e_3\\
						\wavecoord_3 \\
						-\wavecoord_3
					\end{pmatrix}.
				\end{equation*}
		\end{enumerate}
		In either case we say that $\mathcal{X} \in \mathcal{W}$ for the \emph{wave set} $\mathcal{W}$.
		Why this set warrants the name ``wave set'' will be discussed in \fref{Section}{sec:slow_fast_decomp} below when we discuss slow-fast decompositions in more detail.
	\end{prop}

	Note that, as was done for the balanced set above, \fref{Proposition}{prop:chara_im_N} characterises the wave set in three ways:
	\begin{itemize}
		\item	it is the image of  $ \mathcal{N} $,
		\item	it is characterised as the null set of the two measurements $PV$ and $M$, and
		\item	it may be parameterised via the coordinates $\wavecoord$ and $w$.
	\end{itemize}

	\begin{proof}
		First we suppose that item 1 holds, i.e. that
		\begin{equation*}
			\left\{
			\begin{aligned}
				& u = \mathbb{P}_L \left( v_h^\perp + (\phi - {\min}_0\, r ) e_3 \right),\\
				& \theta = v_3, \text{ and } \\
				& q = -v_3
			\end{aligned}
			\right.
		\end{equation*}
		for some $(v,\,\phi,\,r) \in \mathbb{L}^2_\sigma $.
		This means that $ u = v_h^\perp + (\phi - {\min}_0\, q) e_3 + \nabla\pi$ for some $\pi \in \mathring{H}^1 $ and so
		\begin{equation*}
			\nabla_h^\perp \cdot u_h + \partial_3 \theta
			= \nabla_h^\perp \cdot ( v_h^\perp + \nabla_h \pi ) + \partial_3 v_3
			= \nabla \cdot v
			= 0
		\end{equation*}
		while clearly $\theta + q = 0$. Thus item 2 holds.

		Now suppose that item 2 holds. We may introduce the new unknown $v = -u_h^\perp + \theta e_3$
		such that now item 2 reads
		\begin{equation*}
			\nabla \cdot v = 0 \text{ and } \theta + q = 0
		\end{equation*}
		(by virtue of \fref{Lemma}{lemma:div_and_curl_of_good_unknown}).
		In other words $v = \wavecoord$ for some divergence-free $\wavecoord$ while $\theta = -q$, and so
		\begin{equation*}
			\left\{
			\begin{aligned}
				&u = \wavecoord_h^\perp + u_3 e_3,\\
				&\theta = \wavecoord_3, \text{ and } \\
				&q = -\wavecoord_3.
			\end{aligned}
			\right.
		\end{equation*}
		Therefore, for $w \vcentcolon= u_3$, we see that item 3 holds as desired since $\nabla_h^\perp \cdot \wavecoord_h - \partial_3 w = - \nabla\cdot u = 0$.

		Finally we suppose that item 3 holds.
		Then item 1 follows immediately from choosing $v = \wavecoord$, $\phi = w$, and $r = 0$
		since $\wavecoord_h^\perp + we_3$ is divergence-free and so the Leray projection of this vector field is equal to itself.
	\end{proof}

	At this stage we have a good understanding of both the balanced set (which is the null set of $ \mathcal{N} $)
	and of the wave set (which is the image of $ \mathcal{N} $).
	We now seek to leverage this understanding to produce the decomposition formally foreshadowed in \eqref{eq:decomp_foreshadow}.

	In order to do this in a manner that will be useful later, especially when it comes to dynamics, we need to discuss dynamics briefly.
	Recall that, at leading order, the dynamics of interest are those governed by $ \partial_t + \mathcal{N} = 0$.
	So how do the null set and image of $ \mathcal{N} $ relate to this dynamical system?

	Well, note that \fref{Proposition}{prop:alt_char_balanced_set} tells us that producing \emph{balanced} states is easy: it suffices to specify $p$ and $M$
	since then a balanced state may be reconstructed from these two quantities.
	If the leading operator $ \mathcal{N} $ were linear, then this would be enough since balanced solutions of $ \partial_t + \mathcal{N} $ would then be guaranteed to be slow,
	meaning that they would have vanishing time derivatives.
	However, due to the presence of moisture and phase boundaries, $ \mathcal{N} $ is \emph{not} linear.
	This means that we must \emph{additionally} impose that balanced solutions be slow.
	In more prosaic terms: since any choice of the coordinates $p$ and $M$ produces a balanced state,
	we must make sure that these coordinates are carefully chosen to also be \emph{slow}.

	This is where our understanding of the image of $ \mathcal{N} $, i.e. of the wave set, comes in handy.
	Indeed: we have characterised this wave set as the set of states whose $PV$ and $M$ vanish.
    This is another way of saying that $PV$ and $M$ are precisely a complete set of \emph{time-dependent} slow measurements! (Any constant measurement would be trivially slow.)

	We have found our solution: to guarantee that our balanced component determined by $p$ and $M$ is also slow, it must be reconstructed solely from the slow measurements, $PV$ and $M$.
	If we think about this geometrically, this is really a matter of \emph{transversality}: we need the wave set $ \mathcal{W}$ to be transverse to the balanced set $\mathcal{B}$
	so that balanced states in $\mathcal{B}$ may be fully characterised by the measurements annihilating the wave set $\mathcal{W}$.
	This is taken care of in the result below.

	We can think about why the balanced set and the wave set need to be transverse another way.
	Recall that we are after a decomposition of the form \eqref{eq:decomp_foreshadow}.
	This decomposition is only unique if its two components are transverse, otherwise the same state could be decomposed in several different ways.

	\begin{lemma}[$PV$ and $M$ characterise balanced states]
	\label{lemma:PV_and_M_charac_balanced_states}
		Given $PV \in H^{-1}$ and $M\in L^2$, where $H^{-1}$ denotes the dual of $ \mathring{H}^1 $,
		there exists a unique $ \mathcal{X}\in \mathcal{B}$ for which $PV( \mathcal{X}) = PV$ and $\mathcal{M} ( \mathcal{X}) = M$.
	\end{lemma}
	\begin{proof}
		By virtue of \fref{Proposition}{prop:alt_char_balanced_set} it suffices to show that there exist a unique $p\in \mathring{H}^1 $ and a unique $\widetilde{M} \in L^2$
		for which $(PV \circ \Phi) (p,\,\widetilde{M}) = PV$ and $(\mathcal{M} \circ \Phi) (p,\,\widetilde{M}) = M$.
		The latter identity tells us that $\widetilde{M} = M$, and since
		\begin{equation*}
			(PV \circ \Phi) (p,\, M) = \Delta p + \frac{1}{2} \partial_3 {\min}_0\, (M-\partial_3 p)
		\end{equation*}
		the result follows from the uniqueness of a solution to nonlinear $PV$-and-$M$ inversion as proved in \cite{remond2024nonlinear}.
	\end{proof}

	Now that we have verified that the wave set and the balanced set are transverse,
	we are ready to obtain the decomposition that this entire section is building towards.

	\begin{theorem}[Pre-decomposition]
	\label{theorem:pre_decomp}
		We have the decomposition $ \mathbb{L}^2_\sigma = \mathcal{B} + \mathcal{W}$ in the sense that, for any $ \mathcal{X}\in \mathbb{L}^2_\sigma $ there exist unique
		$ \mathcal{X}_B \in \mathcal{B}$ and $ \mathcal{X}_W \in \mathcal{W}$,
		called the \emph{balanced component} of $ \mathcal{X}$ and the \emph{wave component} of $ \mathcal{X}$ respectively, such that
		\begin{enumerate}
			\item	$ \mathcal{X} = \mathcal{X}_B + \mathcal{X}_W$ and
			\item	$ PV( \mathcal{X}) = PV ( \mathcal{X}_B)$ and $\mathcal{M} ( \mathcal{X}) = \mathcal{M} ( \mathcal{X}_B)$
				for $\mathcal{M} (u,\,\theta,\,q) = \theta + q$ and $ PV(u,\,\theta,\,q) = \nabla_h^\perp \cdot u_h + \partial_3 \theta$.
		\end{enumerate}
		Moreover $ \mathcal{X}_B$ and $ \mathcal{X}_W$ may be computed explicitly as follows.
		\begin{itemize}
			\item	$ \mathcal{X}_B = \Phi (p,\, M)$ for $\Phi$ as in \fref{Proposition}{prop:alt_char_balanced_set}, where $M \vcentcolon= \theta + q$
				and where $p \in \mathring{H}^1 $ is the unique solution of
				\begin{equation}
					\Delta p + \frac{1}{2} \partial_3 {\min}_0\, (M - \partial_3 p) = PV.
				\label{eq:PV_M_inversion_pre_decomp_statement}
				\end{equation}
			\item	$\mathcal{X}_W = \mathcal{X} - \mathcal{X}_B$, and,
				in light of \fref{Proposition}{prop:chara_im_N}, we also have $ \mathcal{X}_W = \Psi (\wavecoord,\,w)$ where $w = u_3$ and
				\begin{align*}
					\wavecoord
					&= - ( u_h^\perp - u_{B,\,h}^\perp ) + (\theta - \theta_B) e_3\\
					&= - ( u_h^\perp - \nabla_h p) + (\theta - \partial_3 p - \frac{1}{2} {\min}_0\, (M-\partial_3 p)) e_3.
				\end{align*}
				In particular $\wavecoord$ satisfies
				\begin{equation*}
					\nabla\times \left( A_B^{-1} \wavecoord \right) = \partial_3 u_h - \nabla_h^\perp (\theta - H_B q) + (\partial_3 w) e_3,
				\end{equation*}
				where $H_B \vcentcolon= \mathds{1} (q_B < 0) = \mathds{1} (M < \partial_3 p)$ and $A_B \vcentcolon= I - \frac{1}{2} H_B e_3\otimes e_3$,
				such that $A_B^{-1} = I + H_B e_3\otimes e_3$, subject to $\nabla\cdot\wavecoord = 0$ and $\fint A_B^{-1} \wavecoord = - \fint u_h^\perp + \fint (\theta - H_B q) e_3$.
		\end{itemize}
	\end{theorem}
	\begin{proof}
		The uniqueness of this decomposition follows from \fref{Lemma}{lemma:PV_and_M_charac_balanced_states} and the fact that, as per \fref{Proposition}{prop:chara_im_N},
		elements of the wave set $\mathcal{W}$ have vanishing $PV$ and $M$.
		The existence of the decomposition follows from the existence of solutions to nonlinear $PV$-and-$M$ inversion since finding $ \mathcal{X}_B$ is,
		by virtue of \fref{Proposition}{prop:alt_char_balanced_set}, equivalent to finding $p$ and $M$ and since $p$ may be chosen to be the solution of nonlinear $PV$-and-$M$ inversion.

		The latter characterisation of $\wavecoord$ follows from \fref{Lemmas}{lemma:div_and_curl_of_good_unknown} and \ref{lemma:identity_lin_buoyancy_and_xi}:
		\begin{equation*}
			\nabla\times \left( A_B^{-1} \wavecoord \right)
			= \nabla\times \left[ -u_h^\perp + (\theta-H_Bq) e_3 \right] + \nabla\times \nabla p
			= \partial_3 u_h - \nabla_h^\perp (\theta - H_B q) + (\partial_3 w) e_3
		\end{equation*}
		and
		\begin{equation*}
			\fint A_B^{-1} \wavecoord = - \fint u_h^\perp + \fint (\theta - H_B q) e_3.
			\qedhere
		\end{equation*}
	\end{proof}

It is interesting to see 
$\nabla\times \left( A_B^{-1} \wavecoord \right)$ 
arise in the pre-decomposition above, where matrix $A_B^{-1}$
depends on the balanced moisture, $q_B$. 
It would be difficult to guess the form of 
$\nabla\times \left( A_B^{-1} \wavecoord \right)$,
with its coefficient matrix that depends on balanced moisture, $q_B$.
Here, the appearance of $A_B^{-1}$ follows from the
decomposition into nullspace and image,
and from seeking a curl-like operator for determining the
unbalanced coordinate $\wavecoord$.
Also note that $\nabla\times \left( A_B^{-1} \wavecoord \right)$
involves a linear differential operator, once $q_B$ is known,
but it is a nonlinear transformation from 
$(u,\,\theta,\,q)$ to $\wavecoord$.

	As an immediate corollary of this decomposition we make the following observation:
	since both the balanced set and the wave set could be parameterised
    and since any state can be decomposed into a balanced and a wave piece,
    this means that we have obtained a global parametrisation of the entire state space.

	\begin{cor}
	\label{cor:decomposition}
		Any state $\mathcal{X}\in \mathbb{L}^2_\sigma $ can be uniquely parameterised by
		\begin{equation*}
			\left\{
			\begin{aligned}
				u &= \nabla_h^\perp p + \wavecoord_h^\perp + we_3,\\
				\theta &= \partial_3 p + \frac{1}{2} {\min}_0\, (M - \partial_3 p) + \wavecoord_3, \text{ and } \\
				q &= M - \partial_3 p - \frac{1}{2} {\min}_0\, (M - \partial_3 p) - \wavecoord_3
			\end{aligned}
			\right.
		\end{equation*}
		for $p\in \mathring{H}^1 $, $M,\,w\in L^2$, and $\wavecoord\in { \left( L^2 \right)}^3_\sigma$ satisfying $\nabla_h^\perp \cdot \wavecoord_h = \partial_3 w$.
	\end{cor}
	\begin{proof}
		\fref{Theorem}{theorem:pre_decomp} tells us that any $ \mathcal{X}\in \mathbb{L}^2_\sigma $ may be uniquely decomposed as
		\begin{equation*}
			\mathcal{X} = \mathcal{X}_B + \mathcal{X}_W \in \mathcal{B} + \mathcal{W}.
		\end{equation*}
		Then \fref{Propositions}{prop:alt_char_balanced_set} and \ref{prop:chara_im_N} tell us that
		$ \mathcal{X}_B = \Phi (p,\,M)$ and $ \mathcal{X}_W = \Psi(\wavecoord,\,w)$ for some uniquely determined $p$, $M$, $\wavecoord$, and $w$.
		So indeed
		\begin{equation*}
			\mathcal{X} = \Phi (p,\, M) + \Psi (\wavecoord,\, w)
		\end{equation*}
		as desired.
	\end{proof}

	We do not even have to stop there:
	both the balanced set and the wave set had two \emph{dual} descriptions.
	They were described using a parametrisation, which produced \fref{Corollary}{cor:decomposition} above.
	But they were also described using the measurements that annihilate them.
	This suggests that using the full set of measurements, namely both those characterising the balanced set and those characterising the wave set,
	would yield a complete description of the state space.
	We verify this in the result below.

	\begin{prop}[Global nonlinear change of coordinates]
	\label{prop:global_chg_coord}
		Consider the \emph{measurement space} 
		\begin{equation*}
			\mathfrak{M} \vcentcolon= H^{-1} \times L^2 \times { \left( H^{-1} \right) }^2 \times L^2 \times \mathbb{R}^3.
		\end{equation*}
		The map $\mathcal{M} : \mathbb{L}^2_\sigma \to \mathfrak{M} $ given by
		\begin{equation*}
			\mathcal{M} (u,\,\theta,\, q)
			= \begin{pmatrix}
				\nabla_h^\perp \cdot u_h + \partial_3 \theta\\
				\theta + q\\
				\partial_3 u_h - \nabla_h^\perp (\theta - {\min}_0\, q)\\
				u_3\\
				- \fint u_h^\perp + \fint (\theta - {\min}_0\, q) e_3
			\end{pmatrix}
			= \begin{pmatrix}
				PV \\ M \\ j \\ w \\ a
			\end{pmatrix}
		\end{equation*}
		is invertible.
	\end{prop}
	\begin{proof}
		To verify the invertibility of this change of coordinates we will use the coordinates $\pi$, $M$, $\altwavecoord$, and $w$ defined implicitly by
		\begin{equation}
			\left\{
			\begin{aligned}
				u &= \nabla_h^\perp \pi + \altwavecoord_h^\perp + we_3,\\
				\theta &= \partial_3 \pi + \altwavecoord_3 + \frac{1}{2} {\min}_0\, (M - \partial_3 \pi - \altwavecoord_3), \text{ and } \\
				q &= M - \partial_3 \pi - \altwavecoord_3 - \frac{1}{2} {\min}_0\, (M - \partial_3 \pi - \altwavecoord_3)
			\end{aligned}
			\right.
		\label{eq:global_chg_coord_param}
		\end{equation}
		where $\nabla_h^\perp \cdot \altwavecoord_h = \partial_3 w$ and $\nabla\cdot\altwavecoord = 0$.
		Note that the coordinates $(\pi,\,M,\,\altwavecoord,\,w)$ differ from the coordinates $(p,\,M,\,\wavecoord,\,w)$ used in \fref{Corollary}{cor:decomposition}.
		The key advantage of the coordinates used here is that
		\begin{equation}
			\sign q = \sign (M - \partial_3 \pi - \altwavecoord_3),
		\label{eq:global_chg_coord_sign_q}
		\end{equation}
		which is why they are particularly useful here when it comes to inverting $\mathcal{M}$.

		However these coordinates have the less advantageous property that $\pi$ depends on $PV$, $M$, \emph{and} $\altwavecoord$, as we will see below.
		This is different from $p$ (from \fref{Corollary}{cor:decomposition}), which only depends on $PV$ and $M$.
		That is the reason why these $\pi$ and $\altwavecoord$ coordinates are \emph{not} used for the decomposition (instead of $p$ and $\wavecoord$, respectively):
		the resulting balanced component computed from $\pi$ and $M$ would \emph{not} be slow.

		We close this parenthesis regarding the choice of coordinates and go back to the matter of invertibility.
		In light of \eqref{eq:global_chg_coord_sign_q}, \fref{Lemma}{lemma:invert_b_and_M} and \fref{Lemma}{lemma:identity_lin_buoyancy_and_xi}
		tell us that \eqref{eq:global_chg_coord_param} is equivalent to
		\begin{equation}
			\left\{
			\begin{aligned}
				&- u_h^\perp + (\theta - {\min}_0\, q) e_3 = \nabla\pi + \altwavecoord,\\
				&u_3 = w, \text{ and } \\
				&\theta + q = M.
			\end{aligned}
			\right.
		\label{eq:global_chg_coord_param_equiv}
		\end{equation}
		In particular it follows from the uniqueness of the Helmholtz decomposition that these coordinates are uniquely determined by $u$, $\theta$, and $q$.

		So now it only remains to show that, for any tuple $(PV,\,M,\,j,\,w,\,a)$, corresponding coordinates $(\pi,\,M,\,\altwavecoord,\,w)$ exist.
		Plugging \eqref{eq:global_chg_coord_param} into the definition of the measurement map $ \mathcal{M} $ we see, using \eqref{eq:global_chg_coord_param_equiv}, that
		\begin{equation*}
			\left\{
			\begin{aligned}
				PV &= \nabla\pi + \frac{1}{2} \partial_3 {\min}_0\, (M - \partial_3 \pi - \altwavecoord_3),\\
				M &= M,\\
				j &= \partial_3 \altwavecoord_h^\perp - \nabla_h^\perp \altwavecoord_3,\\
				w &= w, \text{ and } \\
				a &= \fint \altwavecoord.
			\end{aligned}
			\right.
		\end{equation*}
		Since $\nabla_h^\perp\cdot\altwavecoord_h = \partial_3 w$ we see using \fref{Lemma}{lemma:horizontal_vertical_decomposition_of_the_curl} that
		\begin{equation*}
			\nabla\times\altwavecoord = j + (\partial_3 w) e_3.
		\end{equation*}
		We are thus guaranteed a (unique) solution $\altwavecoord$ of
		\begin{equation}
			\left\{
			\begin{aligned}
				\nabla\times\altwavecoord &= j + (\partial_3 w)e_3,\\
				\nabla\cdot\altwavecoord &= 0, \text{ and } \\
				\fint \altwavecoord &= a.
			\end{aligned}
			\right.
		\label{eq:global_chg_coord_inv_1}
		\end{equation}
		The solvability of nonlinear $PV$-and-$M$ inversion \cite{remond2024nonlinear} then guarantees a unique solution $\pi$ in $ \mathring{H}^1 $ of
		\begin{equation}
			\Delta\pi + \frac{1}{2} \partial_3 {\min}_0\, \left[ \left( M - \altwavecoord_3 \right) - \partial_3 \pi \right] = PV.
		\label{eq:global_chg_coord_inv_2}
		\end{equation}

		In other words, if we define
		\begin{equation*}
			\widetilde{\mathfrak{C}} \vcentcolon= \left\{
				(\pi,\,M,\,\altwavecoord,\,w) \in \mathring{H}^1 \times L^2 \times { \left( L^2 \right) }_\sigma^2 \times L^2 :
				\nabla_h^\perp\cdot\altwavecoord_h = \partial_3 w
			\right\}
		\end{equation*}
		to be the space where the coordinates $(\pi,\,M,\,\altwavecoord,\,w)$ reside
		and denote by $(u,\,\theta,\,q) = \mathcal{C} (\pi,\,M,\,\altwavecoord,\,w)$ the map defined by \eqref{eq:global_chg_coord_param}, we have shown that
		for any $(PV,\,M,\,j,\,w,\,a)\in \mathfrak{M} $ there exists a unique $(\pi,\,M,\,\altwavecoord,\,w) \in \widetilde{\mathfrak{C}}$
		such that $( \mathcal{M} \circ \mathcal{C} ) (\pi,\,M,\,\altwavecoord,\,w) = (PV,\,M,\,j,\,w,\,a)$.
		This means that we can take $(u,\,\theta,\,q) = \mathcal{C} (\pi,\,M,\,\altwavecoord,\,w)$ to satisfy $\mathcal{M}(u,\,\theta,\,q) = (PV,\,M,\,j,\,w,\,a)$.
		Moreover, since we know that this state $(u,\,\theta,q)$ is uniquely determined by its coordinates $(\pi,\,M,\,\altwavecoord,\,w)$ this proves that there is a \emph{unique}
		state $(u,\,\theta,\,q)$ satisfying $\mathcal{M}(u,\,\theta,\,q) = (PV,\,M,\,j,\,w,\,a)$.
		This proves that $\mathcal{M}$ is invertible, as desired.
		Moreover the inverse may be computed by solving, in order, \eqref{eq:global_chg_coord_inv_1} then \eqref{eq:global_chg_coord_inv_2}
		and then plugging the solutions into \eqref{eq:global_chg_coord_param}.
	\end{proof}

	\begin{remark}[The global change of coordinates in the dry case]
	\label{rmk:global_chg_coord_dry_case}
		The proof of \fref{Proposition}{prop:global_chg_coord} above is a bit of a mouthful.
		In order to shed some light on the underlying ideas, without having to worry about the details of the proof itself,
		we discuss here what an analog global change of coordinates would look like in the dry case.

		In that case we could break down the change of coordinates as a composition of invertible pieces,
		thus illustrating more clearly what the moving pieces are.
		Indeed, in the dry case we have the following chain of inversions:
		\begin{align*}
			\begin{pmatrix}
				u \\ \theta
			\end{pmatrix} \leftrightarrow \begin{pmatrix}
				u_h \\ u_3 \\ \theta
			\end{pmatrix} \leftrightarrow \begin{pmatrix}
				-u_h^\perp + \theta e_3 \\ u_3
			\end{pmatrix} =\vcentcolon \begin{pmatrix}
				g \\ w
			\end{pmatrix} \overset{(1)}{\leftrightarrow} \begin{pmatrix}
				\nabla\cdot g \\ \nabla\times g \\ \fint g \\ w
			\end{pmatrix}
            &\overset{(2)}{=} \begin{pmatrix}
				\nabla_h\cdot u_h + \partial_3 \theta \\ \partial_3 u_h - \nabla_h^\perp \theta + (\partial_3 w) e_3 \\ - \fint u_h^\perp + \fint \theta e_3 \\ w
			\end{pmatrix}
        \\
            &=\vcentcolon \begin{pmatrix}
				PV \\ j + (\partial_3 w) e_3 \\ a \\ w
			\end{pmatrix} \leftrightarrow \begin{pmatrix}
				PV \\ j \\ a \\ w
			\end{pmatrix}
		\end{align*}
		where (1) follows from the Helmholtz decomposition (recorded in \fref{Theorem}{theorem:Helmholtz}) and (2) follows from \fref{Lemma}{lemma:div_and_curl_of_good_unknown} 
		and the fact that the velocity field $u$ is divergence-free.

		It is therefore tempting to assume that a similar chain of invertible maps produces the global nonlinear change of coordinates for the \emph{moist} case recorded
		in \fref{Proposition}{prop:global_chg_coord}.
		This is not the case.
		Of course, we would need to perform the following inversions initially in the moist case, by virtue of \fref{Lemma}{lemma:invert_b_and_M}:
		\begin{equation*}
			\begin{pmatrix}
				u \\ \theta \\ q
			\end{pmatrix} \leftrightarrow \begin{pmatrix}
				u_h \\ u_3 \\ \theta + q \\ \theta - {\min}_0\, q
			\end{pmatrix} =\vcentcolon \begin{pmatrix}
				u_h \\ w \\ M \\ b
			\end{pmatrix}.
		\end{equation*}
		This is not the source of the issue.
		The source of the issue is that, in the moist case, the potential vorticity $PV$ and the thermal wind imbalance $j$ come
		from \emph{different} good unknowns.
		Indeed:
		\begin{equation*}
			PV = \nabla\cdot \left( -u_h^\perp + \theta e_3 \right)
			\text{ while }
			j = { \left[ \nabla\times \left( -u_h^\perp + \left( \theta - {\min}_0\, q \right) e_3 \right) \right] }_h.
		\end{equation*}
		The good unknowns
		\begin{equation*}
			-u_h^\perp + \theta e_3
			\text{ and } 
			-u_h^\perp + \left( \theta - {\min}_0\, q \right) e_3
		\end{equation*}
		are not the same.
		Note that this is \emph{not} a technical issue.
		This comes from the dynamics in the moist case being fundamentally different.
		In particular, the operator in the moist case loses the skew-symmetry present in the dry case,
		and this is precisely the source of the discrepancy between the good unknowns used for $PV$ and $j$.
		Indeed, the $(\theta,\,q)$--dependence in the equation governing $ \partial_t u$ is of the form
		\begin{equation*}
			\partial_t u = \dots + (\theta- {\min}_0\, q) e_3
			= \dots + (\theta -q H)e_3
		\end{equation*}
		for $H \vcentcolon= \mathds{1} (q<0)$
		while the $u$--dependence in the equation governing $ \partial_t \theta$ and $\partial_t q$ is of the form
		\begin{align*}
			&\partial_t \theta = \dots - u_3 \text{ and } \\
			&\partial_t q = \dots + u_3.
		\end{align*}
		In matrix form this reads
		\begin{equation*}
			\partial_t \begin{pmatrix}
				u \\ \theta \\ q
			\end{pmatrix} = \begin{pmatrix}
				\dots	& e_3	& -He_3	\\
				-e_3^T	& \dots	& \dots	\\
				e_3^T	& \dots	& \dots
			\end{pmatrix} \begin{pmatrix}
				u \\ \theta \\ q
			\end{pmatrix}.
		\end{equation*}
		This means that the good unknown for $PV$ comes from the first column, while the good unknown for $j$ comes from the first row.
		Said another way: $PV$ comes from the \emph{image} of the operator, whereas $j$ comes from the \emph{null set} of the operator.
		Since that matrix is \emph{not} skew-symmetric, i.e. since the operator is \emph{not} skew-adjoint,
		it follows that $PV$ and $j$ come from \emph{different} good unknowns.
	\end{remark}

    \begin{remark}  
    \label{rmk:three_descriptions}
    	To conclude this section, we note that we now have \emph{three} descriptions of the state space available to us:
    	\begin{itemize}
    		\item	the ``primitive'' description in terms of the unknowns $u$, $\theta$, and $q$,
    		\item	the description in terms of the coordinates $p$, $M$, $\wavecoord$, and $w$ (see \fref{Corollary}{cor:decomposition} ), and
    		\item	the description in terms of the measurements $PV$, $M$, $j$, $a$, and $w$ (see \fref{Proposition}{prop:global_chg_coord}).
    	\end{itemize}
    	Crucially: the latter two descriptions are compatible with the decomposition!
    	We can summarize this with the following diagram, where we give a name to the maps converting between each of these descriptions.
    
        \begin{center}
    	\begin{tikzcd}
    		&
    			(u,\,\theta,\,q) \in \mathbb{L}^2_\sigma
    			\arrow[start anchor = east, dr, "\mathcal{M}"]
    		\\
    			(p,\,M,\,\wavecoord,\,w) \in \mathfrak{C}
    			\arrow[ur, end anchor = west, "\mathcal{S}"]
    		&&	(PV,\,M,\,j,\,a,\,w) \in \mathfrak{M}
    			\arrow[ll, "\mathcal{C}"]
    	\end{tikzcd}
        \end{center}
    
        Recall that $ \mathbb{L}^2_\sigma $ is the \emph{state space} introduced in \fref{Definition}{def:state_space} and
        that $\mathfrak{M}$ is the \emph{measurement space} defined in \fref{Proposition}{prop:global_chg_coord}. Here $\mathfrak{C}$ is the \emph{coordinate space} defined as follows:
        \begin{equation*}
        \label{eq:def_mathfrak_C}
    			\mathfrak{C} = \left\{ (p,\,M,\,\wavecoord,\,w) \in \mathring{H}^1 \times L^2 \times { \left( L^2 \right) }_\sigma^2 \times L^2 : \nabla_h^\perp\cdot\wavecoord_h = \partial_3 w \right\}
        \end{equation*}
        The state-to-measurement map $\mathcal{M}$ is then as defined in \fref{Proposition}{prop:global_chg_coord}, where it is shown to be invertible, the coordinate-to-state map $\mathcal{S}$ corresponds to the parametrisation recorded in \fref{Corollary}{cor:decomposition}, which is invertible since \fref{Theorem}{theorem:pre_decomp} tells us how to compute its inverse, and the measurement-to-coordinate map $\mathcal{C}$ is well-defined by virtue of the invertibility of the other two maps.
    \end{remark}

	\begin{remark}[Vertically Sheared Horizontal Flows]
	\label{rmk:VSHF}
		This remark discusses vertically sheared horizontal flows, which are particular solutions to the Boussinesq system characterized by vanishing thermodynamic variables
		(namely $\theta$ and $q$ in the moist case) and by a purely $z$--dependent and horizontal velocity $u = u_h(z)$.

		The purpose of this remark is two-fold.
		First it shows what the global change of coordinates of \fref{Corollary}{cor:decomposition} and \fref{Proposition}{prop:global_chg_coord} looks like in a specific case,
		namely when it is applied to these vertically sheared horizontal flow.
		Second it shows one benefit of the decomposition presented here: these flows are accounted for \emph{naturally}.

		First, we recall that combining \fref{Corollary}{cor:decomposition} and \fref{Proposition}{prop:global_chg_coord} gives us \emph{three} descriptions of any given state.
		\begin{itemize}
			\item	The description in physical space in terms of $(u,\,\theta,\,q)$.
			\item	The description in coordinate space in terms of $(p,\,M,\,\wavecoord,\,w)$.
			\item	The description in measurement space in terms of $(PV,M,\,j,\,w,\,a)$.
		\end{itemize}
		Let us now see what vertically sheared horizontal flows look like in each of these three descriptions.
		\begin{itemize}
			\item	In physical space a vertically sheared horizontal flow corresponds to any state $(u,\,\theta,\,q)$
				where $\theta$ and $q$ vanish and where $u=u_h(z)$, i.e. $u_3 = 0$ and $u$ only depends on $z$.
			\item	In coordinate space a vertically sheared horizontal flow corresponds to any state described by $(p,\,M,\,\wavecoord,\,w)$
				where $p$, $M$, and $w$ vanish and where $\wavecoord = \wavecoord_h (z)$.
			\item	In measurement space a vertically sheared horizontal flow corresponds to any any state described by $(PV,\,M,\,j,\,w,\,a)$
				where $PV$, $M$, and $w$ vanish, where $j = j(z)$, and where $a = (a_1,\,a_2,\,0)$ is any purely horizontal vector in $ \mathbb{R}^3$.
		\end{itemize}
		In other words, since a vertically sheared horizontal flows is fully characterized by a map $\gamma : \mathbb{T}^1 \to \mathbb{R}^2$, here are three descriptions
		of the same such flow.
		\begin{align*}
			&\text{Physical space}	&&\text{Coordinate space}	&&\text{Measurement space}		\\
			&u = \gamma_h (z)	&&p = 0				&&PV = 0				\\
			&\theta = 0		&&M = 0				&&M = 0					\\
			&q = 0			&&\wavecoord = -\gamma^\perp_h (z)	&&j = \gamma'(z)			\\
			&			&&w = 0				&&w = 0					\\
			&			&&				&&a = -\fint \gamma_h^\perp (z) dz	\\
		\end{align*}

		Second, we note that as described above our decomposition accounts for these vertically sheared horizontal flows very naturally.
		This comes from the fact that our decomposition leverages a \emph{three-dimensional} Helmholtz decomposition of the \emph{mixed vector} fields $-u^\perp_h + \theta e_3$ (for $PV$)
		and $-u_h^\perp + (\theta - {\min}_0\, q)e_3$ (for $j$).
		This is different from what is often done where a \emph{two-dimensional} Helmholtz decomposition is performed \emph{solely on} the horizontal velocity $u_h$.
		This two-dimensional Helmholtz decomposition then leaves a two-dimensional average term of the form $u_h(z)$ to be determined, 
		which is precisely the term accounting for these vertically sheared horizontal flows.
		In other words: our decomposition accounts for these intrinsically, instead of having to treat them as a separate add-on after the fact.
	\end{remark}

	\begin{remark}
	\label{rmk:decomposition_is_always_valid}
		It is important to note that the decomposition discussed in this section is \emph{always} valid,
    	even when the state in $ \mathbb{L}^2_\sigma $ thus decomposed is \emph{not} a solution of the moist Boussinesq dynamical system.
        In other words: the decomposition of \fref{Theorem}{theorem:pre_decomp} is valid for any state
        in the state space $\mathbb{L}^2_\sigma$, similarly to how the Helmholtz decomposition is valid for any three-dimensional vector field.

        However, as discussed in the next section, the decomposition we introduced in this context does also have additional meaning in the context of the \emph{dynamics} of the moist Boussinesq system \eqref{eq:moist_Boussinesq_u}--\eqref{eq:moist_Boussinesq_q}.
  \end{remark}

\section{Decomposition and slow-fast dynamics}
\label{sec:slow_fast_decomp}

	In \fref{Section}{sec:decomposition} above we have introduced a preliminary decomposition of the leading-order dynamics.
	In particular we have discussed how to then make sure that the balanced component is slow.
    The key step in doing so was ensuring that the pressure, which acts as a streamfunction for the balanced component of the velocity, be slow.
    This differs from previous work \cite{zhang_smith_stechmann_20_asymptotics,Zhang_Smith_Stechmann_2021_JFM,wetzel2019balanced,wetzel2020potential} where the 
    streamfunction constructed was approximately slow but not precisely slow.
 
	A question remains: ``did we extract \emph{all} of the slow components of a state in this manner?''.
	In the dry case, that question may be answered in the affirmative since the fast dynamics are that of linear oscillations which do not feedback into the slow component.
	In the moist case, however, the fast oscillations are nonlinear and so they \emph{may} feed back into the slow component.
	To make this idea precise we appeal to the framework of fast-wave averaging.
	Fast-wave averaging is a technique used to identify the limiting behaviour of certain dynamical systems which admit two timescales.
	The prototypical case in which two timescales would be present is the following dynamical system in $\mathbb{R}^{n_1 + n_2}$:
	\begin{equation*}
		\partial_t \begin{pmatrix}
			x_1 \\ x_2
		\end{pmatrix} = \begin{pmatrix}
			f_1 (x_1) \\ \frac{1}{\varepsilon} f_2 (x_2)
		\end{pmatrix}
	\end{equation*}
	for $(x_1,\, x_2) \in \mathbb{R}^{n_1 + n_2}$, where $\varepsilon > 0$ is a small parameter. 
	In this case we would say that the dynamics of $x_1$ are slow since they are of order one, whereas the dynamics of $x_2$ are fast since they are of order $1 / \varepsilon$.
	In general, dynamical systems will not admit an \emph{exact} two timescale decomposition as in the prototypical example above.
	In general, that decomposition is only approximately valid for finite $\varepsilon$ and we are then interested in the limiting behaviour of the slow and fast components as $\varepsilon \to 0$.

	Where does fast-wave averaging fit in all this? Well: possessing two timescales is a feature a dynamical system may or may not manifest,
	whereas fast-wave averaging is a technique used to identify the limiting behaviour of the slow and fast timescales (for such dynamical systems).
	Fast-wave averaging happens in four steps.
	\begin{enumerate}
		\item	Assume a two-timescale asymptotic expansion of the unknowns (which requires a sub-linear growth condition).
		\item	Plug this expansion into the equations of motion to identify the resulting equations at leading-order $1/\varepsilon$ and at next-order.
		\item	Solve the leading-order equations over the fast timescale in terms of a yet-to-be-determined ``background profile'' at the slow timescale.
		\item	Determine this background profile at the slow timescale
			by using a solvability condition for the next-order equations.
	\end{enumerate}

	We will now run through these four steps for the moist Boussinesq system \eqref{eq:moist_Boussinesq_u}--\eqref{eq:moist_Boussinesq_q}.
	At each step we will highlight the differences with the dry case. This is because the dry case is well-understood
	however in the moist case we only provide here a sketch of the fast-wave averaging process, leaving a rigorous treatment for future work.

	First we assume a two-timescale expansion of the unknowns. This is done exactly as in the dry case and we assume that, for $\mathcal{X} = (u,\,\theta,\,q)$,
	\begin{equation}
		\mathcal{X} = \mathcal{X}^0 (t,\,\tau) \vert_{\tau = t/\varepsilon} + \varepsilon \mathcal{X}^1 (t,\,\tau) \vert_{\tau = t/\varepsilon} + O(\varepsilon^2)
	\label{eq:two_timescale_assumption}
	\end{equation}
	for some $\mathcal{X}^0$ and $\mathcal{X}^1$ to be determined where $t$ is the slow timescale and $\tau = t/\varepsilon$ is the fast timescale,
	subject to the \emph{sub-linear growth condition}
	\begin{equation}
		\mathcal{X}^1 (t,\,\tau) = o(\tau) \text{ in $L^2$ uniformly in $t$ as } \tau \to \infty,
	\label{eq:sublinear_growth_condition}
	\end{equation}
	or equivalently for every $\delta > 0$ there exists $\tau_* > 0$ such that if $\tau > \tau_*$ then
	\begin{equation*}
		\frac{ \norm{\mathcal{X}^1 (t,\,\tau) }{L^2} }{\tau} < \delta
		\text{ for all } t.
	\end{equation*}
	Note that this sub-linearity condition is imposed for the following reason.
	If the next-order term $\mathcal{X}^1$ is allowed to grow linearly in $\tau$ then we would have that
	\begin{equation*}
		\varepsilon \mathcal{X}^1 (t,\,\tau) \vert_{\tau = t/\varepsilon}
		\sim \varepsilon \left( \mathcal{X}^1 (t,\,0) + \frac{t}{\varepsilon} \right)
		= \varepsilon \mathcal{X}^1 (t,\,0) + t,
	\end{equation*}
	i.e. the term $t$ of order one would appear in an expression that ought to be of order $\varepsilon$.
	
	The next step is to insert the two-timescale assumption \eqref{eq:two_timescale_assumption} into the moist Boussinesq system \eqref{eq:moist_Boussinesq_u}--\eqref{eq:moist_Boussinesq_q}
	and read off the leading-order and next-order equations that come out.
	In order to carry this out we record here two preliminary computations.
	First we note that
	\begin{equation*}
		\partial_t \left[ f(t,\,\tau) \middle\vert_{\tau = t/\varepsilon} \right]
		= \left[ (\partial_t f)(t,\,\tau) + \frac{1}{\varepsilon} (\partial_\tau f) (t,\,\tau) \middle] \right\vert_{\tau = t/\varepsilon}.
	\end{equation*}
	Second we note that, introducing the notation $H \vcentcolon= \mathds{1} (q^0 + \varepsilon q^1 < 0)$ and $H^0 \vcentcolon= \mathds{1} (q^0 < 0)$,
	\begin{equation*}
		{\min}_0\, (q_0 + \varepsilon q^1)
		= (q^0 + \varepsilon q^1) H
		= \underbrace{q^0 H^0}_{ {\min}_0\, q^0 } + q^0 (H - H^0) + \varepsilon q^1 H
	\end{equation*}
	where
	\begin{align*}
		q^0 (H - H^0)
		&= q^0 \left[ \mathds{1}(q^0 + \varepsilon q^1 < 0) - \mathds{1} (q^0 < 0) \right]\\
		&= q^0 \left[ \mathds{1}(q^0 + \varepsilon q^1 < 0,\, q^0 \geqslant 0) - \mathds{1} (q^0 < 0,\, q^0 + \varepsilon q^1 \geqslant 0 \right]\\
		&= q^0 \left[ \mathds{1}(0\leqslant q^0 < -\varepsilon q^1) - \mathds{1} (-\varepsilon q^1 \leqslant q^0 < 0) \right]\\
		&= \varepsilon \cdot \frac{q^0}{\varepsilon} \left[
			\mathds{1} \left( 0\leqslant \frac{q^0}{\varepsilon} < - q^1 \right)
			- \mathds{1} \left( -q^1 \leqslant \frac{q^0}{\varepsilon} < 0 \right)
		\right]\\
		& =\vcentcolon \varepsilon R(q^0 ;\, q^1).
	\end{align*}
	In particular we note that $R(q^0;\,q^1) \in O(1)$ provided that $q^1$ is bounded.
 
	Using these two observations we plug \eqref{eq:two_timescale_assumption} into \eqref{eq:moist_Boussinesq_u}--\eqref{eq:moist_Boussinesq_q} and read off the leading-order equations,
	at order $1/\varepsilon$, to be
	\begin{subnumcases}{}
		\partial_\tau u^0 + e_3\times u^0 + \nabla p^0 - \theta^0 e_3 + {\min}_0\, q^0 e_3 = 0,\\
		\partial_\tau \theta^0 + u_3^0 = 0, \text{ and } \\
		\partial_\tau q^0 - u_3^0 = 0
	\end{subnumcases}
	while the next-order equations, at order one (i.e. order $\varepsilon^0$) to be
	\begin{subnumcases}{}
		\partial_\tau u^1 + e_3\times u^1 + \nabla p^1 - \theta^1 e_3 + q^1 H e_3 = -\partial_t u^0 - u^0\cdot\nabla u^0 + R(q^0;\,q^1),\\
		\partial_\tau \theta^1 + u_3^1 = - \partial_t \theta^0 - u^0\cdot\nabla \theta^0, \text{ and } \\
		\partial_\tau q^1 - u_3^1 = -\partial_t q^0 - u^0 \cdot \nabla q^0.
	\end{subnumcases}
	We can also write these equations more compactly as
	\begin{subnumcases}{}
		\partial_\tau \mathcal{X}^0 + \mathcal{N} (\mathcal{X}^0) \text{ and }
		\label{eq:moist_Boussinesq_lead_order_abstract}\\
		\partial_\tau \mathcal{X}^1 + \mathcal{L}_H (\mathcal{X}^1) = - \partial_t \mathcal{X}^0 - u^0\cdot\nabla \mathcal{X}^0 + \mathcal{R} ( \mathcal{X}^0;\, \mathcal{X}^1),
		\label{eq:moist_Boussinesq_next_order_abstract}
	\end{subnumcases}
	for $\mathcal{N}$ as defined in \eqref{eq:def_N}, $\mathcal{L}_H$ its formal linearisation given by
	\begin{equation}
		\mathcal{L}_H \begin{pmatrix}
			u \\ \theta \\ q
		\end{pmatrix} = \begin{pmatrix}
			\mathbb{P}_L \left[ e_3 \times u - (\theta - qH)e_3 \right] \\ u_3 \\ -u_3
		\end{pmatrix},
	\label{eq:def_L_H}
	\end{equation}
	and the remainder term $\mathcal{R}$ given by $\mathcal{R} (u^0,\,\theta^0,\,q^0;\, u^1,\,\theta^1,\,q^1) = (R(q^0;\,q^1),\,0,\,0)$.

	Crucially: we see that the dynamics at leading-order are governed by $\partial_\tau + \mathcal{N} = 0$ for $\mathcal{N}$ as introduced in \eqref{eq:def_N}.
	The sets $\mathcal{B}$ and $\mathcal{W}$ characterised in \fref{Propositions}{prop:alt_char_balanced_set} and \ref{prop:chara_im_N} thus play an essential role since they come from that same operator $\mathcal{N}$.
	Indeed, as discussed in \fref{Proposition}{prop:dyn_pdt_plus_N}, a solution $\mathcal{X}$ of $\partial_\tau \mathcal{X} + \mathcal{N} ( \mathcal{X} ) = 0$ may be written as
	\begin{equation}
		\mathcal{X} (\tau) = \mathcal{X}_B (0) + e^{-\tau \mathcal{L}_{W,\,H_B}} \mathcal{X}_W (0).
	\label{eq:sol_pdtau_plus_N}
	\end{equation}
	Note that here we abused notation a little bit in order to be reminiscent of what happens in the dry case.
	In the dry case the leading order dynamics are $\partial_\tau + \mathcal{L}$ for $\mathcal{L}$ given by
    \begin{equation}
    \label{eq:def_L}
        \mathcal{L} (u,\,\theta) = \begin{pmatrix}
            \mathbb{P}_L \left( u_h^\perp -\theta e_3 \right) \\ u_3
        \end{pmatrix},
    \end{equation}
    as per \eqref{eq:dry_Boussinesq_u}--\eqref{eq:dry_Boussinesq_theta}.
	Since the operator $\mathcal{L}$ is linear and constant, it follows that the solution map of $\partial_\tau + \mathcal{L}$ is precisely
	given by the semigroup $e^{-t\mathcal{L}}$. The operator $\mathcal{L}_{W,\,H_B}$ in the moist case is linear but with \emph{variable} coefficients (namely through its dependence on $H_B$),
	yet we use the same notation to represent the solution map,
	keeping it mind that it is slightly abusive notation since the solution map is \emph{not} a semigroup.

	We are thus ready to carry out the third step of the fast-wave averaging process.
	We deduce from \eqref{eq:sol_pdtau_plus_N} that the solution of the leading-order equations \eqref{eq:moist_Boussinesq_lead_order_abstract} is given by
	\begin{equation*}
		\underline{\mathcal{X}}^0 (t,\,\tau) = \underline{\mathcal{X}}^0_B (t) + e^{- \tau \mathcal{L}_{W,\,H_B}} \underline{\mathcal{X}}_W^0 (t)
	\end{equation*}
	where the background slow profile $\underline{\mathcal{X}}^0 (t) = \underline{\mathcal{X}}_B^0 (t) + \underline{\mathcal{X}}_W^0 (t)$ is to be determined
	(this is done in the next step of the fast-wave averaging process).
	This is where the first essential difference with the dry case manifests itself.
	In the dry case, the leading-order solution takes the form
	\begin{equation}
		\underline{\mathcal{X}}^0 (t,\,\tau) = \underline{ \mathcal{X} }^0_B (t) + e^{-\tau\mathcal{L}_W} \underline{ \mathcal{X} }^0_W (t),
	\label{eq:lead_order_sol_in_dry_case}
	\end{equation}
	where $\mathcal{L}_W$ denotes the restriction of $\mathcal{L}$ from \eqref{eq:def_L} onto its image.
	Crucially: since the operator $\mathcal{L}_W$ is skew-adjoint in $L^2$ the oscillations in \eqref{eq:lead_order_sol_in_dry_case} are linear oscillations with constant frequencies.
	Here, in the moist case, \emph{the frequencies depend on the phase}.

	To see why this matters, suppose for now that we have completed the fourth and last step of the fast-wave averaging process.
	That is, suppose we have determined the dynamics of the slow part $\underline{\mathcal{X}}^0 (t)$.
	Since we are interested in the limiting behaviour of solutions, we then wish to study the limiting behaviour of $\mathcal{X}^0 (t,\,\tau)$ as $\varepsilon\to 0$.
	In the \emph{dry} case this means studying the limiting behaviour of
	\begin{equation*}
		e^{-\tau\mathcal{L}_W} \underline{\mathcal{X}}_W^0 (t).
	\end{equation*}
	Crucially: since the dry oscillations have constant frequencies, these oscillations have vanishing time averages,
	and so the weak limit of this term is simply zero!

	In the moist case we must by contrast study the limiting behaviour of
	\begin{equation*}
		e^{-\tau\mathcal{L}_{W,\,H_B}} \underline{\mathcal{X}}_W^0 (t).
	\end{equation*}
	But now, as shown in an example first recorded in \cite{Zhang_Smith_Stechmann_2021_JFM}
	(see \fref{Corollary}{cor:moist_Bouss_has_wave_sols_with_nonvanish_averages}, and the simple example it arises from in \fref{Proposition}{prop:second_order_ODE_with_nonzero_average_solutions}),
	phase-dependent oscillations do \emph{not} have vanishing time averages.
	This wave term therefore converges weakly to a \emph{non-zero} time average!

	This tells us exactly how to augment our ``pre-decomposition'' above, in \fref{Theorem}{theorem:pre_decomp}, to now make sure we include \emph{all} of the slow components:
	we take the slow balanced piece from the previous section and add to it the time-average of the wave piece.

	This produces the following slow-fast decomposition, where for any space $S$, $L^1 S$ is used to denote $L^1 ([0,\,\infty);\, S)$.
	\begin{theorem}[Slow-fast decomposition]
	\label{theorem:slow_fast_decomposition}
		We have the decomposition
		\begin{equation*}
			L^1 \mathbb{L}^2_\sigma = L^1 \mathcal{B} + \overline{\mathcal{W}} + \widetilde{\mathcal{W}}
		\end{equation*}
		where
		\begin{equation*}
			\overline{\mathcal{W}} \vcentcolon= \left\{ \mathcal{X} \in L^1 \mathcal{W} : \mathcal{X} \text{ is independent of time $\tau$} \right\}
			\text{ and } 
			\widetilde{\mathcal{W}} \vcentcolon= \left\{ \mathcal{X}\in L^1 \mathcal{W} : \lim_{\tau\to\infty} \fint_0^\tau \mathcal{X} (\tau') d\tau' = 0 \right\},
		\end{equation*}
		in the sense that, for every $\mathcal{X}\in L^1\mathbb{L}^2_\sigma $ there exist unique $\mathcal{X}_B \in L^1\mathcal{B}$, $\overline{\mathcal{X}}_W \in \overline{\mathcal{W}}$,
		and $\widetilde{\mathcal{X}}_W \in \widetilde{\mathcal{W}}$ such that
		\begin{enumerate}
			\item	$ \mathcal{X} = \mathcal{X}_B + \overline{ \mathcal{X} }_W + \widetilde{ \mathcal{X} }_W$ and
			\item	$PV ( \mathcal{X}) = PV	( \mathcal{X}_B)$ and $\mathcal{M} ( \mathcal{X}) = \mathcal{M} ( \mathcal{X}_B)$ for $PV$ and $\mathcal{M}$ as in \fref{Theorem}{theorem:pre_decomp}.
		\end{enumerate}
		Moreover $ \mathcal{X}_B (\tau)$ and $ \mathcal{X}_W (\tau)$ may be computed explicitly at every instant $\tau$ as in \fref{Theorem}{theorem:pre_decomp} such that then
		$\overline{\mathcal{X}}_W \vcentcolon= \lim_{\tau\to\infty} \fint_0^\tau \mathcal{X}_W (\tau') d\tau'$ and $\widetilde{\mathcal{X}} \vcentcolon= \mathcal{X} - \overline{ \mathcal{X} }_W$.
	\end{theorem}
	\begin{proof}
		This follows immediately from applying \fref{Theorem}{theorem:pre_decomp} to $ \mathcal{X} (\tau)$ at every instant $\tau \geqslant 0$
		and then constructing $\overline{\mathcal{X}}_W$ and $\widetilde{\mathcal{X}}_W$ as indicated.
	\end{proof}

	Note that the technical aspects of this theorem were already taken care of in the proof of the ``pre-decomposition'' of \fref{Theorem}{theorem:pre_decomp}.
	So there is nothing new to prove here.
	Nonetheless, as explained in this section in the lead-up to \fref{Theorem}{theorem:slow_fast_decomposition},
	this statement now has an \emph{interpretation} that \fref{Theorem}{theorem:pre_decomp} did not.
	Once viewed through the lens of fast-wave averaging, we can indeed interpret \fref{Theorem}{theorem:slow_fast_decomposition} as a \emph{slow-fast} decomposition:
	\begin{equation*}
		\mathcal{X} = \underbrace{\mathcal{X}_B + \overline{ \mathcal{X} }_W}_{ \mathcal{X}_\text{slow}} + \underbrace{\widetilde{ \mathcal{X} }_W}_{ \mathcal{X}_\text{fast}}.
	\end{equation*}
	In other words: the fast-wave averaging framework now, at least formally, ensures that this decomposition accounts for the entirety of the slow components of the solution!

	Finally, for completeness we record the fourth and final step of fast-wave averaging.
	The idea is that we can solve \eqref{eq:moist_Boussinesq_next_order_abstract} for $\mathcal{X}^1$,
	but we are not particularly interested in what that solution looks like.
	What we care about is making sure that this next-order solution satisfies the sub-linear growth condition \eqref{eq:sublinear_growth_condition}.
	Ultimately, this comes down to a condition on the right-hand side forcing terms of \eqref{eq:moist_Boussinesq_next_order_abstract}.
	Since these forcing terms depend on the leading-order solution $\mathcal{X}^0$, enforcing the sub-linear growth condition on $\mathcal{X}^1$
	turns into a condition on $\mathcal{X}^0$.

	In the dry case, this condition takes the form
	\begin{equation*}
		\fint_0^\tau \partial_t \underline{\mathcal{X}}^0 + e^{\tau' \mathcal{L}} \left( e^{-\tau' \mathcal{L}} \bar{u}^0 \cdot \nabla e^{-\tau'\mathcal{L}} \underline{\mathcal{X}}_0 \right) d\tau'
		= 0
	\end{equation*}
	since we may use the fact that the semigroup $e^{- \tau \mathcal{L}}$ is norm-preserving.
	In the moist case however the solution map $e^{-\tau\mathcal{L}_{W,\,H_B}}$ is \emph{not} norm-preserving (see \fref{Proposition}{prop:L_H_energy} )
	and so we cannot perform the same trick.
	Identifying what the solvability condition looks like in the moist case is therefore left open for future work.

	\begin{remark}[Slow-fast interpretation in different limiting regimes]
	\label{rmk:slow_fast_interpretation_different_regimes}
		In this section we have then described how, once in the limiting regime where both the Froude and Rossby numbers approach zero,
		the decomposition introduced in this paper has an additional interpretation as a slow-fast decomposition.
		Nonetheless, this same decomposition may still be of use in other limiting regimes, provided the interpretation of each of its components as slow or fast is
		updated appropriately.

		Consider for example the case where the Froude number is still taken to approach zero but the Rossby number is kept constant.
		In that case the vertically sheared horizontal flows discussed in \fref{Remark}{rmk:VSHF} are no longer fast, but now become slow!
	\end{remark}

\section{Decomposition and energy}
\label{sec:energy_and_decomposition}

	In the dry case, the decomposition of states into their slow and fast components
	has the additional properties that it is an orthogonal decomposition,
	where the inner product used is the one induced by the conserved energy.

	In the moist case, which is the focus of this paper, we therefore ask the following question:
	how does our decomposition of \fref{Theorem}{theorem:pre_decomp} relate to the energy conserved by the moist Boussinesq system?
	The energy in question is given by
	\begin{equation}
		E = \frac{1}{2} \int {\lvert u \rvert}^2 + \theta^2 + {\min}_0^2\, q + M^2,
	\label{eq:conserved_nonlinear_energy}
	\end{equation}
where an arbitrary coefficient of the $M^2$ term has been
set to 1 for simplicity.
For a demonstration that this energy is conserved,
see \fref{Proposition}{prop:cons_en}.

	In the dry case, since the decomposition comes down to linear projections, and since the energy gives rise to an inner product,
	the relation between the decomposition and the energy can be phrased as a simple question: are these projections orthogonal with respect to this inner product?
	In the moist case, both the decomposition and the energy are \emph{nonlinear} and so it is not clear how one would go about even \emph{asking} if the two are related.
	We offer several options in \fref{Sections}{sec:PDE_centric_approach}-\ref{sec:metric_less_approach}, summarizing and contrasting each option, along with their benefits and inconveniences,
	in \fref{Section}{sec:energy_summary}.

	Note that, in each of the three cases discussed below, the strategy is the same.
	Namely we seek to identify a way to quantify how far apart two states $(u_1,\,\theta_1,\,q_1)$ and $(u_2,\,\theta_2,\,q_2)$ are in such a manner that
	\begin{enumerate}
		\item	this agrees with the decomposition of \fref{Theorem}{theorem:pre_decomp}, and gives it further meaning as a projection, in the sense that the closest balanced state to any given state must be its balanced component and
		\item	this agrees with the energy (when the two are comparable).
	\end{enumerate}

	Finally, note that the authors believe that the quantification provided in \fref{Section}{sec:metric_less_approach} is really the right one to use for describing the decomposition from a projection perspective.
	However, it may appear a bit odd at first.
	To better motivate it we therefore discuss, in \fref{Sections}{sec:PDE_centric_approach} and \ref{sec:energy_centric_approach},
	alternative approaches which are less conclusive but help us understand why the \emph{good} approach, discussed in \fref{Section}{sec:metric_less_approach}, must take the form it does.
    Also, the alternative approaches in 
    \fref{Sections}{sec:PDE_centric_approach} and \ref{sec:energy_centric_approach}
    have value in their own right, for other applications
    beyond the focus here on a nonlinear eigenspace decomposition.

	Before we dive into the details of each of these three approaches, we must elucidate the following question:
	what does it mean for two states to be comparable using the energy?

	\begin{remark}[Comparing states using the energy]
	\label{rmk:comparing_states_using_the_energy}
		Recall that the energy of a state $ \mathcal{X}=(u,\,\theta,\,q)$ is 
		\begin{equation*}
			E ( \mathcal{X} ) = \int {\lvert u \rvert}^2 + \theta^2 + {\min}_0^2\, q + {(\theta + q)}^2.
		\end{equation*}
		The question is: given two states $ \mathcal{X}_1$ and $ \mathcal{X}_2$, \emph{when} can we use the energy to quantify the discrepancy
		between $ \mathcal{X}_1$ and $ \mathcal{X}_2$?

		A naive attempt would be to use the energy of their \emph{difference}, i.e.
		\begin{equation}
		\label{eq:diff_E}
			E ( \mathcal{X}_1 - \mathcal{X}_2) = \int {\lvert u_1 - u_2 \rvert}^2 + {(\theta_1 - \theta_2)}^2 + {\min}_0^2\, (q_1 - q_2) + {( (\theta_1+q_1) - (\theta_2+q_2))}^2.
		\end{equation}
		This is not physically meaningful.
		Why?
		Because of the nonlinearity.
		When computing $ E ( \mathcal{X}_1 - \mathcal{X}_2)$ the nonlinear term $ {\min}_0^2\, (q_1 - q_2)$ only turns on when $q_1 - q_2 > 0$.
		However this nonlinearity is not supposed to care about the sign of $q_1-q_2$.
		Instead the nonlinearity is there to account for the fact that, depending on the signs of $q_1$ and $q_2$, the energetic weights of $q_1$ and $q_2$ change.

		What is the takeaway?
		The only time we can directly compare two states $ \mathcal{X}_1$ and $ \mathcal{X}_2$ using the energy is when the signs of $q_1$ and $q_2$ agree everywhere in the spatial domain.
		The sign of $q$ is a mathematical proxy used to tell us which \emph{phase} the state is in at any point in the spatial domain.
		We may thus rephrase the observation above as follows.
		\vspace{1em}
		\begin{center}
			The only time we can directly compare two states $ \mathcal{X}_1$ and $ \mathcal{X}_2$\\
			is when their phases coincide everywhere in the spatial domain.
		\end{center}
		\vspace{1em}
		It will therefore be convenient to encode the spatial location of a state's phases.
		This is done by using the indicator function
		\begin{equation*}
			H \vcentcolon= \mathds{1} (q<0).
		\end{equation*}
		Note that since the water content $q = q(t,\,x)$ is a function of $t$ and $x$, then the indicator $H$ also depends on $t$ and $x$. As above, we will often use notation which suppresses the $(t,x)$-dependence in an effort to improve readability.

		So let us consider two states $ \mathcal{X}_1$ and $ \mathcal{X}_2$ whose phases agree,
		i.e. for which $H_1 = H_2 =\vcentcolon H$. On one hand
		\begin{align*}
			E ( \mathcal{X}_1 )
			&= \int {\lvert u_1 \rvert}^2 + \theta_1^2 + {\min}_0^2\, q_1 + {(\theta_1 + q_1)}^2\\
			&= \int {\lvert u_1 \rvert}^2 + \theta_1^2 + q_1^2 H + {(\theta_1 + q_1)}^2
		\end{align*}
		while on the other hand
		\begin{align*}
			E ( \mathcal{X}_2 )
			&= \int {\lvert u_2 \rvert}^2 + \theta_2^2 + {\min}_0^2\, q_2 + {(\theta_2 + q_2)}^2\\
			&= \int {\lvert u_2 \rvert}^2 + \theta_2^2 + q_2^2 H + {(\theta_2 + q_2)}^2.
		\end{align*}
	A natural way to compare $ \mathcal{X}_1$ and $ \mathcal{X}_2$ using the energy is then to treat the indicator function $H$ as \emph{fixed} in
		the expressions above for both $ E( \mathcal{X}_1)$ and $ E( \mathcal{X}_2)$
		and to consider
		\begin{equation}
		\label{eq:motivate_E_H}
			\int {\lvert u_1 - u_2 \rvert}^2 + {(\theta_1 - \theta_2)}^2 + {(q_1 - q_2)}^2 H + {((\theta_1 + q_1) - (\theta_2 + q_2))}^2.
		\end{equation}
		This is \emph{not} the same as $ E ( \mathcal{X}_1 - \mathcal{X}_2 )$, as can be seen by comparing \eqref{eq:diff_E} with \eqref{eq:motivate_E_H}.
		Nonetheless the quantity appearing in \eqref{eq:motivate_E_H} is, so far, the only way we have of using the energy to compare two states, as long as their phases agree everywhere.
	\end{remark}

    So we give this quantity a name.
    \begin{definition}
    \label{def:energy_fixed_Heaviside_1}
        For \emph{any} indicator function $H$ and \emph{any} state $ \mathcal{X} = (u,\,\theta,\,q) \in \mathbb{L}^2_\sigma $ let
        \begin{equation*}
            E_H ( \mathcal{X} ) \vcentcolon= \int {\lvert u \rvert}^2 + \theta^2 + q^2 H + {(\theta + q)}^2.
        \end{equation*}
    \end{definition}
    Note that this notation may then be used even when (especially when!) the sign of $q$ does not agree with $H$.
    That is precisely what happens in \eqref{eq:motivate_E_H}: $H$ encodes the signs of $q_1$ and $q_2$, \emph{not} the sign of $q_1 - q_2$
    (that was precisely the naive nonphysical approach we sought to avoid!)

    Here we have motivated this ``fixed $H$'' energy by talking about restricting our attention to a subset of states
    whose phases agree.
    There is another reason to look at this energy: it is the quadratic approximation of the conserved energy about a state whose phase are encoded in $H$.
    So this energy also appears in situations where linearisation are used,
    since linearising the dynamical equations is essentially the same as working with the quadratic approximation of the conserved energy.
    This interpretation of the ``fixed $H$'' energy will be more prevalent in later sections, such as \fref{Section}{sec:geometry} where
    geometrical notions relying on tangent spaces (a.k.a. linearisations) are discussed.
    Linearisations and the ``fixed $H$'' energy are
    also relevant for applications such as adjoint models,
    sensitivity analysis, and variational data assimilation,
    which rely on tangent linear approximations to the
    dynamics
    \cite{le1986variational,park1997validity,errico1997adjoint,errico1999examination,mahfouf1999influence,barkmeijer2001tropical,amerault2008tests,doyle2014initial}.

\subsection{PDE-centric approach}
\label{sec:PDE_centric_approach}

	We first discuss a PDE-centric approach to relating the decomposition and the energy.
	Why start here?
	There are two reasons to begin with a PDE-centric approach.

	The first reason to take such an approach is that elliptic PDEs underpin the decomposition in both the case of the Helmholtz decomposition
	and in the dry case. Since a nonlinear elliptic PDE is one of the few structures that have already been identified to be present in the moist case,
	this PDE therefore presents a natural starting point.

	The second reason to take a PDE-centric approach is that a similar approach was fruitful in identifying the slow--fast decomposition.
	Indeed, the slow--fast decomposition was obtained by recognizing the important role played by the slow measurements $PV$ and $M$,
	and then inverting the nonlinear $PV$-and-$M$ inversion PDE \eqref{eq:PV_M_inversion_pre_decomp_statement}
	in order to obtain $p$. The variables $p$ and $M$ were then used to characterise the balanced part of a state.

	Equipped with \fref{Propositions}{prop:alt_char_balanced_set} and \ref{prop:chara_im_N} we can actually go a step further:
	the slow measurements $PV$ and $M$ characterise the balanced component of a state while the fast measurements $j$, $w$, and $a$ characterise its wave component.
	Indeed, we have the following.
	\begin{itemize}
		\item	On one hand \fref{Proposition}{prop:alt_char_balanced_set} tells us that the measurements
			\begin{align*}
				j(u,\,\theta,\,q) = \partial_3 u_h - \nabla_h^\perp (\theta - {\min}_0\, q),\,
				w(u,\,\theta,\,q) = u_3,
				\\
				\text{ and }
				a(u,\,\theta,\,q) = - \fint u_h^\perp + \fint (\theta - {\min}_0\, \theta) e_3
			\end{align*}
			are transversal to the balanced set $\mathcal{B}$, and so we may expect them to characterize elements of the wave set $\mathcal{W}$.
		\item	On the other hand \fref{Proposition}{prop:chara_im_N} tells us that the measurements
			\begin{equation*}
				PV(u,\,\theta,\,q) = \nabla_h^\perp \cdot u_h + \partial_3 \theta \text{ and } \\
				M(u,\,\theta,\,q) = \theta + q
			\end{equation*}
			are transversal to the wave set $\mathcal{W}$, and so we may similarly expect them to characterize elements of the balanced set $\mathcal{B}$.
	\end{itemize}
	Both of these assertions turn out to be true, and thus give rise to the global change of coordinates recorded in \fref{Proposition}{prop:global_chg_coord} which allows us to translate between a state $(u,\,\theta,\,q)$
	and the tuple of measurements $(PV,\,M,\,j,\,w,\,a)$.

	We will now use this change of variable from \fref{Proposition}{prop:global_chg_coord} to define a distance on the state space $ \mathbb{L}^2_\sigma $.
	We will do this by using the change of variables to pull-back a metric on the measurement space $ \mathcal{M} $ and thus produce a metric on the state space $ \mathbb{L}^2_\sigma $
	(see \fref{Proposition}{prop:pullback_metric}).
	The execution of this strategy requires a few steps.
	In order to have a clearer idea of where we are headed, we record below what would happen if we tried doing the same thing in the \emph{dry} case.
	In the dry case we can write the energy in terms of the measurements (namely $PV$, $j$, $w$, and $a$).
	In the language that will be used in this section, we can therefore view the energy, and the $L^2$ metric it corresponds to, as the pullback of a metric on the (dry) measurement space!
	Since we will ultimately, in this section, define the distance of interest in terms of the \emph{coordinates} (think $p$, $M$, $\wavecoord$, and $w$),
	we also record in \fref{Proposition}{prop:dry_Parseval} below how the energy may be written in a third way, in terms of the coordinates (which are $p$, $\wavecoord$, and $w$ in the dry case).

	\begin{prop}[Parseval identity in the dry case]
	\label{prop:dry_Parseval}
		For any $u \in { \left( L^2 \right) }_\sigma^3$ and $\theta\in L^2$ we have that
		\begin{equation*}
			\underbrace{
				\int {\lvert u \rvert}^2 + \theta^2
			}_{E}
			= \int {\lvert \nabla\Delta^{-1} PV \rvert}^2 + {\left\lvert {(\nabla\times)}^{-1} \left( j + (\partial_3 w) e_3,\, a \right) \right\rvert}^2
			= \underbrace{
				\int {\lvert \nabla p \rvert}^2 + {\lvert \wavecoord \rvert}^2 + w^2
			}_{\sim d^2}
		\end{equation*}
		where $p$ solves
		\begin{equation*}
			\Delta p = \underbrace{
				\nabla_h^\perp \cdot u_h + \partial_3 \theta
			}_{PV}
		\end{equation*}
		while $\wavecoord$ solves
		\begin{equation*}
			\nabla\times\wavecoord = \underbrace{
				\partial_3 u_h - \nabla_h^\perp \theta
			}_{j}
			+ (\partial_3 w) e_3
			\text{ subject to }
			\fint \wavecoord = \underbrace{\fint -u_h^\perp + \fint \theta e_3}_{a},
		\end{equation*}
		and where $w \vcentcolon= u_3$.
		Here we write $\wavecoord = {(\nabla\times)}^{-1} (\omega,\,a)$ to mean that $\sigma$ solves
		\begin{equation*}
			\nabla\times\wavecoord = \omega \text{ subject to } \nabla\cdot\wavecoord = 0 \text{ and } \fint \sigma = a.
		\end{equation*}
	\end{prop}
	\begin{proof}
		We recall that the decomposition in the dry case is
		\begin{equation*}
			\begin{pmatrix}
				u \\ \theta
			\end{pmatrix} = \begin{pmatrix}
				\nabla_h^\perp p \\ \partial_3 p
			\end{pmatrix} + \begin{pmatrix}
				\wavecoord_h^\perp + we_3 \\ \wavecoord_3
			\end{pmatrix}
		\end{equation*}
		where $\wavecoord$ is divergence-free and satisfies $\nabla_h^\perp\cdot\wavecoord_h = \partial_3 w$.
		We may then compute that
		\begin{equation*}
			PV
			= \nabla_h^\perp\cdot u_h + \partial_3 \theta
			= \Delta p
		\end{equation*}
		while \fref{Lemma}{lemma:horizontal_vertical_decomposition_of_the_curl} tells us that
		\begin{equation*}
			\nabla\times\wavecoord
			= \partial_3 \wavecoord_h^\perp - \nabla_h^\perp \wavecoord_3 + (\nabla_h^\perp \cdot \wavecoord_h) e_3
			= \partial_3 u_h - \nabla_h^\perp \theta + (\partial_3 w) e_3,
		\end{equation*}
		with moreover
		\begin{equation*}
			\fint \wavecoord = \fint \left( - u_h^\perp + \theta e_3 \right),
		\end{equation*}
		as claimed.
		Finally we may compute that, since $\wavecoord$ is divergence-free and hence $\int \wavecoord\cdot\nabla p = 0$,
		\begin{align*}
			E
			= \frac{1}{2} \int {\lvert u \rvert}^2 + \theta^2
			&= \frac{1}{2} \int {\lvert u_h \rvert}^2 + u_3^2 + \theta^2
			\\&= \frac{1}{2} \int {\lvert \nabla_h p + \wavecoord_h \rvert}^2 + w^2 + {(\partial_3 p + \wavecoord_3)}^2
			\\&= \frac{1}{2} \int {\lvert \nabla p \rvert}^2 + {\lvert \wavecoord \rvert}^2 + w^2,
		\end{align*}
		as desired.
	\end{proof}

	The benefit of the approach carried out in this section,
	where we pull back a metric from measurement space onto state space, is that it will immediately produce a metric which agrees with the decomposition.
	The drawback of this approach is that it does not, inherently, relate to the conserved energy in any way.
	To see why a metric produced in this way will automatically agree with the decomposition relies on the following observation:
	if two states have the same $PV$ and $M$, then they must have the same balanced component.
	This is proved in \fref{Lemma}{lemma:PV_and_M_characterize_balanced_states} below.
	Therefore projecting a state onto the balanced set is the same as keeping its $PV$ and $M$ measurements fixed and sending its remaining measurements, namely $j$, $w$, and $a$, to zero.
	This is why a metric pulled back from the measurement space $ \mathcal{M} $ will necessarily agree with the decomposition.
	As mentioned above, this relies on the following result.

	\begin{lemma}[$PV$ and $M$ characterize balanced components]
	\label{lemma:PV_and_M_characterize_balanced_states}
		Given any $ \mathcal{X}\in \mathbb{L}^2_\sigma $ the unique balanced state in $\mathcal{B}$ whose $PV$ and $M $ agree with the $PV$ and $M$ of $ \mathcal{X}$
		is precisely the balanced component $ \mathcal{X}_B$.
	\end{lemma}
	\begin{proof}
		Let $ \mathcal{X}\in \mathbb{L}^2_\sigma $ and let $ \mathcal{Y} \in \mathcal{B} $ be a balanced state whose $PV$ and $M$ agrees with those of $ \mathcal{X}$,
		i.e.
		\begin{equation*}
			PV ( \mathcal{X} ) = PV ( \mathcal{Y} ) \text{ and } 
			M ( \mathcal{X} ) = M ( \mathcal{Y} ).
		\end{equation*}
		We know from \fref{Proposition}{prop:alt_char_balanced_set} that since $ \mathcal{Y} $ is balanced it is uniquely characterised by $p\in \mathring{H}^1 $ and $M\in L^2$
		via $ \mathcal{Y}  = \Phi (p,\,M)$.
		Moreover we may compute that
		\begin{equation*}
			M(\Phi(p,\,M)) = M
			\text{ and } 
			PV (\Phi (p,\,M)) = \Delta p + \frac{1}{2} \partial_3 {\min}_0\, (M - \partial_3 p).
		\end{equation*}
		By the well-posedness of nonlinear $PV$-and-$M$ inversion, we know that there is a unique $p$ solving
		\begin{equation*}
			\Delta p + \frac{1}{2} \partial_3 {\min}_0\, (M - \partial_3 p) = PV ( \mathcal{Y} ) = PV ( \mathcal{X}),
		\end{equation*}
		and \fref{Proposition}{prop:alt_char_balanced_set} tells us that the resulting state is precisely $ \mathcal{X}_B$.
		This verifies that indeed $ \mathcal{Y} = \mathcal{X}_B$.
	\end{proof}

	Using \fref{Lemma}{lemma:PV_and_M_characterize_balanced_states} above we could then immediately deduce that a distance on $ \mathbb{L}^2_\sigma $ defined via
	\begin{align}
		&d \left( \mathcal{X}_1,\, \mathcal{X}_2 \right)^2
		\vcentcolon= \norm{ \mathcal{M} ( \mathcal{X}_1 ) - \mathcal{M} ( \mathcal{X}_2 ) }{ \mathfrak{M} }^2
	\nonumber\\
		&= \norm{{PV}_1 - {PV}_2}{H^{-1}}^2 + \norm{M_1 - M_2}{L^2}^2 + \norm{j_1 - j_2}{H^{-1}}^2 + \norm{w_1 - w_2}{L^2}^2 + {\lvert a_1 - a_2 \rvert}^2
	\label{eq:naive_distance}
	\end{align}
	would automatically agree with the decomposition.
	However this distance has no chance of agreeing with the energy!
	This is because its dimensions, or units, are all wrong.
	Indeed: this naive distance above treats $PV$ and $M$ on the same footing even though they do not have the same dimensions.
	In order to produce a metric on the state space $ \mathbb{L}^2_\sigma $ pulled back from the measurement space $ \mathcal{M} $ which has a chance of agreeing with the energy
	we must therefore \emph{dimensionalize} the naive distance from \eqref{eq:naive_distance}.

	In order to perform this dimensionalisation we will use the fact that the ``fixed $H$'' energy introduced in \fref{Definition}{def:energy_fixed_Heaviside_1}
	may sometimes be written not in terms of the state $(u,\,\theta,\,q)$ but in terms of its coordinates $(p,\,M,\,\wavecoord,\,w)$.
	To see this, we first define the following.

	\begin{definition}
	\label{def:energy_fixed_Heaviside_coord}
		Fix an indicator function $H$.
		We define, for any $(p,\,M,\,\wavecoord,\,w) \in \mathcal{C}$,
		\begin{equation*}
			\widetilde{E}_H (p,\,M,\,\wavecoord,\,w) \vcentcolon=  \frac{1}{2} \int \underbrace{A_H \nabla p \cdot \nabla p}_{ =\vcentcolon Q_H^+(p)} + \left( 1 + \frac{H}{2} \right) M^2
			+ \underbrace{A_H^{-1} \wavecoord \cdot \wavecoord}_{ =\vcentcolon Q_H^- (\wavecoord)} + w^2,
		\end{equation*}
		where $A_H = I - \frac{H}{2} e_3\otimes e_3$ and so $A_H^{-1} = I + H e_3\otimes e_3$.
	\end{definition}
	\begin{remark}[The elliptic matrix $A_H$ and its inverse]
	\label{rmk:A_H}
		The matrix $A_H$ introduced in \fref{Definition}{def:energy_fixed_Heaviside_coord} above will be encountered many times in the sequel.
		Where does it come from?
		The simplest explanation is that it comes from the linearisation of nonlinear $PV$-and-$M$ inversion.
		Indeed, upon linearising $PV$-and-$M$ inversion
		\begin{equation*}
			\Delta p + \frac{1}{2} \partial_3 {\min}_0\, (M - \partial_3 p) = PV
		\end{equation*}
		we obtain, upon differentiating,
		\begin{equation*}
			\Delta \partial_\alpha \pi + \frac{1}{2} \mathds{1} (M < \partial_3 p) ( \partial_\alpha M - \partial_\alpha \partial_3 \pi) = \partial_\alpha PV.
		\end{equation*}
		In other words $ \partial_\alpha p $ satisfies, for $H \vcentcolon= \mathds{1} (M < \partial_3 p)$,
		\begin{equation*}
			\nabla\cdot ( A_H \nabla \partial_\alpha p) = \partial_\alpha PV - \frac{1}{2} H \partial_\alpha M.
		\end{equation*}

		In particular we may readily verify that the identity
		\begin{equation*}
			\left( 1 + H \right) \left( 1 - \frac{H}{2} \right) = 1
		\end{equation*}
		holds for any $H$ satisfying $H^2 = H$ (which is the case if $H$ is an indicator function),
		and so the inverse of $A_H$ is indeed given by
		\begin{equation*}
			A_H^{-1} = I + H e_3 \otimes e_3.
		\end{equation*}
	\end{remark}
	In general the form of the fixed Heaviside energy provided in \fref{Definition}{def:energy_fixed_Heaviside_coord}
	does not agree with the form of the fixed Heaviside energy provided in \fref{Definition}{def:energy_fixed_Heaviside_1}.
	Sometimes, they do.
	\begin{prop}[Rewriting the energy]
	\label{prop:rewrite_the_energy}
		Let $E_H$ be as introduced in \fref{Definition}{def:energy_fixed_Heaviside_1}
		and $\widetilde{E}_H$ be as in \fref{Definition}{def:energy_fixed_Heaviside_coord}.
		Let $(u,\,\theta,q) \in \mathbb{L}^2_\sigma $ be corresponding coordinates $(p,\,M,\,\wavecoord,\,w)$ as per \fref{Corollary}{cor:decomposition}.
		In general we have that
		\begin{equation*}
			E_H (u,\,\theta,\,q) = \frac{1}{2} \int {\lvert u_h \rvert}^2 + u_3^2 + \left( 1 - \frac{H}{2} \right) {(\theta - qH)}^2 + \left( 1 + \frac{H}{2} \right) {(\theta + q)}^2.
		\end{equation*}
		Moreover, if $q$ and its balanced component $q_B$ share the same interface, i.e. if
		\begin{equation*}
			H \vcentcolon=  \mathds{1} (q < 0) = \mathds{1}(q_B < 0),
		\end{equation*}
		where note that $\mathds{1} (q_B < 0) = \mathds{1}(M < \partial_3 p)$, then
		\begin{equation*}
			E_H(u,\,\theta,\,q) = \widetilde{E}_H(p,\,M,\,\wavecoord,\,w).
		\end{equation*}
	\end{prop}
	\begin{proof}
		The first identity follows from an immediate computation since $H^2 = H$ and hence
		\begin{equation*}
			\theta^2 + q^2 H
			= \left( 1 - \frac{H}{2} \right) {(\theta - qH)}^2 + \frac{H}{2} {(\theta + q)}^2.
		\end{equation*}
		To deduce the second identity we first note that \fref{Corollary}{cor:decomposition} and \fref{Lemma}{lemma:identity_lin_buoyancy_and_xi} tell us that,
		if $H = H_B \vcentcolon= \mathds{1}(q_B < 0) = \mathds{1} (M<\partial_3 p)$, then
		\begin{equation*}
			\left\{
			\begin{aligned}
				&u_h = \nabla_h^\perp p + \wavecoord_h^\perp,\\
				&u_3 = w,\\
				&\theta - qH = \partial_3 p + (1+H)\wavecoord_3, \text{ and } \\
				&\theta + q = M.
			\end{aligned}
			\right.
		\end{equation*}
		Therefore
		\begin{align*}
			&\int {\lvert u_h \rvert}^2 + u_3^2 + \left( 1 - \frac{H}{2} \right) {(\theta - qH)}^2 + \left( 1 + \frac{H}{2} \right) {(\theta + q)}^2
		\\
			&= \int {\lvert \nabla_h p + \wavecoord_h \rvert}^2 + w^2 + \left( 1 - \frac{H}{2} \right) { \left( \partial_3 p + (1+H) \wavecoord_3 \right)}^2 + \left( 1 + \frac{H}{2} \right) M^2.
		\end{align*}
		Finally, for $A_H$ and $A_H^{-1}$ as in \fref{Definition}{def:energy_fixed_Heaviside_coord} we compute that, since $A_H$ and $A_H^{-1}$ are symmetric and since $\wavecoord$ is divergence-free,
		\begin{align*}
			\int {\lvert \nabla_h p + \wavecoord_h \rvert}^2 + \left( 1 - \frac{H}{2} \right) { \left( \partial_3 p + (1+H) \wavecoord \right)}^2
			&= \int A_H (\nabla p + A_H^{-1} \wavecoord) \cdot ( \nabla p + A_H^{-1} \wavecoord)
		\\
			&= \int A_H \nabla p \cdot \nabla p + 2 \underbrace{\int \nabla p \cdot \wavecoord}_{=0} + \int A_H^{-1} \wavecoord \cdot \wavecoord.
		\end{align*}
		The second identity follows.
	\end{proof}
	This is a good point to take a stock.
	Recall that we seek to define a metric pulled back from the measurement space $ \mathcal{M} $ onto the state space $ \mathbb{L}^2_\sigma $.
	Such a metric will automatically agree with the decomposition but, in order to make it agree with the energy, we must dimensionalize it appropriately. 
	The identity recorded in \fref{Proposition}{prop:rewrite_the_energy} above tells us how to do this.
	The last remaining piece required to define our pullback metric is thus to identify how the measurements ($PV$, $M$, $j$, $w$, and $a$)
	will determine the coordinates ($p$, $M$, $\wavecoord$, and $w$).
	Note that such a map is already implicit in \fref{Corollary}{cor:decomposition}, however we will do something slightly different:
	since the ``fixed $H$'' energy already uses a fixed Heaviside, we will leverage this to define a map taking measurements to coordinates in a \emph{linear} way.
	This linearity will make the ensuing pullback distance easier to work with.

	\begin{lemma}[Definition and invertibility of the fixed Heaviside parametrisation $\mathcal{C}_H$]
	\label{lemma:def_and_invertibility_of_C_H}
		Fix an indicator function $H$ and define $A_H \vcentcolon= I - \frac{H}{2} e_3\otimes e_3$, such that $A_H^{-1} = I + H e_3\otimes e_3$
		(since $(1-H/2)(1+H) \equiv 1$).
		We define $\mathcal{C}_H : \mathfrak{M} \to \mathfrak{C}$, where $\mathfrak{C}$ is the \emph{coordinate space} defined in \eqref{eq:def_mathfrak_C}.
		Given $(PV,\,M,\,j,\,w,\,a)\in \mathfrak{M} $, let $M \vcentcolon= M$, $w \vcentcolon= w$, and let $p$ and $\wavecoord$ solve
		\begin{equation}
			\nabla\cdot(A_H\nabla p) = PV - \frac{1}{2} \partial_3 (HM)
		\label{eq:def_C_H_inv_1}
		\end{equation}
		and
		\begin{equation}
			\nabla\times (A_H^{-1}\wavecoord) = j + (\partial_3 w) e_3
			\text{ subject to } \nabla\cdot \wavecoord = 0 \text{ and } \fint A_H^{-1} \wavecoord = a.
		\label{eq:def_C_H_inv_2}
		\end{equation}
		The map $ \mathcal{C}_H$ is well-defined and invertible.
	\end{lemma}
	Note that if $H = H_B \vcentcolon= \mathds{1}(q_B < 0) = \mathds{1}(M<\partial_3 p)$ then $\mathcal{C}_H$ is precisely the map implicit in \fref{Corollary}{cor:decomposition}
	which assigns the coordinates $p$, $M$, $\wavecoord$, and $w$ to a state $(u,\,\theta,\,q)$ through its measurements $PV$, $M$, $j$, $w$, and $a$.
	\begin{proof}[Proof of \fref{Lemma}{lemma:def_and_invertibility_of_C_H}]
		It suffices to show that both \eqref{eq:def_C_H_inv_1} and \eqref{eq:def_C_H_inv_2} have a unique solution once $PV$, $M$, $j$, $w$, and $a$ are fixed.
		The well-posedness of \eqref{eq:def_C_H_inv_1} follows from the fact that $A_H$ is uniformly elliptic since $A_H(x) \geqslant \frac{1}{2} I$ for all $x$.
		The well-posedness of \eqref{eq:def_C_H_inv_2} follows from the fact that $A_H^{-1}$ is uniformly elliptic, since $A_H^{-1}(x) \geqslant I$ for all $x$,
		and so \fref{Proposition}{prop:solvability_div_curl_system} applies to \eqref{eq:def_C_H_inv_2}.
	\end{proof}

	\begin{remark}
	\label{rmk:def_C_H}
		In the definition of $ \mathcal{C}_H$ we are linearising the coordinates, but \emph{not} the measurements.
		This is because it is the full nonlinear measurements which are well-adapted to the balanced $\mathcal{B}$ and to the wave set $\mathcal{W}$,
		so changing those would mean that the resulting distance (in \fref{Definition}{def:dimensionalized_Parseval_distance} below) would fail to be compatible with the energy.
		It should not yet be clear why using linearised versions of the coordinates is particularly helpful.
		This will be discussed in \fref{Section}{sec:metric_less_approach} below.
	\end{remark}

	We are now ready to define a dimensionalized version of the naive pullback metric of \eqref{eq:naive_distance}.

	\begin{definition}[Dimensionalized Parseval distance]
	\label{def:dimensionalized_Parseval_distance}
		Fix an indicator function $H$.
		Define $d_H : \mathbb{L}^2_\sigma \times \mathbb{L}^2_\sigma \to [0,\,\infty)$ as follows:
		\begin{equation*}
			\frac{1}{2} {d_H ( \mathcal{X}_1,\, \mathcal{X}_2 ) }^2
			= \widetilde{E}_H \left( \mathcal{C}_H \left[ \mathcal{M} ( \mathcal{X}_1 ) - \mathcal{M} ( \mathcal{X}_2 ) \right] \right).
		\end{equation*}
		In other words, given $ \mathcal{X}_1,\,\mathcal{X}_2 \in \mathbb{L}^2_\sigma $, let
		\begin{equation*}
			M \vcentcolon= M(\mathcal{X}_1) - M(\mathcal{X}_2),\,
			w \vcentcolon= w(\mathcal{X}_1) - w(\mathcal{X}_2), \text{ and } 
			a \vcentcolon= a(\mathcal{X}_1) - a(\mathcal{X}_2),
		\end{equation*}
		let $p$ solve, for $PV \vcentcolon= {PV}(\mathcal{X}_1) - {PV}(\mathcal{X}_2)$,
		\begin{equation*}
			\nabla\cdot ( A_H \nabla p) = PV - \frac{1}{2} \partial_3 (HM)
		\end{equation*}
		and let $\wavecoord$ solve, for $j \vcentcolon= j(\mathcal{X}_1) - j(\mathcal{X}_2)$,
		\begin{equation*}
			\nabla\times (A_H^{-1}\wavecoord) = j + (\partial_3 w) e_3
			\text{ subject to } \fint A_H^{-1} \wavecoord = a \text{ and } \nabla\cdot\wavecoord = 0.
		\end{equation*}
		Then
		\begin{equation*}
			\frac{1}{2} { d_H( \mathcal{X}_1,\, \mathcal{X}_2 )}^2
			= \int Q_H^+ (p) + \left( 1 + \frac{H}{2} \right) M^2 + Q_H^- (\wavecoord) + w^2
		\end{equation*}
		for $Q_H^\pm$ as introduced in \fref{Definition}{def:energy_fixed_Heaviside_coord}.
	\end{definition}
	\begin{prop}
	\label{prop:dimensionalized_Parseval_distance}
		For any indicator function $H$ the function $d_H$ defined in \fref{Definition}{def:dimensionalized_Parseval_distance} is a metric
		which agrees with the extraction of a balanced component.
		Moreover this distance is compatible with the energy when $H = H_1 = H_{B,\,1} = H_2 = H_{B,\,2}$, i.e. in that case
		\begin{align*}
			\frac{1}{2} {d_H \left( \mathcal{X}_1,\, \mathcal{X}_2 \right)}^2
			&= E_H ( \mathcal{X}_1 - \mathcal{X}_2)
		\\
			&= \int {\lvert u_1 - u_2 \rvert}^2 + {(\theta_1 - \theta_2)}^2 + {(q_1 - q_2)}^2 H + {(M_1 - M_2)}^2
		\end{align*}
		where $H = \mathds{1} (q_1 < 0) = \mathds{1} (q_2 < 0)$.
	\end{prop}
	\begin{proof}
		First we show that $d_H$ is a metric.
		Since $\widetilde{E}_H$ is positive-definite and viewed here as taking arguments in $ \mathfrak{C}$ it is a quadratic form which induces an inner product,
		and hence a metric which we denote $d_E$, on $ \mathfrak{C}$.
		Since $ \mathcal{C}_H$ is linear and since both $ \mathcal{C}_H$ and $ \mathcal{M} $ are invertible
		(see \fref{Proposition}{prop:global_chg_coord} and \fref{Lemma}{lemma:def_and_invertibility_of_C_H} )
		and hence injective it follows from \fref{Proposition}{prop:pullback_metric} that $d_H$ is the pullback metric of $d_E$ under $ \mathcal{C}_H \circ \mathcal{M} $.
		So $d_H$ is indeed a metric on $ \mathbb{L}^2_\sigma $.

		Now we show that $d_H$ conditionally agrees with the energy.
		Suppose that $H = H_1 = H_{B,\,1} = H_2 = H_{B,\,2}$.
		The argument is easier to follow if we rewrite $d_H$ as follows, using the fact that $ \mathcal{C}_H $ is linear:
		\begin{equation*}
			\frac{1}{2} { d_H ( \mathcal{X}_1,\, \mathcal{X}_2 )}^2
			= \widetilde{E}_H \left[
				\mathcal{C}_H \left( \mathcal{M} \left( \mathcal{X}_1 \right) \right)
				- \mathcal{C}_H \left( \mathcal{M} \left( \mathcal{X}_2 \right) \right)
			\right].
		\end{equation*}
		Let $(p_1,\,M_1,\,\wavecoord_1,\,w_1) \vcentcolon= \mathcal{C}_H \left( \mathcal{M} \left( \mathcal{X}_1 \right) \right)$.
		Since $H = H_{B,\,1} = \mathds{1} (M_1 < \partial_3 p_1)$ we see that
		\begin{align}
			&\nabla\cdot \left( A_H \nabla p_1 \right) = {PV}_1 - \frac{1}{2} \partial_3 ( HM_1)
		\nonumber\\
			&\iff \Delta p_1 + \frac{1}{2} \partial_3 {\min}_0\, (M_1 - \partial_3 p_1) = {PV}_1.
		\label{eq:dimensionalize_Parseval_distance_PDE_p}
		\end{align}
		Similarly, since $H_1 = H_{B,\,1}$ we have that
		\begin{align*}
			j_1
			&= \partial_3 u_{h,\,1} - \nabla_h^\perp (\theta_1 - {\min}_0\, q_1)
		\\
			&= \partial_3 u_{h,\,1} - \nabla_h^\perp (\theta_1 - H_{B,\,1}  q_1)
		\end{align*}
		and so $\wavecoord_1$ solves
		\begin{equation}
			\nabla\times (A_H^{-1} \wavecoord_1)
			= \partial_3 u_{h,\,1} - \nabla_h^\perp ( \theta_1 - H_{B,\,1} q_1 ) + (\partial_3 w_1) e_3.
		\label{eq:dimensionalize_Parseval_distance_PDE_xi}
		\end{equation}
		We deduce from \eqref{eq:dimensionalize_Parseval_distance_PDE_p} and \eqref{eq:dimensionalize_Parseval_distance_PDE_xi}
		that $p_1$ and $\wavecoord_1$ are precisely characterized as in \fref{Theorem}{theorem:pre_decomp} by
		\begin{equation*}
			\mathcal{X}_1 = \Phi(p_1,\,M_1) + \Psi(\wavecoord_1,\,w_1).
		\end{equation*}
		In other words: the approximate inversion $ \mathcal{C}_H \circ \mathcal{M} $ is \emph{exact} when $H = H_1 = H_{B,\,1}$.
		Moreover \fref{Proposition}{prop:rewrite_the_energy} tells us that $\widetilde{E}_H$ agrees with $E_H$ and with the nonlinear energy when $H = H_1 = H_{B,\,1}$
		and so we conclude that indeed if $H = H_1 = H_{B,\,1} = H_2 = H_{B,\,2}$ then
		\begin{align*}
			\frac{1}{2} { d_H ( \mathcal{X}_1,\, \mathcal{X}_2 )}^2
			&= E_H ( \mathcal{X}_1 - \mathcal{X}_2 )
		\\
			&= \int {\lvert u_1 - u_2 \rvert}^2 + {(\theta_1 - \theta_2)}^2 + {(q_1 - q_2)}^2 H + {(M_1 - M_2)}^2,
		\end{align*}
		where $H = \mathds{1} (q_1 < 0) = \mathds{1} (q_2 < 0)$.

		Finally we show that $d_H$ agrees with the extraction of a balanced component.
		We fix $ \mathcal{X}\in \mathbb{L}^2_\sigma $ and let $ \mathcal{Y} \in \mathcal{B} $.
		Since $ \mathcal{Y} \in \mathcal{B} $ we know from \fref{Proposition}{prop:alt_char_balanced_set} that
		\begin{equation*}
			j ( \mathcal{Y} ) = 0,\,
			w ( \mathcal{Y} ) = 0, \text{ and }
			a ( \mathcal{Y} ) = 0.
		\end{equation*}
		Let us define
		\begin{equation}
			M \vcentcolon= M ( \mathcal{X} ) - M ( \mathcal{Y} ) \text{ and } 
			PV \vcentcolon= PV ( \mathcal{X} ) - PV ( \mathcal{Y} ),
		\label{eq:dimensionalize_Parseval_distance_energy_agreement_1}
		\end{equation}
		and let $p$ be the solution of
		\begin{equation}
			\nabla\cdot (A_H\nabla p) = PV - \frac{1}{2} \partial_3 (HM)
		\label{eq:dimensionalize_Parseval_distance_energy_agreement_2}
		\end{equation}
		while $\wavecoord$ is the solution of
		\begin{equation*}
			\nabla\times ( A_H^{-1} \wavecoord) = j( \mathcal{X} ) + \left[ \partial_3 w( \mathcal{X} ) \right] e_3
			\text{ subject to } \fint A_H^{-1} \wavecoord = a ( \mathcal{X} ) \text{ and } \nabla\cdot\wavecoord = 0.
		\end{equation*}
		Then we may write the distance between $ \mathcal{X} $ and $ \mathcal{Y} $ via
		\begin{equation*}
			\frac{1}{2} { d_H ( \mathcal{X},\, \mathcal{Y} )}^2
			= \int Q_H^+ (p) + \left( 1 + \frac{H}{2} \right) M^2 + Q_H^- (\wavecoord) + w^2
		\end{equation*}
		for $Q_H^\pm$ as in \fref{Definition}{def:energy_fixed_Heaviside_coord}.
		Crucially: $\wavecoord$ and $w$ only depend on $ \mathcal{X}$!
		So we may simplify the minimization problem as follows:
		\begin{equation*}
			\argmin_{ \mathcal{Y} \in \mathcal{B} } d_H ( \mathcal{X},\, \mathcal{Y} )
			= \argmin_{ \mathcal{Y} \in \mathcal{B} } \int Q_H^+ (p) + \left( 1+\frac{H}{2} \right) M^2.
		\end{equation*}
		In particular, since $Q_H^+ (v) = A_H v\cdot v \geqslant \frac{1}{2} {\lvert v \rvert}^2 $, i.e. since $Q_H^+$ is positive-definite,
		the unique minimizer corresponds to
		\begin{equation*}
			p = 0 \text{ and } M = 0.
		\end{equation*}
		Plugging this into \eqref{eq:dimensionalize_Parseval_distance_energy_agreement_1}, and then into \eqref{eq:dimensionalize_Parseval_distance_energy_agreement_2},
		this means that $M ( \mathcal{X}) = M( \mathcal{Y} )$, and hence that $ PV ( \mathcal{X}) = PV ( \mathcal{Y} )$.
		In other words the unique minimizer is the balanced state $ \mathcal{Y} $ whose $PV$ and $M$ agree with $ \mathcal{X}$.
		By \fref{Lemma}{lemma:PV_and_M_characterize_balanced_states}, this means precisely that $ \mathcal{Y}  = \mathcal{X}_B$, as desired.
	\end{proof}

	This concludes our first attempt at quantifying how far apart two states are in a way that agrees with
	\begin{enumerate}
		\item	the extraction of a balanced component and
		\item	the conserved energy.
	\end{enumerate}
	We see that the first of these two requirements was met by the distance introduced in \fref{Definition}{def:dimensionalized_Parseval_distance} above.
	However the second requirement was only met in the restricted setting where the fixed Heaviside
	agreed with the Heavisides arising from both states being compared.

	In some sense, that only the first requirement is satisfied is not surprising.
	Indeed, the two distances introduced in this section are built on top of the measurements from the global change of coordinates of \fref{Proposition}{prop:global_chg_coord} 
	and these measurements have been chosen precisely to be compatible with the decomposition.
	However, these measurements have, a priori, little to do with the conserved energy and so it is actually quite positive that the distance built in this way
	is in any way comparable to the energy.

	This motivates an alternate approach: what if, instead of using decomposition--amenable measurements as a starting point,
	we used the energy itself as a starting point to define a distance?
	This is the approach we take in the following section.

\subsection{Energy-centric approach}
\label{sec:energy_centric_approach}

	Recall that we are after a way to measure how different two states are.
	Moreover we wish to do so in a manner which is compatible with both the balanced-unbalanced decomposition underpinning \fref{Theorem}{theorem:pre_decomp}
	\emph{and} the conserved energy.
	In the previous section we approached this problem with the first requirement in mind, and eventually discovered that this approach failed to (generically) satisfy the second requirement.
	In this section we proceed the other way around: we enforce the second requirement by putting the conserved energy front and center in the construction of a new distance.
	The question will then be to identify whether or not this metric agrees with the decomposition.

Also, recall the difficulty that arises in trying to define a
distance or metric or inner product that agrees with the energy in \eqref{eq:conserved_nonlinear_energy}:
the nonlinear switch from the minimum function or Heaviside function. 
This nonlinear switch is activated based on the total water $q$,
and it is not clear how to define the nonlinear switch
when two different total water variables, $q_1$ and $q_2$,
are involved in quantifying the distance between two different states.
This difficulty and others are discussed in more detail
in the text preceding \fref{Section}{sec:PDE_centric_approach}.

A key idea is helpful for overcoming the difficulties of the
nonlinear switch: in defining a metric, begin by working
\emph{pointwise in space}. 
Then, depending on which phase that point belongs to, the conserved energy \eqref{eq:conserved_nonlinear_energy} determines a natural inner product to use.

    More precisely, this section proceeds as follows.
    We begin by defining a metric \emph{pointwise} which measures the discrepancy between two states $\mathcal{X}_1 (x)$ and $\mathcal{X}_2 (x)$ when evaluated at a specific point $x$ of the domain $\mathbb{T}^3$. This pointwise metric is defined to agree with the conserved energy density and so it is defined piecewise, taking different forms depending on the sign of $q(x)$.
    We then record the properties of this pointwise metric: it is well-defined, a metric, and agrees with the energy unconditionally. We also record how to compute this distance explicitly in a special case.
    We then make the key observation of this section: this pointwise metric does \emph{not} agree with the extraction of a balanced component, i.e. it does not agree with the decomposition.
    Finally we conclude this section by defining the induced metric on the state space $\mathbb{L}^2_\sigma$ obtained by integrating the pointwise metric over the domain and explain why we do not expect this metric to agree with the decomposition either.

	As mentioned above, a key idea in the definition of this metric is to work \emph{pointwise in space}.
	In particular, depending on which phase that point belongs to, the conserved energy \eqref{eq:conserved_nonlinear_energy} determines a natural inner product to use.
	We introduce appropriate notation with this in mind.
 
	\begin{definition}
	\label{def:pointwise_states_and_ips}
		For any $S_1 = (u_1,\,\theta_1,\,q_1) \in \mathbb{R}^5$ and $S_2 = (u_2,\,\theta_2,\,q_2) \in \mathbb{R}^5$,
		where $u_1,\,u_2 \in \mathbb{R}^3$, we define the inner products
		\begin{equation*}
			{\langle S_1 ,\, S_2 \rangle }_u \vcentcolon= u_1 \cdot u_2 + \theta_1 \theta_2 + q_1 q_2 + (\theta_1 + q_1)(\theta_2 + q_2)
		\end{equation*}
		and
		\begin{equation*}
			{\langle S_1 ,\, S_2 \rangle }_s \vcentcolon= u_1 \cdot u_2 + \theta_1 \theta_2 + (\theta_1 + q_1)(\theta_2 + q_2). 
		\end{equation*}
		We denote by $ \norm{ \,\cdot\, }{u} $ and $ \norm{ \,\cdot\, }{s} $ the corresponding norms.
	\end{definition}

    The inner products introduced in \fref{Definition}{def:pointwise_states_and_ips} are chosen precisely so that the ensuing norms $ \norm{ \,\cdot\, }{u} $ and $ \norm{ \,\cdot\, }{s} $ agree with the conserved energy density. Indeed, when $q < 0$ we have that
    \begin{equation*}
        \norm{(u,\,\theta,\,q)}{u}^2 = {\lvert u \rvert}^2 + \theta^2 + q^2 + {(\theta + q)}^2
    \end{equation*}
    and when $q \geqslant 0$ we have that
    \begin{equation*}
        \norm{(u,\,\theta,\,q)}{s }^2 = {\lvert u \rvert}^2 + \theta^2 + {(\theta + q)}^2,
    \end{equation*}
    such that indeed they agree with the conserved energy density as recorded in \eqref{eq:conserved_nonlinear_energy}.

	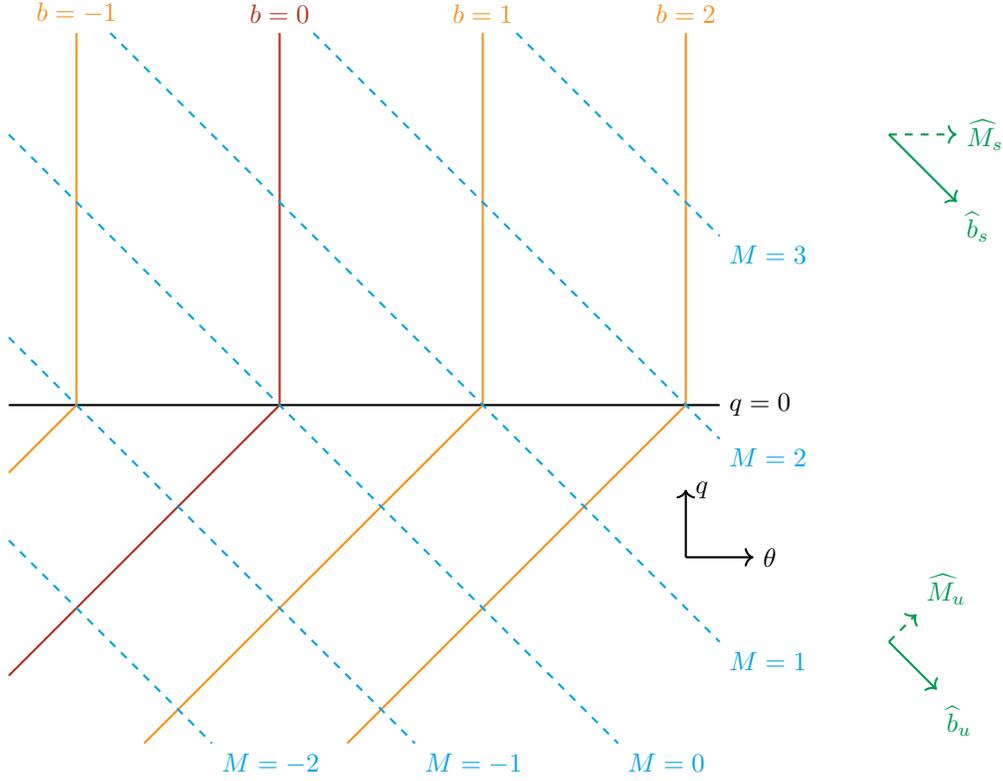
\begin{figure}
		\centering
		\captionsetup{width=0.85\textwidth}
        \begin{tikzpicture}[scale=0.9]
        	\draw[thick]		(-4, 0) --	( 6.5, 0);
        	\node[thick, right] at			( 6.5, 0) {$q = 0$};
        
        	\draw[thick, ->]	(6, -2.25) --	( 7, -2.25);
        	\node[thick, right] at	       		( 7, -2.25) {$\theta$};
        	\draw[thick, ->]	(6, -2.25) --	( 6, -1.25);
        	\node[thick, right] at	     	  	( 6, -1.25) {$q$};
        
        
        	\draw[thick, BrickRed] (-4, -4) --	( 0, 0);
        	\draw[thick, BrickRed] ( 0,  0) --	( 0, 5.5);
        	\node[thick, above, BrickRed] at	( 0, 5.5) {$b=0$};
        
        	\draw[thick, BurntOrange] (-4  , -1  ) --	(-3, 0);
        	\draw[thick, BurntOrange] (-3  ,  0  ) --	(-3, 5.5);
        	\node[thick, above, BurntOrange] at		(-3, 5.5) {$b=-1$};
        
        	\draw[thick, BurntOrange] (-2, -5) --	( 3, 0);
        	\draw[thick, BurntOrange] ( 3,  0) --	( 3, 5.5);
        	\node[thick, above, BurntOrange] at	( 3, 5.5) {$b=1$};
        
        	\draw[thick, BurntOrange] ( 1, -5) --	( 6, 0);
        	\draw[thick, BurntOrange] ( 6,  0) --	( 6, 5.5);
        	\node[thick, above, BurntOrange] at	( 6, 5.5) {$b=2$};
        
        
        	\draw[dashed, thick, Cerulean] ( 3.5,  5.5) --		( 6.5,  2.5);
        	\node[thick, Cerulean, below right] at		( 6.5,  2.5) {$M = 3$};
        
        	\draw[dashed, thick, Cerulean] ( 0.5,  5.5) -- 		( 6.5, -0.5);
        	\node[thick, Cerulean, below right] at		( 6.5, -0.5) {$M = 2$};
        
        	\draw[dashed, thick, Cerulean] (-2.5,  5.5) -- 		( 6.5, -3.5);
        	\node[thick, Cerulean, below right] at		( 6.5, -3.5) {$M = 1$};
        
        	\draw[dashed, thick, Cerulean] (-4  ,  4  ) --		( 5  , -5  );
        	\node[thick, Cerulean, below right] at		( 5  , -5  ) {$M = 0$};
        
        	\draw[dashed, thick, Cerulean] (-4  ,  1  ) --		( 2  , -5  );
        	\node[thick, Cerulean, below right] at		( 2  , -5  ) {$M = -1$};
        
        	\draw[dashed, thick, Cerulean] (-4  , -2  ) -- 		(-1  , -5  );
        	\node[thick, Cerulean, below right] at		(-1  , -5  ) {$M = -2$};
        
        	\draw[dashed, thick, ForestGreen, ->]	( 9  , 4  ) --	(10  , 4  );
        	\node[thick, ForestGreen, right] at		(10  , 4  ) {$\widehat{M}_s$};
        	\draw[thick, ForestGreen, ->]	( 9  , 4  ) --	(10  , 3  );
        	\node[thick, ForestGreen, below right] at	(10  , 3  ) {$\widehat{b}_s$};
        
        	\draw[dashed, thick, ForestGreen, ->]	( 9  ,-3.5) --	( 9.408,-3.092);
        	\node[thick, ForestGreen, above right] at	( 9.408,-3.092) {$\widehat{M}_u$};
        	\draw[thick, ForestGreen, ->]	( 9  ,-3.5) --	( 9.707,-4.207);
        	\node[thick, ForestGreen, below right] at	( 9.707,-4.297) {$\widehat{b}_u$};
        
        \end{tikzpicture}
		\caption{\small
			The vectors $ \left\{ \widehat{M}_s,\, \widehat{b}_s \right\} $ and $ \left\{ \widehat{M}_u,\, \widehat{b}_u \right\} $
			form orthonormal bases with respect to the inner products $ {\langle \,\cdot\,,\,\cdot\, \rangle}_{s} $ and $ {\langle \,\cdot\,,\,\cdot\, \rangle}_{u} $,
			respectively, as they were introduced in \fref{Definition}{def:pointwise_states_and_ips},
			namely $\widehat{M}_s = (0,\,1)$, $\widehat{b}_s = (1,\,-1)$, $ \widehat{M}_u = \frac{1}{\sqrt{6}} (1,\,1)$, and $\widehat{b}_u = \frac{1}{\sqrt{2}} (1,\,-1)$.
			Even though, especially in the saturated region $ \left\{ q > 0 \right\} $, labelling these vectors with the moist variable $M$ and the buoyancy $b$ may at first appear
			perplexing, this is justified by the fact that their corresponding \emph{dual} vectors are precisely (multiples of) $M$ and $b$.
			Indeed: $ {\langle \widehat{M}_s ,\, (\theta,\,q) \rangle }_{s} = \theta + q = M $, $ {\langle \widehat{b}_s ,\, (\theta,\,q) \rangle }_{s} = \theta = b_s$,
			$ {\langle \widehat{M}_u ,\, (\theta,\,q) \rangle }_{u} = \sqrt{\frac{3}{2}} (\theta + q) = \sqrt{\frac{3}{2}} M$,
			and $ {\langle \widehat{b}_u ,\, (\theta,\,q) \rangle }_{u} = \frac{1}{\sqrt{2}} \theta - q = \frac{1}{\sqrt{2}} b_u$.
		}
		\label{fig:snell}
	\end{figure}

    Next we wish to define the \emph{pointwise} metric. As alluded to at the start of this section, this metric will be defined \emph{piecewise} (depending on the sign of $q$) and so we introduce notation which describes each of these pieces.

	\begin{definition}
	\label{def:pointwise_state_space}
		We consider the following subsets of $ \mathbb{R}^5$.
		\begin{itemize}
			\item	$ \mathcal{U} \vcentcolon= \left\{ (u,\,\theta,\,q) \in \mathbb{R}^5 : q < 0 \right\}$ is the set of unsaturated states,
			\item	$ \mathcal{S} \vcentcolon= \left\{ (u,\,\theta,\,q) \in \mathbb{R}^5 : q \geqslant 0 \right\}$ is the set of saturated states, and
			\item	$ \mathcal{I} \vcentcolon= \left\{ (u,\,\theta,\,q) \in \mathbb{R}^5 : q = 0 \right\}$ is the state-space interface between unsaturated states and saturated states.
		\end{itemize}
	\end{definition}

	We have the necessary notation in hand to define the following \emph{Snell-type} metric.
	We refer to this metric as being of Snell-type since shortest paths between states in different phases exhibit the piecewise affine behaviour observed by light passing across a boundary
	(see also \fref{Figure}{fig:snell_refraction}).

	\begin{definition}[Snell-type metric, pointwise]
	\label{def:Snell_type_metric_pointwise}
		We define $d_\text{Snell} : \mathbb{R}^5 \to \mathbb{R}^5 \to [0,\,\infty)$ as follows.
		\begin{itemize}
			\item	If $S_1,\,S_2 \in \mathcal{U} $ then $ d_\text{Snell} (S_1,\,S_2) \vcentcolon= \norm{S_1 - S_2}{u}$.
			\item	If $S_1,\,S_2 \in \mathcal{S} $ then $ d_\text{Snell} (S_1,\,S_2) \vcentcolon= \norm{S_1 - S_2}{s}$.
			\item	If $S_1 \in \mathcal{S} $ and $S_2 \in \mathcal{U}$ then
				\begin{equation*}
					d_\text{Snell} \vcentcolon= \min_{I\in \mathcal{I} }\; \norm{S_1 - I}{s} + \norm{S_2 - I}{u}.
				\end{equation*}
			\item	If $S_1 \in \mathcal{U} $ and $S_2 \in \mathcal{S}$ then $ d_\text{Snell} (S_1,\,S_2) \vcentcolon= d_\text{Snell} (S_2,\,S_1)$.
		\end{itemize}
	\end{definition}

	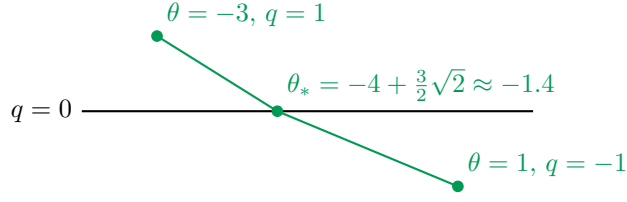
\begin{figure}
		\centering
		\captionsetup{width=0.85\textwidth}
        \begin{tikzpicture}
        	\draw[thick]					(-4  , 0  ) --	( 2  , 0  );
        	\node[thick, left] at				(-4  , 0  ) {$q = 0$};
        	
        	\draw[thick, ForestGreen]			(-3  , 1  ) --	(-1.4, 0  );
        	\draw[thick, ForestGreen]			(-1.4, 0  ) --	( 1  ,-1  );
        
        	\filldraw[ForestGreen]				(-3  , 1  ) circle (0.07);
        	\node[thick, ForestGreen, above right] at	(-3  , 1  ) {$\theta=-3,\,q=1$};
        
        	\filldraw[ForestGreen]				( 1  ,-1  ) circle (0.07);
        	\node[thick, ForestGreen, above right] at	( 1  ,-1  ) {$\theta=1,\,q=-1$};
        
        	\filldraw[ForestGreen]				(-1.4, 0  ) circle (0.07);
        	\node[thick, ForestGreen, above right] at	(-1.4, 0  ) {$\theta_*=-4+\frac{3}{2}\sqrt{2} \approx -1.4$};
        
        \end{tikzpicture}	
		\caption{\small
			As introduced in \fref{Definition}{def:Snell_type_metric_pointwise} we see that here that the ``Snell-type'' metric warrants its name.
			Shown here is the geodesic between two points in the $(\theta,\,q)$ plane which live on different sides of the state-space interface $ \left\{ q = 0 \right\} $.
			As is familiar from Snell's refraction law of optics, the geodesic intersects the interface at an angle.
			Note that we may use \fref{Lemma}{lemma:explicit_computation_of_intersection_point} in order to compute the intersection point and the value of $\theta_*$ exactly.
		}
		\label{fig:snell_refraction}
	\end{figure}

	First we verify that this distance function is well-defined.

	\begin{lemma}[The Snell-type distance is well-defined]
	\label{lemma:d_snell_is_well_defined}
		For any $S_1\in \mathcal{S} $ and any $S_2\in \mathcal{U} $ there is a unique minimizer of
		\begin{equation*}
			\min_{I\in \mathcal{I} } \norm{S_1 - I}{s} + \norm{S_2 - I}{u}.
		\end{equation*}
	\end{lemma}
	\begin{proof}
		To verify that $ d_\text{Snell} $ is well-defined it suffices to verify that, for any $S_1\in \mathcal{S} $ and any $S_2\in \mathcal{U} $,
		there exists a unique minimizer of
		\begin{equation}
			\min_{I\in \mathcal{I} } \norm{S_1 - I}{s} + \norm{S_2 - I}{u}.
		\label{eq:s_snell_well_def_min}
		\end{equation}
		To do so, fix $S\in \mathbb{R}^5 \setminus \mathcal{I} $ and consider $f : \mathbb{R}^4 \to \mathbb{R}^5$ defined by, for every $I = (u_h,\,w,\,\theta,\,0) \in \mathcal{I} $,
		\begin{equation*}
			f(u_h,\,w,\,\theta) = \norm{S-I}{}
		\end{equation*}
		where $ \norm{ \,\cdot\, }{}$ denotes $ \norm{ \,\cdot\, }{s} $ or $ \norm{ \,\cdot\, }{u}$.
		To verify that \eqref{eq:s_snell_well_def_min} has a unique minimizer it suffices to show that $f$ is strictly convex.
		Since $f$ is smooth (as $S$ is not in $ \mathcal{I} $, and $f$ and its derivatives would blow up precisely when $S = I$, but $I\in \mathcal{I} $)
		it suffices to verify that
		\begin{equation*}
			\nabla^2 f > 0.
		\end{equation*}
		So we compute:
		\begin{equation*}
			\partial_i f
			= \partial_i \norm{S-I}{}
			= \frac{ \langle I-S,\, \partial_i I \rangle}{ \norm{S-I}{}}
		\end{equation*}
		and hence
		\begin{equation*}
			\partial_i \partial_j f
			= \partial_j \left(
				\frac{ \langle I-S,\, \partial_i I \rangle}{ \norm{S-I}{}}
			\right)
			= \frac{
				\langle \partial_i I,\, \partial_j I \rangle {\norm{S-I}{}}^2 - \langle \partial_i I,\, S-I \rangle \langle \partial_j I,\, S-I \rangle
			}{
				\norm{S-I}{}^3
			}.
		\end{equation*}

		To make the remaining computation easier to follow, let us write
		\begin{equation*}
			x^i \vcentcolon= \partial_i I
			\text{ and } 
			z \vcentcolon= S-I.
		\end{equation*}
		Then
		\begin{equation*}
			\norm{z}{}^3 \nabla^2 f
			= \langle x^i,\, x^j \rangle \norm{z}{}^2 - \langle x^i,\, z \rangle \langle x^j,\, z \rangle
			=\vcentcolon A
		\end{equation*}
		and so it suffices to show that $A > 0$ (since $ \norm{z}{} \neq 0$ as long as $S\notin \mathcal{I} $).
		To simplify further, let us introduce
		\begin{equation*}
			D \vcentcolon= \text{diag}\, { \left( \norm{x^i}{} \right) }_{i=1}^4.
		\end{equation*}
		Then
		\begin{equation*}
			A > 0 \iff D^T A D > 0
		\end{equation*}
		where, since ${ \left( x^i \right)}_{i=1}^4$ forms an orthogonal set with respect to both $ \norm{ \,\cdot\, }{s} $ and $ \norm{ \,\cdot\, }{u} $,
		we have that $D^T A D = I \norm{z}{}^2 - B$
		for
		\begin{equation*}
			B_{ij} \vcentcolon= \langle \hat{x}^i,\, z \rangle \langle \hat{x}^j,\, z \rangle.
		\end{equation*}
		To conclude the proof it thus suffices to prove that $B < \norm{z}{}^2 I$.
		First we show that $B \leqslant \norm{z}{}^2 I$.
		Well, for any $y\in \mathbb{R}^4$, 
		\begin{equation*}
			By\cdot y
			= y_i \langle \hat{x}^i,\, z \rangle y_j \langle x^j,\, x \rangle
			= y \cdot { \left( \langle \hat{x}^i,\, z \rangle \right)}_{i=1}^4
			\leqslant {\lvert y \rvert}^2 \left( \sum_{i=1}^4 { \langle \hat{x}^i,\, z \rangle}^2 \right).
		\end{equation*}
		Since $ {\left( \hat{x}^i \right)}_{i=1}^4$ is an orthonormal set, we have that
		\begin{equation*}
			\sum_{i=1}^4 { \langle \hat{x}^i,\, z \rangle}^2 \leqslant \norm{z}{}^2,
		\end{equation*}
		which verifies that $B \leqslant \norm{z}{}^2 I$.
		Finally we show that $B < \norm{z}{}^2 I$.
		This follows from the fact that $ { \left( \hat{x}^i \right)}_{i=1}^4$ span $ \mathcal{I} $, but not $ \mathbb{R}^5$, and
		so
		\begin{equation*}
			\sum_{i=1}^4 { \langle \hat{x}^i,\, z \rangle}^2
			= \norm{ \text{proj}_{ \mathcal{I} } z}{}^2
			< \norm{z}{}^2
		\end{equation*}
		since $z = S-I \notin \mathcal{I} $ since $I\in \mathcal{I} $ but $S\notin \mathcal{I} $.
		In other words: since $S\notin \mathcal{I} $ we have that $B < \norm{z}{}^2 I$, and so $\nabla^2 f > 0$.
	\end{proof}

	Now we verify that the distance function introduced in \fref{Definition}{def:Snell_type_metric_pointwise} is indeed a metric.
	We will also show that this Snell-type distance also agrees with the energy, when they are comparable.

	\begin{lemma}[The Snell-type distance is a metric]
	\label{lemma:d_snell_is_metric_and_compatible_with_energy}
		As introduced in \fref{Definition}{def:Snell_type_metric_pointwise}, $ d_\text{Snell} $ is a metric on $ \mathbb{R}^5$.
		Moreover it agrees with the energy density in the sense that, for any $S_1,\,S_2\in \mathbb{R}^5$, if the signs of $q_1$ and $q_2$ agree
		(i.e. if the states $S_1$ and $S_2$ are in the same phase) then
		\begin{equation*}
			{d_\text{Snell} (S_1,\,S_2)}^2 = {\lvert u_1 - u_2 \rvert}^2 + {(\theta_1 - \theta_2)}^2 + {(q_1 - q_2)}^2 h + {((\theta_1 + q_1) - (\theta_2 - q_2))}^2
		\end{equation*}
		where
		\begin{equation*}
			h = \left\{
			\begin{aligned}
				&1 &&\text{ if } q_1,\,q_2 < 0 \text{ and } \\
				&0 &&\text{ if } q_1,\,q_2 \geqslant 0.
			\end{aligned}
			\right.
		\end{equation*}
	\end{lemma}
	\begin{proof}
		Positive-definiteness and symmetry are immediate so we only need to prove that the triangle inequality holds.
		Indeed: the agreement with the energy follows immediately from the definition of $ \norm{ \,\cdot\, }{u} $ and $ \norm{ \,\cdot\, }{s} $ in \fref{Definition}{def:pointwise_states_and_ips}.

		For $S_1,\,S_2,\,S_3 \in \mathbb{R}^5$ we therefore wish to show that
		\begin{equation*}
			d_\text{Snell} (S_1,\, S_2) \leqslant d_\text{Snell} (S_1,\,S_3) + d_\text{Snell} (S_3,\, S_2).
		\end{equation*}
		Clearly the triangle inequality holds when all three states belong to the same phase (either $ \mathcal{U} $ or $ \mathcal{S} $).
		So we need only consider the case where two of the states $S_1$, $S_2$, and $S_3$ are in one phase
		and the last state is in the other.
		Since all that matters is whether or not the intermediate point $S_3$ is alone in its phase, there are without loss of generality two cases to consider.
		\begin{itemize}
			\item	Case 1: $S_1,\,S_2 \in \mathcal{S} $ and $S_3 \in \mathcal{U} $.
				Let us denote by $I_{13}$ the point intersecting the geodesic between $S_1$ and $S_3$ and the state-space interface $ \mathcal{I} $,
				and similarly $I_{23}$ for the point arising from the geodesic between $S_2$ and $S_3$.
				We use the triangle inequality in $ \mathcal{U} $, then the triangle inequality in $S$:
				\begin{align*}
					d_\text{Snell} (S_1,\,S_3) \,+ &\,d_\text{Snell} (S_3,\, S_2)
				\\
					&= d_\text{Snell} (S_1,\,I_{13}) + d_\text{Snell} (I_{13},\, S_3) + d_\text{Snell} (S_3,\, I_{23}) + d_\text{Snell} (I_{23},\, S_2)
				\\
					&\geqslant d_\text{Snell} (S_1,\,I_{13}) + d_\text{Snell} (I_{13},\, I_{23}) + d_\text{Snell} (I_{23},\, S_2)
				\\
					&\geqslant d_\text{Snell} (S_1,\,S_2).
				\end{align*}

				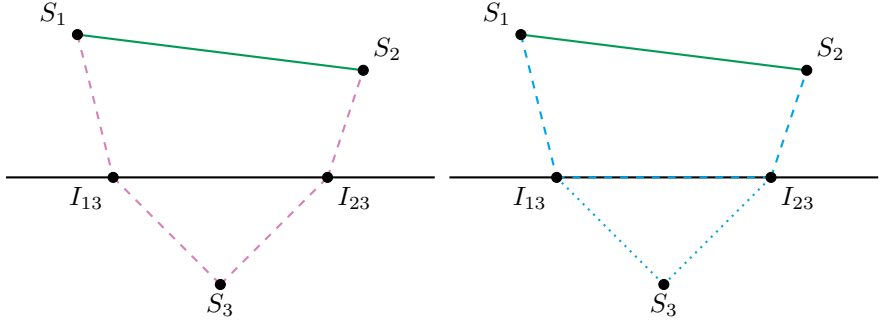
\begin{figure}
					\centering
					\captionsetup{width=0.85\textwidth}
					\begin{subfigure}{0.35\textwidth}
                        \begin{tikzpicture}[scale=0.95]
                        	\draw[thick]				(-3  , 0  ) --	( 3  , 0  );
                        
                        	\draw[thick, ForestGreen]		(-2  , 2  ) --	( 2  , 1.5);
                        
                        	\draw[thick, dashed, Thistle]			(-2  , 2  ) --	(-1.5, 0  );
                        
                        	\draw[thick, dashed, Thistle]			(-1.5, 0  ) --	( 0  ,-1.5);
                        
                        	\draw[thick, dashed, Thistle]			( 0  ,-1.5) --	( 1.5, 0  );
                        
                        	\draw[thick, dashed, Thistle]			( 1.5, 0  ) --	( 2  , 1.5);

                        	\filldraw				(-2  , 2  ) circle (0.07);
                        	\node[thick, above left] (S1) at	(-2  , 2  ) {$S_1$};
                        
                        	\filldraw				( 2  , 1.5) circle (0.07);
                        	\node[thick, above right] (S2) at	( 2  , 1.5) {$S_2$};
                        
                        	\filldraw				( 0  ,-1.5) circle (0.07);
                        	\node[thick, below] (S3) at		( 0  ,-1.5) {$S_3$};
                        
                        	\filldraw				(-1.5, 0  ) circle (0.07);
                        	\node[thick, below left] (I13) at	(-1.5, 0  ) {$I_{13}$};
                        
                        	\filldraw				( 1.5, 0  ) circle (0.07);
                        	\node[thick, below right] (I23) at	( 1.5, 0  ) {$I_{23}$};
                        
                        \end{tikzpicture}
					\end{subfigure}
					\begin{subfigure}{0.35\textwidth}
                        \begin{tikzpicture}[scale=0.95]
                        	\draw[thick]				(-3  , 0  ) --	( 3  , 0  );
                        
                        	\draw[thick, ForestGreen]		(-2  , 2  ) --	( 2  , 1.5);
                        
                        	\draw[thick, dashed, Cerulean]			(-2  , 2  ) --	(-1.5, 0  );
                        
                        	\draw[thick, dashed, Cerulean, dotted]		(-1.5, 0  ) --	( 0  ,-1.5);
                        
                        	\draw[thick, dashed, Cerulean, dotted]		( 0  ,-1.5) --	( 1.5, 0  );
                        
                        	\draw[thick, dashed, Cerulean]			( 1.5, 0  ) --	( 2  , 1.5);

                        	\draw[thick, dashed, Cerulean]			(-1.5, 0  ) --	( 1.5, 0  );

                        	\filldraw				(-2  , 2  ) circle (0.07);
                        	\node[thick, above left] (S1) at	(-2  , 2  ) {$S_1$};
                        
                        	\filldraw				( 2  , 1.5) circle (0.07);
                        	\node[thick, above right] (S2) at	( 2  , 1.5) {$S_2$};
                        
                        	\filldraw				( 0  ,-1.5) circle (0.07);
                        	\node[thick, below] (S3) at		( 0  ,-1.5) {$S_3$};
                        
                        	\filldraw				(-1.5, 0  ) circle (0.07);
                        	\node[thick, below left] (I13) at	(-1.5, 0  ) {$I_{13}$};
                        
                        	\filldraw				( 1.5, 0  ) circle (0.07);
                        	\node[thick, below right] (I23) at	( 1.5, 0  ) {$I_{23}$};
                        
                        \end{tikzpicture}
					\end{subfigure}
					\caption{\small A pictorial depiction of the two steps used in case 1 of the proof of \fref{Lemma}{lemma:d_snell_is_metric_and_compatible_with_energy}.}
					\label{fig:snell_proof_case1}
				\end{figure}

			\item	Case 2: $S_1,\,S_3 \in \mathcal{S} $ and $S_2 \in \mathcal{U} $.
				Let $I_{12}$ and $I_{23}$ be defined similarly to in case 1.
				We use the triangle inequality in $ \mathcal{S} $ and the definition of $ d_\text{Snell} $ as a minimization problem when crossing the interface to deduce that
				\begin{align*}
					d_\text{Snell} (S_1,\, S_3) + d_\text{Snell} (S_3,\, S_2)
					&= d_\text{Snell} (S_1,\,\ S_3) + d_\text{Snell} (S_3,\, I_{23}) + d_\text{Snell} (I_{23},\, S_2)
				\\
					&\geqslant d_\text{Snell} (S_1,\, I_{23}) + d_\text{Snell} (I_{23},\, S_2)
				\\
					&\geqslant d_\text{Snell} (S_1,\,S_2).
				\end{align*}

				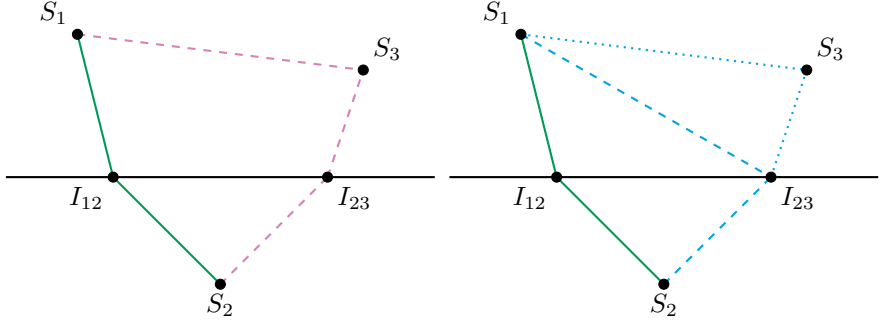
\begin{figure}
					\centering
					\captionsetup{width=0.85\textwidth}
					\begin{subfigure}{0.35\textwidth}
                        \begin{tikzpicture}[scale=0.95]
                        	\draw[thick]				(-3  , 0  ) --	( 3  , 0  );
                        
                        	\draw[thick, dashed, Thistle]			(-2  , 2  ) --	( 2  , 1.5);
                        
                        	\draw[thick, ForestGreen]		(-2  , 2  ) --	(-1.5, 0  );
                        
                        	\draw[thick, ForestGreen]		(-1.5, 0  ) --	( 0  ,-1.5);
                        
                        	\draw[thick, dashed, Thistle]			( 0  ,-1.5) --	( 1.5, 0  );
                        
                        	\draw[thick, dashed, Thistle]			( 1.5, 0  ) --	( 2  , 1.5);

                        	\filldraw				(-2  , 2  ) circle (0.07);
                        	\node[thick, above left] (S1) at	(-2  , 2  ) {$S_1$};
                        
                        	\filldraw				( 2  , 1.5) circle (0.07);
                        	\node[thick, above right] (S2) at	( 2  , 1.5) {$S_3$};
                        
                        	\filldraw				( 0  ,-1.5) circle (0.07);
                        	\node[thick, below] (S3) at		( 0  ,-1.5) {$S_2$};
                        
                        	\filldraw				(-1.5, 0  ) circle (0.07);
                        	\node[thick, below left] (I13) at	(-1.5, 0  ) {$I_{12}$};
                        
                        	\filldraw				( 1.5, 0  ) circle (0.07);
                        	\node[thick, below right] (I23) at	( 1.5, 0  ) {$I_{23}$};
                        
                        \end{tikzpicture}
					\end{subfigure}
					\begin{subfigure}{0.35\textwidth}
                        \begin{tikzpicture}[scale=0.95]
                        	\draw[thick]				(-3  , 0  ) --	( 3  , 0  );
                        
                        	\draw[thick, dashed, Cerulean, dotted]		(-2  , 2  ) --	( 2  , 1.5);
                        
                        	\draw[thick, ForestGreen]		(-2  , 2  ) --	(-1.5, 0  );
                        
                        	\draw[thick, ForestGreen]		(-1.5, 0  ) --	( 0  ,-1.5);
                        
                        	\draw[thick, dashed, Cerulean]			( 0  ,-1.5) --	( 1.5, 0  );
                        
                        	\draw[thick, dashed, Cerulean, dotted]		( 1.5, 0  ) --	( 2  , 1.5);

                        	\draw[thick, dashed, Cerulean]			(-2  , 2  ) --	( 1.5, 0  );

                        	\filldraw				(-2  , 2  ) circle (0.07);
                        	\node[thick, above left] (S1) at	(-2  , 2  ) {$S_1$};
                        
                        	\filldraw				( 2  , 1.5) circle (0.07);
                        	\node[thick, above right] (S2) at	( 2  , 1.5) {$S_3$};
                        
                        	\filldraw				( 0  ,-1.5) circle (0.07);
                        	\node[thick, below] (S3) at		( 0  ,-1.5) {$S_2$};
                        
                        	\filldraw				(-1.5, 0  ) circle (0.07);
                        	\node[thick, below left] (I13) at	(-1.5, 0  ) {$I_{12}$};
                        
                        	\filldraw				( 1.5, 0  ) circle (0.07);
                        	\node[thick, below right] (I23) at	( 1.5, 0  ) {$I_{23}$};
                        
                        \end{tikzpicture}
					\end{subfigure}
					\caption{\small A pictorial depiction of the two steps used in case 2 of the proof of \fref{Lemma}{lemma:d_snell_is_metric_and_compatible_with_energy}.}
					\label{fig:snell_proof_case2}
				\end{figure}

		\end{itemize}
	\end{proof}

	We now seek to determine whether or not this metric, once ported over to measure the distance between spatially-varying states,
	agrees with the balanced-unbalanced decomposition of \fref{Theorem}{theorem:pre_decomp}.

	To do so we will need to understand this Snell-type metric in the pointwise case (which we have discussed so far) in more detail.
	In particular we will rely on the explicit computation of geodesic trajectories in the case where $u_1 = u_2 = 0$.
	In that case, the states live in the $(\theta,\,q)$--plane.
	We consider two states in that plane living on either side of the state-space interface $ \mathcal{I} $, i.e. where one state is saturated and the other is not.
	The geodesic between these two points is then piecewise affine, changing direction when it reaches the interface.
	Crucially: in this simplified setting where $u_1 = u_2 = 0$ we can compute \emph{explicitly} where the geodesic intersects the interface.

	\begin{lemma}[Explicit computation of the intersection point]
	\label{lemma:explicit_computation_of_intersection_point}
		For any $S_1 \in \mathcal{S} $ and any $S_2 \in \mathcal{U} $, if $u_1 = u_2 = 0$ then, for $I(\theta) \vcentcolon= (0,\,0,\,0,\,\theta,\,0) \in \mathcal{I} \subseteq \mathbb{R}^5$,
		\begin{equation*}
			\argmin_{\theta\in \mathbb{R}} \norm{S_1 - I(\theta)}{s} + \norm{S_2 - I(\theta)}{u}
			\in \left\{ \frac{\alpha - \sqrt{\delta}}{\beta},\, \frac{\alpha + \sqrt{\delta}}{\beta} \right\}
		\end{equation*}
		for
		\begin{align*}
			\alpha &= 4q_1^2 (2\theta_2 + q_2) - 6q_2^2 (2\theta_1 + q_1),\\
			\beta &= 8q_1^2 - 24q_2^2, \text{ and } \\
			\delta &= 12q_1^2 q_2^2 {(2\theta_1 + q_1 + 4\theta_2 + 2q_2)}^2.
		\end{align*}
	\end{lemma}
	\begin{proof}
		The key idea is that the first derivative condition for minimizers may be squared to produce a fourth-order polynomial equation.
		Solving this equation using symbolic computation software then produces the claim.
		More precisely, let us define $f: \mathbb{R}\to \mathbb{R}$ via
		\begin{equation*}
			f(\theta) \vcentcolon= \norm{S_1 - I(\theta)}{s} + \norm{S_2 - I(\theta)}{u}.
		\end{equation*}
		Then, proceeding as in the proof of \fref{Lemma}{lemma:d_snell_is_well_defined} we see that
		\begin{equation*}
			f'(\theta)
			= \frac{ {\langle I - S_1,\, \partial_\theta I \rangle}_s }{ \norm{S_1 - I}{s} }
			+ \frac{ {\langle I - S_2,\, \partial_\theta I \rangle}_u }{ \norm{S_2 - I}{u} }
		\end{equation*}
		and so, since $\partial_\theta I = e_\theta$, we see that if $f'(\theta) = 0$ then
		\begin{equation}
			{\langle I - S_2,\, e_\theta \rangle }_{u}^2 \norm{S_1 - I}{s}^2
			= {\langle I - S_1,\, e_\theta \rangle }_{s}^2 \norm{S_2 - I}{u}^2.
		\label{eq:explicity_computation_of_intersection_point_key_eq}
		\end{equation}
		Since
		\begin{equation*}
			{\langle I - S_i,\, e_\theta \rangle}_u
			= {\langle I - S_i,\, e_\theta \rangle}_u
			= 2\theta - 2\theta_i - q_i
		\end{equation*}
		while
		\begin{align*}
			\norm{S_1 - I}{s}^2 &= {(\theta - \theta_1)}^2 + {(\theta - \theta_1 - q_1)}^2 \text{ and } \\
			\norm{S_2 - I}{u}^2 &= \frac{1}{2}  {(\theta - \theta_2 + q_2)}^2 + \frac{3}{2} {(\theta - \theta_2 - q_2)}^2,
		\end{align*}
		where we have used the identity
		\begin{equation*}
			\theta^2 +q^2 h = (1-h/2) {(\theta - qh})^2 + (h/2) {(\theta + q)}^2,
		\end{equation*}
		the equation \eqref{eq:explicity_computation_of_intersection_point_key_eq} comes down to a polynomial equation of order four for $\theta$.
		We may solve this equation explicitly using symbolic computation software, thus producing the claim.
	\end{proof}

	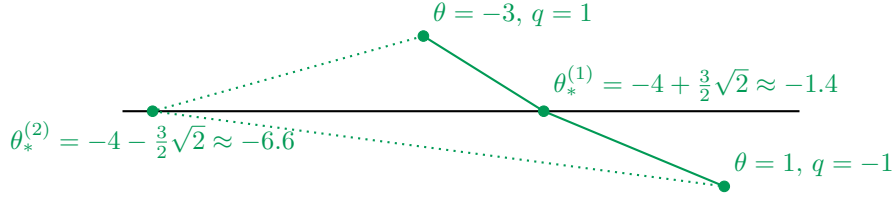
\begin{figure}
		\centering
		\captionsetup{width=0.85\textwidth}
        \begin{tikzpicture}
        	\draw[thick]					(-7  , 0  ) --	( 2  , 0  );
        	
        	\draw[thick, ForestGreen]			(-3  , 1  ) --	(-1.4, 0  );
        	\draw[thick, ForestGreen]			(-1.4, 0  ) --	( 1  ,-1  );
        
        	\draw[thick, ForestGreen, dotted]		(-3  , 1  ) --	(-6.6, 0  );
        	\draw[thick, ForestGreen, dotted]		(-6.6, 0  ) --	( 1  ,-1  );
        
        	\filldraw[ForestGreen]				(-3  , 1  ) circle (0.07);
        	\node[thick, ForestGreen, above right] at	(-3  , 1  ) {$\theta=-3,\,q=1$};
        
        	\filldraw[ForestGreen]				( 1  ,-1  ) circle (0.07);
        	\node[thick, ForestGreen, above right] at	( 1  ,-1  ) {$\theta=1,\,q=-1$};
        
        	\filldraw[ForestGreen]				(-1.4, 0  ) circle (0.07);
        	\node[thick, ForestGreen, above right] at	(-1.4, 0  ) {$\theta_*^{(1)} = -4 + \frac{3}{2}\sqrt{2} \approx -1.4$};
        
        	\filldraw[ForestGreen]				(-6.6, 0  ) circle (0.07);
        	\node[thick, ForestGreen, below] at	(-6.6, 0  ) {$\theta_*^{(2)} = -4 - \frac{3}{2}\sqrt{2} \approx -6.6$};
        
        \end{tikzpicture}
		\caption{\small
			This shows an example where the two possible intersection points are computed from \fref{Lemma}{lemma:explicit_computation_of_intersection_point}.
			In practice, it is very straightforward to identify the correct intersection point:
			the lengths of both paths may readily be computed, and the intersection point corresponding to the shortest path is then the correct one.
		}
		\label{fig:snell_possible_intersection_points}
	\end{figure}

	We now discuss how to extract balanced components in this pointwise setting where the state space is $ \mathbb{R}^5$ (and not $ \mathbb{L}^2_\sigma $).
	The system of interest is
	\begin{subnumcases}{}
		\frac{du_h}{dt} + \frac{1}{\varepsilon} u_h^\perp = 0,\\
		\frac{dw}{dt} + \frac{1}{\varepsilon} (\theta - {\min}_0\, q) = 0,\\
		\frac{d\theta}{dt} + \frac{1}{\varepsilon} w = 0, \text{ and } \\
		\frac{dq}{dt} - \frac{1}{\varepsilon} w = 0,
	\end{subnumcases}
	which fits into our generic framework as it can be written in the form, for $S = (u_h,\,w,\,\theta,\,q)$, as
	\begin{equation*}
		\frac{dS}{dt} + \frac{1}{\varepsilon} N(S) = 0
		\text{ for } N(S) = \begin{pmatrix}
			u_h^\perp \\ \theta - {\min}_0\, q \\ w \\ -w
		\end{pmatrix}.
	\end{equation*}
	As detailed in \fref{Section}{sec:decomposition}, the first step in identifying an appropriate decomposition for this system is to study the kernel / zero set of $N$ and the image of $N$.
	We see that the kernel of $N$ may be characterized (akin to \fref{Proposition}{prop:alt_char_balanced_set}) as a zero set of appropriate measurements or may be explictly parameterised:
	\begin{align*}
		N(S) = 0
		&\iff S \in \ker N =\vcentcolon B
	\\
		&\iff u_h = 0,\, w= 0, \text{ and } b(S) \vcentcolon= \theta - {\min}_0\, q = 0
	\\
		&\iff S = \begin{pmatrix}
			0 \\ 0 \\ \frac{1}{2} {\min}_0\, M \\ M - \frac{1}{2} {\min}_0\, M
		\end{pmatrix}
	\end{align*}
	while, akin to \fref{Proposition}{prop:chara_im_N}, the image of $N$ may be characterised as follows:
	\begin{equation*}
		S \in \im N =\vcentcolon W
		\iff M(S) \vcentcolon= \theta + q = 0
		\iff S = \begin{pmatrix}
			u_h \\ w \\ \eta \\ -\eta
		\end{pmatrix}.
	\end{equation*}
	This leads to the following decomposition of $ \mathbb{R}^5$:
	\begin{equation}
		S = \begin{pmatrix}
			u_h \\ w \\ \theta \\ q
		\end{pmatrix} = \underbrace{\begin{pmatrix}
			0 \\ 0 \\ \frac{1}{2} {\min}_0\, M \\ M - \frac{1}{2} {\min}_0\, M
		\end{pmatrix}}_{\in \ker N = B} + \underbrace{\begin{pmatrix}
			u_h \\ w \\ \eta \\ -\eta
		\end{pmatrix}}_{\in \im N = W}.
	\label{eq:ODE_decomp}
	\end{equation}

	Now let us restrict our attention once again to the $(\theta,\,q)$--plane.
	We seek to determine whether or not the Snell-type metric agrees with the decomposition \eqref{eq:ODE_decomp}.
	In the $(\theta,\,q)$--plane, the balanced component takes the form
	\begin{equation*}
		\begin{pmatrix}
			\theta_B \\ q_B
		\end{pmatrix} = \begin{pmatrix}
			\frac{1}{2} {\min}_0\, M \\ M - \frac{1}{2} {\min}_0\, M
		\end{pmatrix}.
	\end{equation*}
	It turns out that this decomposition \emph{fails to agree} with the Snell-type energy.

	\begin{prop}[The Snell-type distance does not agree with the pointwise decomposition]
	\label{prop:snell_type_dist_does_not_agree_with_pointwise_decomp}
		As introduced in \fref{Definition}{def:Snell_type_metric_pointwise}, $ d_\text{Snell} $ does \emph{not} agree with the extraction of a balanced component,
		in the sense that there exists $ \mathcal{S} \in \mathbb{R}^5$ and $T_* \in B$ such that
		\begin{equation*}
			d_\text{Snell} (S,\,T_*) < d_\text{Snell} (S,\, S_B).
		\end{equation*}
		In other words:
		\begin{equation*}
			S_B \neq \argmin_{T\in B} d_\text{Snell} (S,\,T).
		\end{equation*}
	\end{prop}
	\begin{proof}
		Choose $S = (-3,\,1)$ such that $M = -3+1 = -2$ and so its balanced component is $S_B = (-1,\,-1)$
		and choose $T_* = (-1.1,\,-1.1)$.
		Note that $S$ is saturated whereas $S_B$ and $T_*$ are unsaturated. And, of course, both $S_B$ and $T_*$ belong to the balanced set $B = \ker N$ since they have vanishing $M$.

		On the one hand, equipped with \fref{Lemma}{lemma:explicit_computation_of_intersection_point} we see that
		\begin{equation*}
			d_\text{Snell} (S,\, S_B)
			= \norm{S-I_1}{s} + \norm{S_B - I_1}{u}
		\end{equation*}
		where
		\begin{equation*}
			I_1 \vcentcolon= I(\theta_1)
			\text{ for }
			\theta_1 = -3 + \frac{\sqrt{3}}{2}
		\end{equation*}
		and hence
		\begin{equation*}
			d_\text{Snell} (S,\, S_B)
			= \sqrt{\frac{5}{2} - \sqrt{3}} (1 +\sqrt{3}) > 2.394.
		\end{equation*}
		On the other hand, for $I_2 \vcentcolon= I(\theta_2)$ where $\theta_2 = \frac{11}{526} ( -135 + 17\sqrt{3})$,
		\begin{equation*}
			d_\text{Snell} (S,\,T_*)
			= \norm{S-I_2}{s} + \norm{T_*-I_2}{u}
			= (\sqrt{101488 - 31790 \sqrt{3}} (10 + 11 \sqrt{3}))/2630
			< 2.381.
		\end{equation*}
		Since
		\begin{equation*}
			d_\text{Snell} (S,\,T_*) < d_\text{Snell} (S,\,S_B)
		\end{equation*}
		this proves the claim.
	\end{proof}

	\begin{figure}
		\centering
		\captionsetup{width=0.85\textwidth}
        \begin{tikzpicture}[scale=2]
        	\draw[thick]				(-3.5    , 0    ) --	( 1.5  , 0     );
        	\node[thick, right] at			         		( 1.5  , 0     ) {$q=0$};
        
        	\draw[thick, Cerulean]			(-2.475,-2.475) --	( 0    , 0     );
        	\draw[thick, Cerulean]			( 0    , 0    ) --	( 0    , 1.5   );
        	\node[thick, Cerulean, above right] at	                 	( 0    , 1.5   ) {$b=0$};

        	\draw[thick, dashed, Thistle]			(-3    , 1    ) --	(-2.333, 0    );
        
        	\draw[thick, dashed, Thistle]			(-2.333, 0    ) --	(-1    ,-1    );

        	\draw[thick, ForestGreen]		(-3    , 1    ) --	(-2.667, 0    );
        
        	\draw[thick, ForestGreen]		(-2.667, 0    ) --	(-1.5  ,-1.5  );

        	\draw[thick, Cerulean, dotted]		(-3    , 1    ) --	(-1    ,-1    );

        	\filldraw				(-3    , 1    ) circle (0.04);
        	\node[thick, above left] at		(-3    , 1    ) {$S$};

        	\filldraw				(-1    ,-1    ) circle (0.04);
        	\node[thick, right] at			(-1    ,-1    ) {$S_B$};
        
        	\filldraw				(-1.5  ,-1.5  ) circle (0.04);
        	\node[thick, below] at			(-1.5  ,-1.5  ) {$T_*$};

        	\filldraw				(-2.333, 0    ) circle (0.04);
        	\node[thick, above right] at		(-2.333, 0    ) {$I_1$};
        
        	\filldraw				(-2.667, 0    ) circle (0.04);
        	\node[thick, below left] at		(-2.667, 0    ) {$I_2$};
        
        \end{tikzpicture}
		\caption{\small
			A graphical depiction of \fref{Proposition}{prop:snell_type_dist_does_not_agree_with_pointwise_decomp}.
			Note that the location of the points is not exact: this is an artistic rendering of the setup of the proof of
			\fref{Proposition}{prop:snell_type_dist_does_not_agree_with_pointwise_decomp} used because it makes it easier to see what is going on.
		}
		\label{fig:snell_counterexample}
	\end{figure}
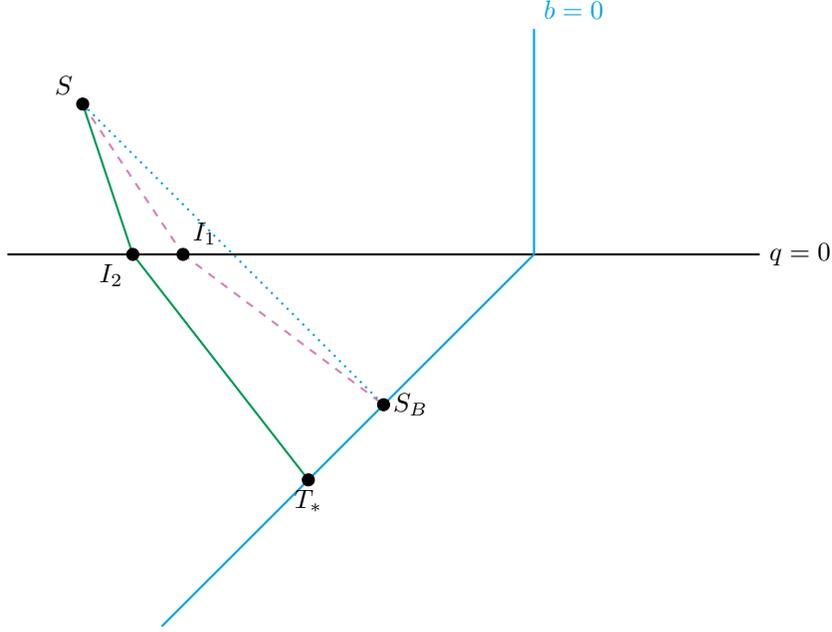

	We conclude this section by extending the Snell-type distance discussed so far \emph{pointwise} to a distance between spatially-varying states.

	\begin{definition}[Snell-type metric, spatially-varying states]
	\label{def:Snell_type_metric_spatially_varying_states}
		We define $d_S : \mathbb{L}^2_\sigma \times \mathbb{L}^2_\sigma \to [0,\,\infty)$ as follows.
		For any $ \mathcal{X},\, \mathcal{Y} \in \mathbb{L}^2_\sigma $,
		\begin{equation*}
			d_S ( \mathcal{X},\, \mathcal{Y} )
			\vcentcolon= \norm{ d_\text{Snell} ( \mathcal{X},\, \mathcal{Y} )}{L^2}
			= { \left( \int d_\text{Snell}^2 ( \mathcal{X}(x),\, \mathcal{Y}(x) ) dx \right)}^{1/2}.
		\end{equation*}
	\end{definition}

	This distance is indeed a metric, and it agrees with the conserved energy when comparable to it.

	\begin{prop}
	\label{prop:snell_type_dist_is_metric_comparable_to_energy}
		As introduced in \fref{Definition}{def:Snell_type_metric_spatially_varying_states}, $d_S$ is a metric on $ \mathbb{L}^2_\sigma $ which agrees with the energy in the sense that,
		if $ \mathcal{X}_1,\, \mathcal{X}_2 \in \mathbb{L}^2_\sigma $ have the same phases, meaning that the signs of $q_1$ and $q_2$ agree everywhere,
		then
		\begin{align*}
			{( d_S ( \mathcal{X}_1,\, \mathcal{X}_2 ) )}^2
			&= E_H ( \mathcal{X}_1 - \mathcal{X}_2 )
		\\
			&= \int {\lvert u_1 - u_2 \rvert}^2 + {(\theta_1 - \theta_2)}^2 + {(q_1 - q_2)}^2 H + {((\theta_1 + q_1) - (\theta_2 + q_2))}^2
		\end{align*}
		for $H \vcentcolon= \mathds{1}(q_1<0) = \mathds{1}(q_2<0)$.
	\end{prop}
	\begin{proof}
		First we verify that $d_S$ is a metric.
		The positive-definiteness and symmetry follow immediately from the fact that $ d_\text{Snell} $ is a metric,
		since the $L^2$ norm is also positive-definite.

		The triangle inequality of $d_S$ then follows from the triangle inequalities of $ d_\text{Snell} $ and of $L^2$
		and from the fact that $s\mapsto \sqrt{s}$ and $s\mapsto s^2$ are increasing:
		\begin{align*}
			d_S ( \mathcal{X}_1,\, \mathcal{X}_2 )
			&= { \left( \int d_\text{Snell}^2 ( \mathcal{X}_1,\, \mathcal{X}_2 ) \right)}^{1/2}
			\\&\leqslant { \left( \int { \left( d_\text{Snell} ( \mathcal{X}_1,\, \mathcal{X}_3) + d_\text{Snell} ( \mathcal{X}_3,\, \mathcal{X}_2) \right)}^2 \right)}^{1/2}
			\\&= \norm{ d_\text{Snell} ( \mathcal{X}_1,\, \mathcal{X}_3 ) + d_\text{Snell} ( \mathcal{X}_3,\, \mathcal{X}_2 ) }{L^2}
			\\&\leqslant \norm{ d_\text{Snell} ( \mathcal{X}_1,\, \mathcal{X}_3 )}{L^2} + \norm{ d_\text{Snell} ( \mathcal{X}_3,\, \mathcal{X}_2 )}{L^2}
			\\&= d_S ( \mathcal{X}_1,\, \mathcal{X}_3) + d_S ( \mathcal{X}_3,\, \mathcal{X}_2).
		\end{align*}

		Finally we prove that $d_S$ agrees with the energy, when comparable.
		So suppose that $q_1(x)$ and $q_2(x)$ have the same phases, i.e. their signs agree everywhere in the spatial domain.
		Then, by virtue of \fref{Lemma}{lemma:d_snell_is_metric_and_compatible_with_energy}, we have that
		\begin{align*}
			d_S^2 ( \mathcal{X}_1,\, \mathcal{X}_2 )
			&= \int d_\text{Snell}^2 \left( \mathcal{X}_1 (x),\, \mathcal{X}_2 (x) \right) dx
			\\&= \int {\lvert u_1 - u_2 \rvert}^2  + {(\theta_1 - \theta_2)}^2 + {(q_1-q_2)}^2 H + {((\theta_1 + q_1) - (\theta_2 + q_2))}^2
			\\&= E_H ( \mathcal{X}_1 - \mathcal{X}_2 )
		\end{align*}
		for $H \vcentcolon= \mathds{1}(q_1<0) = \mathds{1}(q_2<0)$.
	\end{proof}

	\begin{remark}
	\label{rmk:suggest_snell_fails_in_PDE_case}
		Finally we note that, although this distance agrees with the energy, the evidence from the restricted case of spatially constant cases
		strongly suggests that this Snell-type distance does \emph{not} agree with the extraction of a balanced component, especially in light of \fref{Lemma}{lemma:ODE_sols_are_PDE_sols_general} which tells us that spatially constant solutions are honest-to-goodness solutions of the moist Boussinesq system considered here.

        While this Snell-type distance may not agree with
        the balanced--unbalanced decomposition, 
        it could be useful for other applications such as
        data-driven methods that identify nonlinear features 
        in dynamical systems and utilize a metric, distance function,
        and/or kernel function
        \cite{coifman2006diffusion,giannakis2012nonlinear,berry2016variable,das2019delay,giannakis2019data}.
	\end{remark}

\subsection{Metric-less approach}
\label{sec:metric_less_approach}

	To make the approach of this section easy to follow we must be clear about what the last two sections, \fref{Sections}{sec:PDE_centric_approach} and \ref{sec:energy_centric_approach}, achieved.
	\begin{itemize}
		\item	In \fref{Section}{sec:PDE_centric_approach} we devised a metric which agreed with the decomposition, but which failed to agree with the energy (it agreed only \emph{conditionally}).
		\item	In \fref{Section}{sec:energy_centric_approach} we devised a metric which agreed with the energy, but which failed to agree with the decomposition.
	\end{itemize}
	In this section we will see that it is possible to agree with both the decomposition and the energy, at a price.
	That price is that we lose the triangle inequality.
	Instead of a metric, we will therefore be dealing with a statistical divergence.
	\begin{definition}[Statistical divergence]
	\label{def:stat_div}
		Let $\mathcal{H}$ be a Hilbert space.
		A \emph{statistical divergence} on $\mathcal{H}$ is a map $D: \mathcal{H} \times \mathcal{H} \to \mathbb{R} $ satisfying the following properties.
		\begin{enumerate}
			\item	Non-negative: $D(x,\,y) \geqslant 0$ for every $x,\,y\in \mathcal{H}$.
			\item	Definite: $D(x,\,y) = 0$ if and only if $x=y$.
			\item	Limiting quadratic: for every $x\in \mathcal{H}$ there is a positive-definite quadratic form $\mathcal{Q}_x$ such that, for every $y\in \mathcal{H}$,
				\begin{equation*}
					\frac{1}{\varepsilon^2} D(x,\,x+\varepsilon y) \to \mathcal{Q}_x (y)
				\end{equation*}
				as $\varepsilon\downarrow 0$.
				Since this quadratic form gives rise to a positive-definite inner product, $\mathcal{Q}_x$ is often referred to as the \emph{Riemann metric induced by the divergence}.
		\end{enumerate}
	\end{definition}

The prototypical example of a statistical divergence is the relative entropy, 
or Kullback–Leibler divergence, from information theory
\cite{cover2012elements,branicki2012quantifying,chen2014information}.
Another reminiscent notion from the PDE literature is the \emph{relative energy}
\cite{dafermos1979second,feireisl2012relative,feireisl2016singular,breit2017compressible}.
 
	The statistical divergence here is built in a simple way: we use the metric $d_H$ introduced in \fref{Section}{sec:PDE_centric_approach},
	but now let $H$ vary depending on the states whose divergence is being measured.
	\begin{definition}
	\label{def:our_stat_div}
		For any $ \mathcal{X}_1,\, \mathcal{X}_2 \in \mathbb{L}^2_\sigma $ we define
		\begin{equation*}
			D ( \mathcal{X}_1;\, \mathcal{X}_2)
			\vcentcolon= \frac{1}{2} { d_{H_{B_1}} ( \mathcal{X}_1,\, \mathcal{X}_2 ) }^2
		\end{equation*}
		where $d_H$ is as in \fref{Definition}{def:dimensionalized_Parseval_distance},
		where $H_{B_1} \vcentcolon= \mathds{1}(q_{B,\,1} < 0)$ denotes the \emph{balance} interface of the first state $ \mathcal{X}_1 $.
	\end{definition}

    \begin{remark}
    \label{rmk:detailed_computation_stat_div}
        We detail how to compute the statistical divergence $D$ introduced in \fref{Definition}{def:our_stat_div} (similarly to how we detailed the computation of the metric $d_H$ in \fref{Definition}{def:dimensionalized_Parseval_distance}).
        So let $\mathcal{X}_1 = (u_1,\,\theta_1,\,q_1)$ and $\mathcal{X}_2 = (u_2,\,\theta_2,\,q_2)$ be states in $ \mathbb{L}^2_\sigma $.
        
        First we compute $\mathcal{X}_{B,\,1}$ in order to determine the indicator $H_{B_1}$.
        This is done by setting
        \begin{equation*}
            M_1 \vcentcolon= M(\mathcal{X}_1) = \theta_1 + q_1
            \text{ and } 
            {PV}_1 = PV ( \mathcal{X}_1 ) = \nabla^\perp_h\cdot u_{h,\,1} + \partial_3 \theta_1
        \end{equation*}
        and finding the solution $p_1 \in \mathring{H}^1$ of
        \begin{equation*}
            \Delta p_1 + \frac{1}{2} \partial_3 {\min}_0\, ( M_1 - \partial_3 p_1) = {PV}_1.
        \end{equation*}
        The phases of the balanced component of the first state are then encoded in
        \begin{equation*}
            H_{B_1} \vcentcolon= \mathds{1} (q_{B,\,1} < 0) = \mathds{1} (\partial_3 p_1 < M_1).
        \end{equation*}

        Second we proceed as in \fref{Defintion}{def:dimensionalized_Parseval_distance} and compute $d_{H_{B_1}} (\mathcal{X}_1,\, \mathcal{X}_2)$.
        This means that we let
		\begin{align*}
			&PV \vcentcolon= PV(\mathcal{X}_1) - PV(\mathcal{X}_2),\,
			M \vcentcolon= M(\mathcal{X}_1) - M(\mathcal{X}_2),\\
			&j \vcentcolon= j(\mathcal{X}_1) - j(\mathcal{X}_2),\,
			w \vcentcolon= w(\mathcal{X}_1) - w(\mathcal{X}_2), \text{ and }
			a \vcentcolon= a(\mathcal{X}_1) - a(\mathcal{X}_2),
		\end{align*}
      and that we let $p\in \mathring{H}^1$ and $\wavecoord\in {(L^2)}^3_\sigma$ solve
		\begin{equation*}
			\nabla\cdot ( A_{H_{B_1}} \nabla p) = PV - \frac{1}{2} \partial_3 (HM)
		\end{equation*}
		and
		\begin{equation*}
			\nabla\times (A_{H_{B_1}}^{-1}\wavecoord) = j + (\partial_3 w) e_3
			\text{ subject to } \fint A_{H_{B_1}}^{-1} \wavecoord = a \text{ and } \nabla\cdot\wavecoord = 0
		\end{equation*}
        for $A$ and $A^{-1}$ as in \fref{Definition}{def:energy_fixed_Heaviside_coord}.

        So finally we have that
		\begin{equation*}
            D ( \mathcal{X}_1,\, \mathcal{X}_2 )
			= \frac{1}{2} { d_{H_{B_1}}( \mathcal{X}_1,\, \mathcal{X}_2 )}^2
			= \int Q_H^+ (p) + \left( 1 + \frac{H}{2} \right) M^2 + Q_H^- (\wavecoord) + w^2
		\end{equation*}
		for $Q_H^\pm$ as introduced in \fref{Definition}{def:energy_fixed_Heaviside_coord}.
        We note that, in the more concise notation of \fref{Definition}{def:dimensionalized_Parseval_distance} we may write this as
        \begin{equation*}
            D ( \mathcal{X}_1,\, \mathcal{X}_2 ) 
            = \widetilde{E}_{H_{B_{1}}} ( \mathcal{C}_{H_{B_1}} [ \mathcal{M} ( \mathcal{X}_1 ) - \mathcal{M} (\mathcal{X}_2) ] )
        \end{equation*}
    \end{remark}

    \begin{remark}
        We saw in \fref{Remark}{rmk:detailed_computation_stat_div} that we may write the statistical divergence in terms of $\mathcal{C}_H$. This map is an analog of $\mathcal{C}$, which is the map from measurement space to coordinate space which associates to a state's measurements its corresponding coordinates (see \fref{Remark}{rmk:three_descriptions}). The analogous map $\mathcal{C}_H$ proceeds similarly, except that now the three transformations (between state space, measurement space, and coordinate space) are \emph{linear}. The linearity comes from treating a particular indicator function $H$ as fixed.

        More precisely, given an indicator function $H$, we have the following.
        \begin{itemize}
            \item   The ``fixed $H$'' coordinate-to-state map $\mathcal{S}_H : \mathfrak{C} \to \mathbb{L}^2_\sigma$ is given by
            \begin{equation*}
                \mathcal{S}_H (p,\,M,\,\wavecoord,\,w) = (u,\,\theta,\,q)
                \iff \left\{
                    \begin{aligned}
                        u &= \nabla_h^\perp p + \wavecoord_h^\perp\\
                        \theta &= \partial_3 p + \frac{1}{2} H (M - \partial_3 p) + \wavecoord_3\\
                        q &= M - \partial_3 p - \frac{1}{2} H (M - \partial_3 p) - \wavecoord_3,
                    \end{aligned}
                \right.
            \end{equation*}
            where recall that $\mathfrak{C}$ defined in \fref{Remark}{rmk:three_descriptions} is the coordinate space.
            \item   The ``fixed $H$'' state-to-measurement map $\mathcal{M}_H : \mathbb{L}^2_\sigma \to \mathfrak{M}$, where $\mathfrak{M}$ denotes the measurement space as in \fref{Proposition}{prop:global_chg_coord}, is given by
            \begin{equation*}
                \mathcal{M}_H (u,\,\theta,\,q) = (PV,\,M,\,j_H,\,w,\,a_H)
            \end{equation*}
            where
            \begin{equation*}
                j_H (u,\,\theta,\,q) = \partial_3 u_h - \nabla_h^\perp ( \theta - qH)
                \text{ and }
                a_H = - \fint u_h^\perp + \fint (\theta - qH) e_3.
            \end{equation*}
            This means that only $j$ and $a$ are modified into $j_H$ and $a_H$, respectively, but the remaining measurements are unchanged. This is because they are the only nonlinear measurements in the first place since they are the only measurements involving the buoyancy.
            \item   The ``fixed $H$'' measurement-to-coordinate map $\mathcal{C}_H : \mathfrak{M} \to \mathfrak{C}$ is as defined in \fref{Lemma}{lemma:def_and_invertibility_of_C_H}.
        \end{itemize}
        These maps will all be used below to compute the Riemann metric induced by the statistical divergence introduced in \fref{Definition}{def:our_stat_div}.
    \end{remark}
    
	\begin{prop}
	\label{prop:stat_div}
		The map $D : \mathbb{L}^2_\sigma \times \mathbb{L}^2_\sigma \to [0,\,\infty)$ introduced in \fref{Definition}{def:our_stat_div} is a statistical divergence
		in the sense of \fref{Definition}{def:stat_div}.
		In particular, for any $ \mathcal{X},\, \mathcal{Y} \in \mathbb{L}^2_\sigma $,
		\begin{equation*}
			\frac{1}{\varepsilon^2} D( \mathcal{X};\, \mathcal{X} + \varepsilon \mathcal{Y} )
			\to E_B \left( { \left[ \mathcal{Y}  \right] }_{ \mathcal{X} } \right)
		\end{equation*}
		where $ { \left[ \mathcal{Y}  \right] }_{ \mathcal{X} } \vcentcolon= \left( \mathcal{S}_B \circ \mathcal{C}_B \circ \mathcal{M}_H \right) (\mathcal{Y}) $
		is known as the $ \mathcal{X}$--representative of $ \mathcal{Y} $
		(such that if $H_{ \mathcal{X}} = H_{ \mathcal{Y} }$ and $H_{B,\, \mathcal{X}} = H_{B,\, \mathcal{Y} }$ then $ { \left[ \mathcal{Y} \right] }_{ \mathcal{X}} = \mathcal{Y} $).
		Moreover
		\begin{enumerate}
			\item	this statistical divergence agrees with the conserved energy in the sense that, if $H_1 = H_2 = H_{B,\,1} = H_{B,\,2}$ then
				\begin{equation*}
					D( \mathcal{X}_1;\, \mathcal{X}_2 ) = E_H ( \mathcal{X}_1 - \mathcal{X}_2)
				\end{equation*}
				where $H = \mathds{1}(q_1 < 0) = \mathds{1} (q_2 < 0)$ and
			\item	this statistical divergence agrees with the extraction of a balanced component in the sense that
				\begin{equation*}
					\mathcal{X}_B = \argmin_{ \mathcal{Y} \in \mathcal{B} } D ( \mathcal{X};\, \mathcal{Y} ).
				\end{equation*}
		\end{enumerate}
	\end{prop}
	\begin{proof}
		The fact that $D$ is non-negative and definite follows immediately from the positive-definiteness of $d_H$, which is established in \fref{Proposition}{prop:dimensionalized_Parseval_distance}.
		The fact that $D$ agrees with the energy and is compatible with the extraction of a balanced component also follows from the proof of \fref{Proposition}{prop:dimensionalized_Parseval_distance}.
		So finally we turn our attention to the Riemann metric induced by this statistical divergence.
		Let $ \mathcal{X},\, \mathcal{Y} \in \mathbb{L}^2_\sigma$,
		which we write as $ \mathcal{X} = (u,\,\theta,\,q)$ and $ \mathcal{Y} = (v,\,\phi,\,r)$,
		and let $\varepsilon > 0$.

		We note that
		\begin{align*}
			j ( \mathcal{X} + \varepsilon \mathcal{Y} ) - j ( \mathcal{X} )
			&= \partial_3 ( u_h + \varepsilon v_h )
				- \nabla_h^\perp ( \theta + \varepsilon \phi )
				+ \nabla_h^\perp ( {\min}_0\, (q + \varepsilon r))
		\\&\quad
				- \partial_3 u_h + \nabla_h^\perp \theta - \nabla_h^\perp {\min}_0\, q
		\\
			&= \varepsilon \partial_3 v_h - \varepsilon \nabla_h^\perp \phi
				+ \nabla_h^\perp \left[ {\min}_0\, (q + \varepsilon r) - {\min}_0\, q \right]
		\\
			&= \varepsilon \partial_3 v_h - \varepsilon \nabla_h^\perp \phi
				+ \nabla_h^\perp \left[ \varepsilon r \mathds{1} (q+\varepsilon r < 0) + R_\varepsilon (q;\,r) \right]
		\\
			&= \varepsilon j_{H( \mathcal{X}+\varepsilon \mathcal{Y} )} ( \mathcal{Y} ) + \nabla_h^\perp R_\varepsilon (q;\,r)
		\end{align*}
		for
		\begin{align*}
			R_\varepsilon (q;r)
			&\vcentcolon= q \left[ \mathds{1} (q + \varepsilon r<0) - \mathds{1} (q<0) \right]
		\\
			&= \varepsilon \cdot \underbrace{
				\frac{q}{\varepsilon} \left[ \mathds{1} \left( \frac{q}{\varepsilon} + r < 0 \right) - \mathds{1} \left( \frac{q}{\varepsilon} < 0 \right) \right]
			}_{
				\to q (\delta_0 \circ q) r
			}
		\end{align*}
		where $\delta_0$ denotes the Dirac delta distribution at zero such that,
		since $q (\delta \circ q) = 0$, $\frac{1}{\varepsilon} R_\varepsilon (q;\,r) \to 0$ as $\varepsilon\downarrow 0$.

		Therefore, since all other measurements in $ \mathcal{M} $ besides $j$ are linear,
		\begin{align*}
			\mathcal{M} ( \mathcal{X} + \varepsilon \mathcal{Y} ) - \mathcal{M} ( \mathcal{X} )
			= \begin{pmatrix}
				\varepsilon PV ( \mathcal{Y} )\\
				\varepsilon M ( \mathcal{Y} )\\
				\varepsilon j_{H ( \mathcal{X} + \varepsilon \mathcal{Y} )} ( \mathcal{Y} ) + \nabla_h^\perp R_\varepsilon (q;\,r)\\
				\varepsilon w ( \mathcal{Y} )\\
				\varepsilon a ( \mathcal{Y} )
			\end{pmatrix}
			= \varepsilon \begin{pmatrix}
				PV ( \mathcal{Y} )\\
				M ( \mathcal{Y} )\\
				j_{H ( \mathcal{X} + \varepsilon \mathcal{Y} )} ( \mathcal{Y} )\\
				w ( \mathcal{Y} )\\
				a ( \mathcal{Y} )
			\end{pmatrix} + \varepsilon \begin{pmatrix}
				0\\
				0\\
				\nabla_h^\perp \frac{1}{\varepsilon} R\\
				0\\
				0
			\end{pmatrix}
		\\
			=\vcentcolon \varepsilon \mathcal{M}_{H_\varepsilon} ( \mathcal{Y} ) + \varepsilon \mathcal{R},
		\end{align*}
		where as noted above $ \mathcal{R} \to 0$ when $\varepsilon\downarrow 0$.

		So finally we have that
		\begin{align*}
			D ( \mathcal{X};\, \mathcal{X} + \varepsilon \mathcal{Y} )
			&= E_B \left( \mathcal{C}_B \left[ \mathcal{M} ( \mathcal{X} + \varepsilon \mathcal{Y} ) - \mathcal{M} ( \mathcal{X} ) \right] \right)
		\\
			&= E_B \left( \mathcal{C}_B \left[ \varepsilon \mathcal{M}_{H_\varepsilon} ( \mathcal{Y} ) + \varepsilon \mathcal{R} \right] \right)
		\\
			&= \varepsilon^2 E_B \left({ \left[ Y \right] }_{ \mathcal{X}+\varepsilon Y} + \mathcal{C}_B \mathcal{R} \right)
		\end{align*}
		such that indeed, since $H( \mathcal{X} + \varepsilon \mathcal{Y} ) \to H ( \mathcal{X}) = H$,
		\begin{equation*}
			\frac{1}{\varepsilon^2} D ( \mathcal{X};\, \mathcal{X} + \varepsilon \mathcal{Y} )
			\to E_B \left( { \left[ \mathcal{Y}  \right] }_{ \mathcal{X} } \right),
		\end{equation*}
		as claimed.
	\end{proof}

\subsection{Summary}
\label{sec:energy_summary}

    The takeaway message from this section is the following: if we wish to quantify the discrepancy between two states in $ \mathbb{L}^2_\sigma $ in a manner which is both consistent with the decomposition and with the conserved energy then we must give up on the triangle inequality. As \fref{Sections}{sec:PDE_centric_approach} and \ref{sec:energy_centric_approach} show, we are unable to obtain a metric which has both of these properties. We therefore identify, in \fref{Section}{sec:metric_less_approach}, a \emph{statistical divergence} which fails to satisfy the triangle inequality but which is consistent with both the decomposition and the conserved energy. This is summarized in the table below.

	\begin{center}
	\begin{tabular}{l|ccc}
		Approach		& Agrees with			& Agrees with			& Satisfies the\\
					&the decomposition		& the energy			& triangle inequality\\
		\hline
		Parseval-type metric	& \textcolor{ForestGreen}{Yes}	& \textcolor{BurntOrange}{No}	& \textcolor{ForestGreen}{Yes}\\
		Snell-type metric	& \textcolor{BurntOrange}{No}	& \textcolor{ForestGreen}{Yes}	& \textcolor{ForestGreen}{Yes}\\
		Statistical divergence	& \textcolor{ForestGreen}{Yes}	& \textcolor{ForestGreen}{Yes}	& \textcolor{BurntOrange}{No}
	\end{tabular}
	\end{center}

\section{Iterative methods for the decomposition and their geometry}
\label{sec:iterative_methods}

	In this section we discuss how to compute the decomposition of a state into its slow and fast components in \emph{practice}.
	The only computationally non-trivial step in that process is to extract the balanced component, since this requires the inversion of the nonlinear elliptic PDE
	\begin{equation*}
		\Delta p + \frac{1}{2} \partial_3 {\min}_0\, (M - \partial_3 p) = PV
	\end{equation*}
	where $PV$ and $M$ are given and we solve for $p$.

	Our ability to solve this nonlinear elliptic PDE comes from identifying a \emph{variational form} of this equation.
	That is the observation underpinning \cite{remond2024nonlinear}.
	We recall this variational formulation here.

	\begin{theorem}[Variational formulation of $PV$-and-$M$ inversion]
	\label{thm:var_form_PV_and_M_inversion}
		For any $M\in L^2$, $PV\in H^{-1}$, and $p\in \mathring{H}^1 $, $p$ is a global minimizer of $E_\text{var}$, defined as
		\begin{equation*}
			E_\text{var} \vcentcolon= \int_{\mathbb{T}^3} \frac{1}{2} {\lvert \nabla p \rvert}^2  - \frac{1}{4} {{\min}_0\, (M - \partial_3 p)}^2
			+ {\langle PV,\, p \rangle }_{H^{-1}\times \mathring{H}^1 },
		\end{equation*}
		if and only if $p$ is an $H^1$--weak solution of
		\begin{equation*}
			\Delta p + \frac{1}{2} \partial_3 {\min}_0\, (M- \partial_3 p) = PV
		\end{equation*}
		in the sense that
		\begin{equation}
            \label{eq:PV_M_inv_weak}
			\int_{\mathbb{T}^3} \nabla p \cdot \nabla\phi + \frac{1}{2} {\min}_0\, (M-\partial_3 p) \partial_3 \phi = - {\langle PV,\, \phi \rangle }_{H^{-1} \times \mathring{H}^1 } 
			\text{ for every } \phi\in \mathring{H}^1 .
		\end{equation}
	\end{theorem}
	\begin{proof}
		See Lemma 4.9 in \cite{remond2024nonlinear}.
	\end{proof}

	With this variational formulation of $PV$-and-$M$ inversion in hand, we can now leverage iterative methods from optimization
	in order to find computational methods.
	We discuss such methods in \fref{Section}{sec:convergence_guarantees} below, obtaining convergence guarantees.
	In \fref{Section}{sec:geometry} we then relate one of these iterative methods (arising from Newton's method) to the geometry of the balanced set $\mathcal{B}$,
	obtaining a geometric interpretation of this descent method.

\subsection{Convergence guarantees}
\label{sec:convergence_guarantees}

		In this section we introduce and compare several gradient descent methods that can be employed to minimise the variational energy introduced in \fref{Theorem}{thm:var_form_PV_and_M_inversion},
	and thus find solutions of $PV$-and-$M$ inversion.
	We will discuss three descent methods: $\mathring{H}^1$-gradient descent, $L^2$-gradient descent, and Newton descent (which, as we will discuss, can be interpreted locally as a gradient descent),
	recording corresponding convergence rates for each method.
	Ultimately, the convergence rates of interest are those obtained for \emph{backtracking line search} (see \fref{Definition}{def:backtracking})
	which are recorded in \fref{Theorem}{thm:conv_rates_our_en_backtracking}
	but we also record analogous convergence rates in \fref{Theorem}{thm:conv_rates_our_en_exact} for \emph{exact} line search.
	This is done to make the rates obtained in the main result of \fref{Theorem}{thm:conv_rates_our_en_backtracking} easier to contextualise,
	by comparing them with what can be obtained in the idealized scenario of exact line search.

	We then conclude this section with a consistency result that guarantees that minimizers of the variational energy restricted to an (appropriate) finite-dimensional subspace of $\mathring{H}^1$
	do converge to the true minimizer of the energy.
	Put together, the convergence rates of \fref{Theorem}{thm:conv_rates_our_en_backtracking} and the consistency result of \fref{Theorem}{thm:conv_min_approx_prob_to_true_min}
	tell us that we have a ready-to-implement method for finding approximate minimizers which are guaranteed to converge to the true minimizer.

    Before we get started we recall results from \cite{remond2024nonlinear} that will be used in this section.

    \begin{theorem}[Properties of the variational energy]
    \label{thm:prop_var_en}
        The variational energy $E_{var}$ defined in \fref{Theorem}{thm:var_form_PV_and_M_inversion} has the following properties.
        \begin{enumerate}
            \item   $E_{var}$ is $\frac{1}{2}$-strongly convex,
                    weakly lower semi-continuous over $\mathring{H}^1$, and
                    coercive over $\mathring{H}^1$ in the sense that 
                	$
                		E_{var}(p) \geqslant \frac{1}{16}\norm{p}{ \mathring{H}^1 }^2 - \frac{3}{4} \norm{M}{L^2}^2 - 4 \norm{PV}{H^{-1}}^2
                	$
                	for every $p\in \mathring{H}^1 $.
            \item   We may write
                	\begin{equation}
                	\label{eq:write_E_using_e_M}
                		E_{var}(p) = \int_{\mathbb{T}^3} e_M \left( x,\, \nabla p(x) \right) dx + \langle PV,\,p \rangle
                	\end{equation}
                	for $e_M(x,\,u) \vcentcolon= \frac{1}{2} {\lvert u \rvert}^2 - \frac{1}{4} {\min \left( M(x) - u_3,\, 0 \right)}^2$.
            \item   $E_{var}$ is G\^{a}teaux differentiable on $\mathring{H}^1$ with G\^{a}teaux derivative at any $p\in\mathring{H}^1$,
                	denoted by $DE_{var}(p)$, given by
                	\begin{equation}
                	\label{eq:Gateaux_diff_en}
                		DE_{var}(p) \phi = \int_{\mathbb{T}^3} \nabla p \cdot \nabla\phi + \frac{1}{2} \min ( M-\partial_3 p,\, 0) \partial_3 \phi
                			+ \langle PV,\, \phi \rangle \text{ for every } \phi\in\mathring{H}^1.
                	\end{equation}
            \item   The G\^{a}teaux derivative $E_{var}$ is Lipschitz, with specifically the estimate
                	$
                		\norm{ DE_{var}(p_1) - DE_{var}(p_2) }{H^{-1}} \leqslant \frac{3}{2} \norm{p_1 - p_2}{\mathring{H}^1}
                	$
                	for every $p_1,\,p_2\in \mathring{H}^1 $,
                	and so $E_{var}$ is Fr\'{e}chet differentiable.
            \item   Suppose that $E_{var}$ has a minimizer $p^* \in \mathring{H}^1 $.
                	For any $p\in \mathring{H}^1$ the following estimate holds:
                	$
                		\norm{p - p^*}{\mathring{H}^1}^2
                		\leqslant 4 \left( E_{var}(p) - E_{var}(p^*) \right)
                		= 4 \left( E_{var}(p) - \min E_{var} \right).
                	$
        \end{enumerate}
    \end{theorem}
    \begin{proof}
        See Proposition 4.4, Lemma 4.6, Corollary 4.7, and Lemma 4.8 of \cite{remond2024nonlinear}.
    \end{proof}

    In particular \cite{remond2024nonlinear} establishes the well-posedness of nonlinear $PV$-and-$M$ inversion.

    \begin{theorem}[Existence and uniqueness]
    \label{thm:exist_and_unique}
    	For every $M\in L^2$, and $PV\in H^{-1}$ the variational energy $E_{var}$ introduced in \fref{Theorem}{thm:var_form_PV_and_M_inversion}
    	has a unique global minimizer in $\mathring{H}^1$ which is also the unique weak solution of \eqref{eq:PV_M_inv_weak}.

        Moreover we have the following $H^2$ regularity result (\cite{remond2024nonlinear} proves higher-order regularity, but this $H^2$ is all that is required below).
        Let $M\in H^1$ and $PV\in L^2$.
    	Any $\mathring{H}^1$-weak solution $p$ of \eqref{eq:PV_M_inv_weak} belongs to $H^2$ and satisfies
    	$
    		\norm{p}{\dot{H}^2} \leqslant \norm{M}{\dot{H}^1} + 2 \norm{PV}{L^2}.
    	$
    \end{theorem}
    \begin{proof}
        This is Theorem 5.1 and Lemma 5.2 in \cite{remond2024nonlinear}.
    \end{proof}
    
	Now to kick things off in earnest we define the gradient of a functional, taking particular care to highlight the role payed by the choice of the underlying inner product.

	\begin{definition}[Gradient]
	\label{def:grad}
		Consider a functional $ \mathcal{F} : B \to \mathbb{R}$ on a Banach space $B$ equipped with an inner product $ \left( \,\cdot\,,\,\cdot\, \right) $ which is not necessarily complete
		and suppose that there exists a map $D : B\to B$ such that
		$
			D \mathcal{F} (u) \phi = \left( D(u),\, \phi \right)
		$
		for every $u,\,\phi\in B$.
		We call $D$ the \emph{gradient} of $ \mathcal{F} $ with respect to the inner product $ \left( \,\cdot\,,\,\cdot\, \right) $.
		Moreover, if $ \left( \,\cdot\,,\,\cdot\, \right) $ is known as the ``$X$--inner product'' then we write $D = \nabla_X \mathcal{F} $
		and call it the \emph{$X$-gradient}  of $ \mathcal{F} $.
	\end{definition}

	With \fref{Definition}{def:grad} in hand we may now turn our attention towards two gradients of the variational energy that we will use throughout: the $\mathring{H}^1$ and $L^2$-gradients.
	We record expressions for both of these below.

	\begin{lemma}[$\mathring{H}^1$ and $L^2$-gradients for our variational energy]
	\label{lemma:L2_and_H1_grad_of_our_energy}
		Consider the variational energy $E_{var}$ introduced in \fref{Theorem}{thm:var_form_PV_and_M_inversion}.
		The \hyperref[def:grad]{$\mathring{H}^1$-gradient} of $E_{var}$ is given by
		\begin{equation*}
			\left( \nabla_{\mathring{H}^1} E_{var} \right) (p)
			= { \left( -\Delta \right) }^{-1} \left[
				-\Delta p - \frac{1}{2} \partial_3 \left( \min \left( M - \partial_3 p,\, 0 \right) \right) + PV
			\right],
		\end{equation*}
		meaning that $ u = \left( \nabla_{\mathring{H}^1} E_{var} \right) (p)$ is the unique weak solution of
		\begin{equation}
		\label{eq:strong_form_H1_grad_our_energy}
			-\Delta u = -\Delta p - \frac{1}{2} \partial_3 \left( \min \left( M - \partial_3 p,\, 0 \right) \right) + PV
		\end{equation}
		satisfying $\fint_{\mathbb{T}^3} u = 0$, i.e. $u\in\mathring{H}^1$ such that, for every $\phi\in\mathring{H}^1$,
		\begin{equation}
		\label{eq:weak_form_H1_grad_our_energy}
			\int_{\mathbb{T}^3} \nabla u\cdot\nabla \phi
			= \int_{\mathbb{T}^3} \nabla p \cdot \nabla \phi + \frac{1}{2} \min \left( M - \partial_3 p,\, 0 \right) \partial_3 \phi
			+ \langle PV,\, \phi \rangle 
		\end{equation}
		Moreover, if $p\in H^2$, $M\in H^1$, and $PV\in L^2$ then the $L^2$-gradient of $E_{var}$ is given by
		\begin{equation}
		\label{eq:strong_form_L2_grad_our_energy}
			\left( \nabla_{L^2} E_{var} \right) (p) = -\Delta p - \frac{1}{2} \partial_3 \left( \min \left( M - \partial_3 p,\, 0 \right) \right) + PV.
		\end{equation}
	\end{lemma}
	\begin{proof}
		Item 3 of \fref{Theorem}{thm:prop_var_en} tells us that $ u = \nabla_{\mathring{H}^1} E_{var} (p)$ satisfies
		\begin{align*}
			\int_{\mathbb{T}^3} \nabla u \cdot \nabla \phi
			= { \left( u,\, \phi \right) }_{\mathring{H}^1} 
			= { \left( \nabla_{\mathring{H}^1} E_{var} (p) ,\, \phi \right) }_{\mathring{H}^1} 
			= DE_{var}(p) \phi
		\\
			= \int_{\mathbb{T}^3} \nabla p \cdot\nabla \phi + \frac{1}{2} \min \left( M - \partial_3 p,\, 0 \right) \partial_3 \phi
				+ \langle PV,\, \phi \rangle,
		\end{align*}
		which is precisely \eqref{eq:weak_form_H1_grad_our_energy}, the weak formulation of \eqref{eq:strong_form_H1_grad_our_energy}.
		In order to establish \eqref{eq:strong_form_L2_grad_our_energy} it suffices to integrate by parts the expression
		for the G\^{a}teaux derivative of $E_{var}$ obtained in item 3 of \fref{Theorem}{thm:prop_var_en}.
		This yields
		\begin{equation*}
			{ \left( \nabla_{L^2} E_{var} (p),\, \phi \right) }_{L^2} 
			= DE_{var} (p) \phi
			= \int_{\mathbb{T}^3} \left[ -\Delta p - \frac{1}{2} \partial_3 \left( \min \left( M - \partial_3 p,\, 0 \right) \right) + PV \right] \phi,
		\end{equation*}
		from which the claim follows.
	\end{proof}

	\begin{remark}[Comparing the $L^2$ and $ \mathring{H}^1 $ gradients]
	\label{rmk:compare_two_gradients}
		It is not a coincidence that in \fref{Lemma}{lemma:L2_and_H1_grad_of_our_energy} above the $L^2$ and $ \mathring{H}^1 $ gradients are related via
		$\nabla_{ \mathring{H}^1 } E_{var} = {(-\Delta)}^{-1} (\nabla_{L^2} E_{var})$.
		This identity actually follows from the following fact.
		For any Hilbert space $H\subseteq L^2$ and any Fr\'{e}chet differentiable $F:L^2\to\mathbb{R}$,
		if we denote by $R_H : R \to R^*$ the Riesz representation map on $H$ then
		$ \nabla_H F = R_H (\nabla_{L^2} F)$.
		Indeed, for any $u$, $\phi\in H$:
		$
			DF(u) \phi
			= { \left( \nabla_{L^2} F(u),\, \phi \right) }_{L^2} 
			= {\langle \nabla_{L^2} F(u),\, \phi \rangle }_{H^*\times H} 
			= { \left( R_H ( \nabla_{L^2} F )(u),\, \phi \right) }_{H}.
		$

		In our case we have that $R_{ \mathring{H}^1 } = {(-\Delta)}^{-1}$ in the sense that,
		for any $f\in H^{-1}$, $u=R_{ \mathring{H}^1 } F$ if and only if $u$ is the $ \mathring{H}^1 $-weak solution
		of $ -\Delta u = f$ in $\mathbb{T}^3$.
		The general fact above thus tells us that $\nabla_{ \mathring{H}^1 } E_{var} = R_{ \mathring{H}^1 } (\nabla_{L^2} E_{var}) = {( -\Delta)}^{-1} (\nabla_{L^2} E_{var})$,
		as observed in \fref{Lemma}{lemma:L2_and_H1_grad_of_our_energy}.
	\end{remark}

	Besides gradient descent methods using the $\mathring{H}^1$ and $L^2$ gradients,
	the third descent method we will discuss is \emph{Newton descent}.
	This comes from using, locally, the inner product generated by the \emph{Hessian} of the variational energy to define a gradient direction.
	This idea is now made precise.

	\begin{definition}[Newton descent]
	\label{def:Newton_desc}
		Let $ \mathcal{F} : H\to\mathbb{R}$ be a twice Fr\'{e}chet differentiable map over a Hilbert space $H$.
		Suppose there exists $C>0$ such that, for every $u\in H$, $D^2_u \mathcal{F} ( \,\cdot\,, \,\cdot\, )$ induces an inner product,
		written ${ \left( \,\cdot\, , \,\cdot\,  \right) }_{D_u^2 \mathcal{F} } $ and called the \emph{Hessian inner product}, such that
		$
			\frac{1}{C} \norm{\phi}{H} \leqslant \norm{\phi}{D^2_u \mathcal{F} } \leqslant C \norm{\phi}{H}
			\text{ for every } \phi\in H
		$
		(i.e. $\norm{ \,\cdot\, }{D_u^2 \mathcal{F} }$ is \emph{uniformly} comparable to $\norm{ \,\cdot\, }{H}$).
		The negative of the \hyperref[def:grad]{gradient} of $ \mathcal{F} $ with respect to the Hessian inner product is called the \emph{Newton descent direction}.
	\end{definition}

	Having defined Newton descent we now record its form for our variational energy of interest.

	\begin{prop}[Newton descent for our variational energy]
	\label{prop:Newt_desc}
		Let $E_{var}$ be the variational energy introduced in \fref{Theorem}{thm:var_form_PV_and_M_inversion}.
		The \hyperref[def:Newton_desc]{Hessian inner product} induced by $E_{var}$ is, for any $p\in\mathring{H}^1$ and $M\in L^2$,
		\begin{equation}
		\label{eq:Newt_desc_H}
			{ \left( \phi,\, \psi \right) }_{D_p^2 E_{var}}
			= \int_{\mathbb{T}^3} \nabla\phi \cdot \nabla\psi - \frac{1}{2} \mathds{1} (M<\partial_3 p) ( \partial_3 \phi) (\partial_3 \psi)
		\end{equation}
		and so the \hyperref[def:Newton_desc]{Newton descent direction} $\pi$ is the $\mathring{H}^1$-weak solution of
		\begin{equation}
		\label{eq:Newt_desc_S}
			-\nabla\cdot \left[
				\left( I - \frac{1}{2} \mathds{1} (M < \partial_3 p) e_3\otimes e_3 \right) \nabla \pi
			\right] = \Delta p + \frac{1}{2} \partial_3 \left( \min \left( M - \partial_3 p,\, 0 \right) \right) - PV,
		\end{equation}
		for $PV\in H^{-1}$, i.e., for all $\phi\in \mathring{H}^1 $,
		\begin{align}
			&\int_{\mathbb{T}^3} \nabla\pi \cdot\nabla \phi - \frac{1}{2} \mathds{1} (M<\partial_3 p) (\partial_3 \pi) \partial_3 \phi
		\nonumber\\
			&= \int_{\mathbb{T}^3} -\nabla p\cdot\nabla \phi - \frac{1}{2} \min \left( M - \partial_3 p,\, 0 \right) \partial_3 \phi
				- \langle PV,\, \phi \rangle.
		\label{eq:Newt_desc_W}
		\end{align}
	\end{prop}
	\begin{proof}
		Note that $ { \left( \phi,\, \psi \right) }_{D_p^2 E_{var}} = D^2 E_{var}(\phi,\,\psi)$ and
		so \eqref{eq:Newt_desc_H} follows from differentiating the expression for $DE_{var}$ recorded in item 1 of \fref{Theorem}{thm:prop_var_en},
		since $\min( \,\cdot\,, 0)$ is Lipschitz with derivative $\min ( \,\cdot\,,0 ) ' = \mathds{1} ( \,\cdot\, < 0 )$.
		Since the Newton descent direction $\pi$ is characterised by the fact that $ { \left( \pi,\, \phi \right) }_{D^2_p E_{var}} = - DE_{var}(p) \phi$ for every $\phi\in \mathring{H}^1 $,
		\eqref{eq:Newt_desc_W} then follows immediately from \eqref{eq:Newt_desc_H} and the expression for $DE_{var}$ of item 3 of \fref{Theorem}{thm:prop_var_en}.
		To see that \eqref{eq:Newt_desc_S} is indeed the strong form of \eqref{eq:Newt_desc_W} we integrate by parts.
	\end{proof}

	\begin{remark}[Interpretation of Newton descent]
	\label{rmk:interp_Newt_descent_early}
		\fref{Section}{sec:geometry} below details how the Newton descent iteration, recorded in \fref{Proposition}{prop:Newt_desc} above,
		may be interpreted geometrically.

		Another interpretation is already available to us now.
		First: note that once Newton descent enters the so-called ``quadratic regime'', backtracking line search will systematically select steps of unit size.
		Given a current iterate $p_k$, the next iterate $p_{k+1}$ will thus be given by $p_{k+1} = p_k + \pi$ where $\pi$ solves \eqref{eq:Newt_desc_S},
		which we may rewrite as
		\begin{equation*}
			- \nabla \cdot \left( A_{H_k} \nabla \pi \right)
			= \nabla \cdot \left( A_{H_k} \nabla p_k \right)
				+ \frac{1}{2} \partial_3 \left( H_k M \right)
				- PV
		\end{equation*}
		We may rearrange this equation and deduce that the new iterate $p_{k+1} = p_k + \pi$ solves
		\begin{equation*}
			\nabla \cdot \left( A_{H_k} \nabla p_{k+1} \right) = PV - \frac{1}{2} \partial_3 \left( H_k M \right).
		\end{equation*}
		In other words: the next iterate $p_{k+1}$ is obtained by solving the linearised $PV$-and-$M$ inversion
		where the Heaviside characterising the linearisation comes from the previous step.
		Crucially: this is precisely the method that was used in previous works \cite{hu_edwards_smith_stechmann_21,tzou2024} to numerically solve the nonlinear $PV$-and-$M$ inversion.
		This means that the convergence rates presented here for Newton descent are applicable to previous works.
		Moreover this means that the convergence \emph{guarantees} for Newton descent are also the first convergence guarantees of this numerical method which is already used in practice.
	\end{remark}

\vskip0.2in
	The remainder of this section is devoted to the obtention of convergence rates for each of the three descent methods discussed above.
	There are two key assumptions that must be verified about our variational energy in order for such convergence rates to be obtained.
	The first comes from convexity and takes the form of a quadratic lower bound on the variational energy.
	The second, which takes the form of a quadratic \emph{upper} bound on the variational energy, is that to which we turn our attention now.

	The basic ideas are that (1) we will introduce a simpler energy density for which, due to its piecewise quadratic nature, we can easily find such a quadratic upper bound,
	and (2) we will show that this simpler energy density can be turned into our true variational energy density up to some affine transformations.
	The result will then follow from the fact that such affine transformations will only interact with these quadratic upper bounds in predictable ways.
	In order to carry out this argument it is particularly convenient to introduce the notion of a \emph{quadratic remainder}, which we define below.

	\begin{definition}[Quadratic remainder]
	\label{def:quad_remainder}
		Consider a continuously differentiable function $f : \mathbb{R}^d\to\mathbb{R}$.
		For any $x\in\mathbb{R}^d$ we call function $T_x^f : \mathbb{R}^d \to \mathbb{R}$ defined as
		$
			T_x^f (y) \vcentcolon= f(y) - f(x) - D f(x) (y-x)
		$
		for every $y\in\mathbb{R}^d$ the \emph{quadratic remainder of $f$ at $x$}.
	\end{definition}

	We now show that the aforementioned ``simpler'' energy density -- introduced in the result below -- does indeed have a quadratic upper bound.

	\begin{lemma}[Quadratic upper bound for a piecewise quadratic function]
	\label{lemma:quad_up_bd_a_piece_quad_func}
		Consider $f : \mathbb{R}^3 \to \mathbb{R}$ defined by
		$
			f(x) = \frac{1}{2} {\lvert x \rvert}^2 - \frac{1}{4} \min {\left( x_3,\, 0 \right)}^2
		$
		for all $x\in\mathbb{R}^3$.
		For every $x\in\mathbb{R}^3$, 
		$
			f(y) \leqslant f(x) + Df(x) (y-x) + \frac{1}{2} {\lvert y-x \rvert}^2 .
		$
	\end{lemma}
	\begin{proof}
		To make the computations easier we introduce the following function:
		$
			A(x) \vcentcolon= I - \frac{1}{2} \mathds{1} (x_3 < 0) e_3 \otimes e_3
		$,
		such that
		$
			f(x) = \frac{1}{2} A(x) x \cdot x,
		$
		and hence
		$
			\nabla f(x) = A(x) x
		$.
		In particular note that, for every $x$, the largest eigenvalue of $A(x)$ is equal to $1$, and so $A(x)\xi\cdot\xi \leqslant {\lvert \xi \rvert}^2$ for all $\xi\in\mathbb{R}^3$.

		In terms of the \hyperref[def:quad_remainder]{quadratic remainder} $T_x^f$, we may write the target inequality as $T_x^f(y) \leqslant \frac{1}{2} {\lvert y-x \rvert}^2 $,
		so we begin by using the fact that $A$ is pointwise symmetric to rewrite $T_x^f$ in a more convenient way. We obtain
		\begin{equation*}
		\label{eq:int_write_T_using_A}
			T_x^f (y)
			= \frac{1}{2} \left[ A(y) - A(x) \right] y \cdot y + \frac{1}{2} A(x) (y-x) \cdot (y-x).
		\end{equation*}
		Since $ A(y) - A(x) = \frac{1}{2} \left[ \mathds{1}(x_3 < 0 \leqslant y_3) - \mathds{1}(y_3 < 0 \leqslant x_3) \right] e_3\otimes e_3$
		we compute that
		\begin{equation*}
			T_x^f(y) = \left\{
				\begin{aligned}
					& - \frac{y_3^2}{4}	+ \frac{1}{2} A(x) (y-x) \cdot (y-x)		&&\text{if } y_3 < 0 \leqslant x_3,\\
					& \frac{x_3^2}{4} 	+ \frac{1}{2} A(x) (y-x) \cdot (y-x)		&&\text{if } x_3 < 0 \leqslant y_3, \text{ and }\\
					&			  \frac{1}{2} A(x) (y-x) \cdot (y-x)		&&\text{otherwise.}
				\end{aligned}
			\right.
		\end{equation*}
		Only one of these three cases is troublesome since, if $y_3 < 0 \leqslant x_3$ or if $x_3$ and $y_3$ have the same sign then
		$
			T_x^f(y)
			\leqslant \frac{1}{2} A(x) (y-x) \cdot (y-x)
			\leqslant \frac{1}{2} {\lvert y-x \rvert}^2 .
		$
		In the troublesome case, i.e. if $x_3 < 0 \leqslant y_3$, then we need to be more careful.
		In this case writing $z_h = \left( z_1,\, z_2 \right)$ for any $x\in\mathbb{R}^3$ tells us that
		\begin{align*}
			T_x^f(y)
			= \frac{1}{2} A(x) (y-x)\cdot (y-x) + \frac{1}{4} x_3^2
			= \frac{1}{2} {\lvert x_h - y_h \rvert}^2 + \frac{1}{4} {(x_3 - y_3)}^2 + \frac{1}{4} x_3^2.
		\end{align*}
		Crucially: if $x_3 < 0 \leqslant y_3$ then $x_3 y_3 \geqslant 0$ and hence
		$
			\frac{1}{4} {(x_3 - y_3)}^2 + \frac{1}{4} x_3^2 \leqslant \frac{1}{2} {(x_3 - y_3)}^2,
		$
		which proves that $T_x^f (y) \leqslant \frac{1}{2} {\lvert x-y \rvert}^2 $ even in this case, concluding the proof.
	\end{proof}

	With \fref{Lemma}{lemma:quad_up_bd_a_piece_quad_func} in hand we may now use the fact that quadratic upper bounds on \hyperref[def:quad_remainder]{quadratic remainders}
	behave nicely under transformations (see \fref{Lemmas}{lemma:lin_quad_remainder_map}--\ref{lemma:equiv_quad_rem_map_aff_trans}) to deduce that the true variational energy density of interest
	also has a quadratic upper bound.

	\begin{lemma}[Quadratic upper bound for our variational energy density]
	\label{lemma:quad_up_bd_our_en}
		Consider $e:\mathbb{R}^3\to\mathbb{R}$ defined by $e(x) \vcentcolon= \frac{1}{2} {\lvert x \rvert}^2 - \frac{1}{4} {\min(r - x_3,\,0)}^2$, where $r$ is fixed.
		For every $x,\,y\in\mathbb{R}^3$, $e(y) \leqslant e(x) + De(x)(y-x) + \frac{1}{2} {\lvert x-y \rvert}^2 $.
	\end{lemma}
	\begin{proof}
		We note that, for $f$ as in \fref{Lemma}{lemma:quad_up_bd_a_piece_quad_func}, $f(ru_3 - x) = e(x) + a(x)$ for an affine function $a(x) = -rx_3 + \frac{1}{2} r^2$.
		Introducing the affine transformation $\Phi(x) = me_3 - x$ this means that $e = f\circ\Phi - a$.
		Therefore we may use \fref{Lemmas}{lemma:lin_quad_remainder_map}, \ref{lemma:vanishing_quad_rem_charac_aff_func}, and \ref{lemma:equiv_quad_rem_map_aff_trans} to see that
		the \hyperref[def:quad_remainder]{quadratic remainder} of $e$ is given by
		$
			T^e_x
			= T_x^{f\circ\Phi} - T_x^a
			= T^f_{\Phi(x)} \circ \Phi.
		$
		In particular we deduce from \fref{Lemma}{lemma:quad_up_bd_a_piece_quad_func} that
		$
			T_x^e (y)
			\leqslant \frac{1}{2} {\lvert \Phi (y) - \Phi (x) \rvert}^2 
			= \frac{1}{2} {\lvert y-x \rvert}^2 ,
		$
		which is precisely the desired inequality.
	\end{proof}

	Our search for a quadratic upper bound on the variational energy now concludes as we leverage the quadratic upper bound on the variational energy density obtained in \fref{Lemma}{lemma:quad_up_bd_our_en}
	to deduce that the variational energy itself does indeed have a comparable quadratic upper bound.

	\begin{lemma}[Quadratic upper model for our variational energy]
	\label{lemma:fund_upper_mod_ineq_for_our_energy}
		Let $E_{var}$ be the variational energy introduced in \fref{Theorem}{thm:var_form_PV_and_M_inversion}. For every $p,\,r\in \mathring{H}^1 $, the following inequality holds:
		$
		E_{var}(r) \leqslant E_{var}(p) + DE_{var}(p)(r-p) + \frac{1}{2} \norm{r-p}{ \mathring{H}^1 }^2.
		$
	\end{lemma}
	\begin{proof}
		We write $E_{var}$ as in \eqref{eq:write_E_using_e_M}, abusively writing $e_M (\nabla p)$ to mean $e_M (x,\, \nabla p(x))$ in the sequel,
		and note that item 3 of \fref{Theorem}{thm:prop_var_en} tells us that
		\begin{equation}
		\label{eq:fund_upper_int_1}
			DE_{var}(p) \phi = \int_{\mathbb{T}^3} \nabla_u e_M (\nabla p) \cdot \nabla \phi + \langle PV,\, \phi \rangle.
		\end{equation}
		Similarly we may reinterpret \fref{Lemma}{lemma:quad_up_bd_our_en} as telling us that
		\begin{equation}
		\label{eq:fund_upper_int_2}
			e_M(\nabla r) \leqslant e_M (\nabla p) + \nabla_u e_M (\nabla p) \cdot \nabla (r-p) + \frac{1}{2} {\lvert \nabla (r-p) \rvert}^2.
		\end{equation}
		Putting \eqref{eq:write_E_using_e_M}, \eqref{eq:fund_upper_int_1}, and \eqref{eq:fund_upper_int_2} together then yields the claim.
	\end{proof}

	We now have all of the ingredients we need to prove the convergence rates for our three descent methods in the context of \emph{exact} line search.
	First we state a generic version of that result in \fref{Proposition}{prop:key_ineq_conv_rate_grad_desc_exact_Hilbert}
	before specialising it to our energy afterwards in \fref{Theorem}{thm:conv_rates_our_en_exact} .
	Note that the result below is classical -- see for example \cite{boyd_vandenberghe}. We record it here, along with its proof, in order to ensure that
	good care is taken of tracking the constants arising from comparable norms involved in the descent method.
	This is particularly important for us since ultimately it will lead to significant differences between the convergence rate of the $L^2$-gradient descent compared with the other descent methods.

	\begin{prop}[Key inequality for the convergence rate of gradient descent with exact line search in Hilbert spaces]
	\label{prop:key_ineq_conv_rate_grad_desc_exact_Hilbert}
		Let $ {\left( \,\cdot\,,\,\cdot\, \right)}_{X} $ be a complete inner product on a Hilbert space $H$ for which there exists a constant $C>0$ such that
		\begin{equation}
		\label{eq:key_ineq_exact_assum}
			\norm{u}{H} \leqslant C \norm{u}{X} \text{ for every } u\in H.
		\end{equation}
		Let $ \mathcal{F} :H\to \mathbb{R} $ be a continuously Fr\'{e}chet differentiable functional satisfying, for some constants $m,\,M > 0$,
		\begin{equation}
		\label{eq:key_ineq_exact_upper_and_lower}
			\frac{m}{2} \norm{v-u}{H}^2 \leqslant \mathcal{F} (v) - \mathcal{F} (u) - D \mathcal{F} (u) (v-u) \leqslant \frac{M}{2} \norm{v-u}{H}^2
		\end{equation}
		for every $u,\,v\in H$, and let $u^*$ denote the unique minimizer of $ \mathcal{F} $.
		Then
		\begin{equation}
		\label{eq:key_ineq_exact_necessary_decrease}
			\mathcal{F} \left( u - \frac{1}{C^2 M} \nabla_X \mathcal{F} (u) \right)
			\leqslant \mathcal{F} (u) - \frac{1}{2C^2 M} \norm{\nabla_X \mathcal{F} (u)}{X}^2
		\end{equation}
		and
		\begin{equation}
		\label{eq:key_ineq_exact_norm_near_min}
			\mathcal{F} (u) - \mathcal{F} (u^*) \leqslant \frac{1}{2m} \norm{\nabla_H \mathcal{F} (u)}{H}^2
		\end{equation}
		such that
		\begin{equation}
		\label{eq:key_ineq_exact}
			\mathcal{F} \left( u - \frac{1}{C^2 M} \nabla_X \mathcal{F} (u) \right) - \mathcal{F}  ( u^*)
			\leqslant \left( 1 - \frac{m}{C^4M} \right) \left( \mathcal{F} (u) - \mathcal{F} (u^*) \right).
		\end{equation}
	\end{prop}

	\begin{remark}[On the importance of the key inequality]
	\label{rmk:import_key_ineq}
		The so-called ``key inequality'' \eqref{eq:key_ineq_exact} warrants such a lofty name since the actual convergence rate follows from that inequality immediately.
		Indeed, if we denote let $x_0 \in \mathbb{T}^3$ and denote by $x_k$ the iterates obtained via gradient descent with exact line search, i.e.
		\begin{equation*}
			x_k \vcentcolon= \argmin_{t\in\mathbb{R}} \mathcal{F} \left( x_{k-1} - t\nabla \mathcal{F} (x_{k-1}) \right),
		\end{equation*}
		then inducting on the key inequality tells us that
		\begin{equation*}
			\mathcal{F} (x_{k-1}) - \mathcal{F} (x^*)
			\leqslant { \left( 1 - \frac{m}{C^4 M} \right) }^k \left( \mathcal{F} (x_0) - \mathcal{F} (x^*) \right),
		\end{equation*}
		which guarantees that indeed $ \mathcal{F} (x_k)$ converges to $ \mathcal{F} (x_0)$.
	\end{remark}

	\begin{proof}[Proof of \fref{Proposition}{prop:key_ineq_conv_rate_grad_desc_exact_Hilbert}]
		Plugging $v = u - t\nabla_X \mathcal{F} (u)$ into the second inequality in \eqref{eq:key_ineq_exact_upper_and_lower} tells us that
		\begin{align*}
			\mathcal{F} \left( u - t \nabla_X \mathcal{F} (u) \right)
			\leqslant \mathcal{F} (u) - t \norm{\nabla_X \mathcal{F} (u)}{X}^2 + \frac{Mt^2}{2} \norm{\nabla_X \mathcal{F} (u)}{H}^2
		\\
			\leqslant \mathcal{F} (u) + \left( \frac{C^2 M t^2}{2} - t \right) \norm{\nabla_X \mathcal{F} (u)}{X}^2.
		\end{align*}
		Since the quadratic $q(t) \vcentcolon= \frac{C^2 M t^2}{2} - t$ is minimized at $t^* = \frac{1}{C^2M}$, where it attains the value $q(t^*) = - \frac{1}{2C^2 M}$,
		we deduce \eqref{eq:key_ineq_exact_necessary_decrease} from the inequality above.

		To establish \eqref{eq:key_ineq_exact_norm_near_min} we use the first inequality in \eqref{eq:key_ineq_exact_upper_and_lower}.
		We begin by writing $D \mathcal{F} (u) (v-u) = { \left( \nabla_H \mathcal{F} (u),\, v-u \right) }_{H} $, which tells us that with respect to $v$ the right-hand side of that inequality
		is minimized at $v^* = u - \frac{1}{m} \nabla_H \mathcal{F} (u)$.
		Therefore, for any $u,\,v\in H$,
		\begin{equation*}
			F(v)
			\geqslant \left[ \mathcal{F} (u) + { \left( \nabla_H \mathcal{F} (u),\, v-u \right) }_{H} +\frac{m}{2} \norm{v-u}{H}^2  \right] \vert_{v=v^*}
			= \mathcal{F} (u) - \frac{1}{2m} \norm{\nabla_H \mathcal{F} (u)}{H}^2.
		\end{equation*}
		Setting $v = u^*$ produces \eqref{eq:key_ineq_exact_norm_near_min}.

		Finally we turn our attention to the key inequality \eqref{eq:key_ineq_exact}.
		\eqref{eq:key_ineq_exact_assum} and \eqref{eq:key_ineq_exact_norm_near_min} tell us that
		$
			\norm{\nabla_X \mathcal{F} (u)}{X}^2
			\geqslant \frac{2m}{C^2} \left( \mathcal{F} (u) - \mathcal{F} (u^*) \right)
		$
		and so combining this inequality with \eqref{eq:key_ineq_exact_necessary_decrease} yields \eqref{eq:key_ineq_exact}.
	\end{proof}

	We now specialise the generic result proved above to the case of our variational energy.

	\begin{theorem}[Convergence rates for our variational energy with exact line search]
	\label{thm:conv_rates_our_en_exact}
		Let $E_{var}$ be the variational energy introduced in \fref{Theorem}{thm:var_form_PV_and_M_inversion},
		let $V$ be a non-empty subspace of $ \mathring{H}^1 $ (possibly equal to $ \mathring{H}^1 $),
		and let $p^*$ be the unique minimizer of $E_{var}$ in $V$.
		\begin{enumerate}
			\item	Let $p\in V$ and let $\pi = - \nabla_{ \mathring{H}^1 } E_{var} \vert_V (p)$ denote the \hyperref[def:grad]{$ \mathring{H}^1 $-gradient} descent direction
				for $E_{var}\vert_V$ at $p$. Then
				$
					E_{var}( p + \pi ) - E_{var}(p^*) \leqslant \frac{1}{2} \left( E_{var}(p) - E_{var}(p^*) \right).
				$
			\item	Let $p\in V$ and let $\pi$ denote the \hyperref[def:Newton_desc]{Newton descent direction} for $E_{var}\vert_V$ at $p$. Then
				$
					E_{var} \left( p + \frac{1}{4} \pi \right) - E_{var}(p^*) \leqslant \frac{31}{32} \left( E_{var}(p) - E_{var}(p^*) \right).
				$
			\item	Suppose that $V\subseteq H^2$, suppose that there exist a constant $C_n > 0$ such that
				$
					\norm{p}{\mathring{H}^1} \leqslant C_n \norm{p}{L^2} 
				$
				for every $p\in V$, and suppose that $M\in H^1$ and $PV\in L^2$.
				Let $p\in V$ and let $\pi = -\nabla_{L^2} E_{var}\vert_V (p)$ denote the $L^2$-gradient descent direction for $E_{var}\vert_V$ at $p$. Then
				$
					E_{var} \left( p + \frac{1}{C_n^2} \pi \right) - E_{var}(p^*) \leqslant \left( 1 - \frac{1}{2C_n^4} \right) \left( E_{var}(p) - E_{var}(p^*) \right).
				$
		\end{enumerate}
	\end{theorem}
	\begin{proof}
		Since all of these descent directions are gradient descents, we need to verify that the hypotheses of \fref{Proposition}{prop:key_ineq_conv_rate_grad_desc_exact_Hilbert} are verified.
		Item 4 of \fref{Theorem}{thm:prop_var_en} guarantees that $E_{var}$ is continuously Fr\'{e}chet differentiable,
		item 1 of \fref{Theorem}{thm:prop_var_en}, \fref{Lemma}{lemma:fund_upper_mod_ineq_for_our_energy}, and \fref{Lemma}{lemma:equiv_charac_strg_conv} show that \eqref{eq:key_ineq_exact_upper_and_lower}
		is verified with $m = \frac{1}{2} $ and $M=1$, and \fref{Theorem}{thm:exist_and_unique} verifies that $E_{var}$ does indeed have a unique minimizer.
		The only hypothesis left to check is \eqref{eq:key_ineq_exact_assum},
		and the value of the constant $C$ varies from case to case since the inner product $ {\left( \,\cdot\,,\,\cdot\, \right)}_{X} $ used to define the gradient differs in each case.
		\begin{enumerate}
			\item	In this case $ {\left( \,\cdot\,,\,\cdot\, \right)}_{X} = {\left( \,\cdot\,,\,\cdot\, \right)}_{ \mathring{H}^1 } $ and so $C = 1$.
			\item	In this case we may use the \hyperref[def:Newton_desc]{Hessian inner product} recorded in \fref{Proposition}{prop:Newt_desc} to see that
				$
					\norm{\phi}{D^2_p E_{var}}^2
					= \int_{\mathbb{T}^3} {\lvert \nabla\phi \rvert}^2 - \frac{1}{2} \mathds{1}(M<\partial_3 p)
					\geqslant \frac{1}{2} \norm{\phi}{ \mathring{H}^1 }^2,
				$
				i.e. $C = 2$.
			\item	In this case $C = C_n$, by assumption.
		\end{enumerate}
	\end{proof}

	\begin{remark}[On the convergence rate of $L^2$-gradient descent]
	\label{rmk:conv_rate_of_L2_grad}
		We comment on the constant $C_n$ which appears in the third item of \fref{Theorem}{thm:conv_rates_our_en_exact} above.
		Note that in practice such a constant exists since we work with $V$ a finite-dimensional subspace of $ \mathring{H}^1 $, in which case all norms on $V$ are equivalent.
		Moreover, say that in such an instance the dimension of $V$ is given by $n$ (typically commensurate with the number of elements in a finite-element scheme).
		Then, as $n$ approaches infinity, the constant $C_n$ would necessarily grow unboundedly, which would in turn mean that
		$
			1 - \frac{1}{2 C_n^4} \uparrow 1
		$.
		This means that in the case of $L^2$-gradient descent the convergence rate remains linear but the constant involved worsens (i.e. approaches one) as the dimension of the subspace
		we work with increases (e.g. as the mesh is refined).

		Note that this remark also apply when backtracking is used instead of exact line search -- c.f. \fref{Theorem}{thm:conv_rates_our_en_backtracking} below.
	\end{remark}

	We now turn our attention towards the second, and arguably more important, set of convergence rates recorded in this section: those arising from \emph{backtracking} line search.
	For we define backtracking line search below. Recall that, by contrast with exact line search discussed so far, backtracking line search can be implemented immediately in practical situations
	(which is typically not the case for exact line search) while retaining enough favourable structure such that provable convergence rates are still in play with backtracking line search.

	\begin{definition}[Backtracking line search]
	\label{def:backtracking}
		Consider a functional $ \mathcal{F} : B\to \mathbb{R} $ on a Banach space $B$, let $x\in B$ be a current iterate, let $\Delta x\in B$ be a search direction,
		and consider two parameters $\alpha\in (0,\,1/2)$ and $\beta\in (0,\,1)$.
		Backtracking line search proceeds as follows: initiate the step size $t_0 = 1$ and, for $i\geqslant 0$, if
		\begin{equation*}
			f \left( x + t_i \Delta x \right) > f(x) + \alpha t D f(x) \Delta x
		\end{equation*}
		then set $t_{i+1} \vcentcolon= \beta t_i$, while if
		\begin{equation}
		\label{eq:backtracking_exit_condition}
			f \left( x + t_i \Delta x \right) \leqslant f(x) + \alpha t Df(x) \Delta x
		\end{equation}
		then set the step size $t_* \vcentcolon= t_i$ and terminate.
	\end{definition}

	As we did above in the context of exact line search when specialising from the generic result of \fref{Proposition}{prop:key_ineq_conv_rate_grad_desc_exact_Hilbert}
	to \fref{Theorem}{thm:conv_rates_our_en_exact} which applies to the variational energy of interest,
	we first record a generic result for the convergence rate of gradient descent methods with \emph{backtracking} line search.
	As with \fref{Proposition}{prop:key_ineq_conv_rate_grad_desc_exact_Hilbert} above, such a result is classical (see \cite{boyd_vandenberghe})
	and its proof is provided here in order to carefully track the
	constants that appear due to the presence of comparable norms.

	\begin{prop}[Key inequality for the convergence rate of gradient descent with exact line search in Hilbert spaces]
	\label{prop:key_ineq_conv_rate_grad_desc_backtracking_Hilbert}
		Consider the same notation and assumptions as \fref{Proposition}{prop:key_ineq_conv_rate_grad_desc_exact_Hilbert} and let $\alpha\in (0,\,1/2)$ and $\beta\in (0,\,1)$.
		For any $u\in H$ and any $t \in \left[ 0,\, \frac{1}{C^2 M} \right]$,
		\begin{equation}
		\label{eq:backtracking_exit_grad_desc}
			\mathcal{F} \left( u - t \nabla_X \mathcal{F} (u) \right)
			\leqslant \mathcal{F} (u) - \alpha t \norm{\nabla_X \mathcal{F} (u)}{X}^2,
		\end{equation}
		i.e. the \hyperref[def:backtracking]{backtracking} exit condition \eqref{eq:backtracking_exit_condition} holds.
		Moreover, for any $u\in H$ if $t_*$ denotes the step size obtained via backtracking line search with parameters $\alpha$ and $\beta$ and search direction $-\nabla_X \mathcal{F} (u)$
		then
		\begin{equation}
		\label{eq:lower_bound_step_size}
			t_* \geqslant \min \left( 1,\, \frac{\beta}{C^2 M} \right)
		\end{equation}
		and
		\begin{equation*}
			\mathcal{F} \left( u - t_* \nabla_X \mathcal{F} (u) \right) - \mathcal{F} ( u^* )
			\leqslant \left( 1 - \frac{2\alpha m}{C^2} \min \left( 1,\, \frac{\beta}{C^2 M} \right) \right) \left( \mathcal{F} (u) - \mathcal{F} (u^*) \right).
		\end{equation*}
	\end{prop}
	\begin{proof}
		Proceeding as in the proof of \fref{Proposition}{prop:key_ineq_conv_rate_grad_desc_exact_Hilbert} we deduce from \eqref{eq:key_ineq_exact_upper_and_lower} that
		\begin{equation*}
			\mathcal{F} \left( u - t \nabla_X \mathcal{F} (u) \right)
			\leqslant \mathcal{F} (u) + \left( \frac{C^2 M t^2}{2} - t \right) \norm{\nabla_X \mathcal{F} (u)}{X}^2.
		\end{equation*}
		In particular we note that if $t\in \left[ 0,\, \frac{1}{C^2 M} \right]$ then $\frac{C^2 M t^2}{2} - t \leqslant - \frac{t}{2} < -\alpha t$,
		and so indeed \eqref{eq:backtracking_exit_grad_desc} holds.
		As a consequence this means that either $t_* = t_0 = 1$ or $t_* = \beta t_i$ for some $t_i$ which did \emph{not} satisfy the backtracking exit condition,
		and so which must satisfy $t_i > \frac{1}{C^2 M}$. Putting these two options together establishes \eqref{eq:lower_bound_step_size}.
		Finally we note that $t_*$ must satisfy the backtracking exit condition \eqref{eq:backtracking_exit_grad_desc}
		and so combining that inequality with \eqref{eq:key_ineq_exact_upper_and_lower}, \eqref{eq:key_ineq_exact_norm_near_min},
		and \eqref{eq:lower_bound_step_size} allows us to obtain the main inequality.
	\end{proof}

	We now specialise \fref{Proposition}{prop:key_ineq_conv_rate_grad_desc_backtracking_Hilbert} above to our variational energy.

	\begin{theorem}[Convergence rates for our variational energy with backtracking line search]
	\label{thm:conv_rates_our_en_backtracking}
		Let $E_{var}$ be the variational energy introduced in \fref{Theorem}{thm:var_form_PV_and_M_inversion},
		let $V$ be a non-empty subspace of $ \mathring{H}^1 $ (possibly equal to $ \mathring{H}^1 $),
		let $p^*$ denote the unique minimizer of $E_{var}$ in $V$, and let $\alpha\in (0,\,1/2)$ and $\beta\in (0,\,1)$.
		In each case below let $t_*$ denote the step size obtained via \hyperref[def:backtracking]{backtracking} line search with direction $\pi$
		(chosen differently in each case) and parameters $\alpha$ and $\beta$.
		\begin{enumerate}
			\item	Let $p\in V$ and let $\pi = - \nabla_{ \mathring{H}^1 } E_{var}\vert_V (p)$ denote the \hyperref[def:grad]{ $ \mathring{H}^1 $-gradient} descent direction
				for $E_{var}\vert_V$ at $p$. Then
				\begin{equation*}
					t_* \geqslant \beta
					\text{ and } 
					E_{var} ( p + t_* \pi) - E_{var}(p^*) \leqslant (1-\alpha\beta) \left( E_{var}(p) - E_{var}(p^*) \right).
				\end{equation*}
			\item	Let $p\in V$ and let $\pi$ denote the \hyperref[def:Newton_desc]{Newton descent direction} for $E_{var}\vert_V$ at $p$.
				Then
				\begin{equation*}
					t_* \geqslant \frac{\beta}{4}
					\text{ and } 
					E_{var} ( p + t_* \pi ) - E_{var} (p^*) \leqslant \left( 1 - \frac{\alpha\beta}{16} \right) \left( E_{var}(p) - E_{var}(p^*) \right).
				\end{equation*}
			\item	Suppose that $V\subseteq H^2$, suppose that there exist a constant $C_n > 0$ such that
				$
					\norm{p}{\mathring{H}^1} \leqslant C_n \norm{p}{L^2}
				$
				for every $p\in V$, and suppose that $M\in H^1$ and $PV\in L^2$.
				Let $p\in V$ and let $\pi = -\nabla_{L^2} E_{var}\vert_V (p)$ denote the $L^2$-gradient descent direction for $E_{var}\vert_V$ at $p$.
				Then
				$
					t_* \geqslant \min \left( 1,\, \frac{\beta}{C_n^2} \right)
				$
				and
				\begin{equation*}
					E_{var} ( p + t_* \pi ) - E_{var} ( p^*) \leqslant \left( 1 - \frac{\alpha}{C_n^2} \min \left( 1,\, \frac{\beta}{C_n^2} \right) \right) \left( E_{var}(p) - E_{var}(p^*) \right).
				\end{equation*}
		\end{enumerate}
	\end{theorem}
	\begin{proof}
		This follows from \fref{Proposition}{prop:key_ineq_conv_rate_grad_desc_backtracking_Hilbert}
		just as \fref{Theorem}{thm:conv_rates_our_en_exact} follows from \fref{Proposition}{prop:key_ineq_conv_rate_grad_desc_exact_Hilbert},
		noting that the appropriate values of the constant $C$ in each case were identified in the proof of \fref{Theorem}{thm:conv_rates_our_en_exact}.
	\end{proof}

	With the convergence rates in hand we change gear and turn our attention towards a consistency result.
	We seek to prove that minimizers of the variational energy restricted to (appropriate) finite-dimensional subspaces of $\mathring{H}^1$ do converge to the true minimizer.
	In order to do so we first prove the following elementary estimate.

	\begin{lemma}[Explicit local modulus of continuity of the variational energy]
	\label{lemma:explicit_local_mod_cty_en}
		Let $E_{var}$ be the variational energy introduced in \fref{Theorem}{thm:var_form_PV_and_M_inversion}.
		For any $p_1,\,p_2\in \mathring{H}^1$ the following estimate holds.
		\begin{equation}
		\label{eq:explicit_local_mod_cty_en}
			E_{var}(p_1) - E_{var}(p_2)
			\leqslant \left( \frac{1}{2} \norm{M}{L^2} + \frac{3}{4} \norm{\nabla p_1}{L^2} + \frac{3}{4} \norm{\nabla p_2}{L^2} \right) \norm{\nabla (p_1 - p_2)}{L^2}.
		\end{equation}
	\end{lemma}
	\begin{proof}
		We write $g_0 (s) \vcentcolon= \min(s,\,0)$ such that
		\begin{equation*}
			E_{var}(p_1) - E_{var}(p_2)
			= \int_{\mathbb{T}^3} \frac{1}{2} \left( {\lvert \nabla p_1 \rvert}^2 - {\lvert \nabla p_2 \rvert}^2 \right)
				- \frac{1}{4} \left[ g_0 (M-\partial_3 p_1) - g_0 (M-\partial_3 p_2) \right].
		\end{equation*}
		We then estimate
		\begin{equation*}
			\left\vert {\lvert \nabla p_1 \rvert}^2 - {\lvert \nabla p_2 \rvert}^2  \right\rvert 
			= \left\vert \nabla (p_1 + p_2) \right\rvert \; \left\vert \nabla (p_1 - p_2) \right\rvert 
			\leqslant \left( \lvert \nabla p_1 \rvert + \lvert \nabla p_2 \rvert \right) \lvert \nabla (p_1 - p_2) \rvert
		\end{equation*}
		and, since $g_0$ is $1$-Lipschitz, we may proceed in the same way to estimate
		\begin{equation*}
			\left\vert { g_0 ( M - \partial_3 p_1) }^2 - { g_0 ( M - \partial_3 p_2) }^2 \right\rvert 
			\leqslant \left( 2 \lvert M \rvert + \lvert \partial_3 p_1 \rvert + \lvert \partial_3 p_2 \rvert \right) \lvert \partial_3 (p_1 - p_2) \rvert.
		\end{equation*}
		Putting it all together produces the desired inequality.
	\end{proof}

	We are now equipped to prove the last main result of this section, which pertains to numerical consistency.

	\begin{theorem}[Convergence of the minimizers of the approximate variational energy to the minimizer of the true variational energy]
	\label{thm:conv_min_approx_prob_to_true_min}
		Let $V_n \subseteq \mathring{H}^1$ be a closed subspace of $\mathring{H}^1$, define $E_n \vcentcolon = E_{var} \vert_{V_n}$, and let
		$
			p^* \vcentcolon= \argmin E_{var} \text{ and } p^*_n \vcentcolon= \argmin E_n.
		$
		For any $p^\varepsilon_n ,\, p_n \in V_n$,
		\begin{align}
			&\norm{p_n^\varepsilon - p^*}{\mathring{H}^1}
		\label{eq:conv_min_approx_prob_to_true_min_1}\\
			&\leqslant 4 \left[ E_n ( p_n^\varepsilon ) - \min E_n \right]
				+ \left( 8 \norm{M}{L^2} + 12 \norm{PV}{H^{-1}} + 3 \norm{p_n - p^*}{\mathring{H}^1} \right) \norm{p_n - p^*}{\mathring{H}^1}.
		\nonumber
		\end{align}
		In particular, if $p^\varepsilon_n$ and $p_n$ are chosen to satisfy, for some $\alpha,\, C > 0$,
		\begin{equation}
		\label{eq:conv_min_approx_prob_to_true_min_2}
			E_n (p_n^\varepsilon ) - \min E_n \leqslant \varepsilon \text{ and } \norm{p_n - p^*}{\mathring{H}^1} \leqslant C n^{-\alpha} \text{ for all } n
		\end{equation}
		then, for $n$ sufficiently large to ensure $Cn^{-\alpha} \leqslant 1$,
		\begin{equation}
		\label{eq:conv_min_approx_prob_to_true_min_3}
			\norm{ p_n^\varepsilon - p^* }{\mathring{H}^1} \lesssim \varepsilon + \left( 1 + \norm{M}{L^2} + \norm{PV}{H^-1} \right) n^{-\alpha},
		\end{equation}
		i.e. $\varepsilon$-minimizers of the discretised variational energy $E_n$ can be made arbitrarily close to the minimizer of the true variational energy $E_{var}$.
	\end{theorem}
	\begin{proof}
		First note that the $p_n^*$'s are well-defined since the properties of $E_{var}$ parlay over to properties of $E_n$ and the proof of \fref{Theorem}{thm:exist_and_unique}
		can be reproduced for $E_n$.

		Since \eqref{eq:conv_min_approx_prob_to_true_min_3} follows from \eqref{eq:conv_min_approx_prob_to_true_min_1}
		under the additional assumptions of \eqref{eq:conv_min_approx_prob_to_true_min_2} it suffices to establish \eqref{eq:conv_min_approx_prob_to_true_min_1}.
		First we note that, since $p^*$ and $p^*_n$ are minimizers of $E_{var}$ and $E_n$, respectively, item 5 of \fref{Theorem}{thm:prop_var_en}
		(which holds \emph{mutatis mutandis} for $E_n$) tells us that
		\begin{align}
			\norm{p_n^\varepsilon - p^*}{\mathring{H}^1}
			\leqslant \norm{ p_n^\varepsilon - p_n^* }{\mathring{H}^1} + \norm{p_n^* - p^*}{\mathring{H}^1}
		\nonumber\\
			\leqslant 4 \left[ E_n (p_n^\varepsilon) - E_n (p_n^*) \right] + 4 \left[ E_{var}(p_n^*) - E_{var}(p^*) \right].
		\label{eq:conv_min_approx_prob_to_true_min_4}
		\end{align}
		By virtue of \fref{Theorem}{thm:exist_and_unique} and \ref{lemma:explicit_local_mod_cty_en} we may estimate the second term above via
		\begin{align}
			E_{var}(p_n^*) - E_{var}(p^*)
			&= E_n (p_n^*) - E_{var}(p^*)
			= \min E_n - E_{var}(p^*)
			\leqslant E_{var}(p_n) - E_{var}(p^*)
		\nonumber\\
			&\leqslant \left( \frac{1}{2} \norm{M}{L^2} + \frac{3}{4} \left[ \norm{p_n}{\mathring{H}^1} + \norm{p^*}{\mathring{H}^1} \right] \right) \norm{p_n - p^*}{\mathring{H}^1}
		\nonumber\\
			&\leqslant \left( 2 \norm{M}{L^2} + 3 \norm{PV}{H^{-1}} + \frac{3}{4} \norm{p_n - p^*}{\mathring{H}^1} \right) \norm{p_n - p^*}{\mathring{H}^1}.
		\label{eq:conv_min_approx_prob_to_true_min_5}
		\end{align}
		Putting \eqref{eq:conv_min_approx_prob_to_true_min_4} and \eqref{eq:conv_min_approx_prob_to_true_min_5}
		together then produces \eqref{eq:conv_min_approx_prob_to_true_min_1}.
	\end{proof}

\subsection{Geometry}
\label{sec:geometry}

	In this section we do two things.
	\begin{itemize}
		\item	Foreshadowed by \fref{Proposition}{prop:alt_char_balanced_set} we see that $p$ and $M$ may be viewed as \emph{coordinates} along the balanced set $ \mathcal{B} $.
			Moreover $PV$ and $M$ then arise naturally as \emph{dual coordinates} (with respect to some natural inner product arising from the energy).
		\item	Equipped with this geometric language we interpret the Newton descent iteration of \fref{Section}{sec:convergence_guarantees} above in geometric terms.
			We explain how the Newton iteration can be viewed as a flow along the balanced set $ \mathcal{B} $.
			This flow follows the coordinate direction along $M$ while driving $PV$ towards its true value.
	\end{itemize}

	We begin by recalling a definition from \fref{Section}{sec:decomposition}.

	\begin{definition}[Global chart for the balanced set]
	\label{def:global_chart_balanced_set}
		For $ \mathbb{H} \vcentcolon= \mathring{H}^1 \times L^2$ we define $ \Phi : \mathbb{H} \to \mathcal{B} $ via
		\begin{equation*}
			\Phi (p,\, M) = \begin{pmatrix}
				\nabla_h^\perp p\\
				\partial_3 p + \frac{1}{2} {\min}_0\, (M - \partial_3 p)\\
				M - \partial_3 p - \frac{1}{2} {\min}_0\, (M - \partial_3 p)
			\end{pmatrix}.
		\end{equation*}
	\end{definition}

    \begin{figure}
		\centering
		\captionsetup{width=0.85\textwidth}
		\includegraphics[width=0.7\textwidth]{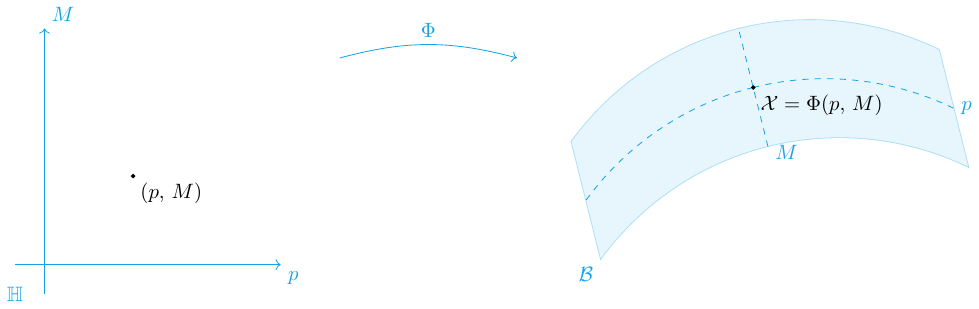}
		\caption{\small A depiction of the global chart for the balanced set introduced in \fref{Definition}{def:global_chart_balanced_set}.}
		\label{fig:global_chart}
	\end{figure}

	We now collect some more properties of this map.

	\begin{prop}
	\label{prop:Pih_is_bi_Lipschitz}
		$\Phi$ is bi--Lipschitz, i.e. Lipschitz with Lipschitz inverse.
	\end{prop}

	In particular, since the parametrisation $\Phi$ of $ \mathcal{B} $ is bi-Lipschitz it follows that the balanced set is an honest-to-goodness manifold. This is depicted in \fref{Figure}{fig:global_chart}.

	\begin{cor}
	\label{cor:bal_set_is_manifold}
		The balanced set $\mathcal{B}$ is a Lipschitz-regular Hilbert manifold, modeled after the Hilbert space $ \mathbb{H} $, with global chart $\Phi$.
	\end{cor}
	\begin{proof}[Proof of \fref{Corollary}{cor:bal_set_is_manifold}]
		This follows from immediately from \fref{Proposition}{prop:Pih_is_bi_Lipschitz}.
	\end{proof}
	\begin{proof}[Proof of \fref{Proposition}{prop:Pih_is_bi_Lipschitz}]
		First we note that the inverse of $\Phi$ is given by
		\begin{equation*}
			\Phi^{-1} (u,\,\theta,\,q)
			= \begin{pmatrix}
				\Delta^{-1} \left[ \nabla_h^\perp \cdot u_h + \partial_3 (\theta - {\min}_0\, q) \right]\\
				\theta + q
			\end{pmatrix}.
		\end{equation*}
		Indeed, suppose that $(u,\,\theta,\,q) = \Phi (p,\,M)$.
		We see that
		\begin{equation*}
			\sign q = \sign (M - \partial_3 p)
		\end{equation*}
		and so \fref{Lemmas}{lemma:invert_b_and_M} and \ref{lemma:div_and_curl_of_good_unknown} tell us that
		\begin{align*}
			(u,\,\theta,\,q) = \Phi (p,\,M)
			&\iff \left\{
				\begin{aligned}
					&u = \nabla_h^\perp p\\
					&\theta = \partial_3 p + \frac{1}{2} {\min}_0\, (M - \partial_3 p)\\
					&q = M - \partial_3 p - \frac{1}{2} {\min}_0\, (M - \partial_3 p)
				\end{aligned}
			\right.
		\\
			&\iff \left\{
				\begin{aligned}
					&u = \nabla_h^\perp p\\
					&\theta + q = M\\
					&\theta - {\min}_0\, q = \partial_3 p
				\end{aligned}
			\right.
		\\
			&\iff \left\{
				\begin{aligned}
					&-u_h^\perp + (\theta - {\min}_0\, q) e_3 = \nabla p\\
					&u_3 = 0\\
					&\theta + q = M
				\end{aligned}
			\right.
		\\
			&\iff \left\{
				\begin{aligned}
					&p = \Delta^{-1} \nabla\cdot \left[ -u_h^\perp + (\theta - {\min}_0\, q) e_3 \right]\\
					&u_3 = 0\\
					&M = \theta + q
				\end{aligned}
			\right.
		\\
			&\iff \left\{
				\begin{aligned}
					&p = \Delta^{-1} \left[ \nabla_h^\perp \cdot u_h + \partial_3 (\theta - {\min}_0\, q) \right]\\
					&M = \theta + q\\
					&u_3 = 0
				\end{aligned}
			\right.
		\\
			&\iff (p,\,M) = \Phi^{-1} (u,\,\theta,\,q).
		\end{align*}
		The fact that both $\Phi$ and $\Phi^{-1}$ are Lipschitz then follows from the fact that $s\mapsto {\min}_0\, s$ is Lipschitz, that $\Delta : \mathring{H}^1 \to H^{-1}$ is an isometry,
		and from \fref{Lemma}{lemma:H_minus_one_norm_of_explicit_PV}.
		Indeed, for $\Phi$ we have that
		\begin{align*}
			\norm{\Phi ( p_1,\,M_1 ) - \Phi (p_2,\, M_2)}{L^2}^2
			&\lesssim \norm{\nabla p_1 - \nabla p_2}{L^2}^2 + \norm{M_1 - M_2}{L^2}^2
		\\
			&= \norm{ (p_1,\,M_1) - (p_2,\,M_2) }{ \mathring{H}^1 \times L^2 }^2
		\end{align*}
		while for $\Phi^{-1}$ we have that
		\begin{align*}
			&\norm{ \Phi^{-1} ( u_1,\,\theta_1,\,q_1) - \Phi^{-1} (u_2,\,\theta_2,\,q_2) }{ \mathring{H}^1 \times L^2 }^2
		\\
			&= \norm{ \Delta^{-1} \left[
				\nabla_h^\perp\cdot u_{h,\,1} + \partial_3 (\theta_1 - {\min}_0\, q_1)
				- \nabla_h^\perp\cdot u_{h,\,2} - \partial_3 (\theta_2 - {\min}_0\, q_2)
			\right]}{ \mathring{H}^1 }^2
		\\
			&\qquad + \norm{ (\theta_1 + q_1) - (\theta_2 + q_2)}{L^2}^2
		\\
			&= \norm{
				\nabla_h^\perp\cdot u_{h,\,1} + \partial_3 (\theta_1 - {\min}_0\, q_1)
				- \nabla_h^\perp\cdot u_{h,\,2} - \partial_3 (\theta_2 - {\min}_0\, q_2)
			}{H^{-1}}^2
		\\
			&\qquad + \norm{ (\theta_1 + q_1) - (\theta_2 + q_2)}{L^2}^2
		\\
			&\lesssim \norm{u_{h,\,1} - u_{h,\,2}}{L^2}^2
				+ \norm{ (\theta_1 - {\min}_0\, q_1) - (\theta_2 - {\min}_0\, q_2)}{L^2}^2
				+ \norm{ (\theta_1 + q_1) - (\theta_2 + q_2)}{L^2}^2
		\\
			&\lesssim \norm{ (u_1,\,\theta_1,\,q_1) - (u_2,\,\theta_2,\,q_2)}{L^2}^2.
			\qedhere
		\end{align*}
	\end{proof}
	\begin{remark}[Alternate form of $\Phi^{-1}$]
	\label{rmk:alt_form_Phi_inverse}
		We typically write $\Phi^{-1}$ in terms of the nonlinear $PV$-and-$M$ inversion,
		and not in terms of the inversion of a Laplacian as is done in the proof of \fref{Proposition}{prop:Pih_is_bi_Lipschitz} above.
		This is because the Laplacian inversion does not use $PV$, but instead uses a potential vorticity based on the \emph{buoyancy}
		(by contrast with $PV$ which is based on the equivalent potential temperature $\theta$) given by
		\begin{equation*}
			{PV}_b = \nabla_h^\perp \cdot u_h + \partial_3 (\theta - {\min}_0\, q).
		\end{equation*}
		Typically we do \emph{not} want to use ${PV}_b$ since it is not slow. That is why $PV$ is usually preferable.
		However, for the purpose of the proof of \fref{Proposition}{prop:Pih_is_bi_Lipschitz}, ${PV}_b$ is preferable to $PV$
		since it lets us prove more easily that $\Phi^{-1}$ is Lipschitz.
	\end{remark}
	
	So far we have been able to endow the balanced set $\mathcal{B}$ with the structure of a Lipschitz-regular manifold.
	Unfortunately, that is as regular as the balanced set gets due to the presence of $ {\min}_0\, q$ in the buoyancy term,
	which is Lipschitz but not differentiable.

	Nonetheless we can formally differentiate $ {\min}_0\, q$ and thus endow the balanced set $ \mathcal{B} $ with a formal differentiable structure.
	In order to do so, we introduce notation for its formal tangent space, for the formal derivative of the chart $\Phi$, and for the inner product we will use (which comes from the conserved energy).

	\begin{definition}
	\label{def:formal_diff_structure_on_balanced_set}
		Fix an indicator function $H$.
		\begin{enumerate}
			\item	We define
				\begin{align*}
					T_H \mathcal{B} \vcentcolon= \bigg\{
						(u,\,\theta,\,q) \in \mathbb{L}^2_\sigma :
						&\;\partial_3 u_h - \nabla_h^\perp (\theta - qH) = 0,\,
						u_3 = 0,\,
					\\
						&\fint u_h = 0, \text{ and } 
						\fint (\theta - qH) = 0
					\bigg\}.
				\end{align*}
			\item	We define $D_H \Phi : \mathbb{H} \to T_H \mathcal{B} $ via
				\begin{equation*}
					D_H \Phi (p,\,M) \vcentcolon= \begin{pmatrix}
						\nabla_h^\perp p\\
						\partial_3 p + \frac{H}{2} (M-\partial_3 p)\\
						M - \partial_3 p - \frac{H}{2} (M-\partial_3 p)
					\end{pmatrix}.
				\end{equation*}
			\item	For any $(u,\,\theta,q)\in \mathbb{L}^2_\sigma$ we define
				\begin{equation*}
					\norm{ (u,\,\theta,q)}{L^2_H}^2 \vcentcolon= \int {\lvert u \rvert}^2 + \theta^2 + q^2 H + {(\theta + q)}^2
				\end{equation*}
				and let $ {\langle \,\cdot\,,\,\cdot\, \rangle}_{L^2_H} $  denote the corresponding innner product.
		\end{enumerate}
	\end{definition}

	As expected, the formal derivative of $\Phi$ completely characterizes the formal tangent space.
	(This would be trivially true if the balanced set $ \mathcal{B} $ was endowed with an honest-to-goodness differentiable structure by its chart $\Phi$.)

	\begin{lemma}
	\label{lemma:formal_tan_space_is_image_of_formal_derivative}
		$T_H \mathcal{B} = \im D_H \Phi$.
	\end{lemma}
	\begin{proof}
		Proceeding similarly to the proof of \fref{Proposition}{prop:alt_char_balanced_set} we introduce the ``good unknown'' (see \fref{Remark}{rmk:good_unknown})
		\begin{equation*}
			v_H \vcentcolon= -u_h^\perp + (\theta - qH) e_3.
		\end{equation*}
		For $ \mathcal{X} = (u,\,\theta,\,q)$ we have on one hand that, as per \fref{Lemma}{lemma:div_and_curl_of_good_unknown},
		\begin{equation*}
			\mathcal{X}\in T_H \mathcal{B}
			\iff \nabla\times v_H,\, \fint v_H = 0, \text{ and } u_3 = 0.
		\end{equation*}
		On the other hand, proceeding as in the proof of \fref{Proposition}{prop:Pih_is_bi_Lipschitz} and using \fref{Lemma}{lemma:invert_b_and_M} tells us that
		\begin{align*}
			\mathcal{X} = D_H \Phi (p,\,M)
			&\iff \left\{
				\begin{aligned}
					&-u_h^\perp + (\theta - qH) e_3 = \nabla p\\
					&u_3 = 0\\
					&\theta + q = 0
				\end{aligned}
			\right.
		\\
			&\iff \left\{
				\begin{aligned}
					&v_H = \nabla p\\
					&u_3 = 0\\
					&\theta + q = M
				\end{aligned}
			\right.
		\end{align*}
		The Helmholtz decomposition (\fref{Corollary}{cor:Helmholtz}) allows us to conclude:
		\begin{align*}
			\mathcal{X}\in T_H \mathcal{B} 
			&\iff \nabla\times v_H = 0,\,\fint v_H = 0, \text{ and } u_3 = 0\\
			&\iff \nabla\times v_H = 0,\,\fint v_H = 0,\,  u_3 = 0, \text{ and } \theta + q =\vcentcolon M\\
			&\iff v_H = \nabla p \text{ for some } p\in \mathring{H}^1 ,\, u_3 = 0, \text{ and } \theta + q = M\\
			&\iff \mathcal{X} = D_H \Phi (p,\,M) \text{ for some } (p,\,M) \in \mathbb{H}\\
			&\iff \mathcal{X}\in \im D_H \Phi.
			\qedhere
		\end{align*}
	\end{proof}

    \begin{figure}
		\centering
		\captionsetup{width=0.85\textwidth}
		\includegraphics[width=0.7\textwidth]{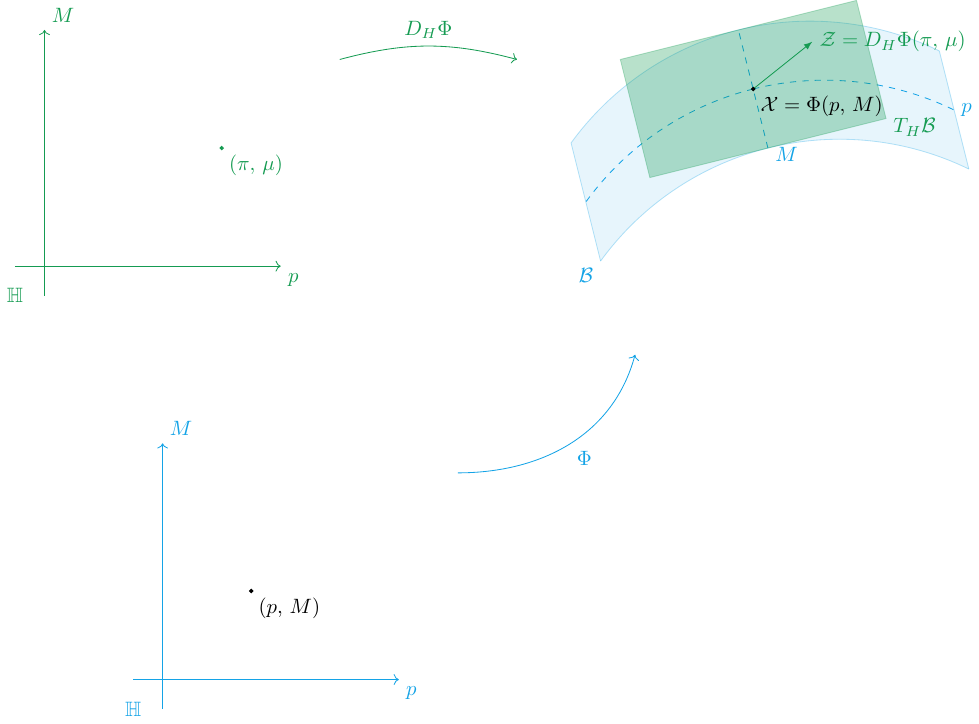}
		\caption{\small
			The formal tangent space to the balanced set introduced in \fref{Definition}{def:formal_diff_structure_on_balanced_set}
			is the image of the formal derivative of the global chart.
		}
		\label{fig:tangent_space}
	\end{figure}

	We now remark that, as suggested by the notation we chose to employ, $D_H\Phi$ may be viewed formally as the derivative of $ \Phi$.
	Since the balanced set $ \mathcal{B} $ is the image of the chart $\Phi$, it follows that the image of its formal derivative $D_H \Phi$
	may be interpreted formally as the tangent space to that same balanced set.
	Since \fref{Lemma}{lemma:formal_tan_space_is_image_of_formal_derivative} above tells us that the image of $D_H \Phi$ is precisely $T_H \mathcal{B} $,
	we see that once again the notation was chosen to be suggestive: $T_H \mathcal{B} $ may be viewed formally as the tangent space to the balanced set $ \mathcal{B} $.
    This is made precise in \fref{Remark}{rmk:D_H_Phi_is_formal_derivative_of_Phi} below and is depicted in \fref{Figure}{fig:tangent_space}.
	\begin{remark}
	\label{rmk:D_H_Phi_is_formal_derivative_of_Phi}
		Formally, we may compute that, for $H \vcentcolon= \mathds{1}(q<0)$ and for any derivative $\partial$,
		\begin{equation*}
			\partial ( {\min}_0\, q ) = (\partial q) \mathds{1} (q<0) = (\partial q) H.
		\end{equation*}
		This shows that $D_H \Phi$ is, formally, the derivative of $\Phi$.
		Since $\Phi$ is a global chart for the Hilbert manifold $ \mathcal{B} $ this means that we may formally view $T_H \mathcal{B} $,
		which is the image of $ D_H \Phi$, as the tangent space of $ \mathcal{B} $.

		An alternative justification, more direct albeit still formal, is as follows.
		Let $ \mathcal{X} = (u,\,\theta,\,q) :(-1,\,1) \to \mathcal{B} $ be a curve of balanced states
		and let $ \mathcal{Z} $ denote its derivative at zero, i.e.
		\begin{equation*}
			\mathcal{Z} = (v,\,\phi,\,r) \vcentcolon= \frac{d}{ds} \mathcal{X}(s) \Big\vert_{s=0}.
		\end{equation*}
		Since $ \mathcal{X}\in \mathcal{B} $ it follows that
		\begin{equation*}
			j( \mathcal{X} ) = 0,\, w( \mathcal{X} ) = 0, \text{ and } a ( \mathcal{X} ) = 0.
		\end{equation*}
		Both $w$ and $a$ are linear, so we deduce immediately that
		\begin{equation*}
			w ( \mathcal{Z} ) = 0 \text{ and } a ( \mathcal{Z} ) = 0.
		\end{equation*}
		What about $j$? Using the formal identity $\partial ( {\min}_0\, q) = (\partial q)H$ discussed above we compute that
		\begin{equation*}
			0
			= \frac{d}{ds} j ( \mathcal{X} ) \Big\vert_{s=0}
			= \frac{d}{ds} \left[ \partial_3 u_h - \nabla_h^\perp (\theta - {\min}_0\, q) \right] \Big\vert_{s=0}
			= \partial_3 v_h - \nabla_h^\perp (\phi - rH)
			=\vcentcolon j_H ( \mathcal{Z} ).
		\end{equation*}
		This means that $ \mathcal{Z} $ belongs to $T_H \mathcal{B} $.
		In other words: if $ \mathcal{X} $ is a curve of balanced states then its formal tangent vector $ \mathcal{Z} $ lies in $T_H \mathcal{B} $.
		This warrants treating formally $T_H \mathcal{B} $ as the tangent space of the balanced set $ \mathcal{B} $.
	\end{remark}

	Note that since the chart $\Phi$ is invertible and since $D_H \Phi$ is its formal derivative,
	we may expect $D_H\Phi$ to inherit this invertibility. This is indeed the case.
	We record this result here since it will be of use later.

	\begin{lemma}
	\label{lemma:D_H_Phi_is_invertible}
		$D_H \Phi$ is invertible.
	\end{lemma}
	\begin{proof}
		This follows as in the proof that $\Phi$ is invertible (\fref{Proposition}{prop:Pih_is_bi_Lipschitz}) since
		$(1+H)(1-H/2) = 1$ and so
		\begin{align*}
			(u,\,\theta,\,q) = D_H \Phi (p,\,M)
			&\iff \left\{
				\begin{aligned}
					&u = \nabla_h^\perp p\\
					&\theta = \partial_3 p + \frac{H}{2} (M - \partial_3 p)\\
					&q = M - \partial_3 p - \frac{H}{2} (M - \partial_3 p)
				\end{aligned}
			\right.
		\\
			&\iff \left\{
				\begin{aligned}
					&-u_h^\perp + (\theta - qH) e_3 = \nabla p\\
					&u_3 = 0\\
					&\theta + q = M
				\end{aligned}
			\right.
		\\
			&\iff \left\{
				\begin{aligned}
					&p = \Delta^{-1} \nabla\cdot \left[ -u_h^\perp + (\theta-qH) e_3 \right]\\
					&u_3 = 0\\
					&M = \theta + q
				\end{aligned}
			\right.
		\\
			&\iff \left\{
				\begin{aligned}
					&p = \Delta^{-1} \left[ \nabla_h^\perp\cdot u_h + \partial_3 (\theta - qH) \right]\\
					&u_3 = 0\\
					&M = \theta + q
				\end{aligned}
			\right.
		\end{align*}
		and so the inverse is given by
		\begin{equation*}
			{(D_H \Phi)}^{-1} (u,\,\theta,\,q) = \begin{pmatrix}
				\Delta^{-1} \left[ \nabla_h^\perp\cdot u_h + \partial_3 (\theta - qH) \right]\\
				\theta + q
			\end{pmatrix}.
            \qedhere
		\end{equation*}
	\end{proof}

	Note that in particular this allows us to view $D_H\Phi$ as recording the \emph{coordinate vectors} arising from the chart $\Phi$. This is made precise in \fref{Remark}{rmk:coord_vec} below and depicted in \fref{Figure}{fig:coordinates}.
 
	\begin{remark}[Coordinate vectors]
	\label{rmk:coord_vec}
		Since $D_H\Phi$ is the formal derivative of $\Phi$ and since $D_H\Phi$ is an invertible linear map whose image is precisely the formal tangent space $T_H \mathcal{B} $,
		we may view $D_H\Phi$ as recording the coordinate vectors corresponding to $\Phi$.
		Indeed, for any $ \mathcal{Z} = (u,\,\theta,\, q)\in T_H \mathcal{B} $ there is a unique $(\pi,\,\mu) \in \mathbb{H} $ such that
		\begin{equation*}
			\mathcal{Z}
			= (D_H\Phi)(\pi,\,\mu)
			= \underbrace{\begin{pmatrix}
				\nabla_h^\perp \pi\\
				\left( 1 - \frac{H}{2} \right)\partial_3 \pi\\
				-\left( 1 - \frac{H}{2} \right)\partial_3 \pi
			\end{pmatrix}}_{ =\vcentcolon e_p (\pi) } + \underbrace{\begin{pmatrix}
				0\\
				\frac{H}{2} \mu\\
				\left( 1 - \frac{H}{2} \right) \mu
			\end{pmatrix}}_{ =\vcentcolon e_M (\mu) }
		\end{equation*}
		where $e_p$ and $e_M$ denote the coordinate ``vectors'' along the $p$ and $M$ coordinates, respectively.
		Note that these are not actually vectors, but rather coordinate \emph{operators}, namely $e_p : \mathring{H}^1 \to T_H \mathcal{B} $ and $e_M : L^2 \to T_H \mathcal{B}$
		such that $T_H \mathcal{B} = \im e_p + \im e_M$.
      This is depicted in \fref{Figure}{fig:coordinates}.
	\end{remark}

    \begin{figure}
		\centering
		\captionsetup{width=0.85\textwidth}
		\includegraphics[width=0.7\textwidth]{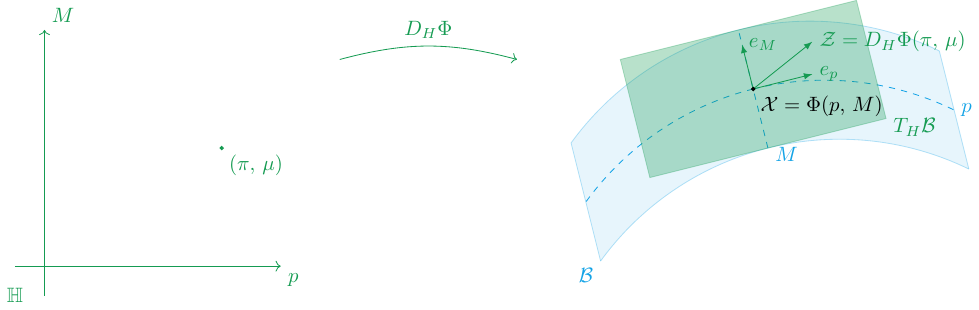}
		\caption{\small
			The ``columns'' of the formal derivative $D_H \Phi$, denoted by $e_p$ and $e_M$, may be viewed as coordinate vectors on the formal tangent space.
		}
		\label{fig:coordinates}
	\end{figure}

	A minor miracle now occurs: while $p$ and $M$ are coordinates along the balanced set $ \mathcal{B} $,
	their dual coordinates turn out to be, essentially, $PV$ and $M$!
	We see this by computing the $L^2_H$--adjoint of $D_H \Phi$.

	\begin{lemma}
	\label{lemma:adjoint_of_D_H_Phi}
		Note that the dual of $ \mathbb{H} $ is $ \mathbb{H}^* = H^{-1} \times L^2$.
		Equipping $T_H \mathcal{B} $ with the inner product $L^2_H$, the adjoint ${(D_H\Phi)}^* : T_H \mathcal{B} \to \mathbb{H}^*$ is given by
		\begin{equation*}
			{(D_H\Phi)}^* (u,\,\theta,\,q)
			= \begin{pmatrix}
				- \nabla_h^\perp \cdot u_h - \partial_3 \theta + \frac{1}{2} \partial_3 \left[ H (\theta + q) \right]\\
				\left( 1 + \frac{H}{2} \right) (\theta + q)
			\end{pmatrix}
			= \begin{pmatrix}
				- PV + \frac{1}{2} \partial_3 (HM)\\
				\left( 1 + \frac{H}{2} \right) M
			\end{pmatrix}.
		\end{equation*}
	\end{lemma}
	\begin{proof}
		For any $(p,\,M) \in \mathbb{H} $ and any $(u,\,\theta,\,q) \in T_H \mathcal{B} $ we compute by integrating by parts that
		\begin{align*}
			{\left\langle D_H \Phi (p,\,M),\, (u,\,\theta,\,q) \right\rangle }_{L^2_H} 
			&= \int u \cdot \nabla_h^\perp p
				+ \theta \left[ \partial_3 p + \frac{H}{2} (M - \partial_3 p) \right]
		\\	&\qquad
				+ q H \left[ M - \partial_3 p - \frac{H}{2} (M - \partial_3 p) \right]
				+ (\theta + q)M
		\\
			&= \int u \cdot \nabla_h^\perp p + \theta \partial_3 p - \frac{H}{2} (\theta + q) \partial_3 p + \left( 1 + \frac{H}{2} \right) (\theta + q) M
		\\
			&= \int - (\nabla_h^\perp\cdot u_h) p - (\partial_3 \theta) p + \frac{1}{2} \partial_3 \left[ H ( \theta + q ) \right] + \left( 1 + \frac{H}{2} \right) (\theta + q) M
		\\
			&= {\left\langle -\nabla_h^\perp\cdot u_h - \partial_3 \theta + \frac{1}{2} \partial_3 [ M (\theta + q)],\, p \right\rangle }_{H^{-1} \times \mathring{H}^1 } 
		\\	&\qquad
				+ {\left\langle \left( 1 + \frac{H}{2} \right) (\theta + q) ,\, M \right\rangle }_{L^2} 
		\\
			&= {\left\langle \begin{pmatrix}
				- \nabla_h^\perp \cdot u_h - \partial_3 \theta + \frac{1}{2} \partial_3 [ H ( \theta + q ) ]\\
				\left( 1 + \frac{H}{2} \right) M
			\end{pmatrix} ,\, \begin{pmatrix}
				p \\ M
			\end{pmatrix} \right\rangle }_{H^* \times H}. \qedhere
		\end{align*}
	\end{proof}

	\begin{remark}[Dual coordinate vectors]
	\label{rmk:dual_coord_vec}
		Just as in \fref{Remark}{rmk:coord_vec} we viewed $D_H \Phi$ as recording the coordinate vectors corresponding to the chart $\Phi$,
		we can view $ {(D_H\Phi)}^* $ as recording their dual coordinate one-forms.
		Indeed, if we define, for any $(u,\,\theta,\,q) \in T_H \mathcal{B} $,
		\begin{align*}
			e_p^\flat (u,\,\theta,\,q) &\vcentcolon= - PV + \frac{1}{2} \partial_3 (HM) \text{ and } \\
			e_M^\flat (u,\,\theta,\,q) &\vcentcolon= \left( 1 + \frac{H}{2} \right) (\theta + q)
		\end{align*}
	then we see that, for any $(u,\,\theta,\,q) \in T_H \mathcal{B} $ and any $(\pi,\,\mu)\in \mathbb{H} $,
	\begin{align*}
		{\langle (u,\,\theta,\,q) ,\, e_p (\pi) \rangle }_{L^2_H} &= {\langle e_p^\flat (u,\,\theta,\,q) ,\, p \rangle }_{ \mathring{H}^1 } \text{ and } \\
		{\langle (u,\,\theta,\,q) ,\, e_M (\mu) \rangle }_{L^2_H} &= {\langle e_M^\flat (u,\,\theta,\,q) ,\, M \rangle }_{ L^2 }.
	\end{align*}
    This is depicted in \fref{Figure}{fig:dual_coordinates}.
	\end{remark}

    \begin{figure}
		\centering
		\captionsetup{width=0.85\textwidth}
		\includegraphics[width=0.7\textwidth]{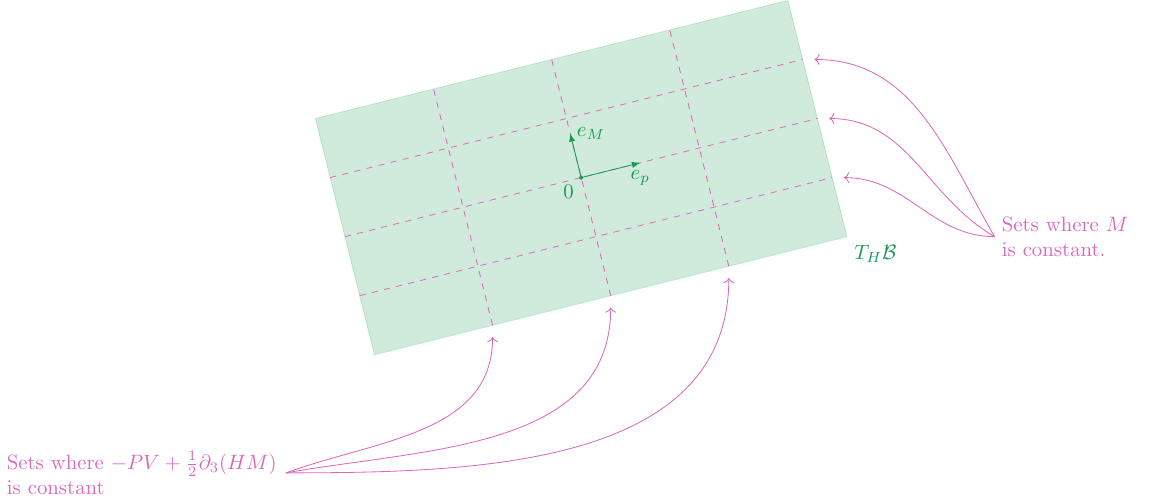}
		\caption{\small
			A depiction of the dual coordinates $e_M^\flat = \left( 1 + \frac{H}{2} \right) M$ and
			$e_p^\flat = - PV + \frac{1}{2} \partial_3 (HM)$ on the formal tangent space $T_H \mathcal{B} $.
		}
		\label{fig:dual_coordinates}
	\end{figure}

	We can see that these coordinates $p$ and $M$ are actually orthogonal.
	This is recorded below, in a result that will also come in handy when interpreting the Newton descent iteration geometrically.

	\begin{lemma}
	\label{lemma:metric_tensor}
		For any $(p,\,M) \in \mathbb{H} $,
		\begin{equation*}
			{(D_H\Phi)}^* (D_H \Phi) (p,\,M) = \begin{pmatrix}
				- \nabla\cdot (A_H \nabla p) \\
				\left( 1 + \frac{H}{2} \right) M
			\end{pmatrix}
		\end{equation*}
		for $A_H \vcentcolon= I - \frac{H}{2} e_3\otimes e_3$ (as usual).
	\end{lemma}
	\begin{remark}
	\label{rmk:metric_tensor}
		Note that $ {(D_H\Phi)}^* ( D_H \Phi) = (-\Delta) \oplus 1$ when $H = 0$,
		which is the canonical isometry from $ \mathbb{H} = \mathring{H}^1 \times L^2 $ to its dual $H^* = H^{-1} \times L^2$.
		This is because when $H=0$ the dry case may be viewed as a special case of the moist Boussinesq system which corresponds to neglecting $q$ (i.e. setting $q = 0$).
	\end{remark}
	\begin{proof}[Proof of \fref{Lemma}{lemma:metric_tensor}]
		Let $(p,\,M) \in \mathbb{H} $ and let us write $(u,\,\theta,\,q) \vcentcolon= D_H \Phi (p,\,M)$.
		Then $M(u,\,\theta,\,q) \vcentcolon= \theta + q = M$ while
		\begin{equation*}
			PV (u,\,\theta,\,q)
			\vcentcolon= \nabla_h^\perp \cdot u_h + \partial_3 \theta
			= \nabla\cdot (A_H \nabla p) + \frac{1}{2} \partial_3 (HM).
		\end{equation*}
		Therefore
		\begin{align*}
			{(D_H\Phi)}^* (D_H\Phi) (p,\,M)
			&= \begin{pmatrix}
				-PV + \frac{1}{2} \partial_3 (HM)\\
				\left( 1 + \frac{H}{2} \right) M
			\end{pmatrix}
		\\
			&= \begin{pmatrix}
				-\nabla\cdot (A_H \nabla p) - \frac{1}{2} \partial_3 (HM) + \frac{1}{2} \partial_3 (HM)\\
				\left( 1 + \frac{H}{2} \right) M
			\end{pmatrix}
		\\
			&= \begin{pmatrix}
				- \nabla\cdot (A_H \nabla p) \\
				1 + H/2
			\end{pmatrix}
		\end{align*}
		as claimed.
	\end{proof}

	\begin{remark}
	\label{rmk:coords_are_orthog}
		Written in block form,
		\begin{equation*}
			{(D_H\Phi)}^* (D_H\Phi) = \begin{pmatrix}
				-\nabla\cdot (A_H \nabla) & 0\\
				0 & 1 + H/2
			\end{pmatrix}
		\end{equation*}
		which is block-diagonal.
		This means that the coordinates $p$ and $M$ of the Hilbert manifold $ \mathcal{B} $, which arise from the global chart $ \Phi $,
		may formally be viewed as $L_H^2$--orthogonal (but not orthonormal).

		In other words, for any $ \mathcal{Z} \in T_H \mathcal{B} $, which may be written uniquely as
		\begin{equation*}
			\mathcal{Z} = e_p (\pi) + e_M (\mu)
		\end{equation*}
		for $(\pi,\,\mu) \in \mathbb{H} $ and $e_p$ and $e_M$ as in \fref{Remark}{rmk:coord_vec}, we have that
		\begin{align*}
			{\langle e_p (\pi) ,\, e_M (\mu) \rangle }_{L^2_H}
			= \int 0 
				+ \left( 1 - \frac{H}{2} \right) (\partial_3 \pi) \cdot \frac{H}{2} \mu
				- \left( 1 - \frac{H}{2} \right) (\partial_3 \pi) \cdot \left( 1 - \frac{H}{2} \right) \mu \cdot H
				+ 0
		\\
			= \int \frac{H}{4} (\partial_3 \pi) \mu - \frac{H}{4} (\partial_3 \pi) \mu
			= 0.
		\end{align*}
		This allows us to refine our comment of \fref{Remark}{rmk:coord_vec}: the decomposition $T_H \mathcal{B} = \im e_p + \im e_M$ is actually $L^2_H$--orthogonal.

		Moreover, for the dual one-forms $e_p^\flat$ and $e_M^\flat$ as defined in \fref{Remark}{rmk:dual_coord_vec} we have that
		\begin{equation}
		\label{eq:coords_are_orthog_intermediate_3}
			\mu = \theta + q \text{ and } 
			\pi \text{ solves } \nabla\cdot (A_H \nabla\pi) = \nabla_h^\perp \cdot u_h + \partial_3 \theta.
		\end{equation}
		Indeed:
		\begin{equation*}
			\left( 1 + \frac{H}{2} \right) (\theta + q)
			= e_M^\flat ( \mathcal{Z} )
			= e_M(\flat ( e_p (\pi) ) + e_M^\flat ( e_M ( \mu ) )
			= 0 + \left( 1 + \frac{H}{2} \right) \mu.
		\end{equation*}
		Similarly: on one hand
		\begin{equation}
		\label{eq:coords_are_orthog_intermediate_1}
			e_p^\flat ( \mathcal{Z} )
			= - \nabla_h^\perp \cdot u_h - \partial_3 \theta + \frac{1}{2} \partial_3 [ H (\theta + q)]
			= - \nabla_h^\perp \cdot u_h - \partial_3 \theta + \frac{1}{2} \partial_3 ( HM )
		\end{equation}
		while on the other hand
		\begin{equation}
		\label{eq:coords_are_orthog_intermediate_2}
			e_p^\flat ( e_p (\pi) ) + e_p^\flat ( e_M (\mu ) )
			= -\nabla\cdot (A_H \nabla\pi) + \frac{1}{2} \partial_3 (H\mu)
			= -\nabla\cdot (A_H \nabla\pi) + \frac{1}{2} \partial_3 (HM).
		\end{equation}
		Equating \eqref{eq:coords_are_orthog_intermediate_1} and \eqref{eq:coords_are_orthog_intermediate_2} yields the equation
		for $\pi$ recorded in \eqref{eq:coords_are_orthog_intermediate_3}, as desired.
	\end{remark}

	\begin{remark}[Comparison of the $L^2_H$--orthogonal decomposition with well-known orthogonal decompositions]
	\label{rmk:comparison_L_2_H_orthog_with_Fourier}
		In \fref{Remark}{rmk:dual_coord_vec} above we have explained how any formal tangent vector
		may be written in a unique way in terms of the ``coordinate vectors'' $e_p$ and $e_M$.

		More precisely, we know that for any formal tangent vector $ \mathcal{Z} = (u,\,\theta,\,q) \in T_H \mathcal{B} $
		there exists a unique pair $(\pi,\,\mu) \in \mathbb{H} $ such that
		\begin{align}
			\mathcal{Z} = e_p ( \pi ) + e_M ( \mu )
			&\text{ where } \left\{
				\begin{aligned} 
					&e_M^\flat ( e_M ( \mu ) ) = e_M^\flat ( \mathcal{Z} ) \text{ and } \\
					&e_p^\flat ( e_p ( \pi ) ) = e_p^\flat ( \mathcal{Z} )
				\end{aligned}
			\right.
		\label{eq:fourier_explanation_comparison}\\
			&\iff \left\{
				\begin{aligned}
					&\mu = \theta + q \text{ and } \\
					&\pi \text{ solves } \nabla\cdot ( A_H \nabla \pi ) = PV.
				\end{aligned}
			\right.
		\nonumber
		\end{align}
		We will now help the reader make better sense of this by showing how their favourite orthogonal decomposition,
		namely that of Fourier series, may be written in this way.

		Fourier series tell us that every $f\in L^2 ( \mathbb{T}^3 )$ has uniquely determined coordinates $ {(a_k)}_k \in l^2 ( \mathbb{Z}^3 )$ such that
		\begin{equation}
		\label{eq:fourier_explanation_1}
			f(x) = \sum_k a_k e^{2\pi ik\cdot x}.
		\end{equation}
		Given a function $f$, the process producing its Fourier series can be split into two steps.
		\begin{enumerate}
			\item	We compute the amplitudes $a_k$.
			\item	We reconstruct $f$ from its amplitudes.
		\end{enumerate}
		Let us rewrite each of these two steps in a manner which will make it easier to directly compare with the $L^2_H$--orthogonal decomposition discussed above.
		\begin{itemize}
			\item 	Step 1: The computation of each amplitude corresponds to a map $a_k : L^2 ( \mathbb{T}^3 ) \to \mathbb{C}$ defined by
				\begin{equation*}
					a_k (f) = \int_{\mathbb{T}^3} f(x) e^{-2\pi i k\cdot x} dx.
				\end{equation*}
			\item	Step 2: We write the reconstruction step in an admittedly unusual manner.
				This will make it easier to compare Fourier series to the $L^2_H$--orthogonal decomposition of our formal tangent vectors.
				
				\vspace{1em}
				If we view each $a_k$ as corresponding to a coordinate direction in $L^2 ( \mathbb{T}^3)$ we may then encode this coordinate direction as a map
				$e_k : \mathbb{C} \to L^2 ( \mathbb{T}^3 )$ given by
				\begin{equation*}
					[ e_k (a) ] (x) = a e^{2\pi ik\cdot k}.
				\end{equation*}
		\end{itemize}
		Combining these formulations of the two steps used to produce Fourier series we may write \eqref{eq:fourier_explanation_1} in the form
		\begin{equation}
		\label{eq:fourier_explanation_2}
			f = \sum_k e_k (a_k) \text{ where } a_k = { \langle f,\, e^{2\pi ik\cdot x} \rangle }_{L^2}.
		\end{equation}

		In particular, in the language we used to describe the formal tangent space $T_H \mathcal{B} $,
		each ``coordinate vector'' $e_k$ has a ``dual coordinate vector'' $e_k^\flat$, which is really just its adjoint and is a map
		$e_k^\flat : L^2 (\mathbb{T}^3) \to \mathbb{C}$.
		Using the $L^2$ inner product on $L^2 ( \mathbb{T}^3)$ and the standard complex inner product on $\mathbb{C}$ defined by $ {\langle x,\, y \rangle }_{\mathbb{C}} = x\bar{y}$
		for every $x,\,y\in \mathbb{C}$ we may compute these dual coordinates vectors:
		for any $k\in \mathbb{Z}^3$, any $f\in L^2 ( \mathbb{T}^3 )$, and any $a\in \mathbb{C}$,
		\begin{align*}
			{\langle e_k^\flat (f) ,\, a \rangle }_{\mathbb{C}}
			= { \langle f,\, e_k(a) \rangle }_{L^2}
			= \int_{\mathbb{T}^3} f(x) \overline{ a e^{2\pi i k\cdot x}} dx
			= \bar{a} \int_{\mathbb{T}^3} f(x) e^{-2\pi ik\cdot x} dx
			= \bar{a} \, a_k(f)
			= {\langle a_k (f),\, a \rangle }_{\mathbb{C}}.
		\end{align*}
		In other words: $e_k^\flat = a_k$.

		We may thus write \eqref{eq:fourier_explanation_2}, already an alternate formulation of \eqref{eq:fourier_explanation_1},
		in a manner even more reminiscent of our treatment of coordinate vectors and dual coordinate vectors in $T_H \mathcal{B} $, namely
		\begin{equation}
		\label{eq:fourier_explanation_3}	
			f = \sum_k e_k (a_k) \text{ where } a_k = e_k^\flat (f).
		\end{equation}

		Finally, to turn this expression into an exact parallel of our treatment of $T_H \mathcal{B} $
		we must use the orthonormality of the Fourier basis. That is, we use the fact that
		\begin{equation*}
			e_k^\flat (e_k (a))
			= a_k (e_k (a))
			= \int_{\mathbb{T}^3} a e^{2\pi ik\cdot x} e^{-2\pi ik\cdot x}dx
			= a.
		\end{equation*}
		We may therefore rewrite \eqref{eq:fourier_explanation_3} in its final form as
		\begin{equation}
		\label{eq:fourier_explanation_4}
			f = \sum_k e_k (a_k) \text{ where } e_k^\flat ( e_k (a_k )) = e_k^\flat (f).
		\end{equation}
		This is exactly the same form as \eqref{eq:fourier_explanation_comparison}!
	\end{remark}

	This section is concerned with the geometry of the balanced set $ \mathcal{B} $ on its own,
	and not with how it is embedded in the broader state space $ \mathbb{L}^2_\sigma $.
	Nonetheless, at this stage it is worth commenting on how the $L^2_H$ inner product endowed on the formal tangent space $T_H \mathcal{B} $ of $ \mathcal{B} $
	interacts with the wave set.
    The result below is depicted in \fref{Figure}{fig:wave_set_orthog_to_balanced_set}.

	\begin{lemma}
	\label{lemma:orthog_wave_and_balanced_set}
		The wave set $ \mathcal{W} $ is $L^2_H$--orthogonal to the formal tangent space $T_H \mathcal{B} $ in the sense that,
		for any $ \mathcal{Z} \in T_H \mathcal{B} $ and any $ \mathcal{Y} \in \mathcal{W}$, $ {\langle \mathcal{Z} ,\, \mathcal{Y}  \rangle }_{L^2_H} = 0$.
	\end{lemma}
	\begin{proof}
		\fref{Lemma}{lemma:formal_tan_space_is_image_of_formal_derivative} tells us that $ \mathcal{Z} = D_H \Phi (\pi,\,\mu)$
		for some $(\pi,\,mu) \in \mathbb{H} $
		while \fref{Proposition}{prop:chara_im_N} tells us that $ \mathcal{Y} = \Psi (\wavecoord,\,w)$ for some $\wavecoord \in \mathring{H}^1_\sigma$ and $w\in L^2$.
		Then
		\begin{align*}
			{\langle \mathcal{Z} ,\, \mathcal{Y}  \rangle }_{L^2_H} 
			&= {\langle D_H \Phi (\pi,\,M) ,\, \Psi (\wavecoord,\,w) \rangle }_{L^2_H}
		\\
			&= \int \nabla_h^\perp \pi \cdot \wavecoord_h^\perp + \left( 1 - \frac{H}{2} \right) (\partial_3 \pi) \cdot \wavecoord_3 + \left( 1 - \frac{H}{2} \right) (\partial_3 \pi) \cdot \wavecoord_3 \cdot H + 0
		\\
			&= \int \nabla_h \pi \cdot \wavecoord_h + (\partial_3 \pi)\wavecoord_3
			= \int \nabla \pi \cdot \wavecoord
			= 0
		\end{align*}
		since $\wavecoord$ is divergence-free.
	\end{proof}

	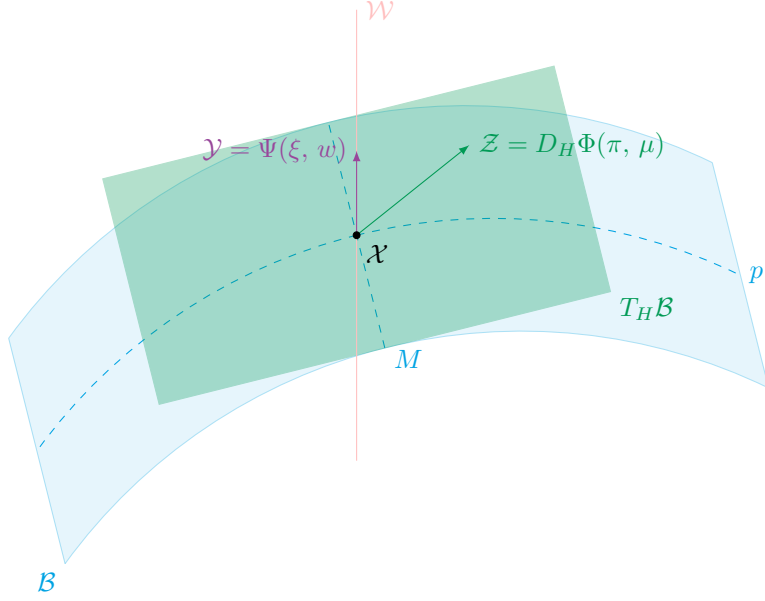
\begin{figure}
		\centering
		\captionsetup{width=0.85\textwidth}
        \begin{tikzpicture}[scale=1.5]
            \def\refx{2}
            \def\refy{-7}
        	\draw[White, opacity=0.3, name path = ManBotRig] (0.25, -1) arc (104.036:64.036:5);
        	\draw[White, opacity=0.3, name path = ManBotLef] (0.25, -1) arc (104.036:144.036:5);
        	\draw[White, opacity=0.3, name path = ManTopRig] (-0.25, 1) arc (104.036:64.036:5);
        	\draw[White, opacity=0.3, name path = ManTopLef] (-0.25, 1) arc (104.036:144.036:5);
        	\tikzfillbetween[of = ManTopLef and ManBotLef, on layer=bg]{White}
        	\tikzfillbetween[of = ManTopRig and ManBotRig, on layer=bg]{White}
        	\tikzfillbetween[of = ManTopRig and ManBotRig, on layer=bg]{Cerulean, opacity=0.1}
        	\tikzfillbetween[of = ManTopLef and ManBotLef, on layer=bg]{Cerulean, opacity=0.1}
        	\draw[Cerulean, opacity=0.3] (0.25, -1) arc (104.036:64.036:5);
        	\draw[Cerulean, opacity=0.3] (0.25, -1) arc (104.036:144.036:5);
        	\draw[Cerulean, opacity=0.3] (-0.25, 1) arc (104.036:64.036:5);
        	\draw[Cerulean, opacity=0.3] (-0.25, 1) arc (104.036:144.036:5);
        	\draw[Cerulean, opacity=0.3] (3.152, 0.645) -- (3.652, -1.355);
        	\draw[Cerulean, opacity=0.3] (-2.584, -2.914) -- (-3.084, -0.914);
        	\node[Cerulean, below left] at (-2.584, -2.914) {$\mathcal{B}$};
        	\draw[Cerulean, dashed] (0, 0) arc (104.036:64.036:5);
        	\draw[Cerulean, dashed] (0, 0) arc (104.036:144.036:5);
        	\node[Cerulean, right] at (3.402, -0.355) {$p$};
        	\draw[Cerulean, dashed] (0.25, -1) -- (-0.25, 1);
        	\node[Cerulean, right] at (0.25, -1-0.1) {$M$};
        	\draw[color=ForestGreen, fill=ForestGreen, opacity=0.3] (1.75, 1.5) -- (2.25, -0.5) -- (-1.75, -1.5) -- (-2.25, 0.5) -- cycle;
        	\node[ForestGreen, right] at (2.25, -0.5-0.15) {$T_H \mathcal{B}$};
        	\draw[ForestGreen, -latex] (0, 0) --	(1, 0.8);
        	\node[ForestGreen, right] at		(1, 0.8) {$\mathcal{Z} = D_H \Phi (\pi,\,\mu)$};
        	\draw[pink] (0, -2) -- (0, 2);
        	\node[pink, right] at (0, 2) {$\mathcal{W}$};
        	\draw[Purple, -latex] (0, 0) -- (0, 0.75);
        	\node[Purple, left] at (0, 0.75) {$\mathcal{Y} = \Psi (\xi,\,w)$};
        	\filldraw (0, 0) circle (0.03);
        	\node[below right] at (0,0) {$\mathcal{X}$};
        \end{tikzpicture}
		\caption{\small
			Depiction of \fref{Lemma}{lemma:orthog_wave_and_balanced_set}: the wave set is $L^2_H$--orthogonal
			to the formal tangent space to the balanced set.
		}
		\label{fig:wave_set_orthog_to_balanced_set}
	\end{figure}

	We are now ready to use the geometric language introduced in this section in order
	to provide a geometric interpretation of the Newton descent iteration presented in \fref{Section}{sec:convergence_guarantees}.

	First, we note that we can rewrite the defining equation of the Newton descent search direction, namely \eqref{eq:Newt_desc_S}, as follows.

	\begin{remark}
	\label{rmk:rewrite_Newton_direction}
		The Newton descent direction $\pi$ solves, for $H \vcentcolon= \mathds{1} (M < \partial_3 p)$,
		\begin{equation*}
			-\nabla\cdot (A_H \nabla\pi ) = PV ( \Phi (p,\, M)) - PV.
		\end{equation*}
		This follows from \eqref{eq:Newt_desc_S} upon observing that $\min_0 (M - \partial_3 p) = H ( M - \partial_3 p )$.

		In other words, and this will play a key role in the geometric interpretation of Newton descent:
		$\pi$ is determined by the potential vorticity discrepancy, i.e. by how far the potential vorticity at the current guess is from the true potential vorticity. 
	\end{remark}

	We are now ready to describe the geometric interpretation of Newton descent.

	\begin{remark}[Geometric interpretation of Newton descent]
	\label{rmk:geo_interp_Newton_desc}
		Consider the following formal flow on the balanced set $ \mathcal{B} $ and its formal tangent space $T_H \mathcal{B} $.
		\begin{equation}
			\left\{
			\begin{aligned}
				&\dot{\mathcal{X}} (s) = {(D_H\Phi)}^{-*} \left( PV - PV(\mathcal{X}),\, 0 \right) \text{ and } \\
				&\mathcal{X}(0) = \Phi (p_0,\, M) 
			\end{aligned}
			\right.
        \label{eq:flow}
		\end{equation}
		where $PV$ and $M$ are given and where the starting point $p_0\in \mathring{H}^1 $ may be freely chosen.

		We claim that, formally, the Newton descent iteration may be viewed as a discretisation of this flow.
		First: what does this flow even do?

		Well, since
		\begin{equation*}
			e_p^\flat = - PV + \frac{1}{2} \partial_3 (HM) \text{ and } 
			e_M^\flat = \left( 1 + \frac{H}{2} \right) M
		\end{equation*}
		we may view $PV$ and $M$ as dual coordinates along the balanced set $ \mathcal{B} $.
		The map $ {(D_H\Phi)}^{-*} : \mathbb{H}^* \to T_H \mathcal{B} $ thus associates to $ (PV,\,M) \in \mathbb{H}^*$
		the corresponding tangent vector in $T_H \mathcal{B} $.
        This means that the flow follows the $PV$--direction, keeping the value of $M$ fixed, so as to approach the given target value of $PV$.
        This is depicted in \fref{Figure}{fig:newton}.
		(An analog of this flow is presented in \fref{Remark}{rmk:simple_flow} below in a much simpler setting in order to make the intuition easier to follow.)

		To see that the Newton iteration is a discretisation of this flow we rewrite the Newton iteration in a different way.

		The Newton iteration proceed as follows.
		\begin{itemize}
			\item	$PV$ and $M$ are given.
			\item	Initialize the iteration by choosing an arbitrary $p_0\in \mathring{H}^1 $.
			\item	Iterate, for $k\geqslant 0$,
				\begin{enumerate}
					\item	Let $H_k \vcentcolon= \mathds{1} (M < \partial_3 p_k)$
						and ${PV}_k \vcentcolon= \Delta p_k + \frac{1}{2} \partial_3 {\min}_0\, (M - \partial_3 p_k)$.
					\item	$\pi$ solves $\nabla\cdot ( A_{H_k} \nabla \pi ) = PV - {PV}_k$.
					\item	Let $p_{k+1} \vcentcolon= p_k + \pi$.
				\end{enumerate}
		\end{itemize}
		Note that step 3 is written this way because, once the Newton iteration reaches the so-called ``quadratic regime''
		(which is guaranteed to occur once iterates are sufficiently close to the true solution),
		then exact line search always produces a unit step size. This is why we can update $p_{k+1} = p_k + t\pi$ with $t=1$.

		We are now ready to relate the Newton iteration to the flow above.
		We may interpret (really, rewrite) the Newton iteration geometrically as follows.
		\begin{itemize}
			\item	$PV$ and $M$ are given.
			\item	Initialize the iteration by choosing an arbitrary $p_0\in \mathring{H}^1 $,
				i.e. choosing an arbitrary $\mathcal{X}_0 = \Phi (p_0,\, M) \in \mathcal{B} $
				with the correct $M$--coordinate.
			\item	Iterate, for $k\geqslant 0$.
				\begin{enumerate}
					\item	Let $(p_k,\, M) \vcentcolon= \Phi^{-1} (\mathcal{X}_k)$
						and ${PV}_k \vcentcolon= PV ( \mathcal{X}_k )$.
					\item	$\pi$ solves $ {(D_H\Phi)}^* (D_H\Phi) (\pi,\,0) = \left( PV - {PV}_k,\, 0 \right)$.
					\item	Let $\mathcal{X}_{k+1} \vcentcolon= \Phi \left( (p_k,\,M) + (\pi,\,0) \right)$.
				\end{enumerate}
		\end{itemize}
		Note that step 2 is the same as above, merely written differently (actually: every step coincides!)
		Indeed \fref{Lemma}{lemma:metric_tensor} tells us that
		\begin{equation*}
			{(D_H\Phi)}^* (D_H\Phi) (\pi,\,0) = \left( -\nabla\cdot ( A_H \nabla \pi),\, 0 \right).
		\end{equation*}

		In other words the Newton iteration may be written in one line as
		\begin{equation*}
			\mathcal{X}_{k+1} = \Phi \left( \Phi^{-1} (\mathcal{X}_k) + (\pi,\,0) \right)
			\text{ where $\pi$ solves }
			{(D_H\Phi)}^* (D_H\Phi) (\pi,\,0) = \left( PV - {PV}_k,\, 0 \right).
		\end{equation*}
		When $\mathcal{X}_{k+1}$ is close to $\mathcal{X}_k$, which equivalently means that $\pi$ is small, we may thus formally approximate
		\begin{equation*}
			\mathcal{X}_{k+1} \sim \mathcal{X}_k + (D_H\Phi)(\pi,\,0).
		\end{equation*}
		We compute that
		\begin{align*}
			(D_H\Phi) (\pi,\,0)
			= (D_H\Phi) { \left[ {(D_H\Phi)}^* (D_H\Phi)  \right] }^{-1} \left( PV - {PV}_k,\, 0 \right)
		\\
			= {(D_H\Phi)}^{-*} \left( PV - PV(\mathcal{X}_k),\, 0 \right)
		\end{align*}
		In other words:
		\begin{equation*}
			\mathcal{X}_{k+1} \sim \mathcal{X}_k + {(D_H\Phi)}^{-*} \left( PV - PV(\mathcal{X}_k),\, 0 \right)
		\end{equation*}
		which verifies that indeed the Newton iteration is a formal discretisation of the flow in \eqref{eq:flow} above.
	\end{remark}

    \begin{figure}
		\centering
		\captionsetup{width=0.85\textwidth}
		\includegraphics[width=0.7\textwidth]{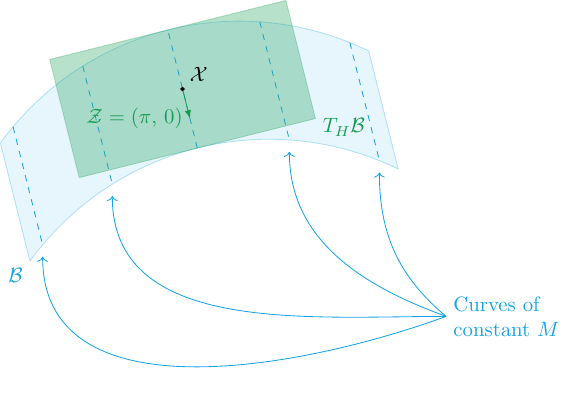}
		\caption{\small
			The Newton descent method follows curves of constant $M$ until the appropriate value of $PV$ is reached: at a point $\mathcal{X}$ on the balanced set, the flow corresponding to Newton descent follows the direction of the formal vector $\mathcal{Z}$ whose $p$--coordinate is given by the Newton descent search direction $\pi$.
		}
		\label{fig:newton}
	\end{figure}

	\begin{remark}[A simple flow]
	\label{rmk:simple_flow}
		We discuss here a very simple analog of the flow presented in \fref{Remark}{rmk:geo_interp_Newton_desc}.
		Consider the following flow on the two-dimensional plane:
		\begin{align*}
			&\dot{x}(s) = x_* - x(s)\\
			&\dot{y}(s) = 0
		\end{align*}
		with initial conditions
		\begin{align*}
			&x(0) = x_0\\
			&y(0) = y_*
		\end{align*}
		where $x_*$ and $y_*$ are given (they play the role of the given target values of $PV$ and $M$) and where $x_0$ (just like $p_0$) may be freely chosen.
		This means that the $x$--coordinate plays the role of the $PV$--coordinate whereas the $y$--coordinate plays the role of the $M$--coordinate.

		Then the new unknown
		\begin{equation*}
			\tilde{x}(s) \vcentcolon= x_* - x(s)
		\end{equation*}
		satisfies
		\begin{equation*}
			\dot{\tilde{x}}(s) = - \dot{x} (s) = - \dot{\tilde{x}}(s)
		\end{equation*}
		and so we deduce that the solution is given by
		\begin{align*}
			&\tilde{x}(s) = \tilde{x} (0) e^{-s} = (x_* - x_0) e^{-s}\\
			&y(s) = y_*.
		\end{align*}
		In terms of the original unknowns $x$ and $y$ this yields
		\begin{align*}
			&x(s) = x_* - \tilde{x}(s) = (1-e^{-s}) x_* + e^{-s} x_0\\
			&y(s) = y_*.
		\end{align*}
		Therefore the solution $(x,\,y)(s)$ approaches the point with coordinates $(x_*,\,y_*)$ exponentially fast,
		along a trajectory of constant $y$.
	\end{remark}

\section{Summary and comparison with Helmholtz and the dry case}
\label{sec:summary}

	We are approaching the conclusion of this paper.
	Before we conclude, we recall in this section that the decomposition we introduced for the moist Boussinesq system under consideration
	has many different facets.
	We summarize these many facets here.
	In order to make it as easy as possible for the reader to follow this discussion, we first recall the many facets of the Helmholtz decomposition
	and of the vortical-wave decomposition of the dry Boussinesq system.

	This is not merely a recapitulation of the discussion of these two classical systems done in the introduction.
	Throughout the paper we have introduced a variety of concepts and terminology for the sake of the moist decomposition,
	but we see that many of these have counterparts in these classical systems.
	It is these counterparts that we summarize first, in \fref{Sections}{sec:summary_Helmholtz} and \ref{sec:summary_dry_Boussinesq} below,
	before summarizing the various facets of the moist decomposition in \fref{Section}{sec:summary_moist_Boussinesq}.

	We could more prosaically call this section: ``the twelve aspects of a decomposition found in both Helmholtz, dry, and moist Boussinesq''.

\subsection{Helmholtz}
\label{sec:summary_Helmholtz}

	We first summarize the various facets of the Helmholtz decomposition.
	\begin{enumerate}
		\item	For any three-dimensional vector field $u$ the Helmholtz decomposition is
			\begin{equation*}
				u = \nabla p + \sigma
			\end{equation*}
			where $\sigma$ is a divergence-free vector field.
		\item	In the context of the Helmholtz decomposition, if we wish to use terminology from atmospheric science and talk about \emph{balanced states},
			a state $u$ will be balanced precisely when it is a solenoidal vector field.
			In other words the balance in question is
			\begin{equation*}
				u_B = \sigma,
			\end{equation*}
			where the subscript $B$ denotes the balanced component of $u$.

			While introducing this terminology in the context of the Helmholtz decomposition may appear overzealous,
			it does serve to make the parallels with the dry Boussinesq system, and ultimately with the \emph{moist} Boussinesq system, more evident.
		\item	Since the second component of the Helmholtz decomposition is the balanced component,
			by analogy with the terminology 
   of this paper we refer to the first component as the \emph{wave component}\footnote{
                This terminology may also be justified more directly as follows, without appealing to how the Helmholtz decomposition is similar to the dry and moist Boussinesq decompositions.
                Indeed, the Helmholtz decomposition may be interpreted as a slow-fast decomposition for the low Mach limit of the compressible Euler equations (see \fref{Remark}{rmk:Helmholtz_as_slow_fast_for_compressible_Euler}) and in this context the first component of the decomposition will precisely obey a wave equation.
			}.
            By analogy with \fref{Proposition}{prop:chara_im_N} we may identify the \emph{slow measurement} in this setting as the measurement which annihilates the wave set.
            Since the wave set is the set of potential vector fields this means that the slow measurements are precisely the \emph{curl} and the \emph{spatial average}:
            $\\nabla p$ is the wave component of $u$ and
			\begin{equation*}
				\nabla\times\nabla p = 0 \text{ and } \fint \nabla p = 0.
			\end{equation*}

		\item	In this case there is no subtlety to identifying the ``good unknown'' since it is the vector field $u$ itself. (See \fref{Remark}{rmk:good_unknown} for a discussion of what is meant by the term ``good unknown'' and what role they play.)
		\item	The elliptic PDE used to reconstruct the balanced component can be viewed as a div-curl system.
			Indeed $\sigma$ solves
			\begin{equation*}
				\nabla\times\sigma = \nabla\times u \text{ subject to } \nabla\cdot\sigma = 0 \text{ and } \fint\sigma = \fint u.
			\end{equation*}
            Note that the forcing is precisely the slow measurements: here the curl and the spatial average are the analogs of the potential vorticity in the dry case
            since they are slow measurements which characterise the balanced component.

			Alternatively we can view this div-curl system as a curl-squared system for the vector potential of $\sigma$.
			Writing $\sigma = c + \nabla\times\psi$ we see that $c = \fint u$ while $\psi$ solves $\nabla\times (\nabla\times\psi) = \nabla\times u$.
			(This curl-squared system may be written as a Laplacian equation since we may without loss of generality impose that the vector potential $\psi$ be divergence-free,
			and in that case the curl-squared operator acting on $\psi$ reduces to a Laplacian.)
		\item	The elliptic PDE used to reconstruct the wave component is Laplace's equation:
			\begin{equation*}
				\Delta p = \nabla\cdot u.
			\end{equation*}
            Note that it is not the fast measurement, the divergence, which plays the role of a forcing term for the wave component.
		\item	Using the decomposition we can obtain a couple of global changes of coordinates.
			We can fully describe the vector field $u$ in three ways.
			\begin{enumerate}
				\item	Describe it using $u$ itself.
				\item	Describe it using the measurements $ \nabla\cdot u$, $ \nabla\times u$, and $\fint u$.
				\item	Describe it using the coordinates $p$ and $\sigma$.
			\end{enumerate}
		\item	Recall that the Helmholtz decomposition may be viewed as a slow-fast decomposition for the low Mach number limit of the compressible Euler equations.
			In that case the leading-order operator of the dynamics is linear, and it is recorded in \fref{Remark}{rmk:Helmholtz_as_slow_fast_for_compressible_Euler}.
			We bring this up now in order to later compare the leading order operators in the case of dry and moist Boussinesq.
			For now it suffices to note that, as is discussed in more detail in \fref{Remark}{rmk:Helmholtz_as_slow_fast_for_compressible_Euler},
			the Helmholtz decomposition is a decomposition into the kernel and the image for that operator: the balanced set is its kernel and the wave set is its image.
		\item	As discussed in \fref{Remark}{rmk:Helmholtz_as_slow_fast_for_compressible_Euler}, the wave component of the Helmholtz decomposition
			can be seen to obey a wave equation.
			Once again, this is not necessarily worthy of note by itself, but we record it here to compare it with the dry and moist Boussinesq decompositions
			in \fref{Sections}{sec:summary_dry_Boussinesq} and \ref{sec:summary_moist_Boussinesq} below.
   
		\item	When viewed as a slow-fast decomposition, the Helmholtz decomposition inherits the additional feature of being orthogonal with respect to the conserved energy,
			which happens to simply be the $L^2$ norm.
		\item	We note that the Helmholtz decomposition and its induced changes of coordinates discussed above yield Parseval--type identities which allow us to rewrite
			the energy in three different ways: in terms of the state $u$, in terms of the coordinates, and in terms of the measurements. We obtain that
			\begin{align*}
				\int {\lvert u \rvert}^2 
				= \int {\lvert \nabla p \rvert}^2 + {\lvert \sigma \rvert}^2 
				= \int {\lvert \nabla\Delta^{-1} (\nabla\cdot u) \rvert}^2 + {\left\lvert {\nabla\times}^{-1} \left( \nabla\times u,\, \fint u  \right) \right\rvert}^2 .
			\end{align*}
			Here we write $\sigma = {(\nabla\times)}^{-1} (\omega,\,a)$ to mean that $\sigma$ solves
			\begin{equation*}
				\nabla\times\sigma = \omega \text{ subject to } \nabla\cdot\sigma = 0 \text{ and } \fint \sigma = a.
			\end{equation*}
        \item   Finally we observe that the extraction of the balanced component of a vector field $u$ may be viewed as a projection. Indeed, if we define the balanced set, in the context of the Helmholtz decomposition, to be $\mathcal{B}_{Hel.} \vcentcolon= {(L^2)}^3_\sigma$, the space of divergence-free three-dimensional vector fields,
            then we see that
            \begin{equation*}
                u_B = \argmin_{v\in\mathcal{B}_{Hel.}} \norm{u-v}{L^2}.
            \end{equation*}
	\end{enumerate}

\subsection{Dry Boussinesq}
\label{sec:summary_dry_Boussinesq}

	We now turn our attention to the various facets of the vortical-wave decomposition used for the dry Boussinesq system.
	Note that the numbering below is the same as the numbering in the discussion of the Helmholtz decomposition in \fref{Section}{sec:summary_Helmholtz} above.
	That is not an accident: it serves to highlight the parallels between the two decompositions.
    Then again we caution the reader that the roles played by the potential and solenoidal pieces are now flipped: $\nabla p$ appeared in the fast waves in \fref{Section}{sec:summary_Helmholtz}, it now appears in the slow balanced components; $\sigma$ appeared in the slow balanced components in \fref{Section}{sec:summary_Helmholtz}, it now appears in the fast waves.
	\begin{enumerate}
		\item	For the dry Boussinesq system discussed in this paper, the decomposition of a state consisting of a vector velocity field $u$ and a scalar potential temperature field $\theta$
			takes the form
			\begin{equation*}
				\begin{pmatrix}
					u \\ \theta
				\end{pmatrix} = \begin{pmatrix}
					\nabla_h^\perp p \\ \partial_3 p
				\end{pmatrix} + \begin{pmatrix}
					\wavecoord_h^\perp + we_3 \\ \wavecoord_3
				\end{pmatrix},
			\end{equation*}
			where $\wavecoord$ is divergence-free.
		\item	In the dry case, the balanced component satisfies the geostrophic and hydrostatic balances:
			\begin{align*}
				u_B &= \nabla_h^\perp p \text{ and } \\
				\theta_B &= \partial_3 p.
			\end{align*}
		\item	The measurement characterising the wave set is, in the dry case, the potential vorticity.
			Indeed in this case the wave set is precisely the set of states whose potential vorticity vanishes:
			\begin{align*}
				PV ( \wavecoord_h^\perp + we_3,\, \wavecoord_3)
                &= \nabla_h^\perp \cdot {( \wavecoord_h^\perp + we_3)}_h + \partial_3 \wavecoord_3
            \\
                &= \nabla_h\cdot\wavecoord_h + \partial_3 \wavecoord_3
                = 0
			\end{align*}
			since $\wavecoord$ is divergence-free.
		\item	In the dry case the good unknown is the following three-dimensional vector field combining both the horizontal velocity and the potential temperature:
			\begin{equation*}
				-u_h^\perp + \theta e_3.
			\end{equation*}
		\item	The elliptic equation used to recover the pressure $p$, and hence the balanced component, is the Laplace/Poisson equation
			\begin{equation*}
				\Delta p = PV.
			\end{equation*}
		\item	The wave component is recovered by setting $w = u_3$ and finding $\wavecoord$ to be the solution of the following div-curl system:
			\begin{equation*}
				\nabla\times \wavecoord
				= \partial_3 u_h - \nabla_h^\perp \theta + (\partial_3 w) e_3
				= j + (\partial_3 w) e_3
				\text{ subject to }
				\nabla\cdot\wavecoord = 0
				\text{ and } 
				\fint \wavecoord = - \fint u_h^\perp + \fint \theta e_3 = a.
			\end{equation*}
		\item	We may describe states of the dry Boussinesq system in three equivalent ways.
			\begin{enumerate}
				\item	Describe them using $u$ and $\theta$.
				\item	Describe them using the measurements $PV$, $j$, $w$, and $a$.
				\item	Describe them using the coordinates $p$, $\wavecoord$, and $w$.
			\end{enumerate}
		\item	The leading-order operator of the dynamics of the dry Boussinesq system considered in this paper is given by
			\begin{equation*}
				\mathcal{L} \begin{pmatrix}
					u \\ \theta
				\end{pmatrix} = \begin{pmatrix}
					\mathbb{P}_L \left( u_h^\perp - \theta e_3 \right) \\ u_3
				\end{pmatrix}.
			\end{equation*}
			The decomposition is adapted to this operator since the balanced set is its kernel and the wave set is its image.
		\item	We may compute immediately that the wave component of a state undergoes linear oscillations under the leading-order operator since, when restricted to the wave set, the square of the operator behaves like multiplication by $-1$:
			\begin{equation*}
				\mathcal{L}^2 \begin{pmatrix}
					\wavecoord_h^\perp + we_3 \\ \wavecoord_3
				\end{pmatrix} = -\begin{pmatrix}
					\wavecoord_h^\perp + we_3 \\ \wavecoord_3
				\end{pmatrix}.
			\end{equation*}
		\item	The conserved energy in the dry case is simply the $L^2$ norm given by
			\begin{equation*}
				E = \int {\lvert u \rvert}^2 + \theta^2.
			\end{equation*}
		\item	The Parseval-type identity in the dry case lets us rewrite the energy in terms of the coordinates or the measurements as follows:
			\begin{equation*}
				\int {\lvert u \rvert}^2 + \theta^2
				= \int {\lvert \nabla p \rvert}^2  + {\lvert \wavecoord \rvert}^2 + w^2
				= \int {\lvert \nabla\Delta^{-1} PV \rvert}^2 + {\left\lvert {\nabla\times}^{-1} \left( j + (\partial_3 w) e_3,\, a \right) \right\rvert}^2.
			\end{equation*}
        \item   For the dry Boussinesq case the balanced set takes the form
                \begin{equation*}
                    \mathcal{B}_{dry} \vcentcolon= \left\{
                        (u,\,\theta) \in {(L^2)}^3_\sigma \times L^2 :
                        u = \nabla_h^\perp p,\, u_3 = 0, \text{ and } \theta = \partial_3 p
                        \text{ for some } p\in\mathring{H}^1
                    \right\}.
                \end{equation*}
                Therefore the extraction of a balanced component once again has a projection perspective:
                \begin{equation*}
                    (u_B,\,\theta_B) = \argmin_{\mathcal{Y} \in \mathcal{B}_{dry}} \norm{(u,\,\theta) - \mathcal{Y}}{L^2}.
                \end{equation*}
	\end{enumerate}

\subsection{Moist Boussinesq}
\label{sec:summary_moist_Boussinesq}

	Finally we perform the same exercise as in \fref{Sections}{sec:summary_Helmholtz} and \ref{sec:summary_dry_Boussinesq}:
	we list the twelve aspects of the moist decomposition presented in this paper, in the same order as the above sections to make the comparison easy to follow.

	Recall that here, in a manner consistent with the rest of this paper (barring the introduction) we write $\theta$ to denote the \emph{equivalent} potential temperature
	(denoted $\theta_e$ in the introduction).
	\begin{enumerate}
		\item	The decomposition of a state $(u,\,\theta,q)$ is given by
			\begin{equation*}
				\begin{pmatrix}
					u \\ \theta \\ q
				\end{pmatrix} = \begin{pmatrix}
					\nabla_h^\perp \\
					\partial_3 p + \frac{1}{2} {\min}_0\, (M - \partial_3 p)\\
					M - \partial_3 p - \frac{1}{2} {\min}_0\, (M - \partial_3 p)
				\end{pmatrix} + \begin{pmatrix}
					\wavecoord_h^\perp + we_3 \\ \wavecoord_3 \\ -\wavecoord_3
				\end{pmatrix},
			\end{equation*}
			where $p$ has vanishing average and $\wavecoord$ is divergence-free.
		\item	In the moist case the balanced component satisfies the geostrophic and hydrostatic balances, meaning that
			\begin{align*}
				u_B &= \nabla_h^\perp p \text{ and } \\
				\theta_B - {\min}_0\, q_B &= \partial_3 p.
			\end{align*}
			To fully characterise the balanced component we must also include the moist variable $M$ and the following ``auxiliary'' balance:
			\begin{equation*}
				\theta_B + q_B = M.
			\end{equation*}
   
		\item	The slow measurements characterising the wave set are now $PV$ and $M$, and indeed
			\begin{equation*}
				PV ( \wavecoord_h^\perp + we_3,\, \wavecoord_3,\, -\wavecoord_3 ) = \nabla\cdot\wavecoord = 0 \text{ and } 
				M ( \wavecoord_h^\perp + we_3,\, \wavecoord_3,\, -\wavecoord_3 ) = \wavecoord_3 - \wavecoord_3 = 0.
			\end{equation*}
		\item	There are \emph{two} good unknowns in the moist case, owing to the fact that the leading-order operator is \emph{not} skew-symmetric (this is discussed in more detail in \fref{Remark}{rmk:global_chg_coord_dry_case}).
            One good unknown is
			\begin{equation*}
				-u_h^\perp + (\theta - {\min}_0\, q) e_3,
			\end{equation*}
			which is related to $j$ and the null set of the leading-order operator,
			while its ``adjoint'' good unknown is
			\begin{equation*}
				-u_h^\perp + \theta e_3,
			\end{equation*}
			which is related to $PV$ and the image of the leading-order operator.
		\item	In order to reconstruct the balanced component we set $M = \theta + q$ and then let $p$ be the solution of nonlinear $PV$-and-$M$ inversion
			\begin{equation*}
				\Delta p + \frac{1}{2} \partial_3 {\min}_0\, (M - \partial_3 p) = PV.
			\end{equation*}
		\item	To reconstruct the wave component we let $w = u_3$ and let $\wavecoord$ solve
			\begin{align*}
				\nabla \times \left( A_B^{-1} \wavecoord \right)
				= \partial_3 u_h - \nabla_h^\perp (\theta - \theta H_B) + (\partial_3 w) e_3
				=\vcentcolon j_B + (\partial_3 w) e_3
			\\
				\text{ subject to }
				\nabla\cdot\wavecoord = 0
				\text{ and } 
				\fint A_B^{-1} \wavecoord = - \fint u_h^\perp + \fint (\theta - qH_B) e_3,
			\end{align*}
			where $A_B^{-1} = I + H_B e_3\otimes e_3$ for $H_B \vcentcolon= \mathds{1} (M < \partial_3 p)$.
		\item	We can describe moist Boussinesq states in three ways.
			\begin{enumerate}
				\item	Describe them using $u$, $\theta$, and $q$.
				\item	Describe them using the measurements $PV$, $M$, $j$, $w$, and $a$.
				\item	Describe them using the coordinates $p$, $M$, $\wavecoord$, and $w$.
			\end{enumerate}
		\item	The leading-order operator of the dynamics in the moist case is
			\begin{equation*}
				\mathcal{N} \begin{pmatrix}
					u \\ \theta \\ q
				\end{pmatrix} = \begin{pmatrix}
					\mathbb{P}_L \left[ u_h^\perp - (\theta - {\min}_0\, q ) e_3 \right]\\
					u_3\\
					-u_3
				\end{pmatrix}.
			\end{equation*}
		\item	As detailed in \fref{Section}{sec:slow_fast_decomp} the operator of interest acting on the wave set, denoted by $\mathcal{L}_{W,\,H_B}$,
			is the linearisation of $\mathcal{N}$ about a balanced state with phases characterised by $H_B$. We then observe that the restriction of this operator to the wave set produces oscillatory dynamics.
			By contrast with the case of the Helmholtz decomposition and of the dry Boussinesq system, the frequencies of these oscillations are \emph{phase-dependent}:
			\begin{equation*}
				\mathcal{L}_{W,\,H_B}^2 \begin{pmatrix}
					\wavecoord_h \\ 0 \\ 0
				\end{pmatrix} = - \begin{pmatrix}
					\wavecoord_h \\ 0 \\ 0
				\end{pmatrix} \text{ and } \mathcal{L}_{W,\,H_B}^2 \begin{pmatrix}
					we_3 \\ \wavecoord_3 \\ -\wavecoord_3
				\end{pmatrix} = - (1 + H_B) \begin{pmatrix}
					we_3 \\ \wavecoord_3 \\ -\wavecoord_3
				\end{pmatrix}.
			\end{equation*}
		\item	In the moist case the conserved energy is no longer quadratic.
			It is now nonlinear and given by
			\begin{equation*}
				E = \int {\lvert u \rvert}^2 + \theta^2 + {\min}_0^2\, q + {(\theta + q)}^2.
			\end{equation*}
		\item	In the moist case there is no simple Parseval-type identity which holds: the entirety of \fref{Section}{sec:energy_and_decomposition} is devoted to finding an alternative!
			Ultimately we identify in \fref{Section}{sec:metric_less_approach} a \emph{statistical divergence} given by
			\begin{equation*}
				D ( \mathcal{X}_1;\, \mathcal{X}_2 )
				= \frac{1}{2} \int A_{H_1} \nabla p \cdot \nabla p + \left( 1 + \frac{H_1}{2} \right) {(M_1 - M_2)}^2 + A_{H_1}^{-1} \wavecoord\cdot\wavecoord + {(w_1 - w_2)}^2
			\end{equation*}
			where, for $H_1 \vcentcolon= \mathds{1} (q_{B,\,1} < 0)$,
			\begin{itemize}
				\item	$p$ solves $\nabla\cdot \left( A_{H_1} \nabla p \right) = {PV}_1 - {PV}_2 - \frac{1}{2} \partial_3 \left[ H_1 \left( M_1 - M_2 \right) \right]$ and
				\item	$\wavecoord$ solves $\nabla\times \left( A_{H_1}^{-1} \wavecoord \right) = j_1 - j_2 + \partial_3 (w_1 - w_2) e_3$
					subject to $\nabla\cdot\wavecoord = 0$ and $\fint A_{H_1}^{-1} \wavecoord = a_1 - a_2$.
			\end{itemize}
        \item   The projection perspective in the moist case takes the following form:
                \begin{equation*}
                    \mathcal{X}_B = \argmin_{\mathcal{Y} \in \mathcal{B}} D(\mathcal{X};\;\mathcal{Y}).
                \end{equation*}
	\end{enumerate}

    Above we have summarized in detail how the many aspects of the Helmholtz, dry, and moist Boussinesq decompositions are analogous to one another, highlighting their differences.
    Below we provide a table summarizing \fref{Sections}{sec:summary_Helmholtz}--\ref{sec:summary_moist_Boussinesq} and which may serve either as a reference or as a more visual way to compare and contrast the various aspects of each decomposition.
    
	\begin{sidewaystable}
		{\small
		\begin{tabular}{|c|c|c|c|}
			\hline
						&Helmholtz			& Dry Boussinesq			& Moist Boussinesq							\\ \hline
			Decomposition		&$u = \nabla p + \sigma$
										&$ \begin{pmatrix}
											u \\ \theta
										\end{pmatrix} = \begin{pmatrix}
											\nabla_h^\perp p \\ \partial_3 p
										\end{pmatrix} + \begin{pmatrix}
											\wavecoord_h^\perp + we_3 \\ \wavecoord_3
										\end{pmatrix}$
															& $\begin{pmatrix}
																u \\ \theta \\ q
															\end{pmatrix} = \begin{pmatrix}
																\nabla_h^\perp p \\
																\partial_3 p + \frac{1}{2} {\min}_0\, (M - \partial_3 p)\\
																M - \partial_3 p + \frac{1}{2} {\min}_0\, (M - \partial_3 p)
															\end{pmatrix} + \begin{pmatrix}
																\wavecoord_h^\perp + we_3 \\ \wavecoord_3 \\ -\wavecoord_3
															\end{pmatrix}$								\\ \hline
            Balanced set & $\mathcal{B}_{Hel.} = {(L^2)}^3_\sigma$ & $\mathcal{B}_{dry} = \left\{
                        \begin{aligned}
                            &u = \nabla_h^\perp p,\, u_3 = 0, \text{ and }
                        \\
                            & \theta = \partial_3 p\text{ for some } p\in\mathring{H}^1
                        \end{aligned}
                    \right\}$
            &
			$\mathcal{B} = \left\{
                \begin{aligned}
    				&u_h = \nabla_h^\perp p,\, u_3 = 0,\,
                    \\&\text{and } \theta - {\min}_0\, q = \partial_3 p
    				\\&\text{ for some } p\in \mathring{H}^1
                \end{aligned}
			\right\}$
            \\ \hline
			Slow measurements, which	&$\nabla\times u$	&$PV = \nabla_h^\perp\cdot u_h + \partial_3\theta$		&$PV = \nabla_h^\perp\cdot u_h + \partial_3 \theta$  					\\ 
			characterise balanced components&	and $\fint u$		&					& and $M = \theta+q$							\\ \hline
			Fast measurements, which	&$\nabla\cdot u$	&$j = \partial_3 u_h - \nabla_h^\perp \theta$,		& $j = \partial_3 u_h - \nabla_h^\perp ( \theta - {\min}_0\, q) e_3$,	\\
			characterise wave components	&&$a= - \fint u_h^\perp + \fint \theta e_3$, 		& $a = -\fint u_h^\perp + \fint (\theta - {\min}_0\, q) e)_3$,		\\
							&			&and $w=u_3$						& and $w = u_3$								\\ \hline
			Elliptic PDE for the		&$\nabla\times\sigma = \nabla\times u$		&$\Delta p = PV$			& $\Delta p + \frac{1}{2} \partial_3 {\min}_0\, (M - \partial_3 p) = PV$	\\
			balanced component		&subject to $\nabla\cdot\sigma = 0$						&					&								\\
                                    &and $\fint \sigma = \fint u$                           &                   &                               \\ \hline
			Elliptic PDE for the		&$\Delta p = \nabla\cdot u$	&$\nabla\times\wavecoord = j + (\partial_3 w)e_3$		& $\nabla\times \left( A_B^{-1} \wavecoord \right)
																			  = j_B + (\partial_3 w) e_3$		\\
			wave component			&&subject to $\nabla\cdot\wavecoord = 0$ and			& subject to $\nabla\cdot\wavecoord = 0$ and\\
							&		&$\fint \wavecoord = -\fint u_h^\perp + \fint \theta e_3$	& $\fint A_B^{-1} \wavecoord
																			  = -\fint u_h^\perp + \fint (\theta - qH_B) e_3$ \\ \hline
            Projection perspective
            & $u_B = \argmin\limits_{v\in\mathcal{B}_{Hel.}} \norm{u-v}{L^2}$
            & $(u_B,\,\theta_B) = \argmin\limits_{\mathcal{Y} \in \mathcal{B}_{dry}} \norm{(u,\,\theta) - \mathcal{Y}}{L^2}$
            &  $\mathcal{X}_B = \argmin\limits_{\mathcal{Y} \in \mathcal{B}} D(\mathcal{X};\;\mathcal{Y})$
            \\ \hline
			Global changes of coordinates	&State: $u$								& State: $u$ and $\theta$		& State: $u$, $\theta$, and $q$\\
							&Measurements: $\nabla\cdot u$, $\nabla\times u$, and $\fint u$		& Measurements: $PV$, $j$, $w$, and $a$	& Measurements: $PV$, $M$, $j$, $w$,
																					  and $a$\\
							&Coordinates: $p$ and $\sigma$					& Coordinates: $p$, $\wavecoord$, and $w$	& Coordinates: $p$, $M$, $\wavecoord$, and $w$
																					\\ \hline
			Leading-order operator		&See \fref{Remark}{rmk:Helmholtz_as_slow_fast_for_compressible_Euler}
										& $\mathcal{L} \begin{pmatrix}
											u \\ \theta
										\end{pmatrix} = \begin{pmatrix}
											\mathbb{P}_L (u_h^\perp - \theta e_3) \\ u_3
										\end{pmatrix}$				& $\mathcal{N} \begin{pmatrix}
																u \\ \theta \\ q
															\end{pmatrix} = \begin{pmatrix}
																\mathbb{P}_L \left[ u_h^\perp - (\theta - {\min}_0\, q) e_3 \right]
																\\ u_3 \\ -u_3
															\end{pmatrix}$									\\ \hline
			Conserved energy		& See \fref{Remark}{rmk:Helmholtz_as_slow_fast_for_compressible_Euler}
								& $E = \int {\lvert u \rvert}^2 + \theta^2$		& $E = \int {\lvert u \rvert}^2 + \theta^2 + {\min}_0^2\, q + {(\theta + q)}^2$	\\ \hline
		\end{tabular}
		}
	\end{sidewaystable}

\section{Conclusion}
\label{sec:conclusion}

	In this section we discuss the main contributions of this paper and the main takeaways for the reader.

	The core contribution of this paper is the following.
	We provided a blueprint for decomposing the solutions of $ \partial_t + \mathcal{N} $, specifically when $ \mathcal{N} $ is nonlinear.
	Recall that otherwise, if $ \mathcal{N} $ is linear, then we can proceed via spectral methods and identify eigenvalues and eigenvectors of that operator.

	This blueprint requires a good understanding of the null set and the image of $ \mathcal{N} $,
	which we refer to as the balanced set and the wave set in the context of this paper.

	We actually provide \emph{two} descriptions of each of these sets:
	(1) an explicit parametrisation and (2) a characterisation as the joint zero set of some measurements.
	Note that the measurements annihilating the wave set are known as slow measurements
	while those annihilating the balanced set are known as fast measurements.

	A crucial condition to check in order to follow this general blueprint is that the slow measurements be characteristic of balanced states.
	In other words: there must be only one balanced set with any given $PV$ and $M$.
	Geometrically this is asking that the wave set (i.e. the joint zero set of $PV$ and $M$) be \emph{transversal} to the balanced set.

	This allows us to define a global decomposition of state space:
	\begin{equation*}
		\mathbb{L}^2_\sigma = \mathcal{B} + \mathcal{W}.
	\end{equation*}

	In particular this decomposition gives broader meaning to nonlinear $PV$-and-$M$ inversion.
	Previously, this inversion was understood in the context of balanced states:
	given $PV$ and $M$, the inversion was used to produce the corresponding balanced state.
	Now this inversion has relevance for \emph{any} state: it is used to produce its balanced \emph{component}.

	We then discuss three aspects of this decomposition.

	\begin{enumerate}
		\item	The first aspect we discuss is its interpretation as a slow-fast decomposition.
		We obtain the three-part decomposition
		\begin{equation}
		\label{eq:three_part_decomp}
			\mathcal{X} = 
			\overunderbraces{
				&\br{1}{\text{balanced}} 	&	&\br{3}{\text{wave}}
			}{
				&\mathcal{X}_B			& +	&\overline{ \mathcal{X} }_W	& +	&\left( \mathcal{X}_W - \overline{ \mathcal{X} }_W \right)
			}{
				&\br{3}{\text{slow}}							&	&\br{1}{\text{fast}}
			}
		\end{equation}
		where $ \overline{ \mathcal{X} }_W $ denotes the time average of the wave component.

		Contrast this with the dry case: in that case the wave component undergoes linear oscillations,
		their time averages thus vanish, and so they do \emph{not} contribute to the slow component at lowest order (they do contribute, a little, when the small parameter $\varepsilon$ is finite).
		Here, in the moist case, the wave component undergoes nonlinear oscillations and so its time average contributes to the slow component.

		\item	The second aspect of the decomposition we discuss is its connection to the conserved energy.
		(Recall: when the operator $ \mathcal{N} $ is linear, the eigendecomposition obtained often has favourable orthogonality properties,
		e.g. with respect to the inner product generated by the conserved energy as is the case in the case for the dry Boussinesq system.)
		We obtain a statistical divergence which quantifies the discrepancy between states in a manner consistent with both the conserved energy and the decomposition.

		This statistical divergence is inpired by the Parseval-type identities which hold for the Helmholtz decomposition and the dry Boussinesq decomposition.
		For example in the dry case we ca write the energy either in terms of the state $(u,\,\theta)$ or in terms of the measurements $(PV,\,j,\,w,\,a)$.
		We proceed similarly in the moist case: we use the nonlinear change of variables from the state $(u,\,\theta,\,q)$ to the measurements $(PV,\,M,\,j,\,w,\,a)$
		to define the statistical divergence in terms of these measurements.

		\item	The third aspect of the decomposition we discuss is how to compute it in practice.
		We introduce iterative methods that may be used to compute the decomposition.
		(These methods produce the balanced component $ \mathcal{X}_B $ and the wave component may then be obtained as a residual $ \mathcal{X}_W \vcentcolon= X - \mathcal{X}_B$.)
		One of these methods, the Newton descent method, has the added benefit of shedding light on past work on this moist Boussinesq system.

		This Newton descent iteration corresponds to a previously used numerical method.
		Our analysis here thus provides the first rigorous proof that this numerical method converges.

		This Newton descent iteration may also be interpreted geometrically: it is precisely the discretisation of the flow along the balanced set where the moist variable $M$
		is kept constant and where we follow the potential vorticity coordinate direction until the desired potential vorticity is attained.
	\end{enumerate}

	\fref{Table}{tab:main_results} below serves as an easy reference for the main results of this paper.
	This table also lists a few important auxiliary results, remarks, and definitions which either highlight the key ideas underpinning the main results or provide essential context for these results.

	\begin{table}
		\centering
		\begin{tabular}{ll}
			\textbf{Result 1:} Decomposition			& \fref{Theorem}{theorem:pre_decomp}\\
			The balanced set					& \fref{Proposition}{prop:alt_char_balanced_set}\\
			The wave set						& \fref{Proposition}{prop:chara_im_N}\\
			Nonlinear change of variables				& \fref{Proposition}{prop:global_chg_coord}\\
			\hline
			\textbf{Result 2:} Slow-fast interpretation		& \fref{Theorem}{theorem:slow_fast_decomposition}\\
			of the decomposition\\
			\hline
			\textbf{Result 3:} Statistical divergence 		& \fref{Proposition}{prop:stat_div}\\
			agreeing with the decomposition\\
			and the conserved energy\\
			Definition of the					& \fref{Definition}{def:our_stat_div}\\
			statistical divergence					& (and \fref{Definition}{def:dimensionalized_Parseval_distance})\\
			\hline
			\textbf{Result 4:} Convergence rate for			& \fref{Theorem}{thm:conv_rates_our_en_backtracking}\\
			the Newton descent iteration\\
			The Newton descent					& \fref{Proposition}{prop:Newt_desc}\\
			search direction\\
			Justification of a numerical method			& \fref{Remark}{rmk:interp_Newt_descent_early}\\
			used in the literature\\
			Geometric interpretation				& \fref{Remark}{rmk:geo_interp_Newton_desc}
		\end{tabular}
		\caption{\small The main results of this paper and some important auxiliary results.}
		\label{tab:main_results}
	\end{table}

	We now pause to highlight the main difficulty in constructing a decomposition with the favourable properties discussed above.
	Naively, the difficulty may appear to be the nonlinearities or the lack of skew-symmetry.
	Both of these challenges come from the phase boundaries, and so the real difficulty is more specifically this: the phase boundaries are \emph{dynamic}, meaning that they evolve over time.

	To explain this we consider what would happen if this difficulty were removed.
	That is, suppose that $(u,\,\theta,\,q)$ is a solution of the moist Boussinesq system \eqref{eq:moist_Boussinesq_u}--\eqref{eq:moist_Boussinesq_q} whose phase boundaries do \emph{not} evolve over time.
	In other words, if $H \vcentcolon= \mathds{1} (q<0)$, suppose that $ \partial_t H = 0$.

	Then we may treat $H$ as a fixed indicator function and observe that $ \mathcal{N} (u,\,\theta,\,q) = \mathcal{L}_H (u,\,\theta,\,q)$.
	Crucially: $\mathcal{L}_H$ is $L^2_H$--skew-symmetric (see \fref{Lemma}{lemma:L_H_is_L_2_H_skew}).
	This means that, when the phase boundaries are \emph{static}, the leading-order operator becomes both linear \emph{and} skew-symmetric.
	Both symptomatic difficulties (nonlinearities and lack of skew-symmetry) are removed when the root difficulty (dynamic phase boundaries) is removed.

\subsection{Looking forward}
\label{sec:look_forward}

In this section we record some thoughts on future research direction stemming from the topics discussed here.

	\begin{itemize}
		\item	Can we prove that the PQG system (see \cite{smith2017precipitating}), or a variant thereof is the rigorous fast-wave averaging limit of the moist Boussinesq system, in the limit of fast rotation and strong stratification
			(i.e. as both the Froude and Rossby numbers approach zero, at the same rate)?

            In the case of balanced, or well-prepared initial data we expect to be obtain to obtain the PQG system as a rigorous limit of the moist Boussinesq system discussed here.

            However, in the case of unbalanced, or ill-prepared initial data, there is mounting evidence that the appropriate limiting system is not the PQG system.
            The crux of the matter are the nonlinear oscillations of the fast waves:
            for ill-prepared data these waves are present and in the system and, in the moist case, they feed back into the slow components through their average. We thus expect that, at the very least, a modification needs to be made to the PQG system if we hope to obtain it as a rigorous limit of the moist Boussinesq system for generic initial data.
            
			Since this is a challenging problem, and here are two more restricted questions to look at first and whose resolution will shed light on the general case.
			\begin{itemize}
				\item	Treat the ODE case first. The ODE problem is recorded in \fref{Lemma}{lemma:ODE_sols_are_PDE_sols_general}, which also shows that ODE solutions are PDE solutions.
					Crucially: this ODE system possesses one of the key difficulties found in the PDE case: its wave solutions, i.e. its fast components,
					are nonlinear oscillations (which are piecewise linear oscillations, meaning that depending on the sign of $q$ they have fixed frequencies,
					but their frequencies change as the sign of $q$ changes) with \emph{non-zero} averages.
				\item	Treat the case of steady interfaces. In other words, look at solutions to the full moist Boussinesq system whose interface (i.e. the zero level set of $q$)
					is fixed in time and does not evolve dynamically, even though $q$ itself evolves dynamically.
					Such solutions are readily obtained by considering velocity fields whose support is localized away from the interface.
					In this case, the Heaviside $H = \mathds{1} (q<0)$ which appears in the moist Boussinesq system can be considered to be fixed,
					and so the leading-order operator is linear!
					This removes a key difficulty from the general, nonlinear, moist Boussinesq system due to dynamically-evolving phase boundaries.

                    It is important to note what we expect to glean from studying this ``fixed $H$'' case. Indeed, if we consider this example from a physical perspective it is severely contrived and hence there is little hope of extracting meaningful physical intuition from this test case. Nonetheless, studying this ``fixed $H$'' case may prove helpful in isolating the mathematical difficulties one at a time, and this is where we may benefit from studying this simpler case first before handling the general case.
			\end{itemize}
			Note that carrying out the fast-wave averaging process will require a precise identification of the solvability condition for the next-order equations,
			something which we have not carried out, even formally, in this paper.
			Recall that this next-order solvability condition is used to determine the behaviour of the ``background profile'' at the slow timescale which arises
			from the solution of the leading-order equation.
		\item	Can we find an explicit example of a PDE state for which its Snell-type distance projection does not agree with our decomposition?
			We strongly suspect that such examples exist since they are easy to find in the ODE case.
        \item 
            In \fref{Section}{sec:geometry} we discussed how the iterative method
            used to solve nonlinear $PV$-and$M$
            inversion numerically could be interpreted geometrically.
            This interpretation viewed the iterative method as the discretisation
            of a flow \emph{along} the balanced set.
            It may be possible to re-interpret this same iterative method
            in a different way as follows.
            First, the detailed geometric understanding of the balanced set,
            in terms of a chart, coordinate vectors, and dual coordinate vectors,
            would need to be extended to the entire state space,
            thus describing both the balanced set \emph{and} the wave set.
            We may then be able to re-interpret this same iterative method geometrically
            as a flow in state space which starts at some (possibly unbalanced) initial state
            and flows \emph{towards} the balanced set,
            only reaching it when the iterative method terminates.
            The appeal of such an interpretation would be that it matches the following physical intuition:
            a balanced state is obtained by starting from \emph{any} state and then progressively removing
            wave components (i.e. imbalances) until we reach the balanced set.
        \item The two metrics/distance functions defined in \fref{Sections}{sec:PDE_centric_approach} and \ref{sec:energy_centric_approach} may potentially be used for other applications, such as adjoint models or data-driven models.
        The ``fixed $H$'' energy from \fref{Section}{sec:PDE_centric_approach} appears in situations where linearisation is used,
        and hence it could find application to adjoint models,
        sensitivity analysis, and variational data assimilation,
        which rely on tangent linear approximations to the
        dynamics
        \cite{le1986variational,park1997validity,errico1997adjoint,errico1999examination,mahfouf1999influence,barkmeijer2001tropical,amerault2008tests,doyle2014initial}.
        The ``fixed $H$'' energy also appears here in the numerical/iterative methods, in \fref{Section}{sec:geometry} where
        geometrical notions relying on tangent spaces (a.k.a. linearisations) are discussed.
        The Snell-like distance in \fref{Section}{sec:energy_centric_approach} here 
        could be useful for applications such as
        data-driven methods that identify nonlinear features 
        in dynamical systems and utilize a metric, distance function,
        and/or kernel function
        \cite{coifman2006diffusion,giannakis2012nonlinear,berry2016variable,das2019delay,giannakis2019data}.
        \item Can the ocean (or other geophysical or astrophysical fluids) have a nonlinear eigenspace decomposition, similar to the atmospheric decomposition presented here? The nonlinearity in the buoyancy here was a nonlinear switch due to clouds, but the blueprint of the techniques here may also be applicable for other types of nonlinearity. For instance, in oceanography, a nonlinear equation of state describes the nonlinear dependence of oceanic buoyancy on temperature and salinity, and plays an important role in a variety of processes \cite{klocker2010influence,hieronymus2013buoyancy,nycander2015nonlinear,roquet2015defining,korn2021global}. Also, PV inversion is used in oceanography as well \cite{muller1995ertel,ochoa2001geostrophy}, and it would be interesting to investigate whether the nonlinear eigenspace decomposition here may be applicable for incorporating the effects of a nonlinear equation of state for seawater.
	\end{itemize}


\section*{Acknowledgments}

The authors are grateful for the support of NSF DMS-1907667.
A.R.-T. gratefully acknowledges support from the Simons Foundation, the Isaac Newton Institute for Mathematical Sciences, and the Department of Applied Mathematics and Theoretical Physics of the University of Cambridge where part of this work was carried out.

\appendix
\section{Appendix}
\label{app:misc}

\subsection{Helmholtz decomposition}
\label{app:Helmholtz}

    In this section we record a precise statement of the Helmholtz decomposition.
    This is a well-known result which plays an important role in this paper and so we record it here, if only to record it in a manner with the notation consistent with the notation in the rest of this paper.

	\begin{theorem}[Helmholtz]
	\label{theorem:Helmholtz}
		For every $u\in L^2$ there exists a unique $p \in \mathring{H}^1 $ and a unique $\sigma \in L^2_\sigma$ such that $u = \sigma + \nabla p$.
	\end{theorem}
	\begin{proof}
		We define $p$ to be the solution of $\nabla p = \nabla\cdot u$, $\psi$ to be the solution of $\nabla\times (\nabla\times\psi) = \nabla\times u$,
		and let $c = \fint u$.
		Then $u = \nabla p + \sigma$ for $\sigma = \nabla\times \psi + c$ (as can be directly verified using the Fourier transform since here we are working on the three-dimensional torus).
	\end{proof}

    In particular we will use the following corollary of the Helmholtz decomposition.
 
	\begin{cor}
	\label{cor:Helmholtz}
		For every $u\in L^2$, if $\nabla\times u = 0$ and $\fint u = 0$ then there exists a unique $p\in \mathring{H}^1 $ such that $u = \nabla p$.
	\end{cor}
	\begin{proof}
		It follows from the Helmholtz decomposition that $u = \nabla p + \sigma$ for some uniquely determined $p$ and $\sigma$.
		Since $\sigma$ is uniquely determined by the curl and the average of $u$, which both vanish, we see that indeed $\sigma = 0$.
	\end{proof}

    To conclude this section we discuss, as was alluded to in the introduction, how
    the Helhmoltz decomposition may be viewed as a slow-fast decomposition for the compressible Euler equations.

	\begin{remark}[The Helmholtz decomposition is a slow-fast decomposition for the compressible Euler equations]
	\label{rmk:Helmholtz_as_slow_fast_for_compressible_Euler}
		In this remark we discuss how the Helmholtz decomposition may be interpreted as a slow-fast decomposition for compressible Euler equations:
		the slow piece is governed by the incompressible Euler equations whereas the fast piece is governed by the wave equation of acoustics.

		As starting point we take the following isentropic compressible Euler equations.
		\begin{subnumcases}{}
			\partial_t \rho + \nabla\cdot (\rho u) = 0 \text{ and } \\
			\rho ( \partial_t u + u\cdot\nabla u) + \frac{1}{\varepsilon^2} \nabla p = 0
		\end{subnumcases}
		where $\varepsilon$ denotes the Mach number.
		Here we are considering the ideal gas law $p = \rho^\gamma$ for some $\gamma > 1$.

		Then, for $D_t \vcentcolon= \partial_t + \cdot\nabla$ denoting the advective derivative,
		\begin{equation*}
			D_t p
			= \gamma\rho^{\gamma - 1}D_t \rho
			= - \gamma\rho^{\gamma-1} \rho \nabla\cdot u
			= -\gamma p\nabla\cdot u
		\end{equation*}
		and so we may rewrite the system in terms of the pressure $p$ and the velocity field $u$ only as
		\begin{subnumcases}{}
			\partial_t p + u\cdot\nabla p + \gamma p \nabla\cdot u = 0 \text{ and } \\
			p^{1/\gamma} ( \partial_t u + u\cdot\nabla u) + \frac{1}{\varepsilon^2} \nabla p = 0.
		\end{subnumcases}
		We may now \emph{symmetrize} the system by changing the pressure coordinate and using instead $q \vcentcolon= \frac{p - \bar{p}}{\varepsilon}$
		for some constant $\bar{p}$, such that $p = \bar{p} + \varepsilon  q$, which yields
		\begin{subnumcases}{}
			\partial_t q + u\cdot\nabla q + \frac{\gamma p}{\varepsilon} \nabla\cdot u = 0 \text{ and } \\
			p^{1/\gamma} ( \partial_t u + u\cdot\nabla u) + \frac{1}{\varepsilon} \nabla q = 0,
		\end{subnumcases}
		which is now symmetrized.

		We now make the two timescale assumption (see \fref{Section}{sec:slow_fast_decomp})
		\begin{equation*}
			\mathcal{X} (t) = \mathcal{X} \left( t,\,\frac{t}{\varepsilon} \right),
		\end{equation*}
		denoting the fast timescale $\tau \vcentcolon= t/\varepsilon$
		\emph{and} we make the expansion assumption
		\begin{equation*}
			\mathcal{X} = \mathcal{X}_0 + \varepsilon \mathcal{X}_1 + O(\varepsilon^2)
		\end{equation*}
		for $\mathcal{X} = u,\,q$.
		In particular we see that
		\begin{equation*}
			p = \bar{p} + \varepsilon q_0 + \varepsilon^2 q_1 + O (\varepsilon^3).
		\end{equation*}
		This yields
		\begin{equation*}
			\frac{1}{\varepsilon} ( \partial_\tau q_0 + \gamma \bar{p} \nabla\cdot u_0 ) + O(1) = 0
		\end{equation*}
		and
		\begin{equation*}
			\bar{p}^{1/\gamma} \cdot\frac{1}{\varepsilon} \partial_\tau u_0 + \frac{1}{\varepsilon} \nabla q_0 + O (1) = 0.
		\end{equation*}
		This leading order this comes down to
		\begin{subnumcases}{}
			\partial_\tau q_0 + \gamma \bar{p} \nabla\cdot u_0 = 0 \text{ and } \\
			\bar{p}^{1/\gamma} \partial_\tau u_0 + \nabla q_0 = 0.
		\end{subnumcases}

		From now on in this remark we will focus on the leading order system, so we drop the subscript and consider
		\begin{subnumcases}{}
			\partial_t q + \gamma \bar{p} \nabla\cdot u = 0 \text{ and } \\
			\bar{p}^{1/\gamma} \partial_t u + \nabla q = 0.
		\end{subnumcases}
		Therefore the energy
		\begin{equation}
		\label{eq:acoustic_energy}
			E = \frac{1}{2} \int \gamma \bar{p}^{1 + 1/\gamma} {\lvert u \rvert}^2 + q^2
		\end{equation}
		is conserved since
		\begin{align*}
			\int \gamma \bar{p}^{1 + 1/\gamma} u\cdot \partial_t u + q \partial_t q
			&= \int \gamma \bar{p}^{1 + 1/\gamma} u \cdot \left( - \frac{1}{\bar{p}^{1/\gamma}} \nabla q \right)
				+ q \cdot \left( -\gamma \bar{p} \nabla\cdot u \right)
		\\
			&= -\gamma \bar{p} \int \left[ u\cdot\nabla q + q(\nabla\cdot u) \right]
			= 0 .
		\end{align*}

		We now set $\gamma = \bar{p} = 1$ for simplicity, reducing the leading order system to
		\begin{subnumcases}{}
			\partial_t q + \nabla\cdot u = 0 \text{ and } \label{eq:Helmholtz_slow_fast_system_a}\\
			\partial_t u + \nabla q = 0 \label{eq:Helmholtz_slow_fast_system_b}
		\end{subnumcases}
		which we may rewrite as
		\begin{equation*}
			\partial_t \begin{pmatrix}
				q \\ u
			\end{pmatrix} = \mathcal{L} \begin{pmatrix}
				q \\ u
			\end{pmatrix} \text{ where } \mathcal{L} = \begin{pmatrix}
				0 & \text{div} \\ \text{grad} & 0
			\end{pmatrix}.
		\end{equation*}
		When $\gamma = \bar{p} = 1$ the conserved energy reduces to the standard $L^2$ inner product.
		Here we note that here the leading order operator $\mathcal{L}$ is skew-symmetric with respect to $L^2$
		since the $L^2$ adjoint of $\text{div}$ is $-\text{grad}$.

		We then decompose ${(L^2)}^3 = {(L^2)}^3_\sigma + \mathring{H}^1 $, i.e. writing $u = \sigma + \nabla\phi$ where $\sigma$ is divergence-free, such that
		\begin{equation*}
			\begin{pmatrix}
				q \\ u
			\end{pmatrix} = \begin{pmatrix}
				q \\ \nabla \phi
			\end{pmatrix} + \begin{pmatrix}
				0 \\ \sigma
			\end{pmatrix}.
		\end{equation*}
		The first term in this decomposition is the fast acoustic term.
		The second term in this decomposition is the slow term obeying incompressible Euler.
		In particular: the kernel of the leading order operator $\mathcal{L}$ is precisely the space $ \left\{ 0 \right\} \times {(L^2)}^3_\sigma$ to which the second component
		belongs, and so this is an honest-to-goodness slow-fast decomposition.
		(We could also readily verify that the image of $\mathcal{L}$ is the space in which the first component lives.)

		In other words: the Helmholtz decomposition is precisely the slow-fast decomposition for this compressible Euler system!

		Moreover considering $ \partial_t \eqref{eq:Helmholtz_slow_fast_system_a} - \nabla\cdot \eqref{eq:Helmholtz_slow_fast_system_b}$ produces
		\begin{equation*}
			( \partial_t^2 - \Delta ) q = 0
		\end{equation*}
		and so, for $S$ denoting the semigroup associated with the wave equation above,
		\begin{equation*}
			q(t) = S \left( \frac{t}{\varepsilon} \right) \bar{q} (t),
		\end{equation*}
		where $\bar{q} (t)$ would then be determined by a solvability condition arising from the next-order dynamics.
	\end{remark}

\subsection{Useful computations, identities, and facts}
\label{app:useful_facts}

        In this section we record a variety of computations, identities, and facts which are of use throughout the paper.

        First we record a useful splitting of the curl operator into its horizontal and vertical components.

	\begin{lemma}[Horizontal-vertical decomposition of the curl]
	\label{lemma:horizontal_vertical_decomposition_of_the_curl}
		For any sufficiently regular vector field $v$, $\nabla \times v = \partial_3 v_h^\perp - \nabla_h^\perp v_3 + ( \nabla_h^\perp \cdot v_h) e_3$.
	\end{lemma}
	\begin{proof}
		Since $\nabla\times v = (\partial_2 v_3 - \partial_3 v_2,\, \partial_3 v_1 - \partial_1 v_3,\, \partial_1 v_2 - \partial_2 v_1)$ the result follows from appropriately grouping terms.
	\end{proof}

        We now record a couple of identities concerning the divergence and curl of
        a vector field of the form of the ``good unknowns'' (see \fref{Remark}{rmk:good_unknown}). This is a very useful identity since it helps us tie the good unknown to the potential vorticity $PV$ and to the thermal wind imbalance $j$ (see also \fref{Remark}{rmk:good_unknown}).
 
	\begin{lemma}
	\label{lemma:div_and_curl_of_good_unknown}
		Consider a sufficiently regular vector field $v$ and scalar field $\phi$. Then
		\begin{equation*}
			\nabla\cdot ( - v_h^\perp + \phi e_3 ) = \nabla_h^\perp \cdot v_h + \partial_3 \phi
			\text{ and }
			\nabla\times ( -v_h^\perp + \phi e_3 ) = \partial_3 v_h - \nabla_h^\perp \phi - (\nabla_h\cdot v_h) e_3.
		\end{equation*}
		In particular if $v$ is divergence-free then $\nabla\times ( -v_h^\perp + \phi e_3 ) = \partial_3 v_h - \nabla_h^\perp \phi + ( \partial_3 v_3) e_3$.
	\end{lemma}
    \begin{proof}
        This follows immediately from \fref{Lemma}{lemma:horizontal_vertical_decomposition_of_the_curl}.
    \end{proof}

        We now record a useful algebraic inversion which tells us that we may equivalently
        characterise the thermodynamic variables $\theta$ and $q$ by instead using the
        buoyancy $b = \theta - {\min}_0\,q$ and the moist variable $M = \theta + q$.

	\begin{lemma}
	\label{lemma:invert_b_and_M}
		Given any two scalars $b$ and $M$ the nonlinear algebraic system
		\begin{equation*}
			\left\{
			\begin{aligned}
				& \theta + q = M \text{ and } \\
				& \theta - {\min}_0\, q = b
			\end{aligned}
			\right.
		\end{equation*}
		is invertible, with solution given uniquely by
		\begin{equation*}
			\left\{
			\begin{aligned}
				& \theta = b + \frac{1}{2} {\min}_0\, (M - b) \text{ and } \\
				& q = M - b - \frac{1}{2} {\min}_0\, (M - b).
			\end{aligned}
			\right.
		\end{equation*}
	\end{lemma}
	\begin{proof}
		We can do this easily by introducing $H = \mathds{1} (q<0)$.
		Treating $H$ as fixed, the system then becomes linear (and invertible) and we obtain
		\begin{equation*}
			\left\{
			\begin{aligned}
				& \theta =  b + \frac{1}{2} H (M-b) \text{ and } \\
				& q = (1 - H/2) (M-b).
			\end{aligned}
			\right.
		\end{equation*}
		In particular we note that $\sign q = \sign (M - b)$ and so $H(M-b) = {\min}_0\, (M-b)$, from which the claim follows.
	\end{proof}

        Below we relate another identity related to ``good unknowns''.
        This time we obtain an identity for the vector field $-u_h^\perp + (\theta-qH) e_3$
        which appears in the ``fixed $H$'' version of the measurement-to-coordinates map
        (see \fref{Lemma}{lemma:def_and_invertibility_of_C_H} and \fref{Remark}{rmk:def_C_H}).
 
	\begin{lemma}
		\label{lemma:identity_lin_buoyancy_and_xi}
		If, for some indicator function $H$,
		\begin{align*}
			u &= \nabla_h^\perp p + \wavecoord_h^\perp + we_3,\\
			\theta &= \partial_3 p + \frac{H}{2} (M - \partial_3 p) + \wavecoord_3, \text{ and } \\
			q &= M - \partial_3 p - \frac{H}{2} (M - \partial_3 p) - \wavecoord_3
		\end{align*}
		for some divergence-free $\wavecoord$ satisfying $\nabla_h^\perp \cdot \wavecoord_h = \partial_3 w$
		then
		\begin{equation*}
			- u_h^\perp + (\theta - qH) e_3 = \nabla p + A_H^{-1} \wavecoord
		\end{equation*}
		where $A_H = I - \frac{H}{2} e_3\otimes e_3$ and so $A_H^{-1} = I _ H e_3\otimes e_3$
		(since $(1-H/2)(1+H) \equiv 1$).
	\end{lemma}
	\begin{proof}
		Since ${\left(v^\perp\right)}^\perp = -v_h$ we see immediately that $-u_h^\perp = \nabla_h p + \wavecoord_h$.
		Then we compute that
		\begin{equation*}
			\theta - qH
			= \left( 1 - \frac{H}{2} \right)\partial_3 p + \frac{H}{2} M + \wavecoord_3 - \frac{1}{2}H (M-\partial_3 p) + H \wavecoord_3
			= \partial_3 p + (1 + H)\wavecoord_3.
		\end{equation*}
		The claim follows from combining these two observations.
	\end{proof}

        We now record a well-posedness result for a div-curl system
        involving an elliptic matrix.
        This is used to verify that the ``fixed $H$'' version of the measurement--to--coordinate map is invertible.

	\begin{prop}[Solvability of a div-curl system]
	\label{prop:solvability_div_curl_system}
		For any uniformly elliptic $B\in L^\infty$, $f\in { \left( \mathring{H}^1_\sigma \right)}^*$, and $g\in \mathbb{R}^3$
		there is a unique $\wavecoord \in { \left( L^2 \right)}^3_\sigma$ satisfying
		\begin{equation}
			\nabla\times(B\wavecoord) = f,\, \nabla\cdot\wavecoord = 0, \text{ and } \fint \wavecoord = g.
		\label{eq:solvability_div_curl_system}
		\end{equation}
	\end{prop}
	\begin{proof}
		First we show that a solution exists.
		We let $\psi \in \mathring{H}^1_\sigma$ be the unique solution of
		\begin{equation*}
			\nabla\times(\nabla\times\psi) = f
		\end{equation*}
		and then let $\pi\in \mathring{H}^1 $ be the unique solution of
		\begin{equation*}
			\nabla\cdot \left( B^{-1} \nabla\pi \right) = -\nabla\cdot \left[ B^{-1} \left( \nabla\times \psi + g \right) \right].
		\end{equation*}
		Then
		\begin{equation*}
			\wavecoord \vcentcolon= B^{-1} \left( \nabla \pi + \nabla\times \psi + g \right)
		\end{equation*}
		is a solution.
		Indeed
		\begin{equation*}
			\fint B\wavecoord = g
			\text{ and } 
			\nabla\times (B\wavecoord) = \nabla\times (\nabla\times \psi) = f
		\end{equation*}
		while
		\begin{equation*}
			\nabla\cdot \wavecoord 
			= \nabla\cdot \left( B^{-1} \nabla\pi \right) + \nabla\cdot \left[ B^{-1} \left( \nabla\times \psi + g \right) \right]
			= 0.
		\end{equation*}

		The uniqueness of a solution then follows from the uniqueness of the Helmholtz decomposition
		since we have characterized the solution $\wavecoord$ by the Helmholtz decomposition of $B\wavecoord$.
		If $\wavecoord_1$ and $\wavecoord_2$ are solutions of \eqref{eq:solvability_div_curl_system} then their Helmholtz decompositions
		\begin{equation*}
			B\wavecoord_i = \nabla\pi_i + \nabla\times\psi_i + c_i,
		\end{equation*}
		for $i=1,\,2$, where $c_i \in \mathbb{R}^3$, satisfy
		\begin{equation*}
			c_1 = \fint B\wavecoord_1 = g = \fint B\wavecoord_2 = c_2
		\end{equation*}
		and
		\begin{equation*}
			\nabla\times (\nabla\times \psi_1)
			= \nabla\times (B\wavecoord_1)
			= f
		 	= \nabla\times (B\wavecoord_2)
			= \nabla\times (\nabla\times \psi_2),
		\end{equation*}
		which ensures that $c_1 = c_2$ and $\psi_1 = \psi_2$.
		This means that
		\begin{equation*}
			B(\wavecoord_1 - \wavecoord_2) = \nabla (\pi_1 - \pi_2)
		\end{equation*}
		and so
		\begin{equation*}
			\nabla\cdot \left[ B^{-1} \nabla (\pi_1 - \pi_2) \right]
			= \nabla\cdot (\wavecoord_1 - \wavecoord_2)
			= 0,
		\end{equation*}
		which also ensures that $\pi_1 = \pi_2$, and hence that $\wavecoord_1 = \wavecoord_2$.
	\end{proof}

        The result below pertains to the leading order dynamics discussed in \fref{Section}{sec:slow_fast_decomp}.
        It tells us how the balanced and wave sets relate to the dynamics of solutions of
        $\partial_t + \mathcal{N}$.

	\begin{prop}
	\label{prop:dyn_pdt_plus_N}
		Let $\mathcal{X} = (u,\,\theta,\,q)$ be a solution of $\partial_t \mathcal{X} + \mathcal{N} \mathcal{X} = 0$ for $\mathcal{N}$ as defined in \eqref{eq:def_N}.
		Then
		\begin{equation*}
			\mathcal{X} = \underbrace{ \begin{pmatrix}
				\nabla_h^\perp p \\
				\partial_3 p + \frac{1}{2} {\min}_0\, (M- \partial_3 p)\\
				M - \partial_3 p - \frac{1}{2} {\min}_0\, (M - \partial_3 p)
			\end{pmatrix}}_{ \mathcal{X}_B } + \underbrace{ \begin{pmatrix}
				\wavecoord_h^\perp + we_3 \\ \wavecoord_3 \\ - \wavecoord_3
			\end{pmatrix}}_{ \mathcal{X}_W}
		\end{equation*}
		where $ \mathcal{X}_B \in \mathcal{B}$ and $ \mathcal{X}_W \in \mathcal{W}$ for $\mathcal{B}$ and $\mathcal{W}$ as in \fref{Propositions}{prop:alt_char_balanced_set} and \ref{prop:chara_im_N},
		respectively, for $M \vcentcolon= \theta + q$ and $p$ the solution of
		\begin{equation*}
			\Delta p + \frac{1}{2} \partial_3 {\min}_0\, (M-\partial_3 p) = \underbrace{ \nabla_h^\perp \cdot u_h + \partial_3 \theta}_{PV},
		\end{equation*}
		and where $\mathcal{X}_W \vcentcolon= \mathcal{X} - \mathcal{X}_B$.
		Moreover
		\begin{equation*}
			\partial_t \mathcal{X}_B = 0
			\text{ and } 
			\partial_t \mathcal{X}_W = \begin{pmatrix}
				\mathbb{P}_L ( -\wavecoord_h - (1+H_B)\wavecoord_3 e_3) \\ w \\ -w
			\end{pmatrix} =\vcentcolon \mathcal{L}_{W,\,H_B} \mathcal{X}_W \in \mathcal{W}
		\end{equation*}
		for $H_B \vcentcolon= \mathds{1} (q_B < 0) = \mathds{1} (M < \partial_3 p)$.
		Finally: if $ \mathcal{X}_W = (\wavecoord_h,\,0,\,0)$ then
		\begin{equation*}
			\mathcal{L}^2_{W,\,H_B} \mathcal{X}_W = - \mathcal{X}_W
		\end{equation*}
		and if $ \mathcal{X}_W = (we_3,\,\wavecoord_3,\,-\wavecoord_3)$ then
		\begin{equation*}
			\mathcal{L}^2_{W,\,H_B} \mathcal{X}_W = - (1+H_B) \mathcal{X}_W.
		\end{equation*}
		In other words: the $\mathcal{B}$--component is constant whereas the dynamics of the $\mathcal{W}$--component induce an $ \mathcal{X}_B$--dependent dynamical system on $\mathcal{W}$.
	\end{prop}
	\begin{proof}
		The existence and uniqueness of $p$ is established in \cite{remond2024nonlinear}.
		The fact that the $\mathcal{B}$--component is constant follows from the fact that $\partial_t PV = \partial_t M = 0$ and the fact that $\partial_t p$ solves
		\begin{equation*}
			\nabla\cdot (A_B \nabla \partial_t p) = 0
		\end{equation*}
		for $A_B \vcentcolon= I - \frac{1}{2} \mathds{1} (M < \partial_3 p)e_3\otimes e_3)$ uniformly elliptic, such that indeed $\partial_t p$ vanishes.
		Then $\mathcal{X}_W \vcentcolon= \mathcal{X} - \mathcal{X}_B$ belongs to $\mathcal{W}$ and takes the prescribed form by virtue of \fref{Proposition}{prop:chara_im_N}.
		Moreover we note that $\im \mathcal{L}_{W,\,H_B}\subseteq \mathcal{W}$ since
		\begin{equation*}
			PV \begin{pmatrix}
				\mathbb{P}_L ( -\wavecoord_h - (1+H_B)\wavecoord_3 e_3) \\ w \\ -w
			\end{pmatrix} = - \nabla_h^\perp \cdot \wavecoord_h + \partial_3 w = 0
		\end{equation*}
		and
		\begin{equation*}
			\mathcal{M} \begin{pmatrix}
				\mathbb{P}_L ( -\wavecoord_h - (1+H_B)\wavecoord_3 e_3) \\ w \\ -w
			\end{pmatrix} = w - w = 0.
		\end{equation*}
		Finally the identities involving $\mathcal{L}^2$ follow from direct computations.
	\end{proof}

        We now turn our attention to identities involving the conserved energy.
        First we record the computation verifying that the energy is indeed conserved.

	\begin{prop}[Conserved energy]
	\label{prop:cons_en}
		Let $(u,\,\theta,\,q)$ be a solution of the moist Boussinesq system \eqref{eq:moist_Boussinesq_u}--\eqref{eq:moist_Boussinesq_q}.
		Then the energy
		\begin{equation*}
			E = \frac{1}{2} \int {\lvert u \rvert}^2 + \theta^2 + {\min}_0^2\, q + {(\theta + q)}^2
		\end{equation*}
		is conserved, meaning that $ \frac{dE}{dt} = 0$.
	\end{prop}
	\begin{proof}
		We immediately note that, for $M \vcentcolon= \theta + q$,
		\begin{equation*}
			( \partial_t + u\cdot\nabla ) M = 0.
		\end{equation*}
		Therefore, since $u$ is divergence-free,
		\begin{equation*}
			\frac{d}{dt} \int \frac{1}{2} {(\theta + q)}^2
			= \int ( \partial_t +u\cdot\nabla ) \left( \frac{1}{2} M^2 \right)
			= \int M ( \partial_t + u\cdot\nabla ) M
			= 0.
		\end{equation*}
		Now observe that $\frac{d}{ds} {\min}_0^2 s = {\min}_0\, s$ and so
		\begin{equation*}
			\left( \partial_t + u\cdot\nabla \right) \left( \frac{1}{2} {\min}_0^2\, q \right)
			= \left( {\min}_0\, q \right) \left( \partial_t +u\cdot\nabla \right) q.
		\end{equation*}
		Therefore
		\begin{align*}
			&\frac{d}{dt} \int {\lvert u \rvert}^2 + \frac{1}{2} \theta^2 + \frac{1}{2} {\min}_0^2\, q
		\\
			&= \int u \cdot ( \partial_t + u \cdot \nabla ) u + \theta ( \partial_t + u \cdot \nabla ) \theta + ( {\min}_0\, q) ( \partial_t + u \cdot \nabla ) q
		\\
			&= \frac{1}{\varepsilon} - u\cdot e_3\times u - u\cdot\nabla p + u_3 (\theta - {\min}_0\, q) - \frac{1}{\varepsilon} \int \theta u_3 + \frac{1}{\varepsilon} \int ( {\min}_0\, q) u_3
		\\
			&= \frac{1}{\varepsilon} \int (\nabla\cdot u) p
			= 0,
		\end{align*}
		and so $ \frac{dE}{dt} = 0$ as claimed.
	\end{proof}

        We now turn our attention to another variant of the energy, namely the
        ``fixed $H$'' energy used in \fref{Section}{sec:energy_and_decomposition}.
        We first show that the (formally) linearised operator is skew-adjoint with respect to that energy, which is also known in the geometric context of \fref{Section}{sec:geometry} as the $L^2_H$ inner product.

	\begin{lemma}
	\label{lemma:L_H_is_L_2_H_skew}
		Let $H$ be an indicator function.
		The $L^2_H$--adjoint of $\mathcal{L}_H$ is $-\mathcal{L}_H$.
	\end{lemma}
	\begin{proof}
		For any $(u,\,\theta,\,q)$ and $(v,\,\phi,\,r)$ in $ \mathbb{L}^2_\sigma $ we compute that
		\begin{align*}
			{\langle \mathcal{L}_H (u,\,\,\theta,\,q) ,\, (v,\,\phi,\,r) \rangle }_{L^2_H} 
			&= \int \mathbb{P}_L \left( u_h^\perp - \theta e_3 + q H e_3  \right) \cdot v + u_3\phi - u_3 r H + 0
		\\
			&= \int -u \cdot v_h^\perp - \theta v_3 + qv_3 H + u_3\phi - u_3 rH
		\\
			&= \int u\cdot \left( -v_h^\perp + \phi e_3 - rHe_3 \right) + \theta ( -v_3) + q\cdot v_3 \cdot H
		\\
			&= {\langle (u,\,\theta,\,q) ,\, \left( \mathbb{P}_L ( -v_h^\perp + \phi e_3 - rH e_3 ) ,\, -v_3,\, v_3 \right) \rangle }_{L^2_H} 
		\end{align*}
		and so indeed
		\begin{equation*}
			\mathcal{L}_H^* \begin{pmatrix}
				v \\ \phi \\ r
			\end{pmatrix} = \begin{pmatrix}
				\mathbb{P}_L ( -v_h^\perp + \phi e_3 - rHe_3 ) \\ -v_3 \\ v_3
			\end{pmatrix} = - \mathcal{L}_H \begin{pmatrix}
				v \\ \phi \\ r
			\end{pmatrix}. \qedhere
		\end{equation*}
	\end{proof}

        We now record another property of the ``fixed $H$'' energy,
        showing that if the indicator is time-dependent, this energy \emph{fails}
        to be conserved.
    
	\begin{prop}
	\label{prop:L_H_energy}
		Suppose that $H$ is a given indicator function and that $\mathcal{X} = (u,\,\theta,\,q)$ is a solution of $\partial_t \mathcal{X} + \mathcal{L}_H \mathcal{X} = 0$
		for $\mathcal{L}_H$ as defined in \eqref{eq:def_L_H}.
		For
		\begin{equation*}
			E_H ( \mathcal{X} ) \vcentcolon= \frac{1}{2} \int {\lvert u \rvert}^2  + \theta^2 + q^2 H + {(\theta + q)}^2
		\end{equation*}
		we have that
		\begin{equation*}
			\frac{d}{dt} E_H = \frac{1}{2} \int q^2 \partial_t H.
		\end{equation*}
		In particular $E_H$ is constant if and only if the $H$--interface is contained in the $q$--interface, i.e. $ \partial\left\{ H = 1 \right\} \subseteq \left\{ q = 0 \right\}$.
	\end{prop}
	\begin{proof}
		This follows from the fact that if $ \mathcal{X}$ solves $\partial_t \mathcal{X} + \mathcal{L} \mathcal{X}$ then $M \vcentcolon= \theta + q$ satisfies $\partial_t M = 0$ and
		so we may compute that
		\begin{equation*}
			\int u\cdot\partial_t u + \theta \partial_t \theta + q (\partial_t q) H + M (\partial_t M)
			= \int u_h \cdot u_h^\perp - u_3 (\theta - qH) + \theta u_3 - q H u_3
			= 0
		\end{equation*}
		(in other words: if $H$ is constant in time then the operator $\mathcal{L}_H$ is skew-adjoint with respect to the inner product induced by $E_H$).
	\end{proof}

        We now turn our attention to a standard result verifying that pullback metrics
        are themselves metrics. This plays a crucial role in \fref{Section}{sec:energy_and_decomposition} in order to eventually define the statistical divergence we introduce in this paper and which agrees with both the decomposition and the conserved energy.

	\begin{prop}[Pullback metric]
	\label{prop:pullback_metric}
		Let $S$ be a set, let $(M,\,d)$ be a metric space, and let $F:S\to M$ be injective.
		The map $F^* d :S\times S \to [0,\,\infty)$ defined by
		\begin{equation*}
			F^* d (x,\,y) = d(F(x),\, F(y))
		\end{equation*}
		for every $x,\,y\in S$ is a metric on $S$
		known as the \emph{pullback metric} of $d$ under $F$.
	\end{prop}
	\begin{proof}
		The pullback metric immediately inherits symmetry and the triangle inequality from $d$.
		Its positive-definiteness then follows from the positive-definiteness of $d$ and the injectivity of $F$: for any $x,\,y\in S$,
		\begin{equation*}
			F^* d (x,\,y) = 0
			\iff d (F(x),\, F(y)) = 0
			\iff F(x) = F(y)
			\iff x = y.
			\qedhere
		\end{equation*}
	\end{proof}

        The last result we record in this appendix is a result used in \fref{Section}{sec:geometry}
        to verify that the balanced set is an honest-to-goodness Lipschitz--regular Hilbert manifold.
        The result in question is an elementary estimate.
 
	\begin{lemma}
	\label{lemma:H_minus_one_norm_of_explicit_PV}
		For any $v \in {(L^2)}^3$ and any $\phi \in L^2$,
		\begin{equation*}
			\norm{ \nabla_h^\perp \cdot v_h + \partial_3 \phi }{H^{-1}} \leqslant \norm{v_h}{L^2} + \norm{\phi}{L^2}.
		\end{equation*}
	\end{lemma}
	\begin{proof}
		This follows from a direct computation since, for any $\psi\in \mathring{H}^1 $,
		\begin{align*}
			{\langle \nabla_h^\perp \cdot v_h + \partial_3 \phi,\, \psi \rangle }_{H^{-1} \times \mathring{H}^1}
			&= \int ( \nabla_h^\perp \cdot v_h + \partial_3 \phi) \psi
		\\
			&= - \int v_h \cdot\nabla_h^\perp \psi + \phi \partial_3 \psi
			\leqslant \left( \norm{v_h}{L^2} + \norm{\phi}{L^2} \right) \norm{\psi}{ \mathring{H}^1 },
		\end{align*}
		as desired.
	\end{proof}

\subsection{Tools from convex analysis}
\label{app:convex}

In this section we record various results from convex analysis that are of use to us.
This begins with a result on various characterisations of strong convexity
and concludes with results pertaining to the \hyperref[def:quad_remainder]{quadratic remainder}
(which are used in \fref{Section}{sec:convergence_guarantees} to obtain a quadratic upper bound on our variational energy, one of the two key ingredients used to deduce convergence rates for various descent methods).

First we recall a few first-order characterisations of strong convexity.

\begin{lemma}[Equivalent characterisations of strong convexity]
\label{lemma:equiv_charac_strg_conv}
	Let $ \left( H,\, \langle \,\cdot\,,\,\cdot\, \rangle  \right)$ be a Hilbert space and let $f: H\to \mathbb{R}$ be Fr\'{e}chet differentiable.
	For any $\mu > 0$ the following are equivalent, where each inequality holds for every $x,\,y\in H$ and $\theta\in [0,\,1]$.
	\begin{align}
		&f(\theta x + (1-\theta) y) \leqslant \theta f(x) + (1-\theta) f(y) - \theta(1-\theta) \frac{\mu}{2} \norm{x-y}{}^2	\label{eq:equiv_charac_strg_conv_1}\\
		&f(y) \geqslant f(x) + \langle Df(x),\, y-x \rangle + \frac{\mu}{2} \norm{x-y}{}^2					\label{eq:equiv_charac_strg_conv_2}\\
		&\langle Df(x) - Df(y) ,\, x-y \rangle \geqslant \mu \norm{x-y}{}^2							\label{eq:equiv_charac_strg_conv_3}
	\end{align}
	If any of these inequalities hold recall that we say that $f$ is \emph{$\mu$-strongly convex}.
\end{lemma}
\begin{proof}
	See Definition 2.1.2 and Theorem 2.1.9 in \cite{nesterov}.
\end{proof}

We now turn our attention towards properties of the \hyperref[def:quad_remainder]{quadratic remainder}.
First we note that the map $f \mapsto T^f$ is linear.

\begin{lemma}[Linearity of the quadratic remainder map]
\label{lemma:lin_quad_remainder_map}
	For any $x\in\mathbb{R}^d$ the map $ \hyperref[def:quad_remainder]{T_x} : C^1 ( \mathbb{R}^d;\, \mathbb{R} ) \to C^0 ( \mathbb{R}^d;\, \mathbb{R} )$
	defined by $f\mapsto \hyperref[def:quad_remainder]{T_x^f}$ is linear.
\end{lemma}
\begin{proof}
	In light of the definition of the quadratic remainder this follows immediately from the derivative being linear.
\end{proof}

Second we observe that quadratic remainders can be used to characterise affine functions.
This is also the first of two results teasing out the relation between affine functions and quadratic remainders

\begin{lemma}[Vanishing quadratic remainders characterise affine functions]
\label{lemma:vanishing_quad_rem_charac_aff_func}
	A continuously differentiable function $f: \mathbb{R}^d\to\mathbb{R}$ is affine, meaning that there exist $v\in\mathbb{R}^d$ and $s\in\mathbb{R}$ for which
	$
		f(x) = v^T x + s \text{ for every } x\in\mathbb{R}^d,
	$
	if and only if its \hyperref[def:quad_remainder]{quadratic remainder} $T_x^f$ vanishes everywhere,
	i.e. $T_x^f (y) = 0$ for every $x,\,y\in\mathbb{R}^d$.
\end{lemma}
\begin{proof}
	The ``only if'' direction follows from a direct computation.
	Conversely suppose that $T_x^f \equiv 0$.
	In particular this means that $T_0^f (y) = 0$, which may be rearranged into $f(y) = f(0) + {\nabla f(0)}^T y$, proving that $f$ is affine.
\end{proof}

To conclude this section we note how the quadratic remainder behaves under affine transformations -- this is the second of two results on the relation between affine functions and quadratic remainders.

\begin{lemma}[Equivariance of the quadratic remainder map under affine transformations]
\label{lemma:equiv_quad_rem_map_aff_trans}
	Consider a continuously differentiable function $f : \mathbb{R}^d\to\mathbb{R}$ and an affine function $\Phi : \mathbb{R}^d \to \mathbb{R}^d$.
	Then
	$
		\hyperref[def:quad_remainder]{T_x^{f\circ\Phi} (y)} = T_{\Phi(x)}^f \left( \Phi(y) \right) \text{ for every } x,\,y\in\mathbb{R}^d.
	$
\end{lemma}
\begin{proof}
	We compute that
	\begin{equation}
	\label{eq:int_T_aff_comp}
		T_x^{f\circ\Phi} (y) = f \left( \Phi (y) \right) - f \left( \Phi (x) \right) - { \left( D\Phi \right) }^T (x) (\nabla f) \left( \Phi (x) \right) \cdot (y-x).
	\end{equation}
	Since we seek to write $T^{f\circ\Phi}$ in terms of $T^f$ the key step in \eqref{eq:int_T_aff_comp} is to write $y-x$ in terms of $\Phi (y) - \Phi (x)$.
	To do so we note that $\Phi$ is affine and so \fref{Lemma}{lemma:vanishing_quad_rem_charac_aff_func} tells us that $T^\Phi$ vanishes, i.e.
	$
		\Phi (y) - \Phi (x) = D\Phi (x) (y - x)
	$
	for every $x,\,y\in\mathbb{R}^d$.
	Plugging this into \eqref{eq:int_T_aff_comp} yields
	\begin{align*}
		T_x^{f\circ\Phi} (y)
		&= f \left( \Phi (y) \right) - f \left( \Phi(x) \right) - (\nabla f) \left( \Phi (x) \right) \cdot(D\Phi) (x) (y-x)
	      \\&= f \left( \Phi (y) \right) - f \left( \Phi (x) \right) - (\nabla f) \left( \Phi (x) \right) \cdot \left( \Phi (y) - \Phi (x) \right)
		= T_{\Phi(x)}^f \left( \Phi (y) \right),
	\end{align*}
	as claimed.
\end{proof}
 
\subsection{Exact nonlinear waves}
\label{app:nonlinear_waves}

    In this section we record results pertaining to the nonlinear oscillations with non-vanishing averages which are present in the moist Boussinesq system.
    The key results are \fref{Proposition}{prop:second_order_ODE_with_nonzero_average_solutions}, which shows that a simple ODE leads to nonlinear oscillations with non-vanishing averages, and \fref{Corollary}{cor:moist_Bouss_has_wave_sols_with_nonvanish_averages} which explains how this solutions of this simple ODE give rise to honest-to-goodness solutions of the moist Boussinesq system with oscillations of the same nature.

    First we find an exact solution of this simple ODE.

	\begin{prop}
	\label{prop:second_order_ODE_with_nonzero_average_solutions}
		The ODE $q'' + q + {\min}_0\, q = 0$ admits periodic weak solutions with nonzero average over their period.
	\end{prop}
	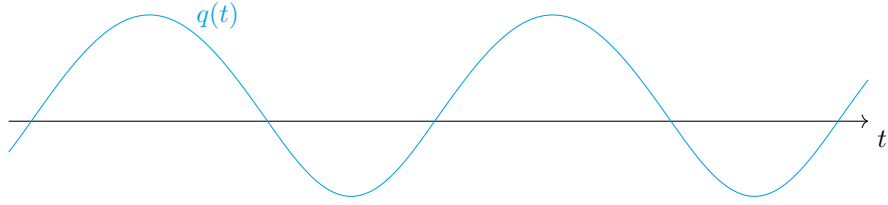
\begin{figure}
		\centering
		\captionsetup{width=0.85\textwidth}
        \begin{tikzpicture}
            \def\T{5.36303}
            \def\tbar{0.920151}
            \def\leftoffset{-0.3}
            \def\rightoffset{0.4}
            \def\sqrttwo{1.41421}
        	\draw[->] (\leftoffset, 0) -- (2*\T + \rightoffset, 0) node[below right] {$t$};
        	\node[Cerulean, right] at (pi/2 + 0.5, \sqrttwo) {$q(t)$};
        	\draw[Cerulean, smooth, domain=\leftoffset	:0,			variable=\t] plot ({\t}, {		sin(deg( sqrt(2)*(\t + \T - \tbar)	))});
        	\draw[Cerulean, smooth, domain=0		:pi,			variable=\t] plot ({\t}, {sqrt(2) *	sin(deg( \t				))});
        	\draw[Cerulean, smooth, domain=pi		:\T,			variable=\t] plot ({\t}, {		sin(deg( sqrt(2)*(\t - \tbar)		))});
        	\draw[Cerulean, smooth, domain=\T		:\T+pi,			variable=\t] plot ({\t}, {sqrt(2) *	sin(deg( \t - \T			))});
        	\draw[Cerulean, smooth, domain=\T+pi		:2*\T,			variable=\t] plot ({\t}, {		sin(deg( sqrt(2)*(\t - \T - \tbar)	))});
        	\draw[Cerulean, smooth, domain=2*\T		:2*\T+\rightoffset,	variable=\t] plot ({\t}, {sqrt(2) *	sin(deg( \t - 2*\T			))});
        \end{tikzpicture}
		\caption{\small
          \small 
          An exact solution of $q'' + q + {\min}_0\, q = 0$ which has nonzero averages (over its period).
            Note that this is \emph{not} a sine wave shifted upwards slightly:
            the oscillation frequencies and amplitudes are \emph{different} when $q<0$
            and when $q\geqslant 0$.
        }
		\label{fig:nonlinear_oscillations}
	\end{figure}
	\begin{proof}
		Define
		\begin{equation*}
			Q(t) \vcentcolon= \left\{
			\begin{aligned}
				&\sqrt{2} \sin t			&&\text{if } 0 \leqslant t \leqslant \pi \text{ and } \\
				&\sin \sqrt{2} \left(t-\bar{t}\,\right)	&&\text{if } \pi \leqslant t < T
			\end{aligned}
			\right.
		\end{equation*}
		for $T \vcentcolon= \pi + \frac{\pi}{\sqrt{2}}$ and $\bar{t} \vcentcolon= \pi - \frac{\pi}{\sqrt{2}}$,
		and note that $Q$ is differentiable.
		Then define $q: [0,\,\infty) \to \mathbb{R} $ to be the $T$--periodic extention of $Q$,
		noting that $q$ is also differentiable.
		Moreover: $q'$ is continuous on $ \left\{ 0,\,\pi \right\} + T\mathbb{N}$, i.e. on the set where it changes sign.
		Therefore $q$ is a weak solution of $q'' + q + {\min}_0\, q = 0$ since
		\begin{equation*}
			q'' = q \text{ when } q > 0 \text{ and } q'' =  2q \text{ when } q < 0.
		\end{equation*}
		Finally the simple substitution $s = \sqrt{2}(t - \pi) + \pi$ shows that
		\begin{equation*}
			\int_0^T q
			= \int_0^T Q
			= \int_0^\pi \sqrt{2} \sin t \,dt + \int_\pi^{2\pi} \frac{1}{\sqrt{2}} \sin s\, ds
			= \left( \sqrt{2} - \frac{1}{\sqrt{2}} \right) \int_0^\pi \sin t \,dt
			\neq 0,
		\end{equation*}
		as desired.
	\end{proof}

     We now seek to use \fref{Proposition}{prop:second_order_ODE_with_nonzero_average_solutions} above to show that the moist Boussinesq system \eqref{eq:moist_Boussinesq_u}--\eqref{eq:moist_Boussinesq_q} admits solutions with the same nonlinear oscillations.
     As an intermediate step we consider an ODE system contained within the moist Boussinesq system (this notion is made precise in \fref{Lemma}{lemma:ODE_sols_are_PDE_sols} below).

	\begin{cor}
	\label{cor:ODE_system_with_nonzero_average_solutions}
		The ODE system
		\begin{subnumcases}{}
			\frac{dw}{dt} = \theta - {\min}_0\, q,		\label{eq:ODE_model_w}\\
			\frac{d\theta}{dt} = -w, \text{ and }		\label{eq:ODE_model_theta}\\
			\frac{dq}{dt} = w				\label{eq:ODE_model_q}
		\end{subnumcases}
		admits periodic solutions with nonzero averages over their period.
		Moreover these solutions may be chosen to satisfy $M \vcentcolon= \theta + q = 0$.
	\end{cor}
	\begin{proof}
		Suppose that $M \vcentcolon= \theta + q$ vanishes.
		Then $q = -\theta$ and so the ODE system reduces to
		\begin{equation*}
			\frac{dw}{dt} = -q- {\min}_0\, q \text{ and } \frac{dq}{dt} = w
		\end{equation*}
		such that $q$ must satisfy
		\begin{equation*}
			\frac{d^2 q}{dt^2} = \frac{dw}{dt} = - q - {\min}_0\, q,
		\end{equation*}
		or equivalently
		\begin{equation}
			q'' + q + {\min}_0\, q = 0.
		\label{eq:second_order_ODE_int}
		\end{equation}
		\fref{Proposition}{prop:second_order_ODE_with_nonzero_average_solutions} then guarantees the existence of a periodic weak solution $q_*$ of \eqref{eq:second_order_ODE_int}
		with nonzero average over its period.
		Then $(w_*,\,\theta_*,\,q_*)$ is a solution of the original ODE system \eqref{eq:ODE_model_w}--\eqref{eq:ODE_model_q} satisfying $\theta_* + q_* = 0$ provided that
		\begin{equation*}
			\theta_* = -q_*
			\text{ and } 
			w_* = \frac{dq}{dt}.
		\end{equation*}
		So finally, for the period $T$ and for $\bar{t}$ being defined as in \fref{Proposition}{prop:second_order_ODE_with_nonzero_average_solutions} we have that
		\begin{equation*}
			w_* \vert_{ (0,\,T) }
			= \left\{
			\begin{aligned}
				&\sqrt{2} \cos t						&&\text{if } 0\leqslant t < \pi \text{ and } \\
				&\sqrt{2} \cos \sqrt{2} \left(t - \bar{t} \, \right)		&&\text{if } \pi \leqslant t < T
			\end{aligned}
			\right.
		\end{equation*}
		and so
		\begin{equation*}
			\int_0^T w_*
			= \int_0^\pi \sqrt{2} \cos t \,dt + \int_\pi^{2\pi} \cos s \,ds
			= (\sqrt{2} - 1) \int_0^\pi \cos t \,dt
			\neq 0.
		\end{equation*}
		So indeed $w_*$ and $\theta_*$ also have nonzero averages over their period.
	\end{proof}

    As alluded to earlier we now show that solutions of the ODE system of \fref{Corollary}{cor:ODE_system_with_nonzero_average_solutions} give rise to solutions of the moist Boussinesq system \eqref{eq:moist_Boussinesq_u}--\eqref{eq:moist_Boussinesq_q}.

	\begin{lemma}
	\label{lemma:ODE_sols_are_PDE_sols}
		Let $\bar{w},\,\bar{\theta},\,\bar{q} : [0,\,\infty)\to \mathbb{R} $ be weak solutions of the ODE system \eqref{eq:ODE_model_w}--\eqref{eq:ODE_model_q}.
		If we define
		\begin{equation*}
			u(t,\,x) = \bar{w}(t) e_3,\,
			\theta (t,\,x) = \bar{\theta} (t), \text{ and } \\
			q (t,\,x) = \bar{q} (t)
		\end{equation*}
		then $(u,\,\theta,\,q)$ is a weak solution of the moist Boussinesq system \eqref{eq:moist_Boussinesq_u}--\eqref{eq:moist_Boussinesq_q}
		with $\varepsilon = 1$.
	\end{lemma}
	\begin{proof}
		When $u$, $\theta$, and $q$ are constant in space the advective terms vanish.
		Since the pressure necessarily solves $-\Delta p = \nabla\cdot (u\cdot\nabla u)$ it must therefore also vanish.
		Finally, since $e_3\times u = u_h^\perp$, which also vanishes when $u = \bar{w} e_3$,
		the claim follows.
	\end{proof}

    We are finally ready to show that the moist Boussinesq system \eqref{eq:moist_Boussinesq_u}--\eqref{eq:moist_Boussinesq_q} admits solutions with nonlinear oscillations whose time averages do not vanish.

	\begin{cor}
	\label{cor:moist_Bouss_has_wave_sols_with_nonvanish_averages}
		The moist Boussinesq system \eqref{eq:moist_Boussinesq_u}--\eqref{eq:moist_Boussinesq_q}
		admits periodic weak solutions with nonzero averages over their period.
		Moreover these solutions may be chosen to belong to the wave set $\mathcal{W}$ at every instant (for $\mathcal{W}$ as defined in \fref{Proposition}{prop:chara_im_N}).
	\end{cor}
	\begin{proof}
		This follows from \fref{Corollary}{cor:ODE_system_with_nonzero_average_solutions} and \fref{Lemma}{lemma:ODE_sols_are_PDE_sols} since
		together they guarantee the existence of weak solutions of the moist Boussinesq system \eqref{eq:moist_Boussinesq_u}--\eqref{eq:moist_Boussinesq_q}
		with $M \vcentcolon= \theta + q = 0$.
		Moreover we know from \fref{Lemma}{lemma:ODE_sols_are_PDE_sols} that these solutions take the form $u = \bar{w} e_3$ and $\theta = \bar{\theta} $.
		Therefore these solutions have vanishing potential vorticity, and they satisfy $M=0$.
		\fref{Proposition}{prop:chara_im_N} therefore tells us that, indeed, these solutions belong to the wave set $\mathcal{W}$.
	\end{proof}

    To conclude this section we record a result analogous to \fref{Lemma}{lemma:ODE_sols_are_PDE_sols} above: now the full velocity is incorporated into the ODE system, not just its vertical component. This is of independent interest from \fref{Lemma}{lemma:ODE_sols_are_PDE_sols} since it is used in the study of the Snell-type metric in \fref{Section}{sec:energy_centric_approach}.

    \begin{lemma}
    \label{lemma:ODE_sols_are_PDE_sols_general}
        Let $\bar{u},\,\bar{\theta},\,\bar{q} : [0,\,\infty) \to \mathbb{R}$ be weak solutions of the ODE system
		\begin{subnumcases}{}
			\frac{du}{dt} = - u_h^\perp + (\theta - {\min}_0\, q) e_3,		\label{eq:ODE_model_gen_u}\\
			\frac{d\theta}{dt} = -u_3, \text{ and }		\label{eq:ODE_model_gen_theta}\\
			\frac{dq}{dt} = u_3				\label{eq:ODE_model_gen_q}
		\end{subnumcases}
		If we define
		\begin{equation*}
			u(t,\,x) = \bar{u}(t),\,
			\theta (t,\,x) = \bar{\theta} (t), \text{ and } \\
			q (t,\,x) = \bar{q} (t)
		\end{equation*}
		then $(u,\,\theta,\,q)$ is a weak solution of the moist Boussinesq system \eqref{eq:moist_Boussinesq_u}--\eqref{eq:moist_Boussinesq_q}
		with $\varepsilon = 1$.
	\end{lemma}
    \begin{proof}
        This follows exactly as for \fref{Lemma}{lemma:ODE_sols_are_PDE_sols}.
    \end{proof}

\bibliographystyle{alpha-bis}
\bibliography{KSS-PV-references,references}

\newcommand{\etalchar}[1]{$^{#1}$}
\begin{thebibliography}{CZHK{\etalchar{+}}15}

\bibitem[AZD08]{amerault2008tests}
C.~Amerault, X.~Zou, and J.~Doyle.
\newblock Tests of an adjoint mesoscale model with explicit moist physics on
  the cloud scale.
\newblock {\em Monthly Weather Review}, 136(6):2120--2132, 2008.

\bibitem[Ban21]{bannister2021balance}
R.~N. Bannister.
\newblock Balance conditions in variational data assimilation for a
  high-resolution forecast model.
\newblock {\em Quarterly Journal of the Royal Meteorological Society},
  147(738):2917--2934, 2021.

\bibitem[BANZ97]{babinetal1997}
A.~Babin, M.~A., B.~Nicolaenko, and Y.~Zhou.
\newblock On the asymptotic regimes and the strongly stratified limit of
  rotating {B}oussinesq equations.
\newblock {\em Thoeret. Computat. Fluid Dyn.}, 9:223--251, 1997.

\bibitem[Bar95]{bartello95}
P.~Bartello.
\newblock Geostrophic adjustment and inverse cascades in rotating, stratified
  turbulence.
\newblock {\em J. Atmos. Sci.}, 52(24):4410--4428, 1995.

\bibitem[BBP{\etalchar{+}}01]{barkmeijer2001tropical}
J.~Barkmeijer, R.~Buizza, T.~Palmer, K.~Puri, and J.-F. Mahfouf.
\newblock Tropical singular vectors computed with linearized diabatic physics.
\newblock {\em Quarterly Journal of the Royal Meteorological Society},
  127(572):685--708, 2001.

\bibitem[BCF14]{buhlercalliesferrari2014}
O.~Buhler, J.~Callies, and R.~Ferrari.
\newblock Wave–vortex decomposition of one-dimensional ship-track data.
\newblock {\em J. Fluid Mech.}, 756:1007--1026, 2014.

\bibitem[BCZT14]{bousquet_coti_zelati_temam_14}
A.~Bousquet, M.~Coti~Zelati, and R.~Temam.
\newblock Phase transition models in atmospheric dynamics.
\newblock {\em Milan J. Math.}, 82(1):99--128, 2014.

\bibitem[BFH17]{breit2017compressible}
D.~Breit, E.~Feireisl, and M.~Hofmanov{\'a}.
\newblock Compressible fluids driven by stochastic forcing: the relative energy
  inequality and applications.
\newblock {\em Communications in Mathematical Physics}, 350:443--473, 2017.

\bibitem[BH16]{berry2016variable}
T.~Berry and J.~Harlim.
\newblock Variable bandwidth diffusion kernels.
\newblock {\em Applied and Computational Harmonic Analysis}, 40(1):68--96,
  2016.

\bibitem[BHK23]{baumer2023scaling}
D.~B{\"a}umer, S.~Hittmeir, and R.~Klein.
\newblock Scaling approaches to quasigeostrophic theory for moist,
  precipitating air.
\newblock {\em Journal of the Atmospheric Sciences}, 80(7):1771--1786, 2023.

\bibitem[Blu72]{blumen1972geostrophic}
W.~Blumen.
\newblock Geostrophic adjustment.
\newblock {\em Reviews of Geophysics}, 10(2):485--528, 1972.

\bibitem[BM04]{biello2004effect}
J.~A. Biello and A.~J. Majda.
\newblock The effect of meridional and vertical shear on the interaction of
  equatorial baroclinic and barotropic {R}ossby waves.
\newblock {\em Studies in Applied Mathematics}, 112(4):341--390, 2004.

\bibitem[BM12]{branicki2012quantifying}
M.~Branicki and A.~J. Majda.
\newblock Quantifying uncertainty for predictions with model error in
  non-{G}aussian systems with intermittency.
\newblock {\em Nonlinearity}, 25(9):2543, 2012.

\bibitem[Bre87]{bretherton1987theory}
C.~S. Bretherton.
\newblock A theory for nonprecipitating moist convection between two parallel
  plates. {P}art {I}: {T}hermodynamics and “linear” solutions.
\newblock {\em Journal of Atmospheric Sciences}, 44(14):1809--1827, 1987.

\bibitem[BV04]{boyd_vandenberghe}
S.~Boyd and L.~Vandenberghe.
\newblock {\em Convex optimization}.
\newblock Cambridge University Press, Cambridge, 2004.

\bibitem[Cho68]{chorin1968numerical}
A.~J. Chorin.
\newblock Numerical solution of the {N}avier-{S}tokes equations.
\newblock {\em Mathematics of Computation}, 22(104):745--762, 1968.

\bibitem[CHT{\etalchar{+}}18]{cao_hamouda_temam_tribbia}
Y.~Cao, M.~Hamouda, R.~Temam, J.~Tribbia, and X.~Wang.
\newblock The equations of the multi-phase humid atmosphere expressed as a
  quasi variational inequality.
\newblock {\em Nonlinearity}, 31(10):4692--4723, 2018.

\bibitem[CJTT21]{cao_jia_temam_tribbia}
Y.~Cao, C.~Jia, R.~Temam, and J.~Tribbia.
\newblock Mathematical analysis of a cloud resolving model including the ice
  microphysics.
\newblock {\em Discrete Contin. Dyn. Syst.}, 41(1):131--167, 2021.

\bibitem[CL06]{coifman2006diffusion}
R.~R. Coifman and S.~Lafon.
\newblock Diffusion maps.
\newblock {\em Applied and Computational Harmonic Analysis}, 21(1):5--30, 2006.

\bibitem[CMS15]{chen2015multiscale}
S.~Chen, A.~J. Majda, and S.~N. Stechmann.
\newblock Multiscale asymptotics for the skeleton of the {M}adden-{J}ulian
  oscillation and tropical--extratropical interactions.
\newblock {\em Mathematics of Climate and Weather Forecasting}, 1(1), 2015.

\bibitem[CMS16]{chen2016tropical}
S.~Chen, A.~J. Majda, and S.~N. Stechmann.
\newblock Tropical--extratropical interactions with the {MJO} skeleton and
  climatological mean flow.
\newblock {\em Journal of the Atmospheric Sciences}, 73(10):4101--4116, 2016.

\bibitem[CMT14]{chen2014information}
N.~Chen, A.~J. Majda, and X.~T. Tong.
\newblock Information barriers for noisy {L}agrangian tracers in filtering
  random incompressible flows.
\newblock {\em Nonlinearity}, 27(9):2133, 2014.

\bibitem[CT07]{cao_titi}
C.~Cao and E.~S. Titi.
\newblock Global well-posedness of the three-dimensional viscous primitive
  equations of large scale ocean and atmosphere dynamics.
\newblock {\em Ann. of Math. (2)}, 166(1):245--267, 2007.

\bibitem[CT12]{cover2012elements}
T.~M. Cover and J.~A. Thomas.
\newblock {\em Elements of Information Theory}.
\newblock John Wiley \& Sons, 2012.

\bibitem[CZFTT13]{coti_zelati_fremond_temam_tribbia_13}
M.~Coti~Zelati, M.~Fr\'{e}mond, R.~Temam, and J.~Tribbia.
\newblock The equations of the atmosphere with humidity and saturation:
  uniqueness and physical bounds.
\newblock {\em Phys. D}, 264:49--65, 2013.

\bibitem[CZHK{\etalchar{+}}15]{coti_zelati_huang_kukavica_temam_ziane}
M.~Coti~Zelati, A.~Huang, I.~Kukavica, R.~Temam, and M.~Ziane.
\newblock The primitive equations of the atmosphere in presence of vapour
  saturation.
\newblock {\em Nonlinearity}, 28(3):625--668, 2015.

\bibitem[CZT12]{coti_zelati_temam_12}
M.~Coti~Zelati and R.~Temam.
\newblock The atmospheric equation of water vapor with saturation.
\newblock {\em Boll. Unione Mat. Ital. (9)}, 5(2):309--336, 2012.

\bibitem[Daf79]{dafermos1979second}
C.~M. Dafermos.
\newblock The second law of thermodynamics and stability.
\newblock {\em Archive for Rational Mechanics and Analysis}, 70(2):167--179,
  1979.

\bibitem[Dal93]{daley1993atmospheric}
R.~Daley.
\newblock {\em Atmospheric Data Analysis}.
\newblock Cambridge university press, 1993.

\bibitem[DARR14]{doyle2014initial}
J.~D. Doyle, C.~Amerault, C.~A. Reynolds, and P.~A. Reinecke.
\newblock Initial condition sensitivity and predictability of a severe
  extratropical cyclone using a moist adjoint.
\newblock {\em Monthly Weather Review}, 142(1):320--342, 2014.

\bibitem[DB99]{derber1999reformulation}
J.~Derber and F.~Bouttier.
\newblock A reformulation of the background error covariance in the {ECMWF}
  global data assimilation system.
\newblock {\em Tellus A}, 51(2):195--221, 1999.

\bibitem[DG19]{das2019delay}
S.~Das and D.~Giannakis.
\newblock Delay-coordinate maps and the spectra of {K}oopman operators.
\newblock {\em Journal of Statistical Physics}, 175(6):1107--1145, 2019.

\bibitem[ELS20]{earlyetal2020}
J.~J. Early, M.-P. Lelong, and M.~Sundermeyer.
\newblock A generalized wave-vortex decomposition for rotating {B}oussinesq
  flows with arbitrary stratification.
\newblock {\em J. Fluid Mech.}, 912(A32):1--29, 2020.

\bibitem[EM96]{em96}
P.~F. Embid and A.~J. Majda.
\newblock Averaging over fast gravity waves for geophysical flows with arbitary
  potential vorticity.
\newblock {\em Comm. PDEs}, 21(3-4):619--658, 1996.

\bibitem[EM98]{em98}
P.~F. Embid and A.~J. Majda.
\newblock Low {F}roude number limiting dynamics for stably stratified flow with
  small or finite {R}ossby numbers.
\newblock {\em Geophys. Astrophys. Fluid Dynam.}, 87(1-2):1--50, 1998.

\bibitem[ER99]{errico1999examination}
R.~M. Errico and K.~D. Raeder.
\newblock An examination of the accuracy of the linearization of a mesoscale
  model with moist physics.
\newblock {\em Quarterly Journal of the Royal Meteorological Society},
  125(553):169--195, 1999.

\bibitem[Err97]{errico1997adjoint}
R.~M. Errico.
\newblock What is an adjoint model?
\newblock {\em Bulletin of the American Meteorological Society},
  78(11):2577--2592, 1997.

\bibitem[FJN12]{feireisl2012relative}
E.~Feireisl, B.~J. Jin, and A.~Novotn{\`y}.
\newblock Relative entropies, suitable weak solutions, and weak-strong
  uniqueness for the compressible {N}avier--{S}tokes system.
\newblock {\em Journal of Mathematical Fluid Mechanics}, 14(4):717--730, 2012.

\bibitem[FKN09]{ferguson2009two}
J.~Ferguson, B.~Khouider, and M.~Namazi.
\newblock Two-way interactions between equatorially-trapped waves and the
  barotropic flow.
\newblock {\em Chinese Annals of Mathematics, Series B}, 30(5):539--568, 2009.

\bibitem[FKNZ16]{feireisl2016singular}
E.~Feireisl, R.~Klein, A.~Novotn{\`y}, and E.~Zatorska.
\newblock On singular limits arising in the scale analysis of stratified fluid
  flows.
\newblock {\em Mathematical Models and Methods in Applied Sciences},
  26(03):419--443, 2016.

\bibitem[GG99]{gottwald1999formation}
G.~Gottwald and R.~Grimshaw.
\newblock The formation of coherent structures in the context of blocking.
\newblock {\em Journal of the Atmospheric Sciences}, 56(21):3640--3662, 1999.

\bibitem[Gia19]{giannakis2019data}
D.~Giannakis.
\newblock Data-driven spectral decomposition and forecasting of ergodic
  dynamical systems.
\newblock {\em Applied and Computational Harmonic Analysis}, 47(2):338--396,
  2019.

\bibitem[GK12]{gehne2012spectral}
M.~Gehne and R.~Kleeman.
\newblock Spectral analysis of tropical atmospheric dynamical variables using a
  linear shallow-water modal decomposition.
\newblock {\em Journal of the atmospheric sciences}, 69(7):2300--2316, 2012.

\bibitem[GM12]{giannakis2012nonlinear}
D.~Giannakis and A.~J. Majda.
\newblock Nonlinear {L}aplacian spectral analysis for time series with
  intermittency and low-frequency variability.
\newblock {\em Proceedings of the National Academy of Sciences},
  109(7):2222--2227, 2012.

\bibitem[Gra98]{grabowski}
W.~W. Grabowski.
\newblock Toward cloud resolving modeling of large-scale tropical circulations:
  A simple cloud microphysics parametrisation.
\newblock {\em Journal of Atmospheric Sciences}, 55:3283--3298, 1998.

\bibitem[GS96]{grabowski1996two}
W.~W. Grabowski and P.~K. Smolarkiewicz.
\newblock Two-time-level semi-{L}agrangian modeling of precipitating clouds.
\newblock {\em Monthly weather review}, 124(3):487--497, 1996.

\bibitem[HDMSS13]{hernandez_duenas}
G.~Hernandez-Duenas, A.~J. Majda, L.~M. Smith, and S.~N. Stechmann.
\newblock Minimal models for precipitating turbulent convection.
\newblock {\em J. Fluid Mech.}, 717:576--611, 2013.

\bibitem[HESS21]{hu_edwards_smith_stechmann_21}
R.~Hu, T.~K. Edwards, L.~M. Smith, and S.~N. Stechmann.
\newblock Initial investigations of precipitating quasi-geostrophic turbulence
  with phase changes.
\newblock {\em Res. Math. Sci.}, 8(1):Paper No. 6, 25, 2021.

\bibitem[HKLT17]{hittmeir_klein_li_titi_17}
S.~Hittmeir, R.~Klein, J.~Li, and E.~S. Titi.
\newblock Global well-posedness for passively transported nonlinear moisture
  dynamics with phase changes.
\newblock {\em Nonlinearity}, 30(10):3676--3718, 2017.

\bibitem[HKLT20]{hittmeir_klein_li_titi_20}
S.~Hittmeir, R.~Klein, J.~Li, and E.~S. Titi.
\newblock Global well-posedness for the primitive equations coupled to
  nonlinear moisture dynamics with phase changes.
\newblock {\em Nonlinearity}, 33(7):3206--3236, 2020.

\bibitem[HMR85]{hoskins1985use}
B.~J. Hoskins, M.~E. McIntyre, and A.~W. Robertson.
\newblock On the use and significance of isentropic potential vorticity maps.
\newblock {\em Quart. J. Roy. Met. Soc.}, 111(470):877--946, 1985.

\bibitem[HMRP16]{herbertetal2016}
C.~Herbert, R.~Marino, D.~Rosenberg, and A.~Puoquet.
\newblock Waves and vortices in the inverse cascade of stratified turbulence
  with and without rotation.
\newblock {\em J. Fluid Mech.}, 806:165--204, 2016.

\bibitem[HN13]{hieronymus2013buoyancy}
M.~Hieronymus and J.~Nycander.
\newblock The buoyancy budget with a nonlinear equation of state.
\newblock {\em Journal of Physical Oceanography}, 43(1):176--186, 2013.

\bibitem[Kas78]{kasahara1978}
A.~Kasahara.
\newblock Further studies on a spectral model of global barotropic primitive
  equations with {H}ough harmonic expansion.
\newblock {\em J. Atmos. Sci.}, 35:2043--2051, 1978.

\bibitem[KGK{\etalchar{+}}22]{knippertz2022intricacies}
P.~Knippertz, M.~Gehne, G.~N. Kiladis, K.~Kikuchi, A.~Rasheeda~Satheesh, P.~E.
  Roundy, G.-Y. Yang, N.~{\v{Z}}agar, J.~Dias, A.~H. Fink, et~al.
\newblock The intricacies of identifying equatorial waves.
\newblock {\em Quarterly Journal of the Royal Meteorological Society},
  148(747):2814--2852, 2022.

\bibitem[KM81]{klainerman1981singular}
S.~Klainerman and A.~Majda.
\newblock Singular limits of quasilinear hyperbolic systems with large
  parameters and the incompressible limit of compressible fluids.
\newblock {\em Communications on Pure and Applied Mathematics}, 34(4):481--524,
  1981.

\bibitem[KM82]{klainerman1982compressible}
S.~Klainerman and A.~Majda.
\newblock Compressible and incompressible fluids.
\newblock {\em Communications on Pure and Applied Mathematics}, 35:629--651,
  1982.

\bibitem[KM06]{klein_majda_06}
R.~Klein and A.~J. Majda.
\newblock Systematic multiscale models for deep convection on mesoscales.
\newblock {\em Theor. Comput. Fluid Dyn.}, 20:525--551, 2006.

\bibitem[KM10]{klocker2010influence}
A.~Klocker and T.~J. McDougall.
\newblock Influence of the nonlinear equation of state on global estimates of
  dianeutral advection and diffusion.
\newblock {\em Journal of Physical Oceanography}, 40(8):1690--1709, 2010.

\bibitem[KMS12]{khouider2012climate}
B.~Khouider, A.~J. Majda, and S.~N. Stechmann.
\newblock Climate science in the tropics: waves, vortices and pdes.
\newblock {\em Nonlinearity}, 26(1):R1, 2012.

\bibitem[Kor21]{korn2021global}
P.~Korn.
\newblock Global well-posedness of the ocean primitive equations with nonlinear
  thermodynamics.
\newblock {\em Journal of Mathematical Fluid Mechanics}, 23(3):71, 2021.

\bibitem[KSPK22]{klein2022qg}
R.~Klein, L.~Schielicke, S.~Pfahl, and B.~Khouider.
\newblock {QG--DL--Ekman}: dynamics of a diabatic layer in the
  quasi-geostrophic framework.
\newblock {\em Journal of the Atmospheric Sciences}, 79(3):887--905, 2022.

\bibitem[KSS23]{kooloth2023hamilton}
P.~Kooloth, L.~M. Smith, and S.~N. Stechmann.
\newblock Hamilton's principle with phase changes and conservation principles
  for moist potential vorticity.
\newblock {\em Quarterly Journal of the Royal Meteorological Society},
  149(752):1056--1072, 2023.

\bibitem[LCBA20]{lelongetal2020}
M.-P. Lelong, Y.~Cuypers, and P.~Bouruet-Aubertot.
\newblock Near-inertial energy propagation inside a mediterranean anticyclonic
  eddy.
\newblock {\em J. Phys. Oceanogr.}, 50(8):2271--2288, 2020.

\bibitem[LDT86]{le1986variational}
F.-X. Le~Dimet and O.~Talagrand.
\newblock Variational algorithms for analysis and assimilation of
  meteorological observations: theoretical aspects.
\newblock {\em Tellus A: Dynamic Meteorology and Oceanography}, 38(2):97--110,
  1986.

\bibitem[Lei80]{leith1980}
C.~Leith.
\newblock Nonlinear normal mode initialization and quasi-geostrophic theory.
\newblock {\em J. Atmos. Sci.}, 37:958--968, 1980.

\bibitem[LM92]{lienmuller1992}
R.-C. Lien and P.~Muller.
\newblock Normal-mode decomposition of small-scale oceanic motions.
\newblock {\em J. Phys. Oceanogr.}, 22(12):1583--1595, 1992.

\bibitem[LM20]{lian_ma}
R.~Lian and J.~Ma.
\newblock Existence of a strong solution to moist atmospheric equations with
  the effects of topography.
\newblock {\em Bound. Value Probl.}, pages Paper No. 103, 34, 2020.

\bibitem[Lor81]{lorenc1981global}
A.~Lorenc.
\newblock A global three-dimensional multivariate statistical interpolation
  scheme.
\newblock {\em Monthly Weather Review}, 109(4):701--721, 1981.

\bibitem[Mah99]{mahfouf1999influence}
J.-F. Mahfouf.
\newblock Influence of physical processes on the tangent-linear approximation.
\newblock {\em Tellus A: Dynamic Meteorology and Oceanography}, 51(2):147--166,
  1999.

\bibitem[Maj84]{majda1984compressible}
A.~Majda.
\newblock {\em Compressible Fluid Flow and Systems of Conservation Laws in
  Several Space Variables}, volume~53.
\newblock Springer Science \& Business Media, 1984.

\bibitem[Maj03]{majda_2003}
A.~Majda.
\newblock {\em Introduction to {PDE}s and waves for the atmosphere and ocean},
  volume~9 of {\em Courant Lecture Notes in Mathematics}.
\newblock New York University, Courant Institute of Mathematical Sciences, New
  York; American Mathematical Society, Providence, RI, 2003.

\bibitem[MB01]{majda2001vorticity}
A.~J. Majda and A.~L. Bertozzi.
\newblock {\em Vorticity and Incompressible Flow}.
\newblock Cambridge University Press, 2001.

\bibitem[MB03]{majda2003nonlinear}
A.~J. Majda and J.~A. Biello.
\newblock The nonlinear interaction of barotropic and equatorial baroclinic
  {R}ossby waves.
\newblock {\em Journal of the Atmospheric Sciences}, 60(15):1809--1821, 2003.

\bibitem[MC18]{marques2018diagnosis}
C.~A. Marques and J.~M. Castanheira.
\newblock Diagnosis of free and convectively coupled equatorial waves.
\newblock {\em Mathematical Geosciences}, 50(5):585--606, 2018.

\bibitem[ME98]{me98}
A.~J. Majda and P.~Embid.
\newblock Averaging over fast gravity waves for geophysical flows with
  unbalanced initial data.
\newblock {\em Theoret. Computat. Fluid Dyn.}, 11:155--169, 1998.

\bibitem[MSS19]{marsico2019energy}
D.~H. Marsico, L.~M. Smith, and S.~N. Stechmann.
\newblock Energy decompositions for moist {B}oussinesq and anelastic equations
  with phase changes.
\newblock {\em Journal of the Atmospheric Sciences}, 76(11):3569--3587, 2019.

\bibitem[M{\"u}l95]{muller1995ertel}
P.~M{\"u}ller.
\newblock Ertel's potential vorticity theorem in physical oceanography.
\newblock {\em Reviews of Geophysics}, 33(1):67--97, 1995.

\bibitem[Nes04]{nesterov}
Y.~Nesterov.
\newblock {\em Introductory lectures on convex optimization}, volume~87 of {\em
  Applied Optimization}.
\newblock Kluwer Academic Publishers, Boston, MA, 2004.
\newblock A basic course.

\bibitem[NHR15]{nycander2015nonlinear}
J.~Nycander, M.~Hieronymus, and F.~Roquet.
\newblock The nonlinear equation of state of sea water and the global water
  mass distribution.
\newblock {\em Geophysical Research Letters}, 42(18):7714--7721, 2015.

\bibitem[OS15]{ogrosky2015mjo}
H.~R. Ogrosky and S.~N. Stechmann.
\newblock The {MJO} skeleton model with observation-based background state and
  forcing.
\newblock {\em Quarterly Journal of the Royal Meteorological Society},
  141(692):2654--2669, 2015.

\bibitem[OS16]{ogrosky2016identifying}
H.~R. Ogrosky and S.~N. Stechmann.
\newblock Identifying convectively coupled equatorial waves using theoretical
  wave eigenvectors.
\newblock {\em Monthly Weather Review}, 144(6):2235--2264, 2016.

\bibitem[OSB{\etalchar{+}}01]{ochoa2001geostrophy}
J.~Ochoa, J.~Sheinbaum, A.~Badan, J.~Candela, and D.~Wilson.
\newblock Geostrophy via potential vorticity inversion in the {Y}ucatan
  {C}hannel.
\newblock {\em Journal of Marine Research}, 59(5):725--747, 2001.

\bibitem[OWSS24]{ogrosky2024data}
H.~R. Ogrosky, A.~N. Wetzel, L.~M. Smith, and S.~N. Stechmann.
\newblock Data assimilation errors due to incomplete accounting of balanced
  versus unbalanced moisture.
\newblock {\em submitted}, 2024.

\bibitem[PD97]{park1997validity}
S.~K. Park and K.~K. Droegemeier.
\newblock Validity of the tangent linear approximation in a moist convective
  cloud model.
\newblock {\em Monthly Weather Review}, 125(12):3320--3340, 1997.

\bibitem[PS10]{pauluis2010idealized}
O.~Pauluis and J.~Schumacher.
\newblock Idealized moist {R}ayleigh-{B}{\'e}nard convection with piecewise
  linear equation of state.
\newblock {\em Communications in Mathematical Sciences}, 8(1):295--319, 2010.

\bibitem[RDTM08]{raupp2008resonant}
C.~F. Raupp, P.~L.~S. Dias, E.~G. Tabak, and P.~Milewski.
\newblock Resonant wave interactions in the equatorial waveguide.
\newblock {\em Journal of the Atmospheric Sciences}, 65(11):3398--3418, 2008.

\bibitem[RMBN15]{roquet2015defining}
F.~Roquet, G.~Madec, L.~Brodeau, and J.~Nycander.
\newblock Defining a simplified yet “realistic” equation of state for
  seawater.
\newblock {\em Journal of Physical Oceanography}, 45(10):2564--2579, 2015.

\bibitem[RTSS24]{remond2024nonlinear}
A.~Remond-Tiedrez, L.~M. Smith, and S.~N. Stechmann.
\newblock A nonlinear elliptic {PDE} from atmospheric science: well-posedness
  and regularity at cloud edge.
\newblock {\em Journal of Mathematical Fluid Mechanics}, 26(2):30, 2024.

\bibitem[SM15]{stechmann2015identifying}
S.~N. Stechmann and A.~J. Majda.
\newblock Identifying the skeleton of the {M}adden--{J}ulian oscillation in
  observational data.
\newblock {\em Monthly Weather Review}, 143(1):395--416, 2015.

\bibitem[SO14]{stechmann2014walker}
S.~N. Stechmann and H.~R. Ogrosky.
\newblock The {W}alker circulation, diabatic heating, and outgoing longwave
  radiation.
\newblock {\em Geophysical Research Letters}, 41(24):9097--9105, 2014.

\bibitem[SS17]{smith2017precipitating}
L.~M. Smith and S.~N. Stechmann.
\newblock Precipitating quasigeostrophic equations and potential vorticity
  inversion with phase changes.
\newblock {\em J. Atmos. Sci.}, 74(10):3285--3303, 2017.

\bibitem[SW02]{smithwaleffe2002}
L.~M. Smith and F.~Waleffe.
\newblock Generation of slow large scales in forced rotating stratified
  turbulence.
\newblock {\em J. Fluid Mech.}, 451:145--168, 2002.

\bibitem[Tem69]{temam1969approximation}
R.~Temam.
\newblock Sur l'approximation de la solution des {\'e}quations de
  {N}avier-{S}tokes par la m{\'e}thode des pas fractionnaires ({II}).
\newblock {\em Archive for Rational Mechanics and Analysis}, 33:377--385, 1969.

\bibitem[TL22]{tan_liu}
S.~Tan and W.~Liu.
\newblock The strong solutions to the primitive equations coupled with
  multi-phase moisture atmosphere.
\newblock {\em Phys. D}, 440:Paper No. 133442, 18, 2022.

\bibitem[TS24]{tzou2024}
C.-N. Tzou and S.~N. Stechmann.
\newblock in preparation.
\newblock 2024.

\bibitem[TW15]{temam_wu_15}
R.~Temam and K.~J. Wu.
\newblock Formulation of the equations of the humid atmosphere in the context
  of variational inequalities.
\newblock {\em J. Funct. Anal.}, 269(7):2187--2221, 2015.

\bibitem[TW16]{teman_wang_16}
R.~Temam and X.~Wang.
\newblock Approximation of the equations of the humid atmosphere with
  saturation.
\newblock In L.~Bociu, J.-A. D{\'e}sid{\'e}ri, and A.~Habbal, editors, {\em
  System Modeling and Optimization}, pages 21--42, Cham, 2016. Springer
  International Publishing.

\bibitem[VZ21]{vas-zagar2021}
S.~Vasylkevych and N.~Zagar.
\newblock A high-accuracy global prognostic model for the simulation of rossby
  and gravity wave dynamics.
\newblock {\em Quart. J. Roy. Meteor. Soc.}, 147(736):1989--2007, 2021.

\bibitem[WSS{\etalchar{+}}20]{wetzel2020potential}
A.~N. Wetzel, L.~M. Smith, S.~N. Stechmann, J.~E. Martin, and Y.~Zhang.
\newblock Potential vorticity and balanced and unbalanced moisture.
\newblock {\em J. Atmos. Sci.}, 77(6):1913--1931, 2020.

\bibitem[WSSM19]{wetzel2019balanced}
A.~N. Wetzel, L.~M. Smith, S.~N. Stechmann, and J.~E. Martin.
\newblock Balanced and unbalanced components of moist atmospheric flows with
  phase changes.
\newblock {\em Chinese Annals of Mathematics, Series B}, 40(6):1005--1038,
  2019.

\bibitem[YHS03]{yang2003convectively}
G.-Y. Yang, B.~Hoskins, and J.~Slingo.
\newblock Convectively coupled equatorial waves: A new methodology for
  identifying wave structures in observational data.
\newblock {\em Journal of the atmospheric sciences}, 60(14):1637--1654, 2003.

\bibitem[{\v{Z}}F15]{zagar2015systematic}
N.~{\v{Z}}agar and C.~L. Franzke.
\newblock Systematic decomposition of the {M}adden-{J}ulian {O}scillation into
  balanced and inertio-gravity components.
\newblock {\em Geophysical Research Letters}, 42(16):6829--6835, 2015.

\bibitem[ZGK04]{zgk04qjrms}
N.~Zagar, N.~Gustafsson, and E.~Kallen.
\newblock Variational data assimilation in the tropics: The impact of a
  background-error constraint.
\newblock {\em Quart. J. Roy. Meteor. Soc.}, 130(596):103--125, 2004.

\bibitem[{\v{Z}}KT{\etalchar{+}}15]{zetal15}
N.~{\v{Z}}agar, A.~Kasahara, K.~Terasaki, J.~Tribbia, and H.~Tanaka.
\newblock Normal-mode function representation of global 3-d data sets:
  open-access software for the atmospheric research community.
\newblock {\em Geoscientific Model Development}, 8(4):1169--1195, 2015.

\bibitem[ZNV{\etalchar{+}}23]{zagaretal2023}
N.~Zagar, V.~Neduhal, S.~Vasylkevych, Z.~Zaplotnik, and L.~Tanaka, Hiroshi.
\newblock Decomposition of vertical velocity and its zonal wavenumber kinetic
  energy spectra in the hydrostatic atmosphere.
\newblock {\em J. Atmos. Sci.}, 80(11):2747--2767, 2023.

\bibitem[ZSS21a]{Zhang_Smith_Stechmann_2021_JFM}
Y.~Zhang, L.~M. Smith, and S.~N. Stechmann.
\newblock Effects of clouds and phase changes on fast-wave averaging: a
  numerical assessment.
\newblock {\em Journal of Fluid Mechanics}, 920:A49, 2021.

\bibitem[ZSS21b]{zhang_smith_stechmann_20_asymptotics}
Y.~Zhang, L.~M. Smith, and S.~N. Stechmann.
\newblock Fast-wave averaging with phase changes: asymptotics and application
  to moist atmospheric dynamics.
\newblock {\em J. Nonlinear Sci.}, 31(2):Paper No. 38, 46, 2021.

\bibitem[ZSS22]{zhang_smith_stechmann_22}
Y.~Zhang, L.~M. Smith, and S.~N. Stechmann.
\newblock Convergence to precipitating quasi-geostrophic equations with phase
  changes: asymptotics and numerical assessment.
\newblock {\em Phil. Trans. R. Soc. A.}, 380, 2022.

\end{thebibliography}

\end{document}